\documentclass[10pt,a4paper]{article}

\usepackage{etex} 

\usepackage{amsmath,amsfonts,amssymb,graphicx,multirow,enumitem,array}
\usepackage[nobreak]{cite} 
\usepackage[inner=20mm, outer=20mm, top=20mm, bottom=20mm, head=10mm, foot=10mm]{geometry}
\usepackage{amsthm}
\usepackage{caption}

\usepackage[usenames]{color}
\usepackage[table]{xcolor}

\usepackage{arydshln} 

\usepackage{pstricks}
\usepackage{pst-plot}

\usepackage{tikz} 
\usetikzlibrary{arrows,positioning,shapes}

\usepackage[backref=false]{hyperref}
\hypersetup{
colorlinks=true,
citecolor=red,
linkcolor=darkblue,
urlcolor=darkblue
}

\usepackage[capitalise,noabbrev]{cleveref} 

\numberwithin{figure}{section}
\numberwithin{table}{section}
\numberwithin{equation}{section}

\newcolumntype{C}{>{$}c<{$}} 

\allowdisplaybreaks

\newcommand{\nc}{\newcommand}
\nc{\bib}{\bibitem}
\nc{\ir}{\mathrm{i}}

\nc{\g}{\red} 

\setitemize{leftmargin=*}   
\setenumerate{leftmargin=*} 

\newcommand{\BKL}{\cellcolor{blue!20}} 
\newcommand{\IKL}{\cellcolor{blue!40}} 

\long\def\ignore#1{}
\definecolor{darkblue}{rgb}{0,0,.8}
\definecolor{red}{rgb}{1,0,0}
\definecolor{purple}{rgb}{1,0,1}
\definecolor{coloroflink}{rgb}{0.7,0,1}
\definecolor{darkpurple}{rgb}{1,.2,1}
\definecolor{pink}{rgb}{1,.7,.7}
\definecolor{lightblue}{rgb}{.61,.61,1}
\definecolor{midblue}{rgb}{.75,.75,1}
\definecolor{lightlightblue}{rgb}{.9,.9,1}
\definecolor{lightlightgray}{rgb}{.95,.95,.95}
\definecolor{lightestblue}{rgb}{.96,.96,1}
\definecolor{lightpurple}{rgb}{1,.65,1}
\definecolor{darkgreen}{rgb}{0.180392, 0.545098, 0.341176}

\newtheoremstyle{alexi}                                 %
  {5pt}                                                 
  {5pt}                                                 
  {\itshape}                                            
  {}                                                    
  {\scshape}                                            
  {.}                                                   
  {.5em}                                                
  {\thmname{#1} \thmnumber{#2}\thmnote{\normalfont#3}}  

\theoremstyle{alexi}
\newtheorem{Lemme}{Lemma}[section]

\newtheorem{Proposition}[Lemme]{Proposition}
\newtheorem{Lemma}[Lemme]{Lemma}

\newtheorem{cor}[Lemme]{Corollary}

\newtheorem{Definition}{Definition}

\newtheorem{conj}{Conjecture}
\newtheorem{defn}[Definition]{Definition}

\let\oldproofname=\proofname
\renewcommand{\proofname}{\scshape{\oldproofname}} 

\nc{\Db}{\mbox{$\boldsymbol{D}$}}
\nc{\Dbk}{\mbox{$\boldsymbol{D}^{\textrm{\tiny$(k)$}}\!$}}
\nc{\db}{\mbox{$\boldsymbol{d}$}}
\nc{\Dbb}{\mbox{\boldmath $\bar D$}}
\nc{\Ib}{\mbox{$\boldsymbol{I}$}}
\nc{\Eb}{\mbox{\mathbf{$E$}}}
\nc{\Xb}{\mbox{\mathbf{$X$}}}
\nc{\Kb}{\mbox{\mathbf{$K$}}}
\nc{\sr}{{\scshape \scriptsize (r)}}
\nc{\sq}{{\scshape \scriptsize (q)}}
\nc{\su}{{\scshape \scriptsize (u)}}
\nc{\chit}{\protect\raisebox{0.25ex}{$\chi$}}
\nc{\textoverline}[1]{$\overline{\mbox{#1}}$}
\nc{\sro}{\scshape \scriptsize (\textoverline{r})}
\nc{\sqo}{\scshape \scriptsize (\textoverline{q})}
\nc{\suo}{\scshape \scriptsize (\textoverline{u})}
\nc{\sru}{\scshape \scriptsize (\underline{r})}
\nc{\squ}{\scshape \scriptsize (\underline{q})}
\nc{\suu}{\scshape \scriptsize (\underline{u})}
\nc{\tl}{\mathsf{TL}}						
\nc{\gramprod}[2]{\langle #1 | #2 \rangle}		
\nc{\gramprodk}[2]{\langle #1 | #2 \rangle^{\textrm{\tiny$(k)$}}\!}				
\nc{\grammat}{\mathcal{G}}				
\nc{\Rad}{\mathsf{R}}					
\nc{\tlone}{\mathsf{TL}^{\!_{(1)}}}			
\nc{\btl}{\mathsf{B}}					
\nc{\atl}{\mathsf{A}}						
\nc{\stan}{\mathsf{V}}					
\nc{\stanu}{\mathsf{U}}					
\nc{\Itl}{\mathsf{I}}						
\nc{\Ptl}{\mathsf{P}}						
\nc{\Mtl}{\mathsf{M}}						
\nc{\links}{\mathcal{B}}					
\nc{\linksb}{\mathcal{C}}					
\nc{\ham}{\mathcal{H}}					
\nc{\kac}{\mathsf{K}}						
\nc{\hamk}{\mbox{$\mathcal{H}^{\textrm{\tiny$(k)$}}\!$}}		
\nc{\iom}{\mbox{$\mathcal I$}}				
\nc{\Ik}{I^{\textrm{\tiny$(k)$}}}	
\nc{\Ek}{E^{_{(k)}}}
\nc{\Ekj}[2]{E^{_{(#1)}}_{#2}}
\nc{\Ekb}{\bar E^{_{(k)}}}
\nc{\Yk}{Y^{_{(k)}}}
\nc{\Xk}{X^{_{(k)}}}
\nc{\Kk}{K^{\textrm{\tiny$(k)$}}}
\nc{\Lmn}{L_m^{\textrm{\tiny$(n)$}}}

\def\facegrid#1#2{
\psframe[fillstyle=solid,fillcolor=lightlightblue,linewidth=0pt]#1#2
\psgrid[gridlabels=0pt,subgriddiv=1]#1#2}

\def\loopa{
\psframe[linewidth=.25pt](0,0)(1,1)
\psarc[linewidth=1.5pt,linecolor=blue](1,0){.5}{90}{180}
\psarc[linewidth=1.5pt,linecolor=blue](0,1){.5}{-90}{0}
}
\def\loopb{
\psframe[linewidth=.25pt](0,0)(1,1)
\psarc[linewidth=1.5pt,linecolor=blue](0,0){.5}{0}{90}
\psarc[linewidth=1.5pt,linecolor=blue](1,1){.5}{180}{270}
}

\def\wobbly{
\rput{90}(0,0){
\psplot[linecolor=purple,linewidth=1.5pt]
    {0}{0.11}{x 0.0005 div sin 25 div x mul 0.1 div}
\psplot[linecolor=purple,linewidth=1.5pt]
    {0.1}{0.59}{x 0.0005 div sin 25 div }}       
}
\def\wobblysmall{
\rput{90}(0,0){
\psplot[linecolor=purple,linewidth=1.00pt]
    {0}{0.11}{x 0.0005 div sin 25 div x mul 0.1 div}
\psplot[linecolor=purple,linewidth=1.00pt]
    {0.1}{0.59}{x 0.0005 div sin 25 div }}       
}

\nc{\mince}{0.5pt}
\nc{\elegant}{1.5pt}

\newcommand{\alg}[1]{\mathfrak{#1}} 

\newcommand{\func}[2]{#1 \left( #2 \right)} 

\newcommand{\brac}[1]{\left( #1 \right)}
\newcommand{\tbrac}[1]{\bigl( #1 \bigr)}
\newcommand{\sqbrac}[1]{\left[ #1 \right]}

\newcommand{\set}[1]{\left\{ #1 \right\}}

\newcommand{\st}{\mspace{5mu} : \mspace{5mu}} 

\newcommand{\abs}[1]{\left\lvert #1 \right\rvert}

\newcommand{\ZZ}{\mathbb{Z}}

\newcommand{\QQ}{\mathbb{Q}}
\newcommand{\RR}{\mathbb{R}}

\newcommand{\CC}{\mathbb{C}}

\newcommand{\dd}{\mathrm{d}}   
\newcommand{\ii}{\ir} 
\newcommand{\ee}{\mathsf{e}}   

\newcommand{\wun}{\mathbf{1}}  

\newcommand{\comm}[2]{\bigl[ #1 , #2 \bigr]}

\newcommand{\ket}[1]{\bigl\lvert #1 \bigr\rangle}

\newcommand{\ra}{\rightarrow}
\newcommand{\Ra}{\Rightarrow}
\newcommand{\lra}{\longrightarrow}

\newcommand{\SLA}[2]{\alg{#1} \bigl( #2 \bigr)}                             

\newcommand{\MinMod}[2]{\mathsf{M} \bigl( #1 , #2 \bigr)}                   

\newcommand{\Ver}[1]{\mathcal{V}_{#1}}  
\newcommand{\Irr}[1]{\mathcal{L}_{#1}}  
\newcommand{\Kac}[1]{\mathcal{K}_{#1}}  
\newcommand{\FF}[1]{\mathcal{F}_{#1}}   

\newcommand{\chmap}{\mathrm{ch}}
\newcommand{\Gr}[1]{\bigl[ #1 \bigr]}                               
\newcommand{\tGr}[1]{\bigl[ #1 \bigr]}                              
\newcommand{\ch}[1]{\chmap \Gr{#1}}                                 
\newcommand{\fch}[2]{\ch{#1} \bigl( #2 \bigr)}                      

\newcommand{\modS}{\mathsf{S}}                        

\newcommand{\Smat}[2]{\modS \Gr{#1 \ra #2}}           
\newcommand{\tSmat}[2]{\modS \tGr{#1 \ra #2}}

\newcommand{\fuse}{\mathbin{\times}}                                            
\newcommand{\Grfuse}{\mathbin{\boxtimes}}                                       
\newcommand{\fuscoeff}[3]{\mathsf{N}_{#1 \, #2}^{\hphantom{#1 \, #2} #3}}       

\newcommand{\coprodsymb}{\Delta}
\newcommand{\altcoprodsymb}{\widetilde{\Delta}}
\newcommand{\coproduct}[1]{\coprodsymb \bigl( #1 \bigr)}                        

\newcommand{\cft}{conformal field theory}
\newcommand{\cfts}{conformal field theories}
\newcommand{\uea}{universal enveloping algebra}

\newcommand{\lcfts}{logarithmic conformal field theories}

\newcommand{\opes}{operator product expansions}
\newcommand{\hw}{highest-weight}
\newcommand{\hws}{\hw{} state}
\newcommand{\hwss}{\hw{} states}
\newcommand{\sv}{singular vector}
\newcommand{\svs}{singular vectors}
\newcommand{\ssv}{subsingular vector}
\newcommand{\ssvs}{subsingular vectors}
\newcommand{\hwm}{\hw{} module}
\newcommand{\hwms}{\hw{} modules}

\newcommand{\FFm}{Feigin-Fuchs module}
\newcommand{\FFms}{Feigin-Fuchs modules}
\newcommand{\NGK}{Nahm-Gaberdiel-Kausch}

\newcommand{\rhs}{right-hand side}
\newcommand{\TL}{Temperley-Lieb}
\newcommand{\WJ}{Wenzl-Jones}

\renewcommand{\Im}{\operatorname{Im}}

\DeclareMathOperator{\vspn}{span}
\DeclareMathOperator{\End}{End}

\renewcommand{\ge}{\geqslant}
\renewcommand{\le}{\leqslant}

\newcommand{\algtl}{\textup{a-}\tlone} 
\newcommand{\diatl}{\textup{d-}\tlone} 

\newcommand{\algbtl}{\textup{a-}\btl} 
\newcommand{\diabtl}{\textup{d-}\btl} 

\newcommand{\rs}{\protect{\begin{tikzpicture}[scale=0.3]
                           \path[use as bounding box] (-0.1,-0.1) rectangle (1.4,1.1);
                           \draw (-0.1,1.1) -- (1.5,-0.1);
                           \node at (0.1,0.2) {$r$};
                           \node at (1.3,0.9) {$s$};
                          \end{tikzpicture}}} 

\begin{document}

\begin{center}
\textbf{\huge Boundary algebras and Kac modules \\[.35cm] 
for logarithmic minimal models} \\[1cm]
{\Large Alexi Morin-Duchesne$^\ast$, J{\o}rgen Rasmussen$^\S$ and David Ridout$^\ddag$.}
\\[.5cm]
\emph{{}$^\ast$\parbox[t]{0.85\textwidth}{Institut de Recherche en Math\'ematique et Physique, \\
Universit\'e Catholique de Louvain, Louvain-la-Neuve, B-1348, Belgium.}}
\\[.2cm]
\emph{{}$^\S$\parbox[t]{0.85\textwidth}{School of Mathematics and Physics, \\
University of Queensland, St Lucia, Brisbane, Queensland 4072, Australia.}}
\\[.2cm]
\emph{${}^{\ddag}$\parbox[t]{0.85\textwidth}{Department of Theoretical Physics and Mathematical Sciences Institute, \\
Australian National University, Acton, ACT 2601, Australia.}}
\\[.5cm] 
\texttt{alexi.morin-duchesne\,@\,uclouvain.be}
\qquad
\texttt{j.rasmussen\,@\,uq.edu.au}
\qquad
\texttt{david.ridout\,@\,anu.edu.au}
\end{center}
\bigskip

\begin{abstract}
Virasoro Kac modules were originally introduced indirectly as representations whose characters arise in the continuum scaling limits of certain transfer matrices in logarithmic minimal models, described using Temperley-Lieb algebras. The lattice transfer operators include seams on the boundary that use Wenzl-Jones projectors. If the projectors are singular, the original prescription is to select a subspace of the Temperley-Lieb modules on which the action of the transfer operators is non-singular. However, this prescription
does not, in general, yield representations of the Temperley-Lieb algebras and the Virasoro Kac modules have remained largely unidentified. Here, we introduce the appropriate algebraic framework for the lattice analysis as a quotient of the one-boundary Temperley-Lieb algebra. The corresponding standard modules are introduced and examined using invariant bilinear forms and their Gram determinants. The structures of the Virasoro Kac modules are inferred from these results and are found to be given by finitely generated submodules of Feigin-Fuchs modules. Additional evidence for this identification is obtained by comparing the formalism of lattice fusion with the fusion rules of the Virasoro Kac modules.  These are obtained, at the character level, in complete generality by applying a Verlinde-like formula and, at the module level, in many explicit examples by applying the Nahm-Gaberdiel-Kausch fusion algorithm. 
\end{abstract}

\newpage
\tableofcontents
\newpage

%
\section{Introduction}
%

The minimal models introduced by Belavin, Polyakov and Zamolodchikov \cite{BPZ84} are central to conformal field theory \cite{DiFMS}. A minimal model is characterised by a pair of co-prime integers, $1<p<p'$, and is often denoted 
accordingly by 
$\mathcal{M}(p,p')$. The corresponding central charge $c$ and conformal weights $\Delta_{r,s}$
are given by
\begin{equation}
 c=1-6\frac{(p'-p)^2}{pp'},\qquad \Delta_{r,s}=\frac{(rp'-sp)^2-(p'-p)^2}{4pp'},
\label{cD}
\end{equation}
where $r=1,2,\ldots,p-1$ and $s=1,2,\ldots,p'-1$.
These weights satisfy $\Delta_{r,s}=\Delta_{p-r,p'-s}$ and there is an irreducible Virasoro
representation associated with each distinct conformal weight. 
Moreover, these representations are the only indecomposable representations in the
model and the minimal models are examples of \emph{rational} conformal field theories.

At criticality, the restricted solid-on-solid models solved by Andrews, Baxter and Forrester \cite{ABF84,FB85} offer 
lattice realisations of the minimal models. Corresponding to each of the irreducible Virasoro representations in a given minimal model, there is a Yang-Baxter integrable 
boundary condition \cite{Skly88,Cardy89,BP01} for the lattice realisation:
In the continuum scaling limit (or scaling limit, for short), the eigenvalue spectrum of the corresponding transfer matrix (or of the associated Hamiltonian) gives rise to the character of the irreducible representation. In this way, the Hamiltonian of the lattice model becomes the first conformal integral of motion $\iom_1=L_0-\frac c{24}$.

Logarithmic conformal field theory has its roots in work by Rozansky and Saleur \cite{RozQua92} 
and Gurarie \cite{Gur93}, but the first thorough analysis of such a theory appeared in a series
of papers by Gaberdiel and Kausch \cite{GabInd96,GabRat96,GabLoc99} on a theory
with central charge $c=-2$. Their theory
is not a minimal model, at least not from the perspective of the Virasoro algebra, but it may be regarded as minimal with respect to an extended symmetry algebra $\mathcal{W}_{1,2}$ related to that of symplectic fermions \cite{KauSym00}.
The central charge and conformal highest weights of the Virasoro representations are nevertheless of the 
form \eqref{cD}, but with $p=1$, $p'=2$ and no upper bounds on the Kac labels $r$ and $s$.
Subsequently, evidence mounted \cite{CCCPV2000,ASA02,Ras04,Ras07,EbeVir06} 
suggesting that every Virasoro minimal model can be augmented to
a logarithmic conformal field theory of the same central charge. This was made concrete almost ten years ago 
when such logarithmic models were realised algebraically as conformal field theories with $\mathcal{W}_{p,p'}$ symmetry \cite{FeiLog06} and conjectured to be the
scaling limits of a series of exactly solvable lattice models $\mathcal{LM}(p,p')$ \cite{PRZ06}.
In these models, the co-prime integers $p$ and $p'$ satisfy $1\leq p<p'$, thus covering the value $c=-2$ 
(the $\mathcal{W}_{1,p'}$ models were introduced as \cfts{} much earlier \cite{Kau91,Flo96}). 
We emphasise that the present work deals with the so-called Virasoro picture and thus ignores possible
extensions of the Virasoro algebra such as the $\mathcal{W}_{p,p'}$ algebras underlying the W-extended picture \cite{FeiLog06,PRR08,RasWLM09}.

As lattice theories, the logarithmic minimal models $\mathcal{LM}(p,p')$ describe non-intersecting, densely packed 
loops on a square lattice. Mathematically, this can be formalised in terms of the Temperley-Lieb algebra 
$\tl_n(\beta)$ \cite{TL}, where $\beta$ denotes the fugacity of the loops and $n$ is the width of the lattice.
The models admit infinitely many distinct Yang-Baxter integrable boundary conditions, among which the
so-called $(r,s)$-type, or Kac, 
boundary conditions play a prominent role. Matrix realisations of the corresponding transfer operators are well-defined, although their construction does not yield representations
of the full underlying Temperley-Lieb algebras. It was nevertheless argued, based on supporting numerical explorations \cite{PRZ06}, that they give rise to the Virasoro characters 
\begin{equation}
 \chit_{r,s}(q) = q^{-c/24} \frac{q^{\Delta_{r,s}}(1-q^{rs})}{\prod_{j=1}^\infty (1-q^j)}\qquad \text{(\(r,s\in\mathbb{Z}_+\))}
\label{chitrs}
\end{equation}
in the scaling limit. It is stressed that these characters do not, in general, 
correspond to irreducible representations. Because of their definition in terms of Kac labels, they were 
baptised (Virasoro) Kac characters \cite{RasFus07a,RasFus07}. 
To each logarithmic minimal model $\mathcal{LM}(p,p')$,
one can thus associate an infinitely extended Kac table of conformal weights $\Delta_{r,s}$ of the form \eqref{cD}, with $r,s\in\mathbb{Z}_+$. 

Unless a Kac character happens to be irreducible, the structure of the corresponding Virasoro module 
is not determined by the character alone. In fact, the Virasoro module structures associated with the Kac boundary conditions
have remained largely unknown and a primary goal of this paper is to remedy this situation.
By combining lattice analyses with applications of the Nahm-Gaberdiel-Kausch algorithm \cite{NahQua94,GabInd96},
an explicit conjecture for these modules, in the logarithmic minimal models $\mathcal{LM}(1,p')$, was presented 
in \cite{RasCla11} as certain finitely generated submodules of Feigin-Fuchs modules \cite{FeiRep90}.
A key objective here is to extend this conjecture to the general logarithmic minimal models
$\mathcal{LM}(p,p')$ and to substantiate it by providing new and independent evidence.

For getting insight into the conformal properties of logarithmic minimal models, a fundamental paradigm 
is that much information is already encoded at the lattice level. In particular, certain representation-theoretic properties
of the Temperley-Lieb algebra or the transfer matrices are thus expected to `survive' in the scaling limit.
It is therefore of great importance to understand this limit better mathematically, especially how the Virasoro
algebra arises. In terms of the generators of the Temperley-Lieb algebra, Koo and Saleur \cite{KooAss94} have
proposed explicit expressions that are believed to realise the Virasoro modes in the scaling limit.
Common structures between finite systems and their conformal counterparts have also been examined
through quantum groups, initially for the XXZ spin chain \cite{PS90}, but more recently for the $g\ell(1|1)$ 
super-spin chain \cite{GaiCon11,GaiBim13} 
and to realise W-algebraic structures from XXZ spin chains \cite{GST13}. Recent advances \cite{GaiAss12,GaiPer14} have extended this to spin chains with periodic boundary conditions, leading to non-trivial predictions for the structures of the corresponding bulk \cfts{}. We refer to the recent review paper \cite{GaiLog13} for more details.

Fusion has played a crucial role in unravelling Virasoro module structures in logarithmic minimal models \cite{GabInd96,FucNon04,EbeVir06,FeiMod06,FeiKaz06,RidPer07,SemNot07,RidLog07,RidPer08,GabFus09,WooFus10,TsuTen12}. 
Extending ideas originating with Cardy \cite{Cardy86,Cardy89}, fusion can also
be implemented on the lattice, allowing one to construct new (lattice) representations from pairs of Kac boundary conditions \cite{PRZ06}. In some cases, the corresponding Hamiltonian 
acts non-diagonalisably on these new representations, a property believed to persist in the scaling limit. The ensuing Virasoro representations thus exhibit non-trivial Jordan blocks in $L_0$. In contrast, the Hamiltonians associated with the individual Kac boundary conditions are believed to be diagonalisable for all system sizes \cite{PRZ06}. For the simplest boundary conditions, this statement was proved recently \cite{MDRRSA15}. Imposing associativity and distributivity on the fusion rules has led to concrete conjectures for a variety of fusion algebras associated with the logarithmic minimal models \cite{PR07,ReaAss07,RasFus07a,RasFus07,RasCla11,BGT12}. 
These studies also give insight into the structure of the Virasoro modules associated
with the Kac boundary conditions. With reference to the Kac labels in \eqref{chitrs}, it is 
believed, for example, that the
modules are indecomposable unless $r=kp$, $s=k'p'$ and $k,k'>1$, in which case they are 
completely reducible but not irreducible. Our findings confirm this.

Some of the difficulties encountered in mathematically describing the $(r,s)$-type boundary conditions 
can be traced back to the fact that the lattice construction does not, in general, 
yield representations of the associated Temperley-Lieb algebra. This lack of a representation-theoretic framework is readily appreciated, at least indirectly, since the
well known Temperley-Lieb representation theory \cite{JonInd83,MarPot91,GW93,W95,RidSta12} 
does not accommodate the conformal structures that one might expect to see in the scaling limit.
It is then natural to suspect that the \emph{boundary} (or \emph{blob}) Temperley-Lieb 
algebras \cite{MS93,MW00,NRdG05,Nic06} could resolve this issue.
When the strip has at most one non-trivial boundary condition, this is indeed what we find, although the appropriate 
algebraic set-up requires taking a particular \emph{quotient} of the boundary algebra. 
Due to the appearance of so-called seams in the construction of the boundary conditions, 
we call these quotient algebras \emph{boundary seam algebras}. Particular care must be given
to the cases where the loop fugacity $\beta=q+q^{-1}$ is expressed in terms of a root of unity $q$. 

By construction, the modules obtained from the lattice with Kac boundary conditions can then be interpreted as 
\emph{standard modules} 
over the boundary seam algebras. As such, we refer to them as 
\emph{lattice Kac modules}. 
Originally defined somewhat heuristically in \cite{PRZ06,PRannecy,PRV12,PTC14}, these modules finally have a clear mathematical meaning. Here, we examine their structures 
using an invariant bilinear form defined on each of these standard modules, generalising the well known similar 
form \cite{MarPot91,W95,GraCel96,RidSta12} on the standard modules over the Temperley-Lieb algebra.

One of our main results is an explicit general conjecture (\cref{TheConjecture}) for the Virasoro modules arising 
as the scaling limit of the lattice Kac modules. 
These Virasoro modules are called (\emph{Virasoro}) \emph{Kac modules} and are defined as particular finitely 
generated submodules of Feigin-Fuchs modules.
In fact, the specific identification of the Kac modules is subtle. 
Feigin-Fuchs modules of chain or braid type structures come in pairs where one is obtained from the other by reversing all embedding arrows. This pairing is inherited by their submodules and means that the
contragredient module of a Kac module is, in general, not a Kac module.
As discussed in \cite{GabFus09}, see also \cite{RasCla11} for $\mathcal{LM}(1,p')$, the conformal fusion rules are believed to be
invariant under the interchange of each Kac module with its contragredient counterpart. The analysis of the fusion 
rules therefore cannot distinguish which should arise in the scaling limit of the lattice models. 
In stark contrast, the lattice Kac modules, as modules over the boundary seam algebras,
have unambiguous (albeit currently unknown) module structures. These are believed to persist in the scaling limit, so we will use the invariant bilinear form defined on these modules to determine the limiting Virasoro structures, thereby singling out the Virasoro Kac modules over their contragredients. 
Our results cover sufficiently many cases to justify \cref{TheConjecture}. 
Even for $\mathcal{LM}(1,p')$, this constitutes considerable new insight, as compared with \cite{RasCla11}. 

Additional evidence for the Kac module structure comes from conformal field theory considerations.
Motivated by the natural interpretation of lattice Kac modules as certain (lattice) fusion products of simpler lattice
Kac modules, we analyse the corresponding fusion products of Virasoro Kac modules in two independent ways. 
Using a Verlinde-like formula introduced in \cite{CreRel11}, we confirm the expected result, 
at the level of characters, for all Kac modules in any given logarithmic minimal model $\mathcal{LM}(p,p')$. 
This Verlinde-like formula falls under the umbrella of the \emph{standard module formalism} developed in \cite{CreLog13,RidVer14} for modular properties of logarithmic \cfts{}.  Additional tests of this formula may be found in \cite{CreMod12,BabTak12,CreMod13,RidMod14,RidBos14}.
We then apply the Nahm-Gaberdiel-Kausch fusion algorithm and again find exact agreement, at the level of the modules, in all the cases that we consider. It is highly non-trivial that these results, obtained directly at the conformal level, confirm the lattice analysis and conjectures.

For the structure of this paper, the bulk is divided into three phases.
The first phase (\cref{sec:LatMod}) concerns the algebraic description of the lattice construction of 
Kac boundary conditions. After a brief review of the logarithmic minimal models, the 
Temperley-Lieb algebras and the boundary Temperley-Lieb algebras, 
we define the boundary seam algebras that play a central role in our analysis.
We then define the standard modules over these
algebras and the corresponding invariant bilinear forms, deriving a formula for the Gram determinants. 
Some of the details and proofs are relegated to the appendices. 
In particular, \cref{app:TLrep} summarises the representation theory of the \TL{} algebras and \cref{app:TLone} reviews the one-boundary \TL{} algebras, including
a proof of the equivalence between its definition via diagrams and that via generators and relations (we have not found such a proof in the literature).
Finally, \cref{app:Seams} contains technical proofs pertaining to presentations of
the boundary seam algebras themselves, while
\cref{app:btlstan} is devoted to proving representation-theoretic results for these algebras, in particular those involving standard modules and Gram determinants.
 
The second phase (\cref{sec:Kacmod}) is concerned with establishing the connection
between the lattice models and their conformal scaling limits. We first define lattice Kac modules 
and their (lattice) fusion, before defining (Virasoro) 
Kac modules and discussing how they are believed to arise as scaling limits of lattice Kac modules. 
After briefly reviewing how one can guess the character of a limiting Virasoro module from numerical lattice data, and the limitations inherent in this procedure, we turn to an
investigation of the structure of the lattice Kac modules as modules over the boundary seam algebras.  
This structure is (partially) uncovered by using the Gram determinant of the module's invariant bilinear form.
We present numerical experiments 
and comparisons of the results with those expected in the continuum 
in \cref{sec:LatticeData}.  This evidence all 
supports \cref{TheConjecture} which precisely identifies the scaling limits of lattice Kac modules with Virasoro Kac modules. Background information on the representation theory
of the Virasoro algebra is found in \cref{sec:Back}, whereas \cref{sec:Gramdata} tabulates the
data that we have analysed with the aid of computer programs.

In the third phase (\cref{sec:CFTanalysis}), we present (conformal)
fusion results that provide strong evidence for the correctness of \cref{TheConjecture} from a conformal field-theoretic 
perspective.  This begins by deriving the modular transformation properties of the characters of the \FFms{}.  We then 
employ a Verlinde-like formula to determine the character of a fusion product of any two Kac modules.  The information 
so-obtained is then combined with the structure theory of Virasoro modules \cite{RohRed96,RidSta09} and explicit 
fusion computations using the Nahm-Gaberdiel-Kausch algorithm to identify the fusion product of two Kac modules in 
many examples. We present two example computations in detail in order to illustrate the methods used and the 
subtleties encountered.

Finally, \cref{sec:Conc} contains some concluding remarks. Here, we summarise the results of the paper as well as outline 
questions that remain unanswered and future directions that this work suggests.

%
\section{Lattice models and diagrammatic algebras} \label{sec:LatMod}
%

In the first phase of this work, our goal is to provide a rigorous algebraic framework for the study of transfer operators with $(r,s)$-type boundary conditions in logarithmic minimal models. \cref{sec:LMM} reviews the definition of logarithmic minimal models \cite{PRZ06} as lattice loop models, in particular their description in terms of the Temperley-Lieb algebra (\cref{sec:TLa}) and transfer tangles with boundary seams (\cref{sec:bdyseams}). \cref{sec:boundaryTLs} then describes these transfer tangles in the context of the one-boundary Temperley-Lieb algebra (\cref{sec:TLone}). It turns 
out that the natural description is in terms of quotients of the one-boundary Temperley-Lieb algebras that
we call the \emph{boundary seam algebras} 
(\cref{sec:seamsandbtl}). \cref{sec:reps,sec:Gram} then define representations and invariant bilinear forms, 
respectively, for the boundary seam algebras. 

%
\subsection{Logarithmic minimal models}\label{sec:LMM}
%

Logarithmic minimal models \cite{PRZ06} are lattice models defined in terms of non-local observables. They have an underlying structure, the Temperley-Lieb algebra $\tl_n$ \cite{TL}, which precisely encodes this non-locality. The representation theory of this algebra was studied by Jones \cite{JonInd83}, Martin \cite{MarPot91}, Goodman and Wenzl \cite{GW93} and Westbury \cite{W95}. 
A recent article by Ridout and Saint-Aubin \cite{RidSta12} gives a review that aims to be accessible to non-experts. 
\cref{sec:TLa} below recalls the definition of the algebra $\tl_n$, the Wenzl-Jones projectors \cite{WenHec88}, the 
standard modules and the invariant bilinear form \cite{W95,GraCel96}. Further basic results of the representation 
theory of the Temperley-Lieb algebra are reviewed in \cref{app:TLrep}.

The logarithmic minimal models are described by evolution operators called \emph{transfer tangles}. 
These operators are not matrices, but are instead constructed as elements of the diagrammatic algebra $\tl_n$ \cite{KauInv90}. 
Of particular interest for conformal field theory are transfer tangles with integrable boundary conditions. 
\cref{sec:bdyseams} reviews the definition of transfer tangles on lattices whose geometries match that of a strip with non-trivial Kac boundary seams on one side.

\subsubsection{Temperley-Lieb algebras}\label{sec:TLa}

\paragraph{Diagrammatic and algebraic definitions.}
Let $n \in \ZZ_+$. The objects that $\tl_n$ describes are formal linear combinations of diagrams called 
\emph{connectivities}. A connectivity is a diagram drawn in a box where $2n$ nodes, equally divided between the top and bottom edges of the box, are connected pairwise by non-intersecting arcs. 
Two connectivities are considered equal if the connections of their nodes are identical. To illustrate,
\begin{equation}
a_1 = \
\begin{pspicture}[shift=-0.4](-0.0,0)(2.4,1)
\pspolygon[fillstyle=solid,fillcolor=lightlightblue](0,0)(0,1)(2.4,1)(2.4,0)(0,0)
\psbezier[linecolor=blue,linewidth=\elegant]{-}(0.6,1)(0.6,0.5)(1.8,0.5)(1.8,1)
\psarc[linecolor=blue,linewidth=\elegant]{-}(1.2,1){0.2}{180}{360}
\psline[linecolor=blue,linewidth=\elegant]{-}(0.2,0)(0.2,1)
\psline[linecolor=blue,linewidth=\elegant]{-}(2.2,0)(2.2,1)
\psarc[linecolor=blue,linewidth=\elegant]{-}(0.8,0){0.2}{0}{180}
\psarc[linecolor=blue,linewidth=\elegant]{-}(1.6,0){0.2}{0}{180}
\end{pspicture} \qquad {\rm and} \qquad
a_2 = \
\begin{pspicture}[shift=-0.4](-0.0,0)(2.4,1)
\pspolygon[fillstyle=solid,fillcolor=lightlightblue](0,0)(0,1)(2.4,1)(2.4,0)(0,0)
\psarc[linecolor=blue,linewidth=\elegant]{-}(0.4,1){0.2}{180}{360}
\psarc[linecolor=blue,linewidth=\elegant]{-}(1.2,1){0.2}{180}{360}
\psbezier[linecolor=blue,linewidth=\elegant]{-}(0.2,0)(0.2,0.5)(1.8,0.5)(1.8,1)
\psline[linecolor=blue,linewidth=\elegant]{-}(2.2,0)(2.2,1)
\psarc[linecolor=blue,linewidth=\elegant]{-}(1.2,0){0.2}{0}{180}
\psbezier[linecolor=blue,linewidth=\elegant]{-}(0.6,0)(0.6,0.5)(1.8,0.5)(1.8,0)
\end{pspicture} \label{eq:a12}
\end{equation}
are two distinct connectivities in $\tl_6$. We shall refer to linear combinations of connectivities as \emph{tangles}.

The algebraic structure of the Temperley-Lieb algebras depends upon a parameter $\beta$, often called the loop fugacity. For now, $\beta$ is taken to be a 
\emph{formal} parameter, in which case the linear combinations of connectivities have coefficients in some complex field of functions.\footnote{In this section, we may take the field to be the rational function field $\CC(\beta)$.  However, this will not suffice in \cref{sec:bdyseams} where we shall require more complicated functional dependences involving $\beta$ and additional formal variables.}
We shall eventually consider the specialisation to $\beta\in\mathbb{C}$, in which case the specialised Temperley-Lieb algebras will be denoted by $\tl_n(\beta)$. 

The product $a_1 a_2$  of two connectivities is defined using the following diagrammatic recipe. One draws $a_1$ below $a_2$, identifies the top edge of $a_1$ with the bottom edge of $a_2$, erases this identified edge, counts the number $b$ of (closed) loops, and finally replaces them by a multiplicative factor of $\beta^b$.
An example is useful:
\begin{equation}
a_1 a_2 =\ \begin{pspicture}[shift=-0.9](-0.0,0)(2.4,2)
\pspolygon[fillstyle=solid,fillcolor=lightlightblue](0,0)(0,2)(2.4,2)(2.4,0)(0,0)
\psline[linecolor=blue, linewidth=\elegant]{-}(0.2,0)(0.2,1)
\psline[linecolor=blue, linewidth=\elegant]{-}(2.2,0)(2.2,1)
\psarc[linecolor=blue,linewidth=\elegant]{-}(1.2,1){0.2}{180}{360}
\psbezier[linecolor=blue,linewidth=\elegant]{-}(0.6,1)(0.6,0.5)(1.8,0.5)(1.8,1)
\psarc[linecolor=blue,linewidth=\elegant]{-}(0.8,0){0.2}{0}{180}
\psarc[linecolor=blue,linewidth=\elegant]{-}(1.6,0){0.2}{0}{180}
\psline[linewidth=\mince]{-}(0,2)(2.4,2)(2.4,1)(0,1)(0,2)
\rput(0,1){
\psarc[linecolor=blue,linewidth=\elegant]{-}(0.4,1){0.2}{180}{360}
\psarc[linecolor=blue,linewidth=\elegant]{-}(1.2,1){0.2}{180}{360}
\psarc[linecolor=blue,linewidth=\elegant]{-}(1.2,0){0.2}{0}{180}
\psbezier[linecolor=blue,linewidth=\elegant]{-}(0.6,0)(0.6,0.5)(1.8,0.5)(1.8,0)
\psbezier[linecolor=blue,linewidth=\elegant]{-}(0.2,0)(0.2,0.5)(1.8,0.5)(1.8,1)
\psline[linecolor=blue,linewidth=\elegant]{-}(2.2,0)(2.2,1)
}
\end{pspicture}
\ = \beta^2\ 
\begin{pspicture}[shift=-0.4](-0.0,0)(2.4,1)
\pspolygon[fillstyle=solid,fillcolor=lightlightblue](0,0)(0,1)(2.4,1)(2.4,0)(0,0)
\psarc[linecolor=blue,linewidth=\elegant]{-}(0.4,1){0.2}{180}{360}
\psarc[linecolor=blue,linewidth=\elegant]{-}(1.2,1){0.2}{180}{360}
\psbezier[linecolor=blue,linewidth=\elegant]{-}(0.2,0)(0.2,0.5)(1.8,0.5)(1.8,1)
\psline[linecolor=blue,linewidth=\elegant]{-}(2.2,0)(2.2,1)
\psarc[linecolor=blue,linewidth=\elegant]{-}(0.8,0){0.2}{0}{180}
\psarc[linecolor=blue,linewidth=\elegant]{-}(1.6,0){0.2}{0}{180}
\end{pspicture}\ = \beta^2 a_3. \label{eq:a3}
\end{equation}
This product is linearly extended to all the tangles of $\tl_n$, turning the space of tangles into an associative algebra. This algebra is finite-dimensional with dimension 
\begin{equation}\label{eq:dimtl}
 \dim \tl_n=\frac 1{n+1} \binom{2n}{n}.
\end{equation}
The dimension does not change upon specialising to $\beta \in \CC$.

While the definition given above is purely diagrammatic, the Temperley-Lieb algebra also has a
presentation in terms of generators and relations:  an identity element $I$ 
and $n-1$ elements denoted by $e_j$, $j = 1, \dots, n-1$, satisfying
\begin{equation}
e_j^2=\beta e_j, \qquad e_j e_{j\pm1} e_j = e_j,  \qquad e_i e_j = e_j e_i
 \quad (|i-j|>1).
 \label{eq:defTL}
\end{equation}
The generators are identified with connectivities as follows:
\begin{equation}
I =
\ 
\begin{pspicture}[shift=-0.525](-0.0,-0.25)(2.4,0.8)
\pspolygon[fillstyle=solid,fillcolor=lightlightblue](0,0)(0,0.8)(2.4,0.8)(2.4,0)(0,0)
\psline[linecolor=blue,linewidth=\elegant]{-}(0.2,0)(0.2,0.8)
\psline[linecolor=blue,linewidth=\elegant]{-}(0.6,0)(0.6,0.8)
\psline[linecolor=blue,linewidth=\elegant]{-}(1.0,0)(1.0,0.8)
\rput(1.4,0.4){$...$}
\psline[linecolor=blue,linewidth=\elegant]{-}(1.8,0)(1.8,0.8)
\psline[linecolor=blue,linewidth=\elegant]{-}(2.2,0)(2.2,0.8)
\rput(0.2,-0.25){$_1$}
\rput(0.6,-0.25){$_2$}
\rput(1.0,-0.25){$_3$}
\rput(2.2,-0.25){$_n$}
\end{pspicture}
\ , \qquad
e_j = \ 
\begin{pspicture}[shift=-0.3](-0.0,0)(3.2,0.8)
\pspolygon[fillstyle=solid,fillcolor=lightlightblue](0,0)(0,0.8)(3.2,0.8)(3.2,0)(0,0)
\psline[linecolor=blue,linewidth=\elegant]{-}(0.2,0)(0.2,0.8)
\rput(0.6,0.4){$...$}
\psline[linecolor=blue,linewidth=\elegant]{-}(1.0,0)(1.0,0.8)
\psarc[linecolor=blue,linewidth=\elegant]{-}(1.6,0.8){0.2}{180}{360}
\psarc[linecolor=blue,linewidth=\elegant]{-}(1.6,0){0.2}{0}{180}
\psline[linecolor=blue,linewidth=\elegant]{-}(2.2,0)(2.2,0.8)
\rput(2.6,0.4){$...$}
\psline[linecolor=blue,linewidth=\elegant]{-}(3.0,0)(3.0,0.8)
\rput(0.2,-0.25){$_1$}
\rput(1.4,-0.25){$_{\phantom{+}j\phantom{+}}$}
\rput(3.0,-0.25){$_n$}
\end{pspicture}
\ .
\label{eq:iso}
\end{equation}
A proof of the equivalence of this algebraic description and the diagrammatic one presented above may be found in \cite{RidSta12}.  In particular, every
connectivity can then be written as a word in the generators. For the examples in \eqref{eq:a12} and \eqref{eq:a3}, we have $a_1 = e_2 e_4 e_3$, $a_2 = e_3 e_2 e_4 e_1 e_3$ and $a_3 = e_2 e_4 e_1 e_3$.

\paragraph{Projectors.} 
The \emph{Wenzl-Jones projectors} \cite{JonInd83,WenHec88,KauTem94} comprise an important family of elements of $\tl_n$. They are defined recursively by
\begin{equation}
P_1 = I, \qquad P_i = P_{i-1} - \frac {U_{i-2}(\tfrac \beta 2)}{U_{i-1}(\tfrac \beta 2)} P_{i-1} e_{i-1}P_{i-1} \quad(i=2, \dots, n),
\label{eq:WJs}
\end{equation}
where $U_k(x)$ denotes the $k$-th Chebyshev polynomial of the second kind. We will 
omit the argument of $U_k$ in what follows, understanding that it is always
$\tfrac \beta 2$. The first two non-trivial Wenzl-Jones projectors are
\begin{equation}
P_2 = I - \frac 1 \beta \, e_1, \qquad P_3 = I - \frac \beta{\beta^2 - 1} (e_1 + e_2) + \frac 1{\beta^2 - 1} (e_1e_2 + e_1e_2).
\end{equation}
It is clear from these examples that the \WJ{} projectors need not be defined when $\beta$ is specialised to a complex number.  This
illustrates the advantage of treating $\beta$ as a formal parameter. 
The following properties follow straightforwardly from \eqref{eq:WJs}:
\begin{subequations} \label{eq:WJprop}
\begin{gather}
(P_i)^2 = P_i, \qquad P_i e_j = e_j P_i = 0 \quad \text{($1\le j < i$),} \qquad P_i e_j = e_j P_i \quad \text{($j \ge i+1$),} \label{eq:WJprop1} \\
P_i P_j = P_{\max (i,j)}, \qquad e_i P_i e_i = \frac{U_i}{U_{i-1}} e_i P_{i-1}. \label{eq:WJprop2}
\end{gather}
\end{subequations}

As defined above, the projector $P_i$ only involves the $e_j$ with $1 \le j<i$, so it only acts non-trivially on the nodes $1$ through $i$. It is, however, possible to define Wenzl-Jones projectors
acting on other sets of consecutive nodes.
Indeed, denoting the projector $P_i$ diagrammatically by
$\begin{pspicture}[shift=-0.05](-0.03,0.00)(1.07,0.3)
\pspolygon[fillstyle=solid,fillcolor=pink](0,0)(1,0)(1,0.3)(0,0.3)(0,0)\rput(0.5,0.15){$_i$}
\end{pspicture}
$, we see that one may define, for example,
\begin{equation}
\begin{pspicture}[shift=-0.3](0,-0.0)(1.25,0.8)
\pspolygon[fillstyle=solid,fillcolor=lightlightblue](0,0)(0,0.8)(1.25,0.8)(1.25,0)
\psline[linewidth=1.5pt,linecolor=blue]{-}(0.2,0)(0.2,0.8)
\psline[linewidth=1.5pt,linecolor=blue]{-}(0.6,0)(0.6,0.8)
\psline[linewidth=1.5pt,linecolor=blue]{-}(1.0,0)(1.0,0.8)
\pspolygon[fillstyle=solid,fillcolor=pink](0.45,0.25)(0.45,0.55)(1.15,0.55)(1.15,0.25)\rput(0.8,0.4){\small $_2$}
\end{pspicture} \ = \
\begin{pspicture}[shift=-0.3](0,-0.0)(1.2,0.8)
\pspolygon[fillstyle=solid,fillcolor=lightlightblue](0,0)(0,0.8)(1.2,0.8)(1.2,0)
\psline[linewidth=1.5pt,linecolor=blue]{-}(0.2,0)(0.2,0.8)
\psline[linewidth=1.5pt,linecolor=blue]{-}(0.6,0)(0.6,0.8)
\psline[linewidth=1.5pt,linecolor=blue]{-}(1.0,0)(1.0,0.8)
\end{pspicture} \ - \frac 1 \beta \ 
\begin{pspicture}[shift=-0.3](0,-0.0)(1.2,0.8)
\pspolygon[fillstyle=solid,fillcolor=lightlightblue](0,0)(0,0.8)(1.2,0.8)(1.2,0)
\psline[linewidth=1.5pt,linecolor=blue]{-}(0.2,0)(0.2,0.8)
\psarc[linewidth=1.5pt,linecolor=blue]{-}(0.8,0){0.2}{0}{180}
\psarc[linewidth=1.5pt,linecolor=blue]{-}(0.8,0.8){0.2}{180}{0}
\end{pspicture}
\: .
\end{equation}
This then involves a Wenzl-Jones projector $P_2'$ acting non-trivially on nodes $2$ and $3$ (instead of $1$ and $2$).  The
projector $P_2'$ is obtained from $P_2$ simply by shifting the index of $e_1$ by one: $P_2' = I - \frac{1}{\beta} e_2$.  
Similar shifts define more general projectors in the obvious manner.

\paragraph{Standard modules.} 

A \emph{link state} is a diagram drawn above a horizontal line marked with $n$ nodes, in which the nodes are either connected pairwise by non-intersecting arcs, or are connected by a vertical line, called a \emph{defect}, to infinity.
The number $d$ of defects of a link state is clearly constrained to have the same parity as $n$. As with Temperley-Lieb diagrams, two link states are considered equal if their nodes are connected identically. 
We denote by $\links_n^d$ the set of link states with $n$ nodes and $d$ defects. For $n=5$ and $d=1$, for example, there are five distinct link states:
\begin{equation}
\mathcal B_5^1 = \Big\{ \,
\psset{unit=0.8}
\begin{pspicture}[shift=-0.08](-0.0,0)(2.2,0.5)
\psline[linewidth=\mince](0,0)(2.0,0)
\psline[linecolor=blue,linewidth=\elegant]{-}(0.2,0)(0.2,0.5)
\psarc[linecolor=blue,linewidth=\elegant]{-}(1.2,0){0.2}{0}{180}
\psbezier[linecolor=blue,linewidth=\elegant](0.6,0)(0.6,0.6)(1.8,0.6)(1.8,0)
\rput(2.15,0){,}
\end{pspicture} \ 
\begin{pspicture}[shift=-0.08](-0.0,0)(2.2,0.5)
\psline[linewidth=\mince](0,0)(2.0,0)
\psline[linecolor=blue,linewidth=\elegant]{-}(1.8,0)(1.8,0.5)
\psarc[linecolor=blue,linewidth=\elegant]{-}(0.8,0){0.2}{0}{180}
\psbezier[linecolor=blue,linewidth=\elegant](0.2,0)(0.2,0.6)(1.4,0.6)(1.4,0)
\rput(2.15,0){,}
\end{pspicture}\
\begin{pspicture}[shift=-0.08](-0.0,0)(2.2,0.5)
\psline[linewidth=\mince](0,0)(2.0,0)
\psline[linecolor=blue,linewidth=\elegant]{-}(1.8,0)(1.8,0.5)
\psarc[linecolor=blue,linewidth=\elegant]{-}(0.4,0){0.2}{0}{180}
\psarc[linecolor=blue,linewidth=\elegant]{-}(1.2,0){0.2}{0}{180}
\rput(2.15,0){,}
\end{pspicture} \
\begin{pspicture}[shift=-0.08](-0.0,0)(2.2,0.5)
\psline[linewidth=\mince](0,0)(2.0,0)
\psline[linecolor=blue,linewidth=\elegant]{-}(0.2,0)(0.2,0.5)
\psarc[linecolor=blue,linewidth=\elegant]{-}(0.8,0){0.2}{0}{180}
\psarc[linecolor=blue,linewidth=\elegant]{-}(1.6,0){0.2}{0}{180}
\rput(2.15,0){,}
\end{pspicture}\ 
\begin{pspicture}[shift=-0.08](-0.0,0)(2.0,0.5)
\psline[linewidth=\mince](0,0)(2.0,0)
\psline[linecolor=blue,linewidth=\elegant]{-}(1.0,0)(1.0,0.5)
\psarc[linecolor=blue,linewidth=\elegant]{-}(0.4,0){0.2}{0}{180}
\psarc[linecolor=blue,linewidth=\elegant]{-}(1.6,0){0.2}{0}{180}
\end{pspicture}\,
\Big\}.
\label{eq:B51}
\end{equation}
 
The \emph{standard representation $\rho_n^d$} of the algebra $\tl_n$ is constructed on the vector space 
spanned by the link states of $\links_n^d$:
\begin{equation}
\rho_n^d \colon \tl_n \rightarrow \End(\vspn \links_n^d) \qquad \text{(\(0\le d \le n\), \(d = n \bmod{2}\)).} 
\label{eq:rhod}
\end{equation} 
We define the standard action $aw$ for connectivities $a\in \tl_n$ and link states $w \in \links_n^d$, understanding that this is extended linearly to all tangles in $\tl_n$ and all linear combinations of the link states of $\links_n^d$.
This action is defined diagrammatically, with the recipe closely resembling the rule for 
multiplying connectivities. We draw $a$ below $w$, identify
the nodes of $w$ with those of the top edge of $a$, count the number $b$ of (closed) loops, erase them, read the new link state from the connections of the bottom $n$ nodes, and set $aw$ to be this new link state with a multiplicative factor of $\beta^b$. A final modification is performed: if the resulting link state has fewer than $d$ defects, then $aw$ is set to zero. This last rule ensures that $\vspn \links_n^d$ is indeed invariant under the Temperley-Lieb action. Here are two examples for $(n,d) = (5,1)$ and $(n,d) = (5,3)$:
\begin{equation}
\psset{unit=0.8}
\begin{pspicture}[shift=-0.4](-0.0,0)(2.0,1.6)
\pspolygon[fillstyle=solid,fillcolor=lightlightblue](0,0)(0,1)(2.0,1)(2.0,0)(0,0)
\psbezier[linecolor=blue,linewidth=\elegant]{-}(1.8,1)(1.8,0.5)(1.0,0.5)(1.0,0)
\psarc[linecolor=blue,linewidth=\elegant]{-}(0.4,1){0.2}{180}{0}
\psarc[linecolor=blue,linewidth=\elegant]{-}(1.2,1){0.2}{180}{0}
\psarc[linecolor=blue,linewidth=\elegant]{-}(0.4,0){0.2}{0}{180}
\psarc[linecolor=blue,linewidth=\elegant]{-}(1.6,0){0.2}{0}{180}
\psbezier[linecolor=blue,linewidth=\elegant]{-}(0.6,1)(0.6,1.6)(1.8,1.6)(1.8,1)
\psarc[linecolor=blue,linewidth=\elegant]{-}(1.2,1){0.2}{0}{180}
\psline[linecolor=blue,linewidth=\elegant]{-}(0.2,1)(0.2,1.6)
\end{pspicture} \ = \beta \
\begin{pspicture}[shift=-0.02](-0.0,0)(2.0,0.5)
\psline[linewidth=\mince](0,0)(2.0,0)
\psline[linecolor=blue,linewidth=\elegant]{-}(1.0,0)(1.0,0.5)
\psarc[linecolor=blue,linewidth=\elegant]{-}(0.4,0){0.2}{0}{180}
\psarc[linecolor=blue,linewidth=\elegant]{-}(1.6,0){0.2}{0}{180}
\end{pspicture}\ , \qquad
\begin{pspicture}[shift=-0.4](-0.0,0)(2.0,1.6)
\pspolygon[fillstyle=solid,fillcolor=lightlightblue](0,0)(0,1)(2.0,1)(2.0,0)(0,0)
\psbezier[linecolor=blue,linewidth=\elegant]{-}(1.8,0)(1.8,0.5)(1.0,0.5)(1.0,1)
\psbezier[linecolor=blue,linewidth=\elegant]{-}(1.4,0)(1.4,0.5)(0.6,0.5)(0.6,1)
\psarc[linecolor=blue,linewidth=\elegant]{-}(0.8,0){0.2}{0}{180}
\psarc[linecolor=blue,linewidth=\elegant]{-}(1.6,1){0.2}{180}{0}
\psline[linecolor=blue,linewidth=\elegant]{-}(0.2,0)(0.2,1.0)
\psarc[linecolor=blue,linewidth=\elegant]{-}(0.4,1){0.2}{0}{180}
\psline[linecolor=blue,linewidth=\elegant]{-}(1.0,1)(1.0,1.6)
\psline[linecolor=blue,linewidth=\elegant]{-}(1.4,1)(1.4,1.6)
\psline[linecolor=blue,linewidth=\elegant]{-}(1.8,1)(1.8,1.6)
\end{pspicture} \ = 0.
\label{eq:standardexample}
\end{equation}

This gives rise to $\lfloor \frac{n+2}2 \rfloor$ representations of the Temperley-Lieb 
algebra, one for each integer $d$ subject to the constraints in \eqref{eq:rhod}. 
We denote by $\stan_n^d$ the \emph{standard module}, that is the span of $\links_n^d$ endowed with the \TL{}
action just defined.  
As we shall see, the standard modules are irreducible over complex function fields.  
Moreover, their dimensions are given by 
\begin{equation}
\dim \stan_n^d = \begin{pmatrix}n\\ \tfrac {n-d}2\end{pmatrix} - \begin{pmatrix}n\\ \tfrac {n-d-2}2\end{pmatrix}.
\label{eq:dimV}
\end{equation}
The structures of the standard modules over the specialised Temperley-Lieb algebras $\tl_n(\beta)$ depend
on $\beta\in\mathbb{C}$ and can be more complicated. These structures are described in \cref{app:TLrep}, as are those of the irreducible and projective $\tl_n(\beta)$-modules. 

\paragraph{Invariant bilinear forms.} 

A standard tool for investigating the representation theory of cellular algebras \cite{GraCel96}, such as the \TL{} 
algebra, are the invariant bilinear forms $\gramprod{\cdot}{\cdot}$ defined on each standard module.  
These forms take values in the field over which the algebras are defined.

We define a bilinear form $\gramprod{\cdot}{\cdot}$ on $\stan_n^d$ by giving its value $\gramprod{w_1}{w_2}$ on the 
link states $w_1,w_2$ of the basis $\links_n^d$.  The definition is diagrammatic: 
One performs a vertical flip of $w_2$, identifying 
its horizontal segment with that belonging to $w_1$ so that their nodes coincide. The result is a collection of arcs living 
above and below a horizontal line that either form loops or connect defects pairwise. The defects of $w_1$ and $w_2$ 
become \emph{upward} and \emph{downward} pointing, respectively, in this procedure and $\gramprod{w_1}{w_2}$ is 
zero unless every upward defect is connected to a downward one. If this condition is met, then the form evaluates to 
$\gramprod{w_1}{w_2} = \beta^b$, where $b$ is the number of (closed) loops in the diagram. It is readily seen that this 
bilinear form is symmetric. Examples of evaluations of these forms are
\begin{equation}
\psset{unit=0.7}
\Big\langle \
\begin{pspicture}[shift=-0.12](-0.0,0)(3.6,0.5)
\psline[linewidth=\mince](0,0)(3.6,0)
\psline[linecolor=blue,linewidth=\elegant]{-}(0.2,0)(0.2,0.5)
\psarc[linecolor=blue,linewidth=\elegant]{-}(0.8,0){0.2}{0}{180}
\psarc[linecolor=blue,linewidth=\elegant]{-}(1.6,0){0.2}{0}{180}
\psarc[linecolor=blue,linewidth=\elegant]{-}(2.8,0){0.2}{0}{180}
\psbezier[linecolor=blue,linewidth=\elegant](2.2,0)(2.2,0.6)(3.4,0.6)(3.4,0)
\end{pspicture} \ \Big | \
\begin{pspicture}[shift=-0.12](-0.0,0)(3.6,0.5)
\psline[linewidth=\mince](0,0)(3.6,0)
\psline[linecolor=blue,linewidth=\elegant]{-}(3.4,0)(3.4,0.5)
\psarc[linecolor=blue,linewidth=\elegant]{-}(1.2,0){0.2}{0}{180}
\psbezier[linecolor=blue,linewidth=\elegant](0.6,0)(0.6,0.6)(1.8,0.6)(1.8,0)
\psbezier[linecolor=blue,linewidth=\elegant](0.2,0)(0.2,0.9)(2.2,0.9)(2.2,0)
\psarc[linecolor=blue,linewidth=\elegant]{-}(2.8,0){0.2}{0}{180}
\end{pspicture}\ \Big \rangle = \beta ^2, \qquad
\Big \langle
\begin{pspicture}[shift=-0.12](-0.0,0)(3.6,0.5)
\psline[linewidth=\mince](0,0)(3.6,0)
\psline[linecolor=blue,linewidth=\elegant]{-}(0.2,0)(0.2,0.5)
\psline[linecolor=blue,linewidth=\elegant]{-}(1.4,0)(1.4,0.5)
\psline[linecolor=blue,linewidth=\elegant]{-}(2.6,0)(2.6,0.5)
\psarc[linecolor=blue,linewidth=\elegant]{-}(0.8,0){0.2}{0}{180}
\psarc[linecolor=blue,linewidth=\elegant]{-}(2.0,0){0.2}{0}{180}
\psarc[linecolor=blue,linewidth=\elegant]{-}(3.2,0){0.2}{0}{180}
\end{pspicture} \ \Big | \
\begin{pspicture}[shift=-0.12](-0.0,0)(3.6,0.5)
\psline[linewidth=\mince](0,0)(3.6,0)
\psline[linecolor=blue,linewidth=\elegant]{-}(3.4,0)(3.4,0.5)
\psarc[linecolor=blue,linewidth=\elegant]{-}(0.8,0){0.2}{0}{180}
\psbezier[linecolor=blue,linewidth=\elegant](0.2,0)(0.2,0.6)(1.4,0.6)(1.4,0)
\psline[linecolor=blue,linewidth=\elegant]{-}(1.8,0)(1.8,0.5)
\psline[linecolor=blue,linewidth=\elegant]{-}(2.2,0)(2.2,0.5)
\psarc[linecolor=blue,linewidth=\elegant]{-}(2.8,0){0.2}{0}{180}
\end{pspicture}\ \Big \rangle = 0,
\end{equation}
which are easily read off from the diagrams
\begin{equation}
\psset{unit=0.7}
\begin{pspicture}[shift=-0.4](-0.0,-0.5)(3.6,0.5)
\psline[linewidth=\mince](0,0)(3.6,0)
\psline[linecolor=blue,linewidth=\elegant]{-}(0.2,0)(0.2,0.5)
\psarc[linecolor=blue,linewidth=\elegant]{-}(0.8,0){0.2}{0}{180}
\psarc[linecolor=blue,linewidth=\elegant]{-}(1.6,0){0.2}{0}{180}
\psarc[linecolor=blue,linewidth=\elegant]{-}(2.8,0){0.2}{0}{180}
\psbezier[linecolor=blue,linewidth=\elegant](2.2,0)(2.2,0.6)(3.4,0.6)(3.4,0)
\psline[linecolor=blue,linewidth=\elegant]{-}(3.4,0)(3.4,-0.5)
\psarc[linecolor=blue,linewidth=\elegant]{-}(1.2,0){-0.2}{0}{180}
\psbezier[linecolor=blue,linewidth=\elegant](0.6,0)(0.6,-0.6)(1.8,-0.6)(1.8,0)
\psbezier[linecolor=blue,linewidth=\elegant](0.2,0)(0.2,-0.9)(2.2,-0.9)(2.2,0)
\psarc[linecolor=blue,linewidth=\elegant]{-}(2.8,0){-0.2}{0}{180}
\end{pspicture} \qquad {\rm and} \qquad
\begin{pspicture}[shift=-0.4](-0.0,-0.5)(3.6,0.5)
\psline[linewidth=\mince](0,0)(3.6,0)
\psline[linecolor=blue,linewidth=\elegant]{-}(0.2,0)(0.2,0.5)
\psline[linecolor=blue,linewidth=\elegant]{-}(1.4,0)(1.4,0.5)
\psline[linecolor=blue,linewidth=\elegant]{-}(2.6,0)(2.6,0.5)
\psarc[linecolor=blue,linewidth=\elegant]{-}(0.8,0){0.2}{0}{180}
\psarc[linecolor=blue,linewidth=\elegant]{-}(2.0,0){0.2}{0}{180}
\psarc[linecolor=blue,linewidth=\elegant]{-}(3.2,0){0.2}{0}{180}
\psline[linecolor=blue,linewidth=\elegant]{-}(3.4,0)(3.4,-0.5)
\psarc[linecolor=blue,linewidth=\elegant]{-}(0.8,0){-0.2}{0}{180}
\psbezier[linecolor=blue,linewidth=\elegant](0.2,0)(0.2,-0.6)(1.4,-0.6)(1.4,0)
\psline[linecolor=blue,linewidth=\elegant]{-}(1.8,0)(1.8,-0.5)
\psline[linecolor=blue,linewidth=\elegant]{-}(2.2,0)(2.2,-0.5)
\psarc[linecolor=blue,linewidth=\elegant]{-}(2.8,0){-0.2}{0}{180}
\end{pspicture}\ .
\end{equation}
The \emph{Gram matrix} of the bilinear form $\gramprod{\cdot}{\cdot}$, in the link state basis $\links_n^d$, 
is denoted by $\grammat_n^d$. For $(n,d) = (5,1)$, in the ordered basis \eqref{eq:B51}, we have
\begin{equation}
\grammat_5^1 = \begin{pmatrix}
\beta^2 & 1 & \beta & \beta & 1\\
1 & \beta^2 & \beta & \beta & 1\\
\beta & \beta & \beta^2 & 1 & \beta\\
\beta & \beta & 1 & \beta^2 & \beta\\
1 & 1 & \beta & \beta & \beta^2\\
\end{pmatrix}.
\end{equation}

The bilinear form defined above is invariant under the action of $\tl_n$, meaning that
\begin{equation}
\gramprod{w_1}{a w_2} =\gramprod{a^\dagger w_1}{w_2} \qquad \text{(\(a \in \tl_n\), \(w_1,w_2 \in \stan_n^d\)),}
\end{equation}
where $a^\dagger \in \tl_n$ is the adjoint of $a$, obtained for a connectivity
$a$ by flipping it vertically, or equivalently by reversing the order of composition of products of the $e_i$ in the algebraic formalism. In particular, the generators and the \WJ{} projectors are self-adjoint: 
$I^{\dag} = I$, $e_j^\dagger = e_j$ and $P_j^{\dag} = P_j$. For this last result, we remark that this adjoint is 
extended linearly, so without complex conjugation, to general tangles.

This invariant bilinear form appears naturally in certain computations involving link states. Any connectivity $a$ can be obtained from a pair of link states $w_1,w_2 \in \links_n^d$, with the same number $d$ of defects, by drawing $w_2$ upside down above $w_1$ and connecting upward and downward pointing defects together. There is a unique way of doing this given that the resulting arcs may not cross. Then, if $w_3$ is also in $\links_n^d$, the action of $a$ on $w_3$ is given by
\begin{equation}
a\, w_3 = \gramprod{w_2}{w_3} \cdot w_1.
\end{equation}
The first example in \eqref{eq:standardexample}, for instance, corresponds to 
$w_1 = \psset{unit=0.6cm}
\begin{pspicture}[shift=-0.02](-0.0,0)(2.0,0.5)
\psline[linewidth=\mince](0,0)(2.0,0)
\psline[linecolor=blue,linewidth=\elegant]{-}(1.0,0)(1.0,0.5)
\psarc[linecolor=blue,linewidth=\elegant]{-}(0.4,0){0.2}{0}{180}
\psarc[linecolor=blue,linewidth=\elegant]{-}(1.6,0){0.2}{0}{180}
\end{pspicture}$\,, 
$w_2= \psset{unit=0.6cm}
\begin{pspicture}[shift=-0.02](-0.0,0)(2.0,0.5)
\psline[linewidth=\mince](0,0)(2.0,0)
\psline[linecolor=blue,linewidth=\elegant]{-}(1.8,0)(1.8,0.5)
\psarc[linecolor=blue,linewidth=\elegant]{-}(0.4,0){0.2}{0}{180}
\psarc[linecolor=blue,linewidth=\elegant]{-}(1.2,0){0.2}{0}{180}
\end{pspicture}$ and 
$w_3=
\psset{unit=0.6cm}
\begin{pspicture}[shift=-0.02](-0.0,0)(2.0,0.5)
\psline[linewidth=\mince](0,0)(2.0,0)
\psline[linecolor=blue,linewidth=\elegant]{-}(0.2,0)(0.2,0.5)
\psarc[linecolor=blue,linewidth=\elegant]{-}(1.2,0){0.2}{0}{180}
\psbezier[linecolor=blue,linewidth=\elegant](0.6,0)(0.6,0.6)(1.8,0.6)(1.8,0)
\end{pspicture}
$\,. 

The virtue of having an invariant form on the standard module $\stan_n^d$ is that its \emph{radical},
\begin{equation}
 \Rad_n^d = \{w \in \stan_n^d \st \gramprod{v}{w} = 0 \text{ for all } v \in \stan_n^d\},
\end{equation}
is a submodule of the standard module $\stan_n^d$. Moreover, the quotient $\stan_n^d/\Rad_n^d$ is either zero or irreducible \cite{GraCel96}. A non-trivial radical is equivalent to 
the existence of eigenvectors of the Gram matrix $\grammat_n^d$ with eigenvalues equal to zero, so $\det \grammat_n^d$ reveals at least part of the structure of $\stan_n^d$. This \emph{Gram determinant} 
has a closed-form expression given by \cite{W95,RidSta12}
\begin{equation}
\det \grammat_n^d = \prod_{j=1}^{\frac{n-d}2} \Big(\frac{U_{d+j}}{U_{j-1}}\Big)^{\dim \stan_n^{d+2j}}.
\label{eq:GramdetTL}
\end{equation}
This is always non-zero when we work over a complex function field, hence the standard modules of $\tl_n$ are always irreducible in this setting.  However, it is clear that specialising $\beta$ to a complex number may lead to $\det \grammat_n^d = 0$, hence a reducible standard module.
Despite the denominator in
the expression \eqref{eq:GramdetTL}, the determinant is, by construction, polynomial in $\beta$ and therefore well-defined for all $\beta \in \CC$.

\paragraph{Specialising $\boldsymbol \beta$.}
As mentioned above, until \cref{sec:Kacmod}, we choose to work with $\beta$ as a formal parameter, in order to exploit the existence of the \WJ{} projectors. 
However, the investigation of the logarithmic minimal models $\mathcal{LM}(p,p')$ requires that $\beta$ be specialised to values in $\mathbb C$ and, in particular, to 
\begin{equation}
\beta = q + q^{-1}, \quad \text{with} \quad q = \ee^{i \lambda} \quad \text{and} \quad \lambda = \frac{(p'-p)\pi}{p'} \qquad \text{(\(p,p' \in \ZZ_+\), \(1 \le p<p'\), \(\gcd (p,p') = 1\))}.
\label{eq:rootsof1}
\end{equation}
These values of $q$ (and $\beta$) will be identified as 
\emph{roots of unity}, while the other values in $\mathbb C$ will be termed 
\emph{generic}.  Note that $q=1$ and $-1$ are generic values according to this definition. The representation theory of $\tl_n(\beta)$, when $\beta$ is specialised to a generic value, is quite different to that resulting from specialising to a root of unity. In particular, the representation theory for generic values is completely reducible --- all (finite-dimensional) modules are direct sums of irreducible
modules. This is not true for $q$ a root of unity: Some Wenzl-Jones projectors have singularities and some Gram determinants are zero, meaning that some standard modules become reducible, although they remain indecomposable.  We refer to \cref{app:TLrep} for further discussion.

\subsubsection{Transfer tangles with boundary seams} \label{sec:bdyseams}

The logarithmic minimal models are defined in terms of transfer operators that are elements of the diagrammatic algebras. Here, we are interested in the geometry of a strip and study transfer operators with boundary conditions on the right that take the form of Kac boundary triangles, following the ideas and conventions of Pearce, Rasmussen and Zuber \cite{PRZ06}. The 
\emph{double-row transfer tangle} $\Dbk(u,\xi)$  
is a two-parameter tangle in $\tl_{n+k}$,\footnote{Here, we generalise the function field over which $\tl_{n+k}$ is defined so as to incorporate the formal parameters $u$ and $\xi$.  The notation $\tl_n$ will subsequently denote the Temperley-Lieb algebra over this function field.} 
defined (diagrammatically) by
\begin{equation} 
\psset{unit=0.8}
\Dbk (u,\xi) = \frac 1 \beta \ \ 
\begin{pspicture}[shift=-1.5](-0.5,-.6)(6,2)
\facegrid{(0,0)}{(5,2)}
\psarc[linewidth=0.025]{-}(0,0){0.16}{0}{90}
\psarc[linewidth=0.025]{-}(1,1){0.16}{90}{180}
\psarc[linewidth=0.025]{-}(1,0){0.16}{0}{90}
\psarc[linewidth=0.025]{-}(2,1){0.16}{90}{180}
\psarc[linewidth=0.025]{-}(4,0){0.16}{0}{90}
\psarc[linewidth=0.025]{-}(5,1){0.16}{90}{180}
\rput(2.5,0.5){$\ldots$}
\rput(2.5,1.5){$\ldots$}
\rput(3.5,0.5){$\ldots$}
\rput(3.5,1.5){$\ldots$}
\psarc[linewidth=1.5pt,linecolor=blue]{-}(0,1){0.5}{90}{-90}
\psarc[linewidth=1.5pt,linecolor=blue]{-}(5,1){0.5}{-90}{90}
\pspolygon[fillstyle=solid,fillcolor=lightlightblue,linewidth=1pt](5,1)(6,2)(6,0)
\rput(5.6,0.9){\small$u,\xi$}
\rput(5.6,1.225){\small$_{(k)}$}
\rput(0.5,.5){$u$}
\rput(0.5,1.5){$u$}
\rput(1.5,.5){$u$}
\rput(1.5,1.5){$u$}
\rput(4.5,.5){$u$}
\rput(4.5,1.5){$u$}
\rput(2.5,-0.5){$\underbrace{\qquad \hspace{2.4cm} \qquad}_n$}
\end{pspicture} \ \ .
\label{eq:Duk}
\end{equation}
Each square tile above is called a 
\emph{face operator} and is a linear combination of two diagrams,
\begin{equation} 
\psset{unit=.8cm}
\begin{pspicture}[shift=-.40](1,1)
\facegrid{(0,0)}{(1,1)}
\psarc[linewidth=0.025]{-}(0,0){0.16}{0}{90}
\rput(.5,.5){$u$}
\end{pspicture}
\ = 
s_1(-u)\;\begin{pspicture}[shift=-.40](1,1)
\facegrid{(0,0)}{(1,1)}
\rput[bl](0,0){\loopa}
\end{pspicture}
\;+s_0(u)\;
\begin{pspicture}[shift=-.40](1,1)
\facegrid{(0,0)}{(1,1)}
\rput[bl](0,0){\loopb}
\end{pspicture}\ , \qquad\quad s_k(u)=\frac{\sin (u+k\lambda)}{\sin\lambda}.
\label{1x1}
\end{equation}
The parameter $u$ is called the 
\emph{spectral parameter}, while $\lambda$ is the 
\emph{crossing parameter} that parametrises $\beta$ through the relation 
\begin{equation}
 \beta = 2 \cos \lambda = U_1 = s_2(0),
\end{equation}
see \eqref{eq:rootsof1}. We remark that the Chebyshev polynomials at $\frac{\beta}{2}$ can be written as $U_{k} = s_{k+1}(0)$. 

The triangle on the right in \eqref{eq:Duk} is called the \emph{Kac boundary triangle} of seam width $k$. It depends on a formal boundary parameter $\xi$ which will be specialised later to a complex number. The trivial case $k = 0$ does not depend on $u$ nor $\xi$ and is referred to as the 
\emph{vacuum boundary condition}:
\begin{equation}
\begin{pspicture}[shift=-1.2](-0.1,-.3)(1.1,2.3)
\rput(-5,0){
\psarc[linewidth=1.5pt,linecolor=blue]{-}(5,1){0.5}{-90}{90}
\pspolygon[fillstyle=solid,fillcolor=lightlightblue,linewidth=1pt](5,1)(6,2)(6,0)
\rput(5.6,0.9){$u,\xi$}
\rput(5.6,1.2){$_{(0)}$}
}
\end{pspicture}
\ = \, 
\begin{pspicture}[shift=-0.65](-0.1,0.25)(1.1,1.75)
\rput(-5,0){
\psarc[linewidth=1.5pt,linecolor=blue]{-}(5,1){0.5}{-90}{90}
}
\end{pspicture}
\mspace{-20mu}.
\end{equation}
We note that the left boundary of \eqref{eq:Duk} is of this form, reflected about a vertical axis. 
For $k>0$, the boundary triangle is defined in terms of face operators and \WJ{} projectors. 
It takes the form
\begin{equation}  
\begin{pspicture}[shift=-1.2](-0.1,-.3)(1.1,2.3)
\rput(-5,0){
\psarc[linewidth=1.5pt,linecolor=blue]{-}(5,1){0.5}{-90}{90}
\pspolygon[fillstyle=solid,fillcolor=lightlightblue,linewidth=1pt](5,1)(6,2)(6,0)
\rput(5.6,0.9){$u,\xi$}
\rput(5.6,1.2){$_{(k)}$}
}
\end{pspicture}
= \frac{1}{\eta^{(k)}(u,\xi)}\,
\begin{pspicture}[shift=-1.2](-0.4,-.3)(4.5,2.3)
\rput(-5,0){\facegrid{(5,0)}{(9,2)}
\pspolygon[fillstyle=solid,fillcolor=pink](5.1,0)(8.9,0)(8.9,-0.3)(5.1,-0.3)(5.1,0)
\rput(7,-0.15){$_{k}$}
\pspolygon[fillstyle=solid,fillcolor=pink](5.1,2)(8.9,2)(8.9,2.3)(5.1,2.3)(5.1,2)
\rput(7,2.15){$_{k}$}
\psarc[linewidth=0.025]{-}(5,0){0.16}{0}{90}
\psarc[linewidth=0.025]{-}(6,1){0.16}{90}{180}
\psarc[linewidth=0.025]{-}(7,0){0.16}{0}{90}
\psarc[linewidth=0.025]{-}(8,1){0.16}{90}{180}
\psarc[linewidth=0.025]{-}(8,0){0.16}{0}{90}
\psarc[linewidth=0.025]{-}(9,1){0.16}{90}{180}
\psline[linewidth=1.5pt,linecolor=blue](5,0.5)(4.8,0.5)
\psline[linewidth=1.5pt,linecolor=blue](5,1.5)(4.8,1.5)
\rput(5.5,.5){\small$u\!-\!\xi_k$}
\rput(5.5,1.5){\small$u\!+\!\xi_k$}
\rput(6.5,0.5){$\ldots$}
\rput(6.5,1.5){$\ldots$}
\rput(7.5,.5){\small$u\!-\!\xi_2$}
\rput(7.5,1.5){\small$u\!+\!\xi_2$}
\rput(8.5,.5){\small$u\!-\!\xi_1$}
\rput(8.5,1.5){\small$u\!+\!\xi_1$}
\psarc[linewidth=1.5pt,linecolor=blue]{-}(9,1){0.5}{-90}{90}}
\end{pspicture}
\ ,
\label{eq:boundaryop}
\end{equation}
where
\begin{equation} \qquad \xi_j = \xi + j \lambda, \qquad \eta^{(k)}(u,\xi) = \prod_{j=1}^k s_0(u- \xi_{j+1})s_0(u+ \xi_{j-1}).
\end{equation}

Yang-Baxter integrability is built in directly at the level of the diagrammatic algebra and stems from local relations 
satisfied by the face operators and the boundary triangles, see \cite{PRZ06}. Indeed, the transfer tangles $\Dbk(u,\xi)$ 
can be shown to be crossing-symmetric and form a commuting family:
\begin{equation} 
\Dbk(\lambda - u,\xi) = \Dbk(u,\xi), \qquad \comm{\Dbk(u,\xi)}{\Dbk(v,\xi)} = 0,
\label{cross}
\end{equation}
where $u$ and $v$ are formal parameters.
The double-row transfer tangle admits a formal power series expansion in $u$, 
\begin{equation}
\Dbk(u,\xi) = \Ik + \frac{2u}{\sin \lambda} \Big( (\beta^{-1}- n \cos \lambda) \Ik - \hamk \Big) + O(u^2),
\label{eq:Du0} 
\end{equation}
where
\begin{equation}
\Ik = \ 
\begin{pspicture}[shift=-0.4](0,-0.2)(3.25,0.8) 
\pspolygon[fillstyle=solid,fillcolor=lightlightblue](0,0)(0,0.8)(3.25,0.8)(3.25,0)
\psline[linewidth=1.5pt,linecolor=blue]{-}(0.2,0)(0.2,0.8)\rput(0.2,-0.2){\small$_1$}
\psline[linewidth=1.5pt,linecolor=blue]{-}(0.6,0)(0.6,0.8)\rput(0.6,-0.2){\small$_2$}
\rput(1.0,0.4){...}
\psline[linewidth=1.5pt,linecolor=blue]{-}(1.4,0)(1.4,0.8)\rput(1.4,-0.2){\small$_n$}
\psline[linewidth=1.5pt,linecolor=blue]{-}(1.8,0)(1.8,0.8)
\psline[linewidth=1.5pt,linecolor=blue]{-}(2.2,0)(2.2,0.8)
\rput(2.6,0.675){...}
\rput(2.6,0.125){...}
\psline[linewidth=1.5pt,linecolor=blue]{-}(3.0,0)(3.0,0.8)\rput(3.0,-0.2){\small$_{n+k}$}
\pspolygon[fillstyle=solid,fillcolor=pink](1.65,0.25)(1.65,0.55)(3.15,0.55)(3.15,0.25)\rput(2.4,0.4){$_k$}
\end{pspicture}
\label{eq:Ik}
\end{equation}
and
\begin{equation}
\qquad \hamk = -\Ik\sum_{j=1}^{n-1}e_j + \frac{s_k(0)}{s_0(\xi)s_{k+1}(\xi)} \Ik e_n \Ik.
\label{eq:Ham}
\end{equation}
The tangle $\Ik$ plays a prominent role in the description of the boundary seams in terms of the boundary Temperley-Lieb algebras in \cref{sec:boundaryTLs}. The tangle $\hamk$ is the Hamiltonian with a seam of width $k$ and is central to our investigation of Kac modules in \cref{sec:Kacmod}.  We note that the normalisation factors in \eqref{eq:Duk} and \eqref{eq:boundaryop} ensure that the zeroth order coefficient of $\Dbk(u,\xi)$ in \eqref{eq:Du0} is $\Ik$.

We emphasise again that, for now, we treat $\beta$ (and hence $\lambda$), $u$ and $\xi$ as formal parameters, but will specialise them to complex numbers in \cref{sec:Kacmod}.  If $\beta$ is specialised to $0$, in particular, the normalising factor $\frac{1}{\beta}$ of $\Dbk(u,\xi)$ must be removed. To ensure an expansion in $u$ of $\Dbk(u,\xi)$ where the zeroth term is a non-zero
multiple of $\Ik$, while preserving the crossing symmetry \eqref{cross}, it is replaced in \cite{PR07,PRT14} by $\frac{1}{\sin(2u)}$.

%
\subsection{Boundary Temperley-Lieb algebras}\label{sec:boundaryTLs}
%

The boundary seams and triangles defined in the previous section are constructed in terms of the original Temperley-Lieb algebra. From their role as boundary conditions for the transfer tangle, it is not surprising that they can be described using another variety of diagrammatic algebra, the 
\emph{one-boundary} Temperley-Lieb algebra. This section first reviews the definition of the one-boundary Temperley-Lieb algebra $\tlone_n$. We then show that the boundary seams are naturally described in terms of a quotient of $\tlone_n$, the \emph{boundary seam algebra} $\btl_{n,k}$. 

\subsubsection{One-boundary \TL{} algebras} \label{sec:TLone}

The one-boundary \TL{} algebra $\tlone_n$ is a diagrammatic algebra whose elements are linear combinations of 
(generalised) connectivity diagrams. Martin and Saleur \cite{MS93} introduced a two-parameter generalisation of the 
Temperley-Lieb algebras with an extra generator at the boundary, the blob algebra. 
Its representation theory was partially unravelled by Martin and Woodcock \cite{MW00} 
and Graham and Lehrer \cite{GraDia03}. Here, we follow the 
conventions of \cite{PRT14} and define $\tlone_n$ in terms of
three free parameters: $\beta$, the fugacity of loops in the bulk, and $\beta_1$ and $\beta_2$, the fugacities of loops 
rooted in the boundary. The distinction between $\beta_1$ and $\beta_2$ will be discussed below. 
The one-boundary \TL{} algebra defined over a complex function field will be denoted by $\tlone_n$ while the 
specialisation to $\beta,\beta_1,\beta_2 \in \mathbb C$ will be denoted by $\tlone_n(\beta,\beta_1,\beta_2)$.

As for $\tl_n$, the one-boundary \TL{} algebra $\tlone_n$ has two equivalent descriptions, a diagrammatic and an algebraic one.  We prove this equivalence in \cref{app:TLone} as it will be needed to settle the corresponding equivalence for the boundary seam algebras. 
Algebraically, $\tlone_n$ is generated by $n+1$ elements: an identity $I$ and elements $e_j$ with $j = 1, 2, \dots, n$. The generators $e_j$, with $j \neq n$, satisfy the relations \eqref{eq:defTL}, so $\tl_n$ is a subalgebra of $\tlone_n$. The extra generator $e_n$ satisfies the following additional relations:
\begin{equation}
e_i e_n = e_n e_i \quad (i < n-1), \qquad e_{n-1}e_ne_{n-1} = \beta_1 e_{n-1}, \qquad e_n^2 = \beta_2 \,e_n.
\label{eq:defTL1}
\end{equation}

For the diagrammatic definition, let us again draw a rectangular box with $n$ nodes on each of its top and bottom edges. 
A (generalised or boundary) connectivity is a collection of non-intersecting loop segments living inside the box, but these can now either connect nodes pairwise or link them to the right boundary, with each node occupied by exactly one loop segment. For instance,
\begin{equation}
\psset{unit=0.8}
b_1 = \
\begin{pspicture}[shift=-0.4](-0.0,0)(2.4,1)
\pspolygon[fillstyle=solid,fillcolor=lightlightblue](0,0)(2.4,0)(2.4,1)(0,1)
\psarc[linecolor=blue,linewidth=1.5pt]{-}(0.4,0){0.2}{0}{180}
\psarc[linecolor=blue,linewidth=1.5pt]{-}(1.6,0){0.2}{0}{180}
\psarc[linecolor=blue,linewidth=1.5pt]{-}(1.6,1){0.2}{180}{360}
\psbezier[linecolor=blue,linewidth=1.5pt]{-}(0.2,1)(0.2,0.5)(1,0.5)(1,0)
\psbezier[linecolor=blue,linewidth=1.5pt]{-}(0.6,1)(0.6,0.5)(2.2,0.5)(2.2,0)
\psbezier[linecolor=blue,linewidth=1.5pt]{-}(1,1)(1,0.5)(2.2,0.5)(2.2,1)
\end{pspicture}\ , \qquad 
b_2 = \
\begin{pspicture}[shift=-0.4](-0.0,0)(2.4,1)
\pspolygon[fillstyle=solid,fillcolor=lightlightblue](0,0)(2.4,0)(2.4,1)(0,1)
\psline[linecolor=blue,linewidth=1.5pt]{-}(0.2,0)(0.2,1)
\psarc[linecolor=blue,linewidth=1.5pt]{-}(0.8,1){0.2}{180}{0}
\psarc[linecolor=blue,linewidth=1.5pt]{-}(1.6,0){0.2}{0}{180}
\psarc[linecolor=blue,linewidth=1.5pt]{-}(2.4,0){0.2}{90}{180}
\psbezier[linecolor=blue,linewidth=1.5pt]{-}(2.4,0.7)(2.2,0.7)(1.0,0.4)(1.0,0)
\psarc[linecolor=blue,linewidth=1.5pt]{-}(2.0,1){0.2}{180}{0}
\psbezier[linecolor=blue,linewidth=1.5pt]{-}(0.6,0)(0.6,0.5)(1.4,0.5)(1.4,1)
\end{pspicture}
\end{equation}
are two connectivities in $\tlone_6$ that, respectively, have $0$ and $2$ nodes going to the boundary. The number of loop segments going to the boundary is always even and ranges between $0$ and $2n$.

The product $b_1 b_2$ is first defined for connectivities $b_1$ and $b_2$ and then linearly extended to all tangles in $\tlone_n$. It is found by first performing the vertical concatenation of the diagrams of $b_1$ and $b_2$: $b_2$ is drawn on top of $b_1$, with the nodes of the top edge of $b_1$ joined to those of the bottom edge of $b_2$. The intermediate edge is then removed, leaving only the larger rectangle, and the resulting connectivity consists of the arcs connecting the nodes of the bigger rectangle, pairwise or to the boundary. Each closed loop formed in the bulk is erased and replaced by a multiplicative factor of $\beta$, as with $\tl_n$. However, boundary loops can also be formed. These consist of loop segments that start and end at the boundary. They are also removed, and replaced by a multiplicative prefactor of $\beta_1$ or $\beta_2$, with the choice made as follows.

In the diagram obtained from the vertical concatenation of $b_1$ and $b_2$, the arcs ending at the right boundary are even in number. Alternatingly, we assign them an odd $(1)$ or even $(0)$ parity with odd at the bottom. 
The rule for assigning boundary fugacities is then
\begin{equation}
\begin{pspicture}[shift=-0.55](1.3,-0.65)(2.4,0.45)
\pspolygon[fillstyle=solid,fillcolor=lightlightblue,linewidth=0pt,linecolor=white](1.5,-0.35)(2,-0.35)(2,0.35)(1.5,0.35)
\psline{-}(2,-0.35)(2,0.35)
\rput(2.25,-0.2){\tiny $_{(1)}$}
\rput(2.25,0.2){\tiny $_{(0)}$}
\psarc[linecolor=blue,linewidth=1.5pt]{-}(2,0){0.2}{90}{270}
\end{pspicture} \rightarrow \beta_1, \qquad
\begin{pspicture}[shift=-0.55](1.3,-0.65)(2.4,0.45)
\pspolygon[fillstyle=solid,fillcolor=lightlightblue,linewidth=0pt,linecolor=white](1.5,-0.35)(2,-0.35)(2,0.35)(1.5,0.35)
\psline{-}(2,-0.35)(2,0.35)
\rput(2.25,-0.2){\tiny $_{(0)}$}
\rput(2.25,0.2){\tiny $_{(1)}$}
\psarc[linecolor=blue,linewidth=1.5pt]{-}(2,0){0.2}{90}{270}
\end{pspicture} \rightarrow \beta_2.
\label{eq:beta12rule}
\end{equation}
Here is an example of a product between two connectivities:
\begin{equation}
\psset{unit=0.8}
b_1 b_2 = \
\begin{pspicture}[shift=-0.9](-0.0,0)(2.4,2)
\pspolygon[fillstyle=solid,fillcolor=lightlightblue](0,0)(2.4,0)(2.4,1)(0,1)
\psarc[linecolor=blue,linewidth=1.5pt]{-}(0.4,0){0.2}{0}{180}
\psarc[linecolor=blue,linewidth=1.5pt]{-}(1.6,0){0.2}{0}{180}
\psarc[linecolor=blue,linewidth=1.5pt]{-}(1.6,1){0.2}{180}{360}
\psbezier[linecolor=blue,linewidth=1.5pt]{-}(0.2,1)(0.2,0.5)(1,0.5)(1,0)
\psbezier[linecolor=blue,linewidth=1.5pt]{-}(0.6,1)(0.6,0.5)(2.2,0.5)(2.2,0)
\psbezier[linecolor=blue,linewidth=1.5pt]{-}(1,1)(1,0.5)(2.2,0.5)(2.2,1)
\rput(0,1)
{
\pspolygon[fillstyle=solid,fillcolor=lightlightblue](0,0)(2.4,0)(2.4,1)(0,1)
\psline[linecolor=blue,linewidth=1.5pt]{-}(0.2,0)(0.2,1)
\psarc[linecolor=blue,linewidth=1.5pt]{-}(0.8,1){0.2}{180}{0}
\psarc[linecolor=blue,linewidth=1.5pt]{-}(1.6,0){0.2}{0}{180}
\psarc[linecolor=blue,linewidth=1.5pt]{-}(2.4,0){0.2}{90}{180}
\psbezier[linecolor=blue,linewidth=1.5pt]{-}(2.4,0.7)(2.2,0.7)(1.0,0.4)(1.0,0)
\psarc[linecolor=blue,linewidth=1.5pt]{-}(2.0,1){0.2}{180}{0}
\psbezier[linecolor=blue,linewidth=1.5pt]{-}(0.6,0)(0.6,0.5)(1.4,0.5)(1.4,1)
}
\end{pspicture} \ =  \beta  \beta_1 \
\begin{pspicture}[shift=-0.4](-0.0,0)(2.4,1)
\pspolygon[fillstyle=solid,fillcolor=lightlightblue](0,0)(2.4,0)(2.4,1)(0,1)
\psarc[linecolor=blue,linewidth=1.5pt]{-}(0.4,0){0.2}{0}{180}
\psarc[linecolor=blue,linewidth=1.5pt]{-}(1.6,0){0.2}{0}{180}
\psarc[linecolor=blue,linewidth=1.5pt]{-}(2.0,1){0.2}{180}{360}
\psarc[linecolor=blue,linewidth=1.5pt]{-}(0.8,1){0.2}{180}{360}
\psbezier[linecolor=blue,linewidth=1.5pt]{-}(0.2,1)(0.2,0.5)(1,0.5)(1,0)
\psbezier[linecolor=blue,linewidth=1.5pt]{-}(1.4,1)(1.4,0.5)(2.2,0.5)(2.2,0)
\end{pspicture}
\ = \beta \beta_1\, b_3,
\end{equation}
where $b_3$ is the resulting connectivity.

The identification of generators with diagrams for $I$ and the $e_j$, $j = 1, \dots, n-1$, 
is still given by \eqref{eq:iso}, whereas for $e_n$, the identification is
\begin{equation}
e_n =
\begin{pspicture}[shift=-0.35](-0.2,-0.45)(2.6,0.35)
\pspolygon[fillstyle=solid,fillcolor=lightlightblue](0,-0.35)(2.4,-0.35)(2.4,0.35)(0,0.35)
\rput(1.4,0.0){\small$...$}
\psline[linecolor=blue,linewidth=1.5pt]{-}(0.2,0.35)(0.2,-0.35)\rput(0.2,-0.55){$_1$}
\psline[linecolor=blue,linewidth=1.5pt]{-}(0.6,0.35)(0.6,-0.35)\rput(0.6,-0.55){$_2$}
\psline[linecolor=blue,linewidth=1.5pt]{-}(1.0,0.35)(1.0,-0.35)\rput(1.0,-0.55){$_3$}
\psline[linecolor=blue,linewidth=1.5pt]{-}(1.8,0.35)(1.8,-0.35)
\psarc[linecolor=blue,linewidth=1.5pt]{-}(2.4,-0.35){0.2}{90}{180}
\psarc[linecolor=blue,linewidth=1.5pt]{-}(2.4,0.35){0.2}{180}{270}\rput(2.2,-0.55){$_n$}
\end{pspicture} .
\label{eq:ison}
\end{equation}
It is easy to see that this identification is consistent with the relations \eqref{eq:defTL} and \eqref{eq:defTL1}. As shown in \cref{prop:Surjective}, any connectivity can be obtained from a finite product of the generators, for instance, $b_2 = e_4 e_6 e_3 e_2 e_5$.  The algebra $\tlone_n$ is finite-dimensional and its dimension is given by \eqref{eq:diatlcounting}:
\begin{equation} \dim \tlone_n = 
\binom{2n}{n}.
\label{eq:dimTL1}
\end{equation}

As was already noted in \cite{PRT14}, the transformation 
\begin{equation}
e_n' = \frac1\alpha e_n , \qquad e_j' = e_j \quad \text{(\(j = 1, \dots, n-1\))}
\end{equation}
shows that the algebra $\tlone_n$ with parameters $\beta$, $\beta_1$ and $\beta_2$ is isomorphic to that with parameters $\beta$, $\frac{\beta_1}\alpha$ and $\frac{\beta_2}\alpha$.
By choosing $\alpha = \beta_1$, $\tlone_n$ reduces to the two-parameter boundary 
Temperley-Lieb algebra used, for example, in \cite{NRdG05}, while choosing $\alpha = \beta_2$ yields the blob algebra used in \cite{MS93}. (In these references, the extra generator is $e_0$ instead of $e_n$ and lives on the left boundary instead of the right one.) Summarising, the isomorphism class of the algebra $\tlone_n$ only depends upon $\beta_1$ and $\beta_2$ through their ratio. These conclusions carry over to the specialised algebras. 
However, by keeping both as free parameters, we retain the physically important freedom of setting $\beta_1$ or $\beta_2$ to zero.

\subsubsection{Boundary seam algebras} \label{sec:seamsandbtl}

The definition \eqref{eq:boundaryop} of Kac boundary triangles in terms of face operators is useful to analyse certain 
properties of the transfer tangles $\Dbk(u,\xi)$, their commutativity and crossing relations for instance. 
Following \cite{PRZ06}, it is instead convenient, when analysing the underlying algebraic structure, 
to expand each boundary triangle as a linear combination of two diagrams:
\begin{equation}
\psset{unit=0.8}
\begin{pspicture}[shift=-1.2](-0.1,-.3)(1.1,2.3)
\rput(-5,0){
\psarc[linewidth=1.5pt,linecolor=blue]{-}(5,1){0.5}{-90}{90}
\pspolygon[fillstyle=solid,fillcolor=lightlightblue,linewidth=1pt](5,1)(6,2)(6,0)
\rput(5.6,0.9){\small $u,\xi$}
\rput(5.6,1.2){\small $_{(k)}$}
}
\end{pspicture} \; =
\psset{unit=0.8}
\begin{pspicture}[shift=-0.95](-1.0,-.3)(4.2,1.8)
\rput(-5,0){
\psarc[linewidth=1.5pt,linecolor=blue]{-}(4.25,0.75){0.5}{-90}{90}
\psline[linewidth=1.5pt,linecolor=blue]{-}(8.5,0)(8.5,1.5)
\psline[linewidth=1.5pt,linecolor=blue]{-}(7.5,0)(7.5,1.5)
\psline[linewidth=1.5pt,linecolor=blue]{-}(6.5,0)(6.5,1.5)
\psline[linewidth=1.5pt,linecolor=blue]{-}(5.5,0)(5.5,1.5)
\pspolygon[fillstyle=solid,fillcolor=pink](5.1,0.5)(8.9,0.5)(8.9,1.0)(5.1,1.0)(5.1,0.5)
\rput(7,0.75){$_{k}$}
}
\end{pspicture} - \frac{s_{k}(0)s_{0}(2u)}{s_0(\xi + u)s_{k+1}(\xi - u)} \
\begin{pspicture}[shift=-1.20](-0.2,-.6)(4.2,1.8)
\rput(-5,0){
\psarc[linewidth=1.5pt,linecolor=blue]{-}(5,1.5){0.5}{-90}{0}
\psarc[linewidth=1.5pt,linecolor=blue]{-}(5,0){0.5}{0}{90}
\psline[linewidth=1.5pt,linecolor=blue]{-}(8.5,-0.8)(8.5,2.3)
\psline[linewidth=1.5pt,linecolor=blue]{-}(7.5,-0.8)(7.5,2.3)
\psline[linewidth=1.5pt,linecolor=blue]{-}(6.5,-0.8)(6.5,2.3)
\psline[linewidth=1.5pt,linecolor=blue]{-}(5.5,-0.8)(5.5,-0.3)
\psline[linewidth=1.5pt,linecolor=blue]{-}(5.5,1.8)(5.5,2.3)
\pspolygon[fillstyle=solid,fillcolor=pink](5.1,0)(8.9,0)(8.9,-0.5)(5.1,-0.5)(5.1,0)
\rput(7,-0.25){$_{k}$}
\pspolygon[fillstyle=solid,fillcolor=pink](5.1,1.5)(8.9,1.5)(8.9,2.0)(5.1,2.0)(5.1,1.5)
\rput(7,1.75){$_{k}$}
}
\end{pspicture}
.\label{eq:seamdec}
\end{equation}
These diagrams are not tangles in general, but they give rise to elements of $\tl_{n+k}$ when glued to 
$\psset{unit=0.5}
\begin{pspicture}[shift=-0.15](-0.0,-0.35)(2.6,0.35)
\rput(1.4,0.0){\scriptsize$...$}
\psline[linecolor=blue,linewidth=1.5pt]{-}(0.2,0.35)(0.2,-0.35)
\psline[linecolor=blue,linewidth=1.5pt]{-}(0.6,0.35)(0.6,-0.35)
\psline[linecolor=blue,linewidth=1.5pt]{-}(1.0,0.35)(1.0,-0.35)
\psline[linecolor=blue,linewidth=1.5pt]{-}(1.8,0.35)(1.8,-0.35)
\psarc[linecolor=blue,linewidth=1.5pt]{-}(2.4,-0.35){0.2}{90}{180}
\psarc[linecolor=blue,linewidth=1.5pt]{-}(2.4,0.35){0.2}{180}{270}
\end{pspicture} 
$ 
from the right. We denote the result of gluing this (partial) diagram to the boundary triangle of seam width $k$ by
\begin{equation} \label{eq:GluedKacBoundarySeam}
 K^{(k)}(u,\xi)=
\psset{unit=0.8}
\begin{pspicture}[shift=-1.4](0.9,-.5)(7.1,2.3)
\pspolygon[fillstyle=solid,fillcolor=lightlightblue](1,-0.05)(6,-0.05)(6,2.05)(1,2.05)
\psline[linecolor=blue,linewidth=1.5pt]{-}(1.5,-0.05)(1.5,2.05)
\psline[linecolor=blue,linewidth=1.5pt]{-}(2.3,-0.05)(2.3,2.05)
\psline[linecolor=blue,linewidth=1.5pt]{-}(3.1,-0.05)(3.1,2.05)
\rput(3.9,0.4){...}\rput(3.9,1.6){...}
\psline[linecolor=blue,linewidth=1.5pt]{-}(4.7,-0.05)(4.7,2.05)
\psarc[linecolor=blue,linewidth=1.5pt]{-}(6,2.05){0.5}{180}{270}
\psarc[linecolor=blue,linewidth=1.5pt]{-}(6,1.05){0.5}{0}{90}
\psarc[linecolor=blue,linewidth=1.5pt]{-}(6,0.95){0.5}{-90}{0}
\psarc[linecolor=blue,linewidth=1.5pt]{-}(6,-0.05){0.5}{90}{180}
\rput(1,0){\pspolygon[fillstyle=solid,fillcolor=lightlightblue,linewidth=1pt](5,1)(6,2)(6,0)
\rput(5.6,0.9){\small $u,\xi$}
\rput(5.6,1.2){\small $_{(k)}$}}
\rput(3.5,-0.5){$\underbrace{\qquad \hspace{2.2cm} \qquad}_n$}
\end{pspicture}
\; \in \tl_{n+k}. 
\end{equation}
From \eqref{eq:seamdec}, it follows that $K^{(k)}(u,\xi)$ is a linear combination of $\Ik$, as defined in \eqref{eq:Ik}, and $\Ik e_n \Ik$.

This naturally leads us to define a subalgebra\footnote{$\btl_{n,k}$ is a subalgebra of $\tl_{n+k}$ in the sense that it is a subspace of $\tl_{n+k}$ that is closed under addition, multiplication and scalar multiplication.  It does not, however, contain the unit of $\tl_{n+k}$ unless $k=0$ or $1$.}
\begin{equation}
\btl_{n,k} = \big\langle \Ik, \Ekj{k}{j}; j = 1, \dots, n \big\rangle
\label{eq:defbtl}
\end{equation} 
of $\tl_{n+k}$, which we shall refer to as the \emph{boundary seam algebra}.  Aside from $n$ and $k$, it depends upon a single parameter $\beta$.  The boundary seam algebra $\btl_{n,k}$ is unital, with unit $\Ik$, and is generated by tangles $\Ekj{k}{j}$\!, defined by
\begin{subequations}\label{eq:Ek}
\begin{align}
\Ekj{k}{j} &= \Ik e_j = \
\begin{pspicture}[shift=-0.6](0,-0.3)(4.85,0.8)
\pspolygon[fillstyle=solid,fillcolor=lightlightblue](0,0)(0,0.8)(4.85,0.8)(4.85,0)
\psline[linewidth=1.5pt,linecolor=blue]{-}(0.2,0)(0.2,0.8)\rput(0.2,-0.2){\small$_1$}
\rput(0.6,0.4){...}
\psline[linewidth=1.5pt,linecolor=blue]{-}(1.0,0)(1.0,0.8)
\rput(1.4,-0.2){\small$_j$}
\rput(1.84,-0.2){\small$_{j+1}$}
\psarc[linewidth=1.5pt,linecolor=blue]{-}(1.6,0){0.2}{0}{180}
\psarc[linewidth=1.5pt,linecolor=blue]{-}(1.6,0.8){0.2}{180}{0}
\psline[linewidth=1.5pt,linecolor=blue]{-}(2.2,0)(2.2,0.8)
\rput(2.6,0.4){...}
\psline[linewidth=1.5pt,linecolor=blue]{-}(3.0,0)(3.0,0.8)\rput(3.0,-0.2){\small$_{\phantom{+}n\phantom{+}}$}
\psline[linewidth=1.5pt,linecolor=blue]{-}(3.4,0)(3.4,0.8)\rput(3.44,-0.2){\small$_{n+1}$}
\rput(3.8,0.675){...}
\rput(3.8,0.125){...}
\psline[linewidth=1.5pt,linecolor=blue]{-}(4.2,0)(4.2,0.8)
\psline[linewidth=1.5pt,linecolor=blue]{-}(4.6,0)(4.6,0.8)\rput(4.6,-0.2){\small $_{n+k}$}
\pspolygon[fillstyle=solid,fillcolor=pink](3.25,0.25)(3.25,0.55)(4.75,0.55)(4.75,0.25)\rput(4.0,0.4){$_k$}
\end{pspicture}
\qquad (j = 1, \dots, n-1),
\label{eq:Ejk}
\\[0.2cm]
\Ekj{k}{n} &= U_{k-1}\, \Ik e_n \Ik =  U_{k-1} \ \,
\begin{pspicture}[shift=-0.95](0,-0.3)(3.65,1.5)
\pspolygon[fillstyle=solid,fillcolor=lightlightblue](0,0)(0,1.5)(3.65,1.5)(3.65,0)
\psline[linewidth=1.5pt,linecolor=blue]{-}(0.2,0)(0.2,1.5)\rput(0.2,-0.2){\small$_1$}
\psline[linewidth=1.5pt,linecolor=blue]{-}(0.6,0)(0.6,1.5)\rput(0.6,-0.2){\small$_2$}
\rput(1.0,0.75){...}
\psline[linewidth=1.5pt,linecolor=blue]{-}(1.4,0)(1.4,1.5)
\rput(1.8,-0.2){\small$_{\phantom{+}n\phantom{+}}$}\rput(2.24,-0.2){\small$_{n+1}$}
\psarc[linewidth=1.5pt,linecolor=blue]{-}(2.0,0.45){0.2}{0}{180}
\psarc[linewidth=1.5pt,linecolor=blue]{-}(2.0,1.05){0.2}{180}{0}
\psline[linewidth=1.5pt,linecolor=blue]{-}(2.2,0)(2.2,0.2)\psline[linewidth=1.5pt,linecolor=blue]{-}(2.2,1.3)(2.2,1.5)
\psline[linewidth=1.5pt,linecolor=blue]{-}(1.8,0)(1.8,0.45)\psline[linewidth=1.5pt,linecolor=blue]{-}(1.8,1.05)(1.8,1.5)
\psline[linewidth=1.5pt,linecolor=blue]{-}(2.6,0)(2.6,1.5)
\rput(3.0,0.75){...}
\rput(3.0,1.425){...}
\rput(3.0,0.075){...}
\rput(3.0,0.75){...}
\psline[linewidth=1.5pt,linecolor=blue]{-}(3.4,0)(3.4,1.5)
\pspolygon[fillstyle=solid,fillcolor=pink](2.05,0.15)(2.05,0.45)(3.55,0.45)(3.55,0.15)\rput(2.8,0.3){$_k$}
\pspolygon[fillstyle=solid,fillcolor=pink](2.05,1.35)(2.05,1.05)(3.55,1.05)(3.55,1.35)\rput(2.8,1.2){$_k$}
\rput(3.4,-0.2){\small $_{n+k}$}
\end{pspicture}
\ . \label{eq:Enk}
\end{align}
\end{subequations}
For completeness, we remark that when $k=0$, the diagram in the definition of the generator $\Ekj{k}{n}$ does not make sense and $\Ekj{k}{n}$ should be formally regarded as being zero, consistent with $U_{-1} = 0$. The boundary seam algebra $\btl_{n,0}$ will therefore be identified with the \TL{} algebra $\tl_n$.

The Kac boundary triangle \eqref{eq:GluedKacBoundarySeam} can then be written algebraically as 
\begin{equation}  \label{eq:KInBTL}
\Kk\!(u,\xi) = \Ik - \frac{s_{0}(2u)}{s_0(\xi + u)s_{k+1}(\xi - u)}\, \Ekj{k}{n}  
\end{equation}
and is an element of $\btl_{n,k} \subseteq \tl_{n+k}$. The following result generalises this to
the transfer tangle $\Dbk(u,\xi)$. 
\begin{Proposition} \label{sec:dbkinbtl}
The transfer tangle and the Hamiltonian are elements of the boundary seam algebra: 
\begin{equation}
\Dbk(u,\xi), \hamk \in \btl_{n,k}.
\end{equation}
\end{Proposition}
\begin{proof}
Since $\Kk(u,\xi) \in \btl_{n,k}$ by \eqref{eq:KInBTL}, we have $\Kk(u,\xi)=\Ik\Kk(u,\xi)\Ik$, hence 
\begin{equation}
e_i\Kk(u,\xi)e_j=\Ekj{k}{i}\Kk(u,\xi)\Ekj{k}{j} \qquad \text{(\(i,j<n\)).}
\end{equation}
This replacement of $e_i$, with $i<n$, by $\Ekj{k}{i}$ obviously extends to linear combinations of products of these generators.

To apply this to the double-row transfer tangles, we use the diagrammatic identity 
\begin{equation}
\psset{unit=.8cm}
\begin{pspicture}[shift=-.9](-0.5,0)(1,2)
\facegrid{(0,0)}{(1,2)}
\psarc[linewidth=1.5pt,linecolor=blue](0,1){0.5}{90}{-90}
\psarc[linewidth=0.025]{-}(0,0){0.16}{0}{90}
\psarc[linewidth=0.025]{-}(1,1){0.16}{90}{180}
\rput(.5,.5){$u$}
\rput(.5,1.5){$u$}
\end{pspicture}
\ =  \big(s_0(u)\big)^2 \
\begin{pspicture}[shift=-.9](-0.5,0)(1,2)
\facegrid{(0,0)}{(1,2)}
\psarc[linewidth=1.5pt,linecolor=blue](0,1){0.5}{90}{-90}
\rput(0,0){\loopb}
\rput(0,1){\loopa}
\end{pspicture} \ + s_{2}(-2u) \ 
\begin{pspicture}[shift=-.9](0,0)(1,2)
\facegrid{(0,0)}{(1,2)}
\psarc[linewidth=1.5pt,linecolor=blue](1,0){0.5}{90}{180}
\psarc[linewidth=1.5pt,linecolor=blue](1,2){0.5}{180}{270}
\end{pspicture} 
\end{equation}
on \eqref{eq:Duk} $n$ times and find that $\Dbk(u, \xi)$ is expressible as
\begin{equation} \label{eq:AlgDmk}
\beta\, \Dbk(u,\xi)  = \alpha\, \big(s_0(u)\big)^{2n} \Ik + s_2(-2u) \sum_{m = 0}^{n-1} \big(s_0(u)\big)^{2m}\, \db_m^{_{(k)}}, 
\end{equation}
where $\displaystyle \alpha = s_2(0) - \frac{s_{k}(0)s_{0}(2u)}{s_0(\xi + u)s_{k+1}(\xi - u)}$ and
\begin{equation}
\db_m^{_{(k)}}  = \ 
\psset{unit=0.8cm}
\begin{pspicture}[shift=-1.5](-2.85,-.6)(6.1,2.3)
\pspolygon[fillstyle=solid,fillcolor=lightlightblue](1,0)(1,2)(-2.85,2)(-2.85,0)
\facegrid{(1,0)}{(5,2)}
\psarc[linewidth=0.025]{-}(1,0){0.16}{0}{90}
\psarc[linewidth=0.025]{-}(2,1){0.16}{90}{180}
\psarc[linewidth=0.025]{-}(2,0){0.16}{0}{90}
\psarc[linewidth=0.025]{-}(3,1){0.16}{90}{180}
\psarc[linewidth=0.025]{-}(4,0){0.16}{0}{90}
\psarc[linewidth=0.025]{-}(5,1){0.16}{90}{180}
\rput(3.5,0.5){$\ldots$}
\rput(3.5,1.5){$\ldots$}
\psarc[linewidth=1.5pt,linecolor=blue]{-}(1,0){0.5}{90}{180}
\psarc[linewidth=1.5pt,linecolor=blue]{-}(1,2){0.5}{180}{-90}
\psarc[linewidth=1.5pt,linecolor=blue]{-}(5,1){0.5}{-90}{90}
\psline[linewidth=1.5pt,linecolor=blue]{-}(-0.25,0)(-0.25,2)
\rput(-1,1){$...$}
\psline[linewidth=1.5pt,linecolor=blue]{-}(-1.75,0)(-1.75,2)
\psline[linewidth=1.5pt,linecolor=blue]{-}(-2.50,0)(-2.50,2)
\pspolygon[fillstyle=solid,fillcolor=lightlightblue,linewidth=1pt](5,1)(6,2)(6,0)
\rput(5.6,0.9){\small $u,\xi$}
\rput(5.6,1.225){\small$_{(k)}$}
\rput(1.5,.5){$u$}
\rput(1.5,1.5){$u$}
\rput(2.5,.5){$u$}
\rput(2.5,1.5){$u$}
\rput(4.5,.5){$u$}
\rput(4.5,1.5){$u$}
\rput(3,-0.5){$\underbrace{\qquad \hspace{1.6cm} \qquad}_{n-m-1}$}
\rput(-1.375,-0.5){$\underbrace{\quad \hspace{1.25cm} \quad}_{m}$}
\end{pspicture}\ = \ 
\psset{unit=0.7cm}
\begin{pspicture}[shift=-5.4](-2.65,-4.5)(6.1,6)
\psline[linewidth=1.5pt,linecolor=blue]{-}(-2.5,-4)(-2.5,6)
\psline[linewidth=1.5pt,linecolor=blue]{-}(-1.5,-4)(-1.5,6)
\rput(-0.5,-3){$...$}\rput(-0.5,5){...}
\psline[linewidth=1.5pt,linecolor=blue]{-}(0.5,-4)(0.5,6)
\psline[linewidth=1.5pt,linecolor=blue]{-}(1.5,-4)(1.5,6)
\psline[linewidth=1.5pt,linecolor=blue]{-}(2.5,-4)(2.5,6)
\psline[linewidth=1.5pt,linecolor=blue]{-}(3.5,-4)(3.5,6)
\psline[linewidth=1.5pt,linecolor=blue]{-}(4.5,-4)(4.5,6)
\psline[linewidth=1.5pt,linecolor=blue]{-}(5.5,-4)(5.5,6)
\multiput(0,0)(-1,1){4}{\pspolygon[fillstyle=solid,fillcolor=lightlightblue,linewidth=1pt](5,1)(6,2)(5,3)(4,2)\psarc[linewidth=0.025]{-}(5,1){0.21}{45}{135}}\rput(5,2){$u$}\rput(4,3){.}\rput(4.15,2.85){.}\rput(3.85,3.15){.}\rput(3,4){$u$}\rput(2,5){$u$}
\multiput(0,0)(-1,-1){4}{\pspolygon[fillstyle=solid,fillcolor=lightlightblue,linewidth=1pt](5,-1)(6,0)(5,1)(4,0)\psarc[linewidth=0.025]{-}(5,-1){0.21}{45}{135}}\rput(5,0){$u$}\rput(4,-1){$.$}\rput(3.85,-1.15){$.$}\rput(4.15,-0.85){$.$}\rput(3,-2){$u$}\rput(2,-3){$u$}
\pspolygon[fillstyle=solid,fillcolor=lightlightblue,linewidth=1pt](5,1)(6,2)(6,0)
\rput(5.6,0.9){\scriptsize $u,\xi$}
\rput(5.6,1.2){\scriptsize $_{(k)}$}
\rput(4,-4.5){$\underbrace{\qquad \hspace{0.8cm} \qquad}_{n-m-1}$}
\rput(-1,-4.5){$\underbrace{\quad \hspace{1.5cm} \quad}_{m}$}
\end{pspicture}
\ .
\end{equation}
This tangle can be expressed algebraically as
\begin{equation} \label{eq:Algdmk}
\db_m^{_{(k)}}= \Big[\hspace{-0.1cm}\prod_{i=m+1}^{n-1} \hspace{-0.15cm}x_i(u) \Big] \Kk(u,\xi) \Big[\hspace{-0.3cm}\prod_{\substack{j=n-1 \\ (\textrm{step}\, = -1)}}^{m+1} \hspace{-0.3cm}x_j(u) \Big],
\end{equation}
where $x_i(u) = 
\psset{unit=.25cm}
\begin{pspicture}[shift=-.50](-1,-1)(1,1)
\pspolygon[fillstyle=solid,fillcolor=lightlightblue](0,-1)(-1,0)(0,1)(1,0)
\psarc[linewidth=0.045]{-}(0,-1){0.40}{45}{135}
\rput(0,0){\scriptsize $u$}
\end{pspicture}
= s_1(-u)\, I + s_0 (u)\, e_i$ is the algebraic form of the face operator. We mention that
our convention for a product $\prod$ of non-commuting elements, here Temperley-Lieb tangles, is that the ordering is 
from left to right, for instance $\prod_{i=1}^3 a_i = a_1 a_2 a_3$.

Unlike for $\Dbk(u)$, this factorisability of $\db_m^{_{(k)}}$ is possible because the arcs leaving its leftmost face 
operators are not linked. Each $\db_m^{_{(k)}}$ is a linear combination of tangles of the form $a_1\, K^{(k)}(u,\xi)\, a_2$ 
and thus 
\begin{equation}
\db_m^{_{(k)}} = \Big[\prod_{i=n-1}^{m+1} \Xk_i\!(u) \Big] \Kk\!(u,\xi) \Big[\prod_{j=m+1}^{n-1} \Xk_j\!(u) \Big], \qquad \Xk_j\!(u) = s_1(-u) \,\Ik + s_0(u) \Ekj{k}{j}.
\end{equation}
This realises $\Dbk(u,\xi)$ as a linear combination of generators of $\btl_{n,k}$.  As the Hamiltonian is, up to adding a multiple of $\Ik \in \btl_{n,k}$, the first-order coefficient of $\Dbk(u,\xi)$ 
in its formal power series expansion \eqref{eq:Du0} in $u$, it follows that $\hamk \in \btl_{n,k}$.
This last statement was already clear from \eqref{eq:Ham}.
\end{proof}

To investigate the nature of the boundary seam algebras, we note that the generators of $\btl_{n,k}$ satisfy the 
relations
\begin{equation} \label{eq:newBTL}
\begin{aligned} 
\Ik A = A\, \Ik &= A & &\text{($A \in \{\Ik, \Ekj{k}{j}; \ j = 1, \dots, n\}$),} \\ 
\Ekj{k}{i} \Ekj{k}{j} &= \Ekj{k}{j} \Ekj{k}{i} & &\text{($|i-j|>1$),} \\
\Ekj{k}{i} \Ekj{k}{j} \Ekj{k}{i} &= \Ekj{k}{i} & &\text{($|i-j|=1$; \ $i,j \in \{1, \dots, n-1 \}$),} \\ 
\big(\Ekj{k}{j}\big)^2 &=\beta \Ekj{k}{j} & &\text{($j \in \{1, 2, \dots, n-1\}$),} \\
\Ekj{k}{n-1}\Ekj{k}{n}\Ekj{k}{n-1} &= U_{k-1} \, \Ekj{k}{n-1}, \\
\big(\Ekj{k}{n}\big)^2 &= U_{k} \, \Ekj{k}{n}.
\end{aligned}
\end{equation}
These are proven using the defining properties of the \TL{} algebras and the \WJ{} projectors. 
The last equation, for instance, is a consequence of the two properties in \eqref{eq:WJprop2}. 
Comparing with the relations \eqref{eq:defTL} and \eqref{eq:defTL1}, this shows that the boundary seam algebra $\btl_{n,k}$ is a quotient of the one-boundary \TL{} algebra with the parameters $\beta,\beta_1,\beta_2$ related by
\begin{equation}\label{eq:betas}
 \beta_1 = U_{k-1},\qquad \beta_2 = U_k.
\end{equation}
Indeed, we obtain a surjective homomorphism of unital associative algebras,
\begin{equation}\label{eq:surjectiveh}
\mathfrak h\colon \tlone_n \to \btl_{n,k},
\end{equation}
defined on the generators by 
\begin{equation}
 \mathfrak h(I) =\Ik,\qquad \mathfrak h(e_j) = \Ekj{k}{j}\quad \text{(\(j=1,2,\ldots,n\)).} 
\end{equation}
We emphasise that this homomorphism holds over fields of complex functions of $\beta$ after imposing \eqref{eq:betas}.

The relations \eqref{eq:newBTL}, however, do not form a complete set for $\btl_{n,k}$ 
when $n>k$.  This is easy to see for $k=1$. Then, the projector 
$\begin{pspicture}[shift=-0.05](-0.03,0.00)(0.57,0.3)
\pspolygon[fillstyle=solid,fillcolor=pink](0,0)(0.5,0)(0.5,0.3)(0,0.3)(0,0)\rput(0.25,0.15){$_1$}
\end{pspicture}
$
appearing in the definition of $\Ik$ and $\Ekj{k}{j}$ is the identity connectivity 
on one strand, so $\btl_{n,1} = \tl_{n+1}$. As a consequence, the generators of $\btl_{n,1}$ satisfy the extra relation
\begin{equation} \label{eq:BTLRelationk=1}
\Ekj{1}{n}\Ekj{1}{n-1}\Ekj{1}{n} = \Ekj{1}{n} \qquad \text{(\(n>1\)),} 
\end{equation}
which does not hold in $\tlone_{n}$. For general $k$, the diagrams in \eqref{eq:Ek} satisfy, if $n>k$, a single additional algebraically-independent relation: \cref{eq:ClosureRelation} below.
This relation is conveniently expressed in terms of tangles $\Yk_t$ which have the form
\begin{equation}
\psset{unit=0.5}
\Yk_t = 
U_{k-1} \ 
\begin{pspicture}[shift=-2.90](-6.2,-1.6)(6.6,3.6)
\pspolygon[fillstyle=solid,fillcolor=lightlightblue](-6.1,-0.8)(6.6,-0.8)(6.6,3.9)(-6.1,3.9)
\psarc[linewidth=1.5pt,linecolor=blue]{-}(1,3.1){0.5}{-180}{0}
\psarc[linewidth=1.5pt,linecolor=blue]{-}(1,0){0.5}{0}{180}
\psline[linewidth=1.5pt,linecolor=blue]{-}(-5.7,-0.8)(-5.7,3.9)
\rput(-4.9,3.3){...}\rput(-4.9,-0.4){...}
\psline[linewidth=1.5pt,linecolor=blue]{-}(-4.1,-0.8)(-4.1,3.9)
\psline[linewidth=1.5pt,linecolor=blue]{-}(-3.1,-0.8)(-3.1,3.9)
\psline[linewidth=1.5pt,linecolor=blue]{-}(-2.1,-0.8)(-2.1,0)
\psline[linewidth=1.5pt,linecolor=blue]{-}(-2.1,3.1)(-2.1,3.9)
\rput(-1.3,3.3){...}\rput(-1.3,-0.4){...}
\rput(3.25,2.8){...}\rput(3.25,0.3){...}
\psline[linewidth=1.5pt,linecolor=blue]{-}(-0.5,-0.8)(-0.5,0)
\psline[linewidth=1.5pt,linecolor=blue]{-}(-0.5,3.1)(-0.5,3.9)
\psline[linewidth=1.5pt,linecolor=blue]{-}(0.5,-0.8)(0.5,0)
\psline[linewidth=1.5pt,linecolor=blue]{-}(0.5,3.1)(0.5,3.9)
\psline[linewidth=1.5pt,linecolor=blue]{-}(1.5,-0.8)(1.5,-0.3)
\psline[linewidth=1.5pt,linecolor=blue]{-}(1.5,3.4)(1.5,3.9)
\psline[linewidth=1.5pt,linecolor=blue]{-}(2.5,3.4)(2.5,3.9)
\psline[linewidth=1.5pt,linecolor=blue]{-}(2.5,-0)(2.5,-0.8)
\psline[linewidth=1.5pt,linecolor=blue]{-}(4.1,3.4)(4.1,3.9)
\psline[linewidth=1.5pt,linecolor=blue]{-}(4.1,-0)(4.1,-0.8)
\psbezier[linewidth=1.5pt,linecolor=blue]{-}(2.5,3.1)(2.5,1.9)(-0.5,1.9)(-0.5,3.1)
\psbezier[linewidth=1.5pt,linecolor=blue]{-}(2.5,0)(2.5,1.2)(-0.5,1.2)(-0.5,0)
\psbezier[linewidth=1.5pt,linecolor=blue]{-}(4.1,3.1)(4.1,1.3)(-2.1,1.3)(-2.1,3.1)
\psbezier[linewidth=1.5pt,linecolor=blue]{-}(4.1,0)(4.1,1.7)(-2.1,1.7)(-2.1,0)
\psline[linewidth=1.5pt,linecolor=blue]{-}(5.0,-0.8)(5.0,3.9)
\psline[linewidth=1.5pt,linecolor=blue]{-}(6.0,-0.8)(6.0,3.9)
\pspolygon[fillstyle=solid,fillcolor=pink](1.1,0)(6.4,0)(6.4,-0.5)(1.1,-0.5)(1.1,0)
\rput(3.75,-0.25){\small$_{k}$}
\pspolygon[fillstyle=solid,fillcolor=pink](1.1,3.1)(6.4,3.1)(6.4,3.6)(1.1,3.6)(1.1,3.1)
\rput(3.75,3.35){\small$_{k}$}
\rput(-4.4,-1.5){$\underbrace{\quad \hspace{0.7cm}  \quad}_{n-t}$}
\rput(-0.8,-1.5){$\underbrace{\quad \hspace{0.7cm}  \quad}_t$}
\end{pspicture} \qquad \text{($t = 0, \dots, \min(k,n)$).}
\label{eq:Yt}
\end{equation}
For example, $\Yk_0 = U_{k-1}\, \Ik$ and $\Yk_1 = \Ekj{k}{n}$.

The $\Yk_{t+1}$ are built recursively by applying $\Ekj{k}{n} \Ekj{k}{n-1} \cdots \Ekj{k}{n-t} = \prod_{j=0}^t \Ekj{k}{n-j}$ to $\Yk_t$.  For example, $\Yk_1$ is constructed as follows:  $\Ekj{k}{n} \Yk_0 = \Ekj{k}{n} U_{k-1}\, \Ik = U_{k-1}\, \Ekj{k}{n} = U_{k-1}\, \Yk_1$.  Similarly, $\Yk_2$ is obtained from 
\begin{align}
\Ekj{k}{n}\Ekj{k}{n-1}\Yk_1 &= \Ekj{k}{n}\Ekj{k}{n-1}\Ekj{k}{n} = U_{k-1}^2\ 
\psset{unit=0.6}
\begin{pspicture}[shift=-0.75](-0.4,-0.0)(2.85,1.7)
\pspolygon[fillstyle=solid,fillcolor=lightlightblue](-0.4,0)(-0.4,1.7)(2.85,1.7)(2.85,0)
\psline[linewidth=1.0pt,linecolor=blue]{-}(-0.2,0)(-0.2,1.7)
\psline[linewidth=1.0pt,linecolor=blue]{-}(0.6,0)(0.6,1.7)
\rput(0.2,0.5){\small ...}\rput(0.2,1.1){\small ...}
\psline[linewidth=1.0pt,linecolor=blue]{-}(1.0,0)(1.0,0.2)
\psline[linewidth=1.0pt,linecolor=blue]{-}(1.4,0)(1.4,0.2)
\psline[linewidth=1.0pt,linecolor=blue]{-}(1.8,0)(1.8,0.15)
\psline[linewidth=1.0pt,linecolor=blue]{-}(2.2,0)(2.2,1.7)
\psline[linewidth=1.0pt,linecolor=blue]{-}(2.6,0)(2.6,1.7)
\psline[linewidth=1.0pt,linecolor=blue]{-}(1.0,1.5)(1.0,1.7)
\psline[linewidth=1.0pt,linecolor=blue]{-}(1.4,1.5)(1.4,1.7)
\psline[linewidth=1.0pt,linecolor=blue]{-}(1.8,1.55)(1.8,1.7)
\psarc[linewidth=1.0pt,linecolor=blue]{-}(1.6,0.2){0.2}{0}{180}
\psarc[linewidth=1.0pt,linecolor=blue]{-}(1.6,1.5){0.2}{180}{0}
\psbezier[linewidth=1.0pt,linecolor=blue]{-}(1.0,0.2)(1.0,0.5)(1.8,0.5)(1.8,0.8)
\psbezier[linewidth=1.0pt,linecolor=blue]{-}(1.0,1.5)(1.0,1.2)(1.8,1.2)(1.8,0.9)
\pspolygon[fillstyle=solid,fillcolor=pink](1.65,0.1)(1.65,0.2)(2.75,0.2)(2.75,0.1)
\pspolygon[fillstyle=solid,fillcolor=pink](1.65,0.8)(1.65,0.9)(2.75,0.9)(2.75,0.8)
\pspolygon[fillstyle=solid,fillcolor=pink](1.65,1.5)(1.65,1.6)(2.75,1.6)(2.75,1.5)
\end{pspicture} \notag \\[0.3cm]
& = U_{k-1}^2\ 
\psset{unit=0.6} 
\begin{pspicture}[shift=-0.75](-0.4,-0.0)(2.85,1.7)
\pspolygon[fillstyle=solid,fillcolor=lightlightblue](-0.4,0)(-0.4,1.7)(2.85,1.7)(2.85,0)
\psline[linewidth=1.0pt,linecolor=blue]{-}(-0.2,0)(-0.2,1.7)
\psline[linewidth=1.0pt,linecolor=blue]{-}(0.6,0)(0.6,1.7)
\rput(0.2,0.5){\small ...}\rput(0.2,1.1){\small ...}
\psline[linewidth=1.0pt,linecolor=blue]{-}(1.0,0)(1.0,0.2)
\psline[linewidth=1.0pt,linecolor=blue]{-}(1.4,0)(1.4,0.2)
\psline[linewidth=1.0pt,linecolor=blue]{-}(1.8,0)(1.8,0.15)
\psline[linewidth=1.0pt,linecolor=blue]{-}(2.2,0)(2.2,1.7)
\psline[linewidth=1.0pt,linecolor=blue]{-}(2.6,0)(2.6,1.7)
\psline[linewidth=1.0pt,linecolor=blue]{-}(1.0,1.5)(1.0,1.7)
\psline[linewidth=1.0pt,linecolor=blue]{-}(1.4,1.5)(1.4,1.7)
\psline[linewidth=1.0pt,linecolor=blue]{-}(1.8,1.55)(1.8,1.7)
\psline[linewidth=1.0pt,linecolor=blue]{-}(1.8,0.8)(1.8,0.9)
\psarc[linewidth=1.0pt,linecolor=blue]{-}(1.6,0.2){0.2}{0}{180}
\psarc[linewidth=1.0pt,linecolor=blue]{-}(1.6,1.5){0.2}{180}{0}
\psbezier[linewidth=1.0pt,linecolor=blue]{-}(1.0,0.2)(1.0,0.5)(1.8,0.5)(1.8,0.8)
\psbezier[linewidth=1.0pt,linecolor=blue]{-}(1.0,1.5)(1.0,1.2)(1.8,1.2)(1.8,0.9)
\pspolygon[fillstyle=solid,fillcolor=pink](1.65,0.1)(1.65,0.2)(2.75,0.2)(2.75,0.1)
\pspolygon[fillstyle=solid,fillcolor=pink](2.05,0.8)(2.05,0.9)(2.75,0.9)(2.75,0.8)
\pspolygon[fillstyle=solid,fillcolor=pink](1.65,1.5)(1.65,1.6)(2.75,1.6)(2.75,1.5)
\end{pspicture} \ - U_{k-2} U_{k-1} \ 
\begin{pspicture}[shift=-1.10](-0.4,-0.0)(2.85,2.2)
\pspolygon[fillstyle=solid,fillcolor=lightlightblue](-0.4,0)(-0.4,2.4)(2.85,2.4)(2.85,0)
\psline[linewidth=1.0pt,linecolor=blue]{-}(-0.2,0)(-0.2,2.4)
\psline[linewidth=1.0pt,linecolor=blue]{-}(0.6,0)(0.6,2.4)
\rput(0.2,0.7){\small ...}\rput(0.2,1.7){\small ...}
\psline[linewidth=1.0pt,linecolor=blue]{-}(1.0,0)(1.0,0.2)
\psline[linewidth=1.0pt,linecolor=blue]{-}(1.4,0)(1.4,0.2)
\psline[linewidth=1.0pt,linecolor=blue]{-}(1.8,0)(1.8,0.15)
\psline[linewidth=1.0pt,linecolor=blue]{-}(1.8,0.8)(1.8,0.9)
\psline[linewidth=1.0pt,linecolor=blue]{-}(1.8,1.5)(1.8,1.6)
\psline[linewidth=1.0pt,linecolor=blue]{-}(2.2,0)(2.2,0.8)
\psline[linewidth=1.0pt,linecolor=blue]{-}(2.2,1.6)(2.2,2.4)
\psline[linewidth=1.0pt,linecolor=blue]{-}(2.6,0)(2.6,2.4)
\psline[linewidth=1.0pt,linecolor=blue]{-}(1.0,2.2)(1.0,2.4)
\psline[linewidth=1.0pt,linecolor=blue]{-}(1.4,2.2)(1.4,2.4)
\psline[linewidth=1.0pt,linecolor=blue]{-}(1.8,2.25)(1.8,2.4)
\psarc[linewidth=1.0pt,linecolor=blue]{-}(1.6,0.2){0.2}{0}{180}
\psarc[linewidth=1.0pt,linecolor=blue]{-}(1.6,2.2){0.2}{180}{0}
\psarc[linewidth=1.0pt,linecolor=blue]{-}(2.0,0.9){0.2}{0}{180}
\psarc[linewidth=1.0pt,linecolor=blue]{-}(2.0,1.5){0.2}{180}{0}
\psbezier[linewidth=1.0pt,linecolor=blue]{-}(1.0,0.2)(1.0,0.5)(1.8,0.5)(1.8,0.8)
\psbezier[linewidth=1.0pt,linecolor=blue]{-}(1.0,2.2)(1.0,1.9)(1.8,1.9)(1.8,1.6)
\pspolygon[fillstyle=solid,fillcolor=pink](1.65,0.1)(1.65,0.2)(2.75,0.2)(2.75,0.1)
\pspolygon[fillstyle=solid,fillcolor=pink](2.05,0.8)(2.05,0.9)(2.75,0.9)(2.75,0.8)
\pspolygon[fillstyle=solid,fillcolor=pink](2.05,1.5)(2.05,1.6)(2.75,1.6)(2.75,1.5)
\pspolygon[fillstyle=solid,fillcolor=pink](1.65,2.2)(1.65,2.3)(2.75,2.3)(2.75,2.2)
\end{pspicture} \notag \\ 
&= U_{k-1}\, \Yk_1 - U_{k-2} \, \Yk_2, \label{eq:gettingY2} 
\end{align}
where we have used properties \eqref{eq:WJs} and \eqref{eq:WJprop} of the \WJ{} projectors.
In general, this construction results in the relation
\begin{equation}
\Big[\prod_{j=0}^{t} \Ekj{k}{n-j}\Big] \Yk_t = \sum_{i=0}^{t-1} (-1)^i U_{k-1-i} \Big[ \prod_{j=i+2}^t \Ekj{k}{n-j} \Big] \Yk_t + (-1)^t U_{k-1-t} \Yk_{t+1},
\label{eq:Yrec}
\end{equation}
valid for $t = 0, \dots, \min(n,k)-1$. The derivation of this relation uses an identity that was not presented in \cref{sec:TLa}:
\begin{equation}
\begin{pspicture}[shift=-0.05](-0.03,0.00)(1.07,0.3)
\pspolygon[fillstyle=solid,fillcolor=pink](0,0)(1,0)(1,0.3)(0,0.3)(0,0)\rput(0.5,0.15){$_k$} 
\end{pspicture} =  
\frac{1}{U_{k-1}} \Bigg( U_{k-1} \ 
\begin{pspicture}[shift=-0.25](-0.03,-0.3)(1.07,0.3)
\psline[linewidth=1.25pt,linecolor=blue]{-}(0,0)(0,0.3)
\psline[linewidth=1.25pt,linecolor=blue]{-}(0.2,0)(0.2,0.3)
\psline[linewidth=1.25pt,linecolor=blue]{-}(0.4,0)(0.4,0.3)
\rput(0.6,0.15){\scriptsize ...}
\psline[linewidth=1.25pt,linecolor=blue]{-}(0.8,0)(0.8,0.3)
\psline[linewidth=1.25pt,linecolor=blue]{-}(1.0,0)(1.0,0.3)
\rput(0,-0.3){\psline[linewidth=1.25pt,linecolor=blue]{-}(0,0)(0,0.3)
\pspolygon[fillstyle=solid,fillcolor=pink](0.10,0)(1.1,0)(1.1,0.3)(0.10,0.3)\rput(0.6,0.15){$_{k-1}$}}
\end{pspicture}\  
- U_{k-2} \ 
\begin{pspicture}[shift=-0.25](-0.03,-0.3)(1.07,0.3)
\psarc[linewidth=1.25pt,linecolor=blue]{-}(0.1,0){0.1}{0}{180}
\psarc[linewidth=1.25pt,linecolor=blue]{-}(0.1,0.3){0.1}{180}{0}
\psline[linewidth=1.25pt,linecolor=blue]{-}(0.4,0)(0.4,0.3)
\rput(0.6,0.15){\scriptsize ...}
\psline[linewidth=1.25pt,linecolor=blue]{-}(0.8,0)(0.8,0.3)
\psline[linewidth=1.25pt,linecolor=blue]{-}(1.0,0)(1.0,0.3)
\rput(0,-0.3){\psline[linewidth=1.25pt,linecolor=blue]{-}(0,0)(0,0.3)
\pspolygon[fillstyle=solid,fillcolor=pink](0.10,0)(1.1,0)(1.1,0.3)(0.10,0.3)\rput(0.6,0.15){$_{k-1}$}}
\end{pspicture}\ 
+ U_{k-3}\
\begin{pspicture}[shift=-0.25](-0.03,-0.3)(1.07,0.3)
\psarc[linewidth=1.25pt,linecolor=blue]{-}(0.1,0){0.1}{0}{180}
\psarc[linewidth=1.25pt,linecolor=blue]{-}(0.3,0.3){0.1}{180}{0}
\psbezier[linewidth=1.25pt,linecolor=blue]{-}(0.4,0)(0.4,0.15)(0,0.15)(0,0.3)
\psline[linewidth=1.25pt,linecolor=blue]{-}(0.6,0)(0.6,0.3)
\rput(0.8,0.15){\scriptsize ...}
\psline[linewidth=1.25pt,linecolor=blue]{-}(1.0,0)(1.0,0.3)
\rput(0,-0.3){\psline[linewidth=1.25pt,linecolor=blue]{-}(0,0)(0,0.3)
\pspolygon[fillstyle=solid,fillcolor=pink](0.10,0)(1.1,0)(1.1,0.3)(0.10,0.3)\rput(0.6,0.15){$_{k-1}$}}
\end{pspicture}
+ \,\dots\, +  (-1)^{k-1}U_{0}\
\begin{pspicture}[shift=-0.25](-0.03,-0.3)(1.07,0.3)
\psarc[linewidth=1.25pt,linecolor=blue]{-}(0.1,0){0.1}{0}{180}
\psarc[linewidth=1.25pt,linecolor=blue]{-}(0.9,0.3){0.1}{180}{0}
\psbezier[linewidth=1.25pt,linecolor=blue]{-}(0.4,0)(0.4,0.15)(0,0.15)(0,0.3)
\psbezier[linewidth=1.25pt,linecolor=blue]{-}(0.6,0)(0.6,0.15)(0.2,0.15)(0.2,0.3)
\psbezier[linewidth=1.25pt,linecolor=blue]{-}(0.8,0)(0.8,0.15)(0.4,0.15)(0.4,0.3)
\psbezier[linewidth=1.25pt,linecolor=blue]{-}(1.0,0)(1.0,0.15)(0.6,0.15)(0.6,0.3)
\rput(0,-0.3){\psline[linewidth=1.25pt,linecolor=blue]{-}(0,0)(0,0.3)
\pspolygon[fillstyle=solid,fillcolor=pink](0.10,0)(1.1,0)(1.1,0.3)(0.10,0.3)\rput(0.6,0.15){$_{k-1}$}}
\end{pspicture} \ 
 \Bigg).
\label{eq:otherWJrelation}
\end{equation}
We remark that \eqref{eq:Yrec} also makes it clear that the $\Yk_t$ are actually elements of $\btl_{n,k}$, a fact that may not have been obvious from the definition \eqref{eq:Yt}. \cref{sec:aandbiso} goes further and shows that, over a complex function field, every element $\Ik a \Ik$ with $a \in \mathcal \tl_{n+k}$ is in $\btl_{n,k}$ and that $\btl_{n,k}$ and $\atl_{n,k} = \Ik \tl_{n+k} \Ik$ are isomorphic algebras. 

When $t=k$ and $n > k$, the derivation of the relation \eqref{eq:Yrec} instead results in a non-trivial constraint involving $\Yk_k$:
\begin{equation} \label{eq:ClosureRelation}
\Big[\prod_{j=0}^{k} \Ekj{k}{n-j}\Big] \Yk_k = \sum_{i=0}^{k-1} (-1)^i U_{k-1-i} \Big[ \prod_{j=i+2}^k \Ekj{k}{n-j} \Big] \Yk_k  \qquad \text{(\(n>k\)).}
\end{equation}
We will refer to this relation as the \emph{closure relation}.  The case $k=0$ is $\Ekj{k}{n} = 0$ which reduces the boundary seam algebra to the \TL{} algebra $\tl_n$.  Similarly, the
case $k=1$ recovers the \TL{} relation \eqref{eq:BTLRelationk=1} noted above.  
For $n>k$, we prove in \cref{prop:BTLComplete} that the closure relation, together with the one-boundary \TL{} relations \eqref{eq:newBTL}, yields
a complete set of relations of the boundary seam algebra $\btl_{n,k}$.  For $n \le k$, 
as argued in \cref{app:SeamsBetaFormal}, the relations \eqref{eq:newBTL} are already complete.
 
As shown in \cref{sec:dimA}, 
the dimension of the boundary seam algebra $\btl_{n,k}$ is given by
\begin{equation} \label{eq:dimBTL}
\dim \btl_{n,k} = \binom{2n}{n} - \binom{2n}{n-k-1}.
\end{equation}
This is consistent with $\btl_{n,0} = \tl_n$ and $\btl_{n,1} = \tl_{n+1}$. For $n>k$, \eqref{eq:dimTL1} gives 
$\dim \btl_{n,k}< \dim \tlone_n$, in accord with the homomorphism $\mathfrak h$ of \eqref{eq:surjectiveh}
not being an isomorphism in this case.

These results apply to the boundary seam algebras in which $\beta$ is treated as a formal parameter. The specialisation to $\beta \in \mathbb C$, and in particular for values where $U_i = 0$ for some $i \ge 0$, is subtle and a full discussion is deferred until \cref{sec:LatticeKac}, with the details worked out in \cref{app:SeamsBetaC}.  For now, we only remark that the diagrammatic definitions of these algebras can no longer be used when the \WJ{} projectors are not defined. When specialising to roots of unity, we will therefore define the boundary seam algebras algebraically via the relations given in \eqref{eq:newBTL} and a closure relation \eqref{eq:ClosureRelation'} similar to that given in \eqref{eq:ClosureRelation}. These algebraic relations are well-defined for all $\beta \in \CC$. However, the resulting algebra $\btl_{n,k}(\beta)$ is no longer a subalgebra 
of the specialised algebra $\tl_{n+k}(\beta)$. 
Its dimension is discussed in \cref{app:SeamsBetaC}.

\subsection{Representations of boundary seam algebras}
\label{sec:REPS}

\subsubsection{Standard modules}
\label{sec:reps}

The prescriptions used in \cite{PRZ06,PRannecy,PRV12,PTC14} to investigate $\Dbk(u,\xi)$, for $k>0$, were constructed, in general, by replacing the application of projectors in the boundary triangles by a diagrammatic rule, applied by hand, that set arcs appearing in the boundary to zero. 
One goal of this subsection is to clarify this construction and place it in the context of the boundary seam algebras $\btl_{n,k}$. Indeed, we will show that the matrices obtained from these prescriptions do not define representations of $\tl_{n+k}$.  Instead, they yield representations of $\btl_{n,k}$. 
We start the description with an example for $(n,k,d)=(4,2,0)$, where it is recalled that $d$ denotes the number of defects in the parametrisation of a standard module. Technical details are reported in \cref{app:btlstan}.

\paragraph{An example.}

For $(n,k,d)=(4,2,0)$, $\btl_{4,2} \subset \tl_6$ and the standard module $\stan_6^0$ of $\tl_6$ is spanned
by five link states:
\begin{equation}
\links_6^0 = \Big\{\ 
\psset{unit=0.8}
\begin{pspicture}[shift=-0.15](-0.0,0)(2.4,0.5)
\psline{-}(0,0)(2.4,0)
\psarc[linecolor=blue,linewidth=1.5pt]{-}(0.4,0){0.2}{0}{180}
\psarc[linecolor=blue,linewidth=1.5pt]{-}(1.6,0){0.2}{0}{180}
\psbezier[linecolor=blue,linewidth=1.5pt]{-}(1.0,0)(1.0,0.6)(2.2,0.6)(2.2,0)
\rput(2.6,-0.05){,}
\end{pspicture}
\quad
\begin{pspicture}[shift=-0.15](-0.0,0)(2.4,0.5)
\psline{-}(0,0)(2.4,0)
\psarc[linecolor=blue,linewidth=1.5pt]{-}(0.8,0){0.2}{0}{180}
\psarc[linecolor=blue,linewidth=1.5pt]{-}(1.6,0){0.2}{0}{180}
\psbezier[linecolor=blue,linewidth=1.5pt]{-}(0.2,0)(0.2,0.8)(2.2,0.8)(2.2,0)
\rput(2.6,-0.05){,}
\end{pspicture}
\quad
\begin{pspicture}[shift=-0.15](-0.0,0)(2.4,0.5)
\psline{-}(0,0)(2.4,0)
\psarc[linecolor=blue,linewidth=1.5pt]{-}(1.2,0){0.2}{0}{180}
\psbezier[linecolor=blue,linewidth=1.5pt]{-}(0.6,0)(0.6,0.6)(1.8,0.6)(1.8,0)
\psbezier[linecolor=blue,linewidth=1.5pt]{-}(0.2,0)(0.2,0.9)(2.2,0.9)(2.2,0)
\rput(2.6,-0.05){,}
\end{pspicture}
\quad
\begin{pspicture}[shift=-0.15](-0.0,0)(2.4,0.5)
\psline{-}(0,0)(2.4,0)
\psarc[linecolor=blue,linewidth=1.5pt]{-}(0.4,0){0.2}{0}{180}
\psarc[linecolor=blue,linewidth=1.5pt]{-}(1.2,0){0.2}{0}{180}
\psarc[linecolor=blue,linewidth=1.5pt]{-}(2.0,0){0.2}{0}{180}
\rput(2.6,-0.05){,}
\end{pspicture}
\quad
\begin{pspicture}[shift=-0.15](-0.0,0)(2.4,0.5)
\psline{-}(0,0)(2.4,0)
\psarc[linecolor=blue,linewidth=1.5pt]{-}(0.8,0){0.2}{0}{180}
\psarc[linecolor=blue,linewidth=1.5pt]{-}(2.0,0){0.2}{0}{180}
\psbezier[linecolor=blue,linewidth=1.5pt]{-}(0.2,0)(0.2,0.6)(1.4,0.6)(1.4,0)
\rput(2.9,-0.00){.}
\end{pspicture}\ \Big\}\phantom{.}
\label{eq:V60basis}
\end{equation}
In the standard representation $\rho_6^0$, the generators $I^{_{(2)}}_{\phantom{j}}$ and $E_j^{_{(2)}}$ are represented by
\begin{equation}
\begin{gathered}
\rho_6^0(I^{\textrm{\tiny$(2)$}}) =
\begin{pmatrix}
 1 & 0 & 0 & 0 & 0 \\
 0 & 1 & 0 & 0 & 0 \\
 0 & 0 & 1 & 0 & 0 \\
 -\frac{1}{\beta } & 0 & -\frac{1}{\beta } & 0 & 0 \\
 0 & -\frac{1}{\beta } & 0 & 0 & 0 \\
\end{pmatrix}
, \\
\begin{aligned}
\rho_6^0(\Ekj{2}{1}) &=
\begin{pmatrix}
 \beta  & 1 & 0 & 0 & 0 \\
 0 & 0 & 0 & 0 & 0 \\
 0 & 0 & 0 & 0 & 0 \\
 -1 & -\frac{1}{\beta } & 0 & 0 & 0 \\
 0 & 0 & 0 & 0 & 0 \\
\end{pmatrix}
, & \rho_6^0(\Ekj{2}{2}) &=
\begin{pmatrix}
 0 & 0 & 0 & 0 & 0 \\
 1 & \beta  & 1 & 0 & 0 \\
 0 & 0 & 0 & 0 & 0 \\
 0 & 0 & 0 & 0 & 0 \\
 -\frac{1}{\beta } & -1 & -\frac{1}{\beta } & 0 & 0 \\
\end{pmatrix}
, \\
\rho_6^0(\Ekj{2}{3}) &=
\begin{pmatrix}
 0 & 0 & 0 & 0 & 0 \\
 0 & 0 & 0 & 0 & 0 \\
 0 & 1 & \beta  & 0 & 0 \\
 0 & -\frac{1}{\beta } & -1 & 0 & 0 \\
 0 & 0 & 0 & 0 & 0 \\
\end{pmatrix}
, & \rho_6^0(\Ekj{2}{4}) &=
\begin{pmatrix}
 \beta ^2-1 & 0 & -1 & 0 & 0 \\
 0 & \beta ^2-1 & \beta  & 0 & 0 \\
 0 & 0 & 0 & 0 & 0 \\
 \frac{1-\beta ^2}{\beta } & 0 & \frac{1}{\beta } & 0 & 0 \\
 0 & \frac{1-\beta ^2}{\beta } & -1 & 0 & 0 \\
\end{pmatrix}
.
\end{aligned}
\end{gathered}
\end{equation}
All five generators act trivially on the last two link states of \eqref{eq:V60basis}. This is because both states have an arc connecting nodes $5$ and $6$ which is annihilated by the \WJ{} projector $P_2$ that is present in every element of $\btl_{n,k}$. The closure relation \eqref{eq:ClosureRelation} for $k=2$ involves the tangle $Y^{_{(2)}}_2$, represented by
\begin{equation}
\rho_6^0(Y^{_{(2)}}_2) =  \rho_6^0(\beta \Ekj{2}{4} -  \Ekj{2}{4} \Ekj{2}{3} \Ekj{2}{4}) = 
\begin{pmatrix}
 \beta^3- \beta & \beta^2-1 & 0 & 0 & 0 \\
 0 & 0 & 0 & 0 & 0 \\
 0 & 0 & 0 & 0 & 0 \\
 1-\beta^2 & \tfrac{1-\beta^2}\beta & 0 & 0 & 0 \\
 0 & 0 & 0 & 0 & 0 \\
\end{pmatrix}
,
\end{equation}
and it is easy to check that, indeed, $\rho_6^0(\Ekj{2}{4} \Ekj{2}{3} \Ekj{2}{2} Y^{_{(2)}}_2) = \rho_6^0(\beta \Ekj{2}{2} Y^{_{(2)}}_2 - Y^{_{(2)}}_2)$.

Because $I^{(2)}$, as the identity element of $\btl_{4,2}$, is not represented by an identity matrix, $\rho_6^0$ is not a representation of this algebra in the usual sense.\footnote{A representation of a unital associative algebra is usually required to send the unit to the identity map. 
A more general notion of representation would drop this requirement, the unit then being sent instead to an idempotent.  In this generalised sense, $\rho_6^0$ restricts to a representation of $\btl_{4,2}$.} However, 
a representation of $\btl_{4,2}$ is obtained by taking the upper-left $3 \times 3$ minor of the $\rho_6^0$ matrices. We denote this representation by $\rho_{4,2}^0$:
\begin{equation}
\begin{gathered}
\rho_{4,2}^0 (I^{\textrm{\tiny$(2)$}}) = 
\begin{pmatrix}
 1 & 0 & 0 \\
 0 & 1 & 0 \\
 0 & 0 & 1 \\
\end{pmatrix}
, \qquad
\rho_{4,2}^0 (\Ekj{2}{1}) = 
\begin{pmatrix}
 \beta  & 1 & 0 \\
 0 & 0 & 0 \\
 0 & 0 & 0 \\
\end{pmatrix}
, \qquad
\rho_{4,2}^0 (\Ekj{2}{2}) = 
\begin{pmatrix}
 0 & 0 & 0 \\
 1 & \beta  & 1 \\
 0 & 0 & 0 \\
\end{pmatrix}
, \\
\rho_{4,2}^0 (\Ekj{2}{3}) = 
\begin{pmatrix}
 0 & 0 & 0 \\
 0 & 0 & 0 \\
 0 & 1 & \beta  \\
\end{pmatrix}
, \qquad
\rho_{4,2}^0 (\Ekj{2}{4}) = 
\begin{pmatrix}
 \beta ^2-1 & 0 & -1 \\
 0 & \beta ^2-1 & \beta  \\
 0 & 0 & 0 \\
\end{pmatrix}
.
\end{gathered} \label{eq:repex}
\end{equation}
This representation acts on the vector space quotient of $\stan_{6}^0$ by the linear span of the last two link states. We denote this quotient
space and its (representative) link states by 
\begin{equation}
\stan_{4,2}^0 = \vspn
\Big\{\ 
\psset{unit=0.8}
\begin{pspicture}[shift=-0.15](-0.0,0)(2.4,0.5)
\psline{-}(0,0)(2.4,0)
\psline[linecolor=purple,linewidth=2.0pt]{-}(1.6,0)(2.4,0)
\psarc[linecolor=blue,linewidth=1.5pt]{-}(0.4,0){0.2}{0}{180}
\psarc[linecolor=blue,linewidth=1.5pt]{-}(1.6,0){0.2}{0}{180}
\psbezier[linecolor=blue,linewidth=1.5pt]{-}(1.0,0)(1.0,0.6)(2.2,0.6)(2.2,0)
\rput(2.6,-0.05){,}
\end{pspicture}
\quad
\begin{pspicture}[shift=-0.15](-0.0,0)(2.4,0.5)
\psline{-}(0,0)(2.4,0)
\psline[linecolor=purple,linewidth=2.0pt]{-}(1.6,0)(2.4,0)
\psarc[linecolor=blue,linewidth=1.5pt]{-}(0.8,0){0.2}{0}{180}
\psarc[linecolor=blue,linewidth=1.5pt]{-}(1.6,0){0.2}{0}{180}
\psbezier[linecolor=blue,linewidth=1.5pt]{-}(0.2,0)(0.2,0.8)(2.2,0.8)(2.2,0)
\rput(2.6,-0.05){,}
\end{pspicture}
\quad
\begin{pspicture}[shift=-0.15](-0.0,0)(2.4,0.5)
\psline{-}(0,0)(2.4,0)
\psline[linecolor=purple,linewidth=2.0pt]{-}(1.6,0)(2.4,0)
\psarc[linecolor=blue,linewidth=1.5pt]{-}(1.2,0){0.2}{0}{180}
\psbezier[linecolor=blue,linewidth=1.5pt]{-}(0.6,0)(0.6,0.6)(1.8,0.6)(1.8,0)
\psbezier[linecolor=blue,linewidth=1.5pt]{-}(0.2,0)(0.2,0.9)(2.2,0.9)(2.2,0)
\end{pspicture}
\Big\},
\label{eq:V420}
\end{equation}
with the pink line segment indicating the positions of the $k=2$ boundary
nodes. We also remark that, as opposed to $\rho_6^0$, the representation $\rho_{4,2}^0$ is well-defined when we specialise to any $\beta \in \CC$. This includes the case $\beta = 0$, even though the \WJ{} projector $P_2$ appearing in the definition of $I^{\textrm{\tiny$(2)$}}$ and the $\Ekj{2}{j}$ is then not defined.

\paragraph{General case.}

The above example illustrates how to proceed in a representation-theoretic manner for any $n$, $k$ and $d$ satisfying $d=n+k \bmod{2}$.  
In general, the link states in $\stan_{n+k}^d$ that have an arc connecting a pair of boundary
nodes, these being the nodes labelled by $n+1, \dots, n+k$, 
are annihilated by every element of $\btl_{n,k}$. Thus, they span a trivial $\btl_{n,k}$-module which we denote by 
$\stanu_{n,k}^d$. Restricting $\stan_{n+k}^d$ to a $\btl_{n,k}$-module, we can form the quotient
\begin{equation} \label{eq:Vdecomp}
 \stan_{n,k}^d = \frac{\stan_{n+k}^d}{\stanu_{n,k}^d},
\end{equation}
which we will refer to as a \emph{standard module} for the boundary seam algebra $\btl_{n,k}$.  The module 
$\stan_{n,k}^d$
is spanned by the (equivalence classes of the) link states of $\stan_{n+k}^d$ which have no arcs connecting boundary
nodes together.  This set of link states will be denoted by $\links_{n,k}^d$. We mention that it is convenient for what follows to identify these link states, elements of $\stan_{n+k}^d$, with their images in $\stan_{n,k}^d$, trusting that this will not lead to any confusion.  The dimension of the standard $\btl_{n,k}$-modules $\stan_{n,k}^d$ is computed in \cref{sec:dimVnkd}, the result being that 
\begin{equation}
\dim \stan_{n,k}^d = \binom{n}{\frac{n+k-d}{2}} - \binom{n}{\frac{n-k-d-2}{2}}.
\label{eq:dimVnkd}
\end{equation}

It is now clear that the standard representations 
\begin{equation}
\rho_{n,k}^d \colon \btl_{n,k} \rightarrow \mathrm{End}(\stan_{n,k}^d) \qquad \text{(\(0 \le d \le n+k\), \(d = (n+k) \bmod{2}\))} 
\end{equation}
correspond to taking the upper-left matrix minor (of the appropriate size) when we order the basis of $\stan_{n+k}^d$
so that link states with linked boundary nodes appear last. Taking the upper-left minor is precisely the prescription used 
in \cite{PRZ06,PRannecy,PRV12,PTC14}.  We note that, for $k>1$, this does not yield a representation of $\tl_{n+k}$, 
because the $\btl_{n,k}$-module $\stanu_{n,k}^d$ is not invariant under the action of $\tl_{n+k}$.

Our study of $\Dbk(u,\xi)$ in \cref{sec:Kacmod} will require the specialisation of $\beta = q + q^{-1}$ to values in 
$\mathbb C$ and, in particular, to values for which $q$ is a root of unity. 
We recall that the specialised algebra $\btl_{n,k}(\beta)$ is defined using generators and relations, but that it might not have an equivalent diagrammatic description.
\cref{sec:nonsingular} shows that the specialisation of $\rho_{n,k}^d$
is defined for all $\beta \in \CC$: $\rho_{n,k}^{d}(a)$ is well-defined (finite) for all $a \in \btl_{n,k}$, even when some 
projectors in the seam are ill-defined (divergent). In particular, because $\Dbk(u,\xi) \in \btl_{n,k}$, see \cref{sec:dbkinbtl}, the transfer matrices $\rho_{n,k}^d\big(\Dbk(u,\xi)\big)$ are well-defined.

\subsubsection{Invariant bilinear forms}
\label{sec:Gram}

The goal of this subsection is to generalise the invariant bilinear form on the standard $\tl_{n+k}$-modules $\stan_{n+k}^d$ to the 
standard modules $\stan_{n,k}^d$ of $\btl_{n,k}$:
\begin{equation}
\gramprodk{\cdot}{\cdot} \colon \stan_{n,k}^d \times \stan_{n,k}^d \rightarrow \mathbb C.
\end{equation}
The definition uses the fact that $\stan_{n,k}^d$ is the quotient \eqref{eq:Vdecomp} of $\stan_{n+k}^d$ by the trivial $\btl_{n,k}$-submodule $\stanu_{n,k}^d$.
Thus, if $w \in \stan_{n,k}^d$, then there exists $w'\in \stan_{n+k}^d$ whose image in the quotient is $w$.
We recall that when $w'$ may be chosen to be a link state, it will be identified with $w$ for simplicity.

Recalling that $\gramprod{\cdot}{\cdot}$ denotes the invariant bilinear form
of $\tl_{n+k}$, it follows that the bilinear form
\begin{equation}
\gramprodk{w_1}{w_2}=\gramprod{w_1'}{\Ik w_2'}=\gramprod{\Ik w_1'}{\Ik w_2'}
\end{equation}
is well-defined because (the self-adjoint idempotent)
$\Ik$ annihilates $\stanu_{n,k}^d$. We may alternatively define this for link states $w_1$ and $w_2$, in which case our convention is to drop the primes:
\begin{equation}
\gramprodk{w_1}{w_2}=\gramprod{w_1}{\Ik w_2}. 
\end{equation}
This bilinear form is invariant:
\begin{equation}
\gramprodk{w_1}{a w_2} = \gramprodk{a^\dagger w_1}{w_2} \qquad \text{(\(a \in \btl_{n,k}\), \(w_1,w_2 \in \stan_{n,k}^d\)).}
\end{equation}
Moreover, let $\Ik a \Ik \in \btl_{n,k}$ (see \cref{sec:aandbiso}) be such that $a$ is obtained 
from two link states $w_1, w_2\in \links_{n,k}^d$ by flipping $w_2$ and gluing its defects to those of $w_1$. Then, for all $w_3 \in \stan_{n,k}^d$, 
\begin{equation} \label{eq:WallIdentity}
(\Ik a \Ik)\, w_3 = 
w_1 \, \gramprodk{w_2}{w_3}.
\end{equation}

It follows, as in \cite{GraCel96}, that the standard $\btl_{n,k}$-modules $\stan_{n,k}^d$ are irreducible because $\gramprodk{\cdot}{\cdot}$ is not identically zero over a complex function field.  If $\beta$ is specialised to a value in $\mathbb C$ where $\btl_{n,k}(\beta)$ is well-defined diagrammatically,\footnote{We say that a tangle is \emph{well-defined diagrammatically}, at a specialised $\beta \in \CC$, if its decomposition as a linear combination of Temperley-Lieb connectivities, for $\beta$ formal, involves coefficients that do not diverge at the specialised value.}
then the standard modules $\stan_{n,k}^d$ are cyclic and indecomposable, though not necessarily irreducible, 
provided that $\gramprodk{\cdot}{\cdot}$ is not identically zero.\footnote{When $\gramprodk{\cdot}{\cdot}$ is identically zero, one can renormalise it as in \cite{RidSta12} in order to study the corresponding $\stan_{n,k}^d$. We will not detail this case here.}  If $\btl_{n,k}(\beta)$ is not defined diagrammatically, as may happen for certain $k$ when $q$ is a root of unity (see \cref{app:SeamsBetaFormal}), then cyclicity and indecomposability does not follow from \eqref{eq:WallIdentity}. Indeed, the standard modules may then be decomposable. For instance, for $\beta = 0$ and $(n,k,d)=(4,3,3)$, the state
\begin{equation}
2\,
\psset{unit=0.8}
\begin{pspicture}[shift=-0.05](-0.0,0)(2.8,0.5)
\psline[linewidth=\mince](0,0)(2.8,0)
\psline[linecolor=purple,linewidth=2.0pt]{-}(1.6,0)(2.8,0)
\psline[linecolor=blue,linewidth=\elegant]{-}(0.2,0)(0.2,0.5)
\psline[linecolor=blue,linewidth=\elegant]{-}(0.6,0)(0.6,0.5)
\psline[linecolor=blue,linewidth=\elegant]{-}(2.6,0)(2.6,0.5)
\psbezier[linecolor=blue,linewidth=\elegant]{-}(1.0,0)(1.0,0.7)(2.2,0.7)(2.2,0)
\psarc[linecolor=blue,linewidth=\elegant]{-}(1.6,0){0.2}{0}{180}
\end{pspicture}
\, - \,
\begin{pspicture}[shift=-0.05](-0.0,0)(2.8,0.5)
\psline[linewidth=\mince](0,0)(2.8,0)
\psline[linecolor=purple,linewidth=2.0pt]{-}(1.6,0)(2.8,0)
\psline[linecolor=blue,linewidth=\elegant]{-}(0.2,0)(0.2,0.5)
\psline[linecolor=blue,linewidth=\elegant]{-}(2.2,0)(2.2,0.5)
\psline[linecolor=blue,linewidth=\elegant]{-}(2.6,0)(2.6,0.5)
\psarc[linecolor=blue,linewidth=\elegant]{-}(1.2,0){0.2}{0}{180}
\psbezier[linecolor=blue,linewidth=\elegant]{-}(0.6,0)(0.6,0.7)(1.8,0.7)(1.8,0)
\end{pspicture} \, - \,
\begin{pspicture}[shift=-0.05](-0.0,0)(2.8,0.5)
\psline[linewidth=\mince](0,0)(2.8,0)
\psline[linecolor=purple,linewidth=2.0pt]{-}(1.6,0)(2.8,0)
\psline[linecolor=blue,linewidth=\elegant]{-}(1.0,0)(1.0,0.5)
\psline[linecolor=blue,linewidth=\elegant]{-}(2.2,0)(2.2,0.5)
\psline[linecolor=blue,linewidth=\elegant]{-}(2.6,0)(2.6,0.5)
\psarc[linecolor=blue,linewidth=\elegant]{-}(0.4,0){0.2}{0}{180}
\psarc[linecolor=blue,linewidth=\elegant]{-}(1.6,0){0.2}{0}{180}
\end{pspicture} + 
\begin{pspicture}[shift=-0.05](-0.0,0)(2.8,0.5)
\psline[linewidth=\mince](0,0)(2.8,0)
\psline[linecolor=purple,linewidth=2.0pt]{-}(1.6,0)(2.8,0)
\psline[linecolor=blue,linewidth=\elegant]{-}(1.8,0)(1.8,0.5)
\psline[linecolor=blue,linewidth=\elegant]{-}(2.2,0)(2.2,0.5)
\psline[linecolor=blue,linewidth=\elegant]{-}(2.6,0)(2.6,0.5)
\psarc[linecolor=blue,linewidth=\elegant]{-}(0.8,0){0.2}{0}{180}
\psbezier[linecolor=blue,linewidth=\elegant]{-}(0.2,0)(0.2,0.7)(1.4,0.7)(1.4,0)
\end{pspicture}
\end{equation}
spans a one-dimensional direct summand of $\stan_{4,3}^3$.
The general structure of $\stan_{n,k}^d$ in these cases is beyond the scope of this paper and will not be discussed here.

The Gram matrix of the bilinear form $\gramprodk{\cdot}{\cdot}$ is denoted by $\grammat_{n,k}^d$. For $(n,k,d)=(4,2,2)$, in the ordered basis \eqref{eq:B422} for $\links_{4,2}^2$, it is given by 
\begin{equation}
\grammat_{4,2}^2 = 
\begin{pmatrix}
(U_1)^2 & U_1 & U_1 & 1 & U_1 & 0 \\[0.1cm]
 U_1 & (U_1)^2 & 1 & U_1 & 1 & 1 \\[0.1cm]
 U_1 & 1 & U_2 & \frac{U_2}{U_1} & 1 & 0 \\[0.1cm]
 1 & U_1 & \frac{U_2}{U_1} & U_2 & \frac{U_2}{U_1}
  & \frac{U_2}{U_1} \\[0.1cm]
 U_1 & 1 & 1 & \frac{U_2}{U_1} & U_2 & 0 \\[0.1cm]
 0 & 1 & 0 & \frac{U_2}{U_1} & 0 & U_2 \\
\end{pmatrix}, \qquad \det \grammat_{4,2}^2 = \frac{U_4 (U_3)^4}{(U_1)^4}.
\label{eq:Gmat422}
\end{equation}
This form will be crucial in what follows. Its determinant is calculated in \cref{sec:Gramdet} and is given by
\begin{equation}
\det \grammat_{n,k}^d = \prod_{i=1}^{\lfloor \frac k2 \rfloor} \Big(\frac{U_{i-1}}{U_{k-i}}\Big)^{\dim \stan^d_{n,k-2i}} 
\prod_{j=1}^{\frac{n+k-d}2} \Big(\frac{U_{d+j}}{U_{j-1}}\Big)^{\dim \stan^{d+2j}_{n,k}}.
\label{eq:Gnkd}
\end{equation}

%
\section{Kac modules} \label{sec:Kacmod}
%

In the second phase of this work, our goal is to investigate the Virasoro representations, in particular the so-called (\emph{Virasoro}) \emph{Kac modules}, 
that arise in the (continuum) scaling limit known as the \emph{logarithmic minimal models}. 
The latter name comes from the 2006 paper by Pearce, Rasmussen and Zuber \cite{PRZ06}, where the integrability of the underlying loop models was explored from the diagrammatic point of view and where initial conjectures for the limiting Virasoro characters and modules were made, many of them based on eigenvalues found for small system sizes using a computer. A flurry of papers have followed, 
describing various aspects of these models, such as their fusion rules \cite{RasFus07a,RasFus07}, 
Jordan block structures \cite{MDSA11,MDSA13}, and, for the particular model $\mathcal{LM}(1,2)$ that describes 
critical dense polymers \cite{PR07}, the cylinder \cite{PRV10} and (modular invariant) torus partition functions\cite{MDPR13}.

In \cref{sec:LatticeKac}, we use the tools developed in the previous section to propose a definition for what we refer to as \emph{lattice Kac modules}.  These are modules over the boundary seam algebra $\btl_{n,k}$ which are conjectured to correspond, in the scaling limit, to the Kac modules over the Virasoro algebra (see \cref{sec:Kac}). We argue that these are well-behaved for the root of unity cases \eqref{eq:rootsof1} corresponding to the logarithmic minimal models $\mathcal{LM}(p,p')$. \cref{sec:LatticeFusion} reviews the definition of the fusion of these modules on the lattice.

The scaling limits of the lattice Kac modules are investigated in \cref{sec:Kac} 
with the goal being to relate them to a family of modules over the Virasoro algebra, the 
Virasoro Kac modules, defined in \cref{sec:VKac}. We borrow basic results and terminology of Virasoro representation theory, reviewed in \cref{sec:Back}. The scaling limit and the roles played by the tangles $\Dbk(u,\xi)$ and $\hamk$ are described in \cref{sec:scalinglimit}. This culminates with one of our main results, a conjecture for the scaling limit of the lattice Kac modules as Virasoro modules.  This extends a similar conjecture \cite{RasCla11} for $\mathcal{LM}(1,p')$ to all the logarithmic minimal models $\mathcal{LM}(p,p')$.
\cref{sec:characters} then presents a modicum of evidence for the conjecture at the level of the spectra of $\hamk$ and the corresponding Virasoro characters (plenty of additional evidence may be found in \cite{PRZ06,PTC14}). The climax of this development is \cref{sec:Gramidentification}, where we show how to use Gram matrix methods to go beyond characters and Jordan block 
analyses of $L_0$ to obtain non-trivial information concerning the structure of the Kac modules.

%
\subsection{Lattice Kac modules and lattice fusion} 
\label{sec:latticeKacmod}
%

\subsubsection{Lattice Kac modules}\label{sec:LatticeKac} 

We have identified the algebra $\btl_{n,k}$
that describes loop models with boundary seams of width $k$.  By \cref{sec:dbkinbtl}, this
algebra includes the transfer tangles $\Dbk(u,\xi)$ and the Hamiltonian $\hamk$.
The Temperley-Lieb tangles appearing in the formal expansion of $\Dbk(u,\xi)$ in $u$ commute with 
one another and, in the scaling limit, 
are believed to converge to the conformal integrals of motion \cite{SasVir88,EguDef89}. 
These tangles also belong to $\btl_{n,k}$. It might therefore seem that we should restrict our attention to the abelian subalgebra of $\btl_{n,k}$ generated by these integrals of motion.  However, this neglects the fact that we want to relate this algebra, in the scaling limit, to the Virasoro algebra of \cft{}.  While the Hamiltonian is, up to shifts and rescalings, supposed to realise the Virasoro zero mode $L_0$ in this limit, the other integrals of motion do not realise the remaining Virasoro modes $L_m$.  This should be clear from the fact that Virasoro modes do not commute in general.

There are explicit expressions for tangles that are believed to realise the $L_m$ modes in the scaling limit. We refer to these as 
\emph{Virasoro mode approximations}. For the algebra $\tl_n$, one such family of tangles is \cite{KooAss94,GaiCon11}
\begin{equation}\label{eq:modes}
\Lmn = \frac{n}{\pi}\, \Big[ -\frac1{v_s} \sum_{j=1}^{n-1} (e_j-h_{bulk}) \cos\Big(\frac{\pi mj}{n}\Big) + \frac1{v_s^2} \sum_{j=1}^{n-2}\big[e_j,e_{j+1}\big]\sin\Big(\frac{\pi mj}{n}\Big)\Big] + \frac c{24} \delta_{m,0}.
\end{equation}
Here, $h_{bulk}$ is the Hamiltonian bulk free energy, $v_s = \frac{\pi \sin \lambda}{\lambda}$ is the speed of sound (see \eqref{eq:Hj}) and $n$ still refers to the number of nodes of $\tl_n$.

A formula similar to \eqref{eq:modes} was given by Gainutdinov \emph{et al}.\ \cite{GJSV13} for the blob algebra, with the $\Lmn$ again expressed in terms of the generators of the algebra.
We therefore find the following conjecture reasonable.
\begin{conj}
For each $m \in \ZZ$, there exist sequences $\bigl( \Lmn \bigr)_{n \in \ZZ_+}$, where each $\Lmn$ is an element of $\btl_{n,k}$, whose $n \to \infty$ scaling limits satisfy the commutation rules of the Virasoro algebra. \label{conj:highermodes}
\end{conj}
\noindent Assuming this conjecture, \cref{sec:Kac} will investigate the scaling limits 
of loop models with boundary seams, profiting from the indecomposable structures of the diagrammatic algebras to study their limiting Virasoro counterparts. \cref{conj:highermodes} suggests that the algebra of interest for comparing with the continuum behaviour is, in fact, $\btl_{n,k}$ itself. 
We therefore propose the following definition for the \emph{lattice Kac modules}.
\begin{defn}\label{def:latKac}
The lattice Kac module $\kac_{n,k}^d$ is the standard module $\stan_{n,k}^d$ over the algebra $\btl_{n,k}$. 
\end{defn}
\noindent These lattice Kac modules are
characterised by three integers, $n$, $k$ and $d$, and one formal parameter $\beta = q + q^{-1}$. 
Because the eigenvalues of $\Dbk(u,\xi)$ (or just $\hamk$) play a prominent role in defining the scaling limit of the $\kac_{n,k}^d$, see \cref{sec:scalinglimit}, we will add $u$ and $\xi$ (or just $\xi$) to the list of parameters characterising this limit.
\smallskip

We now specialise $\beta$ to values in $\CC$; the parameters $u$ and $\xi$ that appear in the definition of $\Dbk(u,\xi)$ will be left formal for now.
For $q$ a root of unity, labelled as in \eqref{eq:rootsof1} by the pair of integers $(p,p')$, the scaling limit of the lattice Kac modules is believed to be described by a \cft{} called a \emph{logarithmic minimal model}, denoted by $\mathcal {LM}(p,p')$.  For convenience, and following \cite{PRZ06}, we will also use this nomenclature to describe the underlying integrable lattice models.  The first member $\mathcal {LM}(1,2)$ of this family of theories is called 
\emph{critical dense polymers}.  It corresponds to $\beta = 0$, implying that contractible loops are forbidden. It has the peculiarity of being exactly solvable, meaning that the eigenvalues of $\Dbk(u,\xi)$ admit closed expressions for all $n \in \ZZ_{+}$ \cite{PR07,PRV12}.
We also note that the Temperley-Lieb representation theory for $\beta=0$ is the exceptional case discussed at the end of \cref{app:TLrep}. Another member of this family is $\mathcal {LM}(2,3)$, called 
\emph{critical percolation}. It corresponds to $\beta = 1$, implying that contractible loops may be ignored.

When $q$ is a root of unity, there is a subtlety in the lattice algebra description.  Defining $\ell$ to be the smallest positive integer for which $q^{2 \ell} = 1$, so that $\ell = p'$ with the parametrisation \eqref{eq:rootsof1}, the \WJ{} projector $P_\ell$ is undefined because $U_{\ell-1} = 0$, see \eqref{eq:WJs}.  The diagrammatic definition \eqref{eq:defbtl} of $\btl_{n,k}(\beta)$ therefore makes sense for $k< \ell$, as in the formal and generic cases, but not for $k\ge \ell$. In the latter case, this definition may be replaced by a purely algebraic definition in terms of generators and relations, discussed in detail in \cref{app:SeamsBetaC}.
These relations are well-defined for all specialisations due to the careful choice of normalisation of the generator $\Ekj{k}{n}$ in \eqref{eq:Enk}. 
In particular, the values of $\beta_1$ and $\beta_2$ pertaining to the homomorphism $\mathfrak h \colon \tlone_n(\beta,\beta_1,\beta_2) \to \btl_{n,k}(\beta)$, see \eqref{eq:betas},
never diverge and no Chebyshev polynomials appear as denominators in the defining relations. 

The surprising subtlety of the root of unity case is then that the dimension of $\btl_{n,k}(\beta)$, for $k>\ell$, is smaller than in the generic case. Indeed, \cref{sec:dimB} shows that $\btl_{n,k}(\beta) \simeq \btl_{n,k'}(\beta)$, whenever $k,k'>0$ and $k'=k \bmod{\ell}$.  When $k$ is not a multiple of $\ell$, this gives a description of $\btl_{n,k}(\beta)$ in terms of an algebra that is diagrammatically well-defined.  Nevertheless, the action of $\btl_{n,k}(\beta)$ on the lattice Kac modules $\kac_{n,k}^d = \stan_{n,k}^d$ is free of singularities for all $\beta \in \CC$ (\cref{sec:nonsingular}). 
\cref{def:latKac} is thus readily extended to all specialisations $\beta \in \mathbb R$: The specialised lattice Kac module 
$\kac_{n,k}^d$ is the (well-defined) standard module $\stan_{n,k}^d$ over $\btl_{n,k}(\beta)$.
This lack of singularities applies, in particular, to $\Dbk(u,\xi)$ and $\hamk$: 
The representatives $\rho_{n,k}^d(\Dbk(u,\xi))$ and $\rho_{n,k}^d(\hamk)$ are singularity-free. 
These $\btl_{n,k}(\beta)$-modules, as we have defined them, are therefore strong candidates to realise Virasoro modules in the scaling limit.

Finally, we consider the specialisation of the parameters $u$ and $\xi$ to values in $\CC$, though we will typically only consider real values.  In this case, there turn out to be \emph{exceptional points} $(u,\xi)$ for which the transfer tangle $\Dbk(u,\xi)$ is singular as an element of $\btl_{n,k}$, see \eqref{eq:KInBTL}, \eqref{eq:AlgDmk} and \eqref{eq:Algdmk}. Even if we only consider the Hamiltonian $\hamk$ (as we will in the numerical studies of \cref{sec:LatticeData}), then specialising $\xi$ also leads to \emph{exceptional points} at which $\hamk$ diverges, see \eqref{eq:Ham}. \cref{sec:nonsingular} obviously does not apply at these exceptional points and the
corresponding matrices in the representations $\rho_{n,k}^d$ are then also singular.
We therefore have to exclude these points
from \cref{conj1} below.  The exceptional points $\xi$ for $\hamk$ will be identified in \cref{sec:scalinglimit}.

We emphasise that this algebraic development gives a rigorous mathematical framework for the recipes originally described in \cite{PRZ06} to obtain transfer matrices and Hamiltonians. We end this section by recalling a conjecture from that paper.
\begin{conj}[\cite{PRZ06}] \label{conj1}
When $(u,\xi)$ is not exceptional, the linear operators $\rho_{n,k}^d(\Dbk(u,\xi))$ and $\rho_{n,k}^d(\hamk)$ are diagonalisable
on the lattice Kac module $\kac_{n,k}^d$ with real eigenvalues, for all
$\beta\in \mathbb R$, all $n,k \in \ZZ_+$, and all $0\le d \le n+k$ with $d=n+k \bmod{2}$.
\end{conj}
\noindent The statement of this conjecture for the Hamiltonian \mbox{$\mathcal{H}^{\textrm{\tiny$(0)$}}\!$} 
was proven recently in \cite{MDRRSA15}.

\subsubsection{Lattice fusion}\label{sec:LatticeFusion}

A fusion rule in conformal field theory encodes the number of distinct ways that the fields in the model 
arise in the operator product expansion of two given fields.
In representation-theoretic terms, this corresponds to the decomposition of the so-called fusion product of two modules 
over the conformal symmetry algebra, here assumed to be
the Virasoro algebra. Fusion is thus a fundamental
part of conformal field theory and has direct implications for the computation of correlation functions.
However, the determination of the fusion rules in a conformal field theory can be a decidedly non-trivial task.
If a lattice realisation is available, then it may be possible to predict these rules with comparatively less effort, as a prescription for fusion can sometimes be implemented at the lattice level \cite{Cardy86,Cardy89}.

For loop models, a prescription for lattice fusion was given by Pearce, Rasmussen and Zuber \cite{PRZ06}. 
Working with the transfer tangle $\Dbk(u,\xi)$, acting on link states with $d$ defects, they interpreted the defects as a 
boundary condition applied to the right of $\Dbk(u,\xi)$. They then proposed that the spectra and Jordan block structure of $\Dbk(u,\xi)$ 
implied similar results for the generator $L_0$ acting on Virasoro modules labelled by two integers $r$ and $s$. 
We assert in \cref{TheConjecture} below that these Virasoro modules are the Virasoro Kac modules $\Kac{r,s}$ defined 
in \cref{sec:VKac}.
The relation between $k$, $d$, $r$ and $s$ is discussed in \cref{sec:scalinglimit} and evidence supporting this 
conjecture is described in \cref{sec:LatticeData}.

Following the terminology of \cite{PRZ06}, the boundary condition for transfer tangles $\boldsymbol{D}^{\textrm{\tiny$(k)$}}\!(u,\xi) \in \btl_{n,k}$ acting on the lattice Kac modules $\kac_{n,k}^0$ with zero defects (and $k = n \bmod 2$) is said to be of $r$-type.  
This terminology comes about because it is believed, as we shall see in \cref{sec:scalinglimit}, that if $\beta$ is specialised as in \eqref{eq:rootsof1}, then the scaling limit of the lattice Kac modules $\kac_{n,2}^0$, with $n$ even, is a Virasoro module with associated Kac labels of the form $(r,1)$.  We shall refer to the scaling limit as the \emph{Virasoro Kac module} $\Kac{r,1}$, deferring a precise definition until \cref{sec:VKac}.  An example of an $r$-type boundary is afforded by 
the following diagram: 
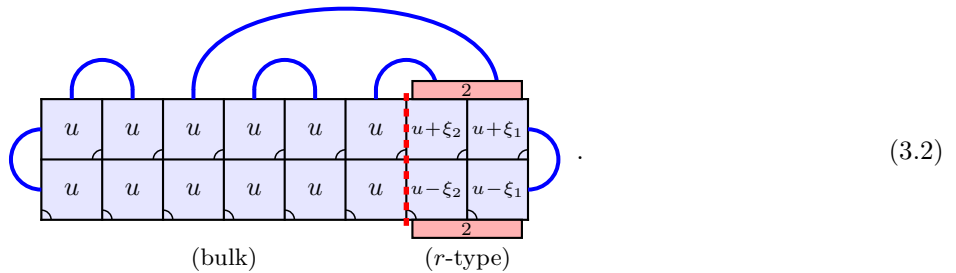
\begin{equation} 
\psset{unit=0.8}
\begin{pspicture}[shift=-1.7](-0.5,-0.7)(8.5,3.4)
\facegrid{(0,0)}{(8,2)}
\psarc[linewidth=0.025]{-}(0,0){0.16}{0}{90}
\psarc[linewidth=0.025]{-}(1,1){0.16}{90}{180}
\psarc[linewidth=0.025]{-}(1,0){0.16}{0}{90}
\psarc[linewidth=0.025]{-}(2,1){0.16}{90}{180}
\psarc[linewidth=0.025]{-}(2,0){0.16}{0}{90}
\psarc[linewidth=0.025]{-}(3,1){0.16}{90}{180}
\psarc[linewidth=0.025]{-}(3,0){0.16}{0}{90}
\psarc[linewidth=0.025]{-}(4,1){0.16}{90}{180}
\psarc[linewidth=0.025]{-}(4,0){0.16}{0}{90}
\psarc[linewidth=0.025]{-}(5,1){0.16}{90}{180}
\psarc[linewidth=0.025]{-}(5,0){0.16}{0}{90}
\psarc[linewidth=0.025]{-}(6,1){0.16}{90}{180}
\psarc[linewidth=0.025]{-}(6,0){0.16}{0}{90}
\psarc[linewidth=0.025]{-}(7,1){0.16}{90}{180}
\psarc[linewidth=0.025]{-}(7,0){0.16}{0}{90}
\psarc[linewidth=0.025]{-}(8,1){0.16}{90}{180}
\psarc[linewidth=1.5pt,linecolor=blue]{-}(0,1){0.5}{90}{-90}
\psarc[linewidth=1.5pt,linecolor=blue]{-}(8,1){0.5}{-90}{90}
\rput(0,0.15){
\psarc[linewidth=1.5pt,linecolor=blue]{-}(1,2){0.5}{0}{180}
\psarc[linewidth=1.5pt,linecolor=blue]{-}(4,2){0.5}{0}{180}
\psarc[linewidth=1.5pt,linecolor=blue]{-}(6,2.0){0.5}{0}{180}
\psbezier[linewidth=1.5pt,linecolor=blue]{-}(2.5,2)(2.5,3.8)(7.5,3.8)(7.5,2)
}
\multiput(0,0)(1,0){8}{\psline[linewidth=1.5pt,linecolor=blue]{-}(0.5,2.00)(0.5,2.16)}
\rput(0.5,0.5){$u$}
\rput(0.5,1.5){$u$}
\rput(1.5,.5){$u$}
\rput(1.5,1.5){$u$}
\rput(2.5,0.5){$u$}
\rput(2.5,1.5){$u$}
\rput(3.5,0.5){$u$}
\rput(3.5,1.5){$u$}
\rput(4.5,.5){$u$}
\rput(4.5,1.5){$u$}
\rput(5.5,.5){$u$}
\rput(5.5,1.5){$u$}
\rput(6.5,.5){\scriptsize$u\!-\!\xi_2$}
\rput(6.5,1.5){\scriptsize$u\!+\!\xi_2$}
\rput(7.5,.5){\scriptsize$u\!-\!\xi_1$}
\rput(7.5,1.5){\scriptsize$u\!+\!\xi_1$}
\pspolygon[fillstyle=solid,fillcolor=pink](6.1,0)(7.9,0)(7.9,-0.3)(6.1,-0.3)(6.1,0)\rput(7.0,-0.15){$_{2}$}
\pspolygon[fillstyle=solid,fillcolor=pink](6.1,2)(7.9,2)(7.9,2.3)(6.1,2.3)(6.1,2)\rput(7.0,2.15){$_{2}$}
\psline[linecolor=red,linewidth=2pt,linestyle=dashed,dash=3pt 3pt]{-}(6,-0.1)(6,2.1)
\rput(3,-0.65){\small(bulk)}
\rput(7,-0.65){\small($r$-type)}
\end{pspicture} \ \ .
\label{eq:diagrtype}
\end{equation}
It illustrates, up to a prefactor, the action of $\boldsymbol{D}^{\textrm{\tiny$(2)$}}\!(u,\xi) \in \btl_{6,2}$ on the link state $
\psset{unit=0.4}
\begin{pspicture}[shift=-0.05](-0.0,0)(3.2,0.5)
\psline{-}(0,0)(2.4,0)
\psline[linecolor=purple,linewidth=2.0pt]{-}(2.4,0)(3.2,0)
\psarc[linecolor=blue,linewidth=1.0pt]{-}(0.4,0){0.2}{0}{180}
\psarc[linecolor=blue,linewidth=1.0pt]{-}(1.6,0){0.2}{0}{180}
\psarc[linecolor=blue,linewidth=1.0pt]{-}(2.4,0){0.2}{0}{180}
\psbezier[linecolor=blue,linewidth=1.0pt]{-}(1.0,0)(1.0,0.9)(3.0,0.9)(3.0,0)
\end{pspicture}
 \in \kac_{6,2}^0$.  The value of $r$ is determined as an explicit function of $k$, $p$ and $p'$ that is given in \cref{TheConjecture} below.

In stark contrast, for transfer tangles with vacuum boundary conditions ($k=0$) acting on lattice Kac modules 
$\kac_{n,0}^d$ (with $d = n \bmod 2$), the boundary is said to be of $s$-type.
This is because, in the scaling limit, the $\kac_{n,0}^d$ are believed to define Virasoro Kac modules
with labels $(1,s)$, where $s = d+1$.
With these conventions, the case $k=d=0$ is thus both of $r$- and $s$-type. 
For $d>0$, instead of being drawn vertically, the defects may be folded towards the right and attached to the right boundary where one places a different type of boundary seam of width $d$. 
An example illustrating this folding procedure 
is the following, where $\boldsymbol{D}^{\textrm{\tiny$(0)$}}\!(u) \in \btl_{5,0}\simeq \tl_5$ acts on $\psset{unit=0.4}
\begin{pspicture}[shift=-0.05](-0.0,0)(2.0,0.5)
\psline{-}(0,0)(2.0,0)
\psarc[linecolor=blue,linewidth=1.0pt]{-}(0.8,0){0.2}{0}{180}
\psline[linecolor=blue,linewidth=1.0pt]{-}(0.2,0)(0.2,0.7)
\psline[linecolor=blue,linewidth=1.0pt]{-}(1.4,0)(1.4,0.7)
\psline[linecolor=blue,linewidth=1.0pt]{-}(1.8,0)(1.8,0.7)
\end{pspicture}
\in \kac_{5,0}^3$ 
(again up to a prefactor):
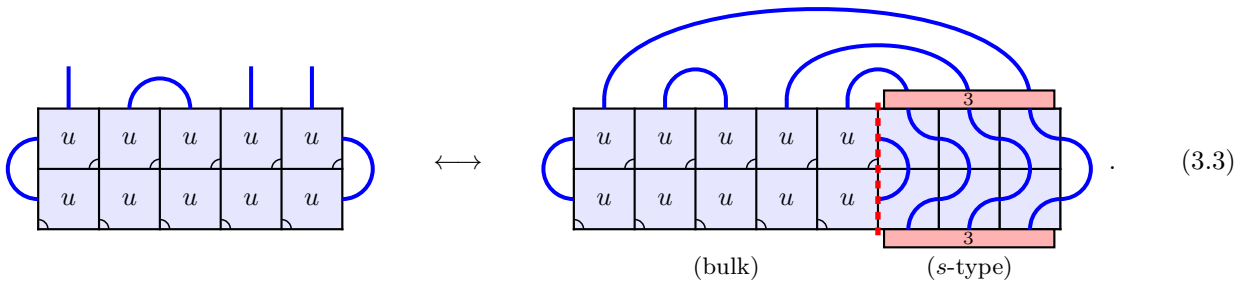
\begin{equation} 
\psset{unit=0.8}
\begin{pspicture}[shift=-1.7](-0.5,-0.7)(5.5,2.6)
\facegrid{(0,0)}{(5,2)}
\psarc[linewidth=0.025]{-}(0,0){0.16}{0}{90}
\psarc[linewidth=0.025]{-}(1,1){0.16}{90}{180}
\psarc[linewidth=0.025]{-}(1,0){0.16}{0}{90}
\psarc[linewidth=0.025]{-}(2,1){0.16}{90}{180}
\psarc[linewidth=0.025]{-}(2,0){0.16}{0}{90}
\psarc[linewidth=0.025]{-}(3,1){0.16}{90}{180}
\psarc[linewidth=0.025]{-}(3,0){0.16}{0}{90}
\psarc[linewidth=0.025]{-}(4,1){0.16}{90}{180}
\psarc[linewidth=0.025]{-}(4,0){0.16}{0}{90}
\psarc[linewidth=0.025]{-}(5,1){0.16}{90}{180}
\psarc[linewidth=1.5pt,linecolor=blue]{-}(0,1){0.5}{90}{-90}
\psarc[linewidth=1.5pt,linecolor=blue]{-}(5,1){0.5}{-90}{90}
\psline[linewidth=1.5pt,linecolor=blue]{-}(0.5,2)(0.5,2.7)
\psline[linewidth=1.5pt,linecolor=blue]{-}(3.5,2)(3.5,2.7)
\psline[linewidth=1.5pt,linecolor=blue]{-}(4.5,2)(4.5,2.7)
\psarc[linewidth=1.5pt,linecolor=blue]{-}(2,2){0.5}{0}{180}
\rput(0.5,0.5){$u$}
\rput(0.5,1.5){$u$}
\rput(1.5,.5){$u$}
\rput(1.5,1.5){$u$}
\rput(2.5,0.5){$u$}
\rput(2.5,1.5){$u$}
\rput(3.5,0.5){$u$}
\rput(3.5,1.5){$u$}
\rput(4.5,.5){$u$}
\rput(4.5,1.5){$u$}
\end{pspicture} 
\qquad \longleftrightarrow \qquad
\begin{pspicture}[shift=-1.7](-0.5,-0.7)(8.5,3.5)
\facegrid{(0,0)}{(8,2)}
\psarc[linewidth=0.025]{-}(0,0){0.16}{0}{90}
\psarc[linewidth=0.025]{-}(1,1){0.16}{90}{180}
\psarc[linewidth=0.025]{-}(1,0){0.16}{0}{90}
\psarc[linewidth=0.025]{-}(2,1){0.16}{90}{180}
\psarc[linewidth=0.025]{-}(2,0){0.16}{0}{90}
\psarc[linewidth=0.025]{-}(3,1){0.16}{90}{180}
\psarc[linewidth=0.025]{-}(3,0){0.16}{0}{90}
\psarc[linewidth=0.025]{-}(4,1){0.16}{90}{180}
\psarc[linewidth=0.025]{-}(4,0){0.16}{0}{90}
\psarc[linewidth=0.025]{-}(5,1){0.16}{90}{180}
\rput(0,0.15){
\psarc[linewidth=1.5pt,linecolor=blue]{-}(2,2){0.5}{0}{180}
\psarc[linewidth=1.5pt,linecolor=blue]{-}(5,2){0.5}{0}{180}
\psbezier[linewidth=1.5pt,linecolor=blue]{-}(3.5,2)(3.5,3.2)(6.5,3.2)(6.5,2)
\psbezier[linewidth=1.5pt,linecolor=blue]{-}(0.5,2)(0.5,4.0)(7.5,4.0)(7.5,2)
}
\psarc[linewidth=1.5pt,linecolor=blue]{-}(0,1){0.5}{90}{-90}
\psarc[linewidth=1.5pt,linecolor=blue]{-}(8,1){0.5}{-90}{90}
\multiput(0,0)(1,0){5}{\psline[linewidth=1.5pt,linecolor=blue]{-}(0.5,2.00)(0.5,2.16)}
\pspolygon[fillstyle=solid,fillcolor=pink](5.1,0)(7.9,0)(7.9,-0.3)(5.1,-0.3)(5.1,0)\rput(6.5,-0.15){$_{3}$}
\pspolygon[fillstyle=solid,fillcolor=pink](5.1,2)(7.9,2)(7.9,2.3)(5.1,2.3)(5.1,2)\rput(6.5,2.15){$_{3}$}
\rput(5,0){\loopa}\rput(5,1){\loopb}
\rput(6,0){\loopa}\rput(6,1){\loopb}
\rput(7,0){\loopa}\rput(7,1){\loopb}
\rput(0.5,0.5){$u$}
\rput(0.5,1.5){$u$}
\rput(1.5,.5){$u$}
\rput(1.5,1.5){$u$}
\rput(2.5,0.5){$u$}
\rput(2.5,1.5){$u$}
\rput(3.5,0.5){$u$}
\rput(3.5,1.5){$u$}
\rput(4.5,.5){$u$}
\rput(4.5,1.5){$u$}
\psline[linecolor=red,linewidth=2pt,linestyle=dashed,dash=3pt 3pt]{-}(5,-0.1)(5,2.1)
\rput(2.5,-0.65){\small(bulk)}
\rput(6.5,-0.65){\small($s$-type)}
\end{pspicture} \ \ .
\label{eq:diagstype}
\end{equation}
We note that the three defects could equivalently be folded to the left boundary.

This folding procedure 
does not affect the action of $\btl_{n,0} \simeq \tl_n$.  This needs interpretation because after folding, the link states have $n+d$ nodes.  As the above example indicates, we embed $\tl_n$ into $\tl_{n+d}$ in the usual fashion except that there are \WJ{} projectors covering the boundary
nodes.  These projectors prevent the number of defects from decreasing, consistent with the standard Temperley-Lieb action.  With them, the standard $\tl_n$-action on link states commutes with the folding procedure in that the result obtained directly agrees with that obtained by adding the $s$-type seam and folding, acting with the embedding of $\tl_n$ into $\tl_{n+d}$, and then removing the seam and the folding. 
For example, the following two computations in $\kac_{5,0}^3$ give equivalent
results:
\begin{equation}
\psset{unit=0.8}
\begin{pspicture}[shift=-0.4](-0.0,0)(2.0,1.6)
\pspolygon[fillstyle=solid,fillcolor=lightlightblue](0,0)(0,1)(2.0,1)(2.0,0)(0,0)
\psbezier[linecolor=blue,linewidth=\elegant]{-}(1.0,1)(1.0,0.5)(0.2,0.5)(0.2,0)
\psbezier[linecolor=blue,linewidth=\elegant]{-}(1.4,1)(1.4,0.5)(0.6,0.5)(0.6,0)
\psarc[linecolor=blue,linewidth=\elegant]{-}(0.4,1){0.2}{180}{0}
\psline[linecolor=blue,linewidth=\elegant]{-}(1.8,0)(1.8,1)
\psarc[linecolor=blue,linewidth=\elegant]{-}(1.2,0){0.2}{0}{180}
\psarc[linecolor=blue,linewidth=\elegant]{-}(0.8,1){0.2}{0}{180}
\psline[linecolor=blue,linewidth=\elegant]{-}(0.2,1)(0.2,1.6)
\psline[linecolor=blue,linewidth=\elegant]{-}(1.4,1)(1.4,1.6)
\psline[linecolor=blue,linewidth=\elegant]{-}(1.8,1)(1.8,1.6)
\end{pspicture} \ = \
\begin{pspicture}[shift=-0.02](-0.0,0)(2.0,0.5)
\psline[linewidth=\mince](0,0)(2.0,0)
\psline[linecolor=blue,linewidth=\elegant]{-}(0.2,0)(0.2,0.5)
\psline[linecolor=blue,linewidth=\elegant]{-}(0.6,0)(0.6,0.5)
\psline[linecolor=blue,linewidth=\elegant]{-}(1.8,0)(1.8,0.5)
\psarc[linecolor=blue,linewidth=\elegant]{-}(1.2,0){0.2}{0}{180}
\end{pspicture}\ \qquad \longleftrightarrow \qquad
\begin{pspicture}[shift=-0.7](-0.0,-0.2)(3.2,2.0)
\pspolygon[fillstyle=solid,fillcolor=lightlightblue](0,0)(0,1)(3.2,1)(3.2,0)(0,0)
\psbezier[linecolor=blue,linewidth=\elegant]{-}(1.0,1)(1.0,0.5)(0.2,0.5)(0.2,0)
\psbezier[linecolor=blue,linewidth=\elegant]{-}(1.4,1)(1.4,0.5)(0.6,0.5)(0.6,0)
\psarc[linecolor=blue,linewidth=\elegant]{-}(0.4,1){0.2}{180}{0}
\psline[linecolor=blue,linewidth=\elegant]{-}(1.8,0)(1.8,1)
\psline[linecolor=blue,linewidth=\elegant]{-}(2.2,0)(2.2,1)
\psline[linecolor=blue,linewidth=\elegant]{-}(2.6,0)(2.6,1)
\psline[linecolor=blue,linewidth=\elegant]{-}(3.0,0)(3.0,1)
\psarc[linecolor=blue,linewidth=\elegant]{-}(1.2,0){0.2}{0}{180}
\rput(0,0.15)
{\psarc[linecolor=blue,linewidth=\elegant]{-}(0.8,1){0.2}{0}{180}
\psarc[linecolor=blue,linewidth=\elegant]{-}(2.0,1){0.2}{0}{180}
\psbezier[linecolor=blue,linewidth=\elegant]{-}(1.4,1)(1.4,1.6)(2.6,1.6)(2.6,1)
\psbezier[linecolor=blue,linewidth=\elegant]{-}(0.2,1)(0.2,2.0)(3.0,2.0)(3.0,1)
}
\multiput(0,0)(0.4,0){5}{\psline[linecolor=blue,linewidth=\elegant]{-}(0.2,1)(0.2,1.2)}
\pspolygon[fillstyle=solid,fillcolor=pink](2.1,0)(3.1,0)(3.1,-0.2)(2.1,-0.2)(2.1,0)\rput(2.6,-0.1){\scriptsize$_{3}$}
\pspolygon[fillstyle=solid,fillcolor=pink](2.1,1)(3.1,1)(3.1,1.2)(2.1,1.2)(2.1,1)\rput(2.6,1.1){\scriptsize$_{3}$}
\end{pspicture} \ = \
\begin{pspicture}[shift=-0.02](-0.0,0)(3.2,0.5)
\psline[linewidth=\mince](0,0)(3.2,0)
\psarc[linecolor=blue,linewidth=\elegant]{-}(1.2,0){0.2}{0}{180}
\psarc[linecolor=blue,linewidth=\elegant]{-}(2.0,0){0.2}{0}{180}
\psbezier[linecolor=blue,linewidth=\elegant]{-}(0.6,0)(0.6,0.6)(2.6,0.6)(2.6,0)
\psbezier[linecolor=blue,linewidth=\elegant]{-}(0.2,0)(0.2,1.0)(3.0,1.0)(3.0,0)
\pspolygon[fillstyle=solid,fillcolor=pink](2.1,0)(3.1,0)(3.1,-0.2)(2.1,-0.2)(2.1,0)\rput(2.6,-0.1){\scriptsize$_{3}$}
\end{pspicture}\ .
\end{equation}
Another way to understand this procedure is to note that the respective vector spaces here are those underlying
$\stan_{n,0}^d$ and $\stan_{n+d}^0$ --- these spaces
have different dimensions in general. The direct and the folded approaches are equivalent only because inserting the \WJ{} projectors effectively makes the subspace spanned by the link states connecting boundary
nodes into a trivial submodule; quotienting by this submodule recovers $\stan_{n,0}^d$
(this is similar to the analysis of \cref{sec:reps}). 

The fusion of $\kac_{n_1,0}^d$ with $\kac_{n_2,k}^0$, denoted by $\kac_{n_1,0}^d \fuse \kac_{n_2,k}^0$, 
is obtained by combining the two types of boundary seams, assuming that $d=n_1 \bmod{2}$ and $k=n_2 \bmod{2}$.  The corresponding transfer tangle is obtained by gluing the open end of a transfer tangle with an $s$-type boundary on the left and an \emph{open boundary} on the right, to that of a transfer tangle with an \emph{open boundary} on the left and an $r$-type boundary on the right.  By an open boundary, we mean one for which
the loop segments reaching it are not connected to a boundary triangle and are instead free ends. As with the boundary triangles discussed in \cref{sec:seamsandbtl}, these open boundary transfer tangles are not elements of the Temperley-Lieb or seam algebras, hence are not tangles in the strict sense that we have adopted in this article.  However, these generalised tangles will become genuine tangles if a proper diagrammatic object is glued to its open boundary, for example a boundary triangle or a second open boundary transfer tangle.  We trust that this abuse of terminology will not cause any confusion.

The transfer tangle corresponding to $\kac_{n_1,0}^d \fuse \kac_{n_2,k}^0$ acts on the quotient of $\stan_{n_1+n_2+d+k}^0$ by the subspace generated by link states with arcs tying two nodes of the same boundary. This is
equivalent to considering the action of the boundary seam algebra $\btl_{n,k}$ with $n = n_1 + n_2$
on link states with $d$ defects. The result is simply the lattice Kac module $\kac_{n,k}^d$ which
does not depend on $n_1$ and $n_2$ individually, but only on their sum. 
In general, the module $\kac_{n_1,0}^d \fuse \kac_{n_2,k}^0$
does not correspond to an $r$-type nor an $s$-type boundary, but instead to a so-called $(r,s)$-type boundary \cite{PRZ06}. 
This terminology stems from the observation that the limiting Virasoro module 
may have Kac labels $(r,s)$ with both $r$ and $s$ larger than $1$, see \eqref{eq:rs1} and \eqref{eq:rs2} below.

To illustrate such an $(r,s)$-type boundary, we note that
the action of $\boldsymbol{D}^{\textrm{\tiny$(2)$}}\!(u,\xi) \in\btl_{5,2}$ on 
the state $
\psset{unit=0.4}
\begin{pspicture}[shift=-0.05](0.4,0)(3.2,0.5)
\psline{-}(0.4,0)(2.4,0)
\psline[linecolor=purple,linewidth=2.0pt]{-}(2.4,0)(3.2,0)
\psarc[linecolor=blue,linewidth=1.0pt]{-}(1.2,0){0.2}{0}{180}
\psarc[linecolor=blue,linewidth=1.0pt]{-}(2.4,0){0.2}{0}{180}
\psline[linecolor=blue,linewidth=1.0pt]{-}(0.6,0)(0.6,0.7)
\psline[linecolor=blue,linewidth=1.0pt]{-}(1.8,0)(1.8,0.7)
\psline[linecolor=blue,linewidth=1.0pt]{-}(3.0,0)(3.0,0.7)
\end{pspicture}\in \kac_{5,2}^3$, 
up to the usual prefactor, may be described diagrammatically in three equivalent fashions:
\begin{equation}
\begin{gathered}
\psset{unit=0.75}
\begin{pspicture}[shift=-1.7](-2.0,-0.7)(8.5,4.0)
\facegrid{(-2,0)}{(8,2)}
\psarc[linewidth=0.025]{-}(1,0){0.16}{0}{90}
\psarc[linewidth=0.025]{-}(2,1){0.16}{90}{180}
\psarc[linewidth=0.025]{-}(2,0){0.16}{0}{90}
\psarc[linewidth=0.025]{-}(3,1){0.16}{90}{180}
\psarc[linewidth=0.025]{-}(3,0){0.16}{0}{90}
\psarc[linewidth=0.025]{-}(4,1){0.16}{90}{180}
\psarc[linewidth=0.025]{-}(4,0){0.16}{0}{90}
\psarc[linewidth=0.025]{-}(5,1){0.16}{90}{180}
\psarc[linewidth=0.025]{-}(5,0){0.16}{0}{90}
\psarc[linewidth=0.025]{-}(6,1){0.16}{90}{180}
\psarc[linewidth=0.025]{-}(6,0){0.16}{0}{90}
\psarc[linewidth=0.025]{-}(7,1){0.16}{90}{180}
\psarc[linewidth=0.025]{-}(7,0){0.16}{0}{90}
\psarc[linewidth=0.025]{-}(8,1){0.16}{90}{180}
\psarc[linewidth=1.5pt,linecolor=blue]{-}(-2,1){0.5}{90}{-90}
\psarc[linewidth=1.5pt,linecolor=blue]{-}(8,1){0.5}{-90}{90}
\rput(0,0.15){
\psarc[linewidth=1.5pt,linecolor=blue]{-}(3,2){0.5}{0}{180}
\psarc[linewidth=1.5pt,linecolor=blue]{-}(6,2){0.5}{0}{180}
\psarc[linewidth=1.5pt,linecolor=blue]{-}(1,2.0){0.5}{0}{180}
\psbezier[linewidth=1.5pt,linecolor=blue]{-}(-0.5,2)(-0.5,3.6)(4.5,3.6)(4.5,2)
\psbezier[linewidth=1.5pt,linecolor=blue]{-}(-1.5,2)(-1.5,4.4)(7.5,4.4)(7.5,2)
}
\multiput(1,0)(1,0){7}{\psline[linewidth=1.5pt,linecolor=blue]{-}(0.5,2.00)(0.5,2.16)}
\rput(1.5,.5){$u$}
\rput(1.5,1.5){$u$}
\rput(2.5,0.5){$u$}
\rput(2.5,1.5){$u$}
\rput(3.5,0.5){$u$}
\rput(3.5,1.5){$u$}
\rput(4.5,.5){$u$}
\rput(4.5,1.5){$u$}
\rput(5.5,.5){$u$}
\rput(5.5,1.5){$u$}
\rput(6.5,.5){\scriptsize$u\!-\!\xi_2$}
\rput(6.5,1.5){\scriptsize$u\!+\!\xi_2$}
\rput(7.5,.5){\scriptsize$u\!-\!\xi_1$}
\rput(7.5,1.5){\scriptsize$u\!+\!\xi_1$}
\pspolygon[fillstyle=solid,fillcolor=pink](6.1,0)(7.9,0)(7.9,-0.3)(6.1,-0.3)(6.1,0)\rput(7.0,-0.15){$_{2}$}
\pspolygon[fillstyle=solid,fillcolor=pink](6.1,2)(7.9,2)(7.9,2.3)(6.1,2.3)(6.1,2)\rput(7.0,2.15){$_{2}$}
\psline[linecolor=red,linewidth=2pt,linestyle=dashed,dash=3pt 3pt]{-}(6,-0.1)(6,2.1)
\psline[linecolor=black,linewidth=2pt]{-}(4,-0.1)(4,2.1)
\rput(-10,0){\pspolygon[fillstyle=solid,fillcolor=pink](8.1,0)(10.9,0)(10.9,-0.3)(8.1,-0.3)(8.1,0)\rput(9.5,-0.15){$_{3}$}
\pspolygon[fillstyle=solid,fillcolor=pink](8.1,2)(10.9,2)(10.9,2.3)(8.1,2.3)(8.1,2)\rput(9.5,2.15){$_{3}$}
\rput(9.5,-0.65){\small($s$-type)}
\rput(8,0){\loopb}\rput(8,1){\loopa}
\rput(9,0){\loopb}\rput(9,1){\loopa}
\rput(10,0){\loopb}\rput(10,1){\loopa}
\psline[linecolor=red,linewidth=2pt,linestyle=dashed,dash=3pt 3pt]{-}(11,-0.1)(11,2.1)}
\rput(3.5,-0.65){\small(bulk)}
\rput(7.0,-0.65){\small($r$-type)}
\end{pspicture}
\qquad  \longleftrightarrow \qquad
\begin{pspicture}[shift=-1.7](0.75,-0.7)(8.0,4.0)
\facegrid{(1,0)}{(8,2)}
\psarc[linewidth=0.025]{-}(1,0){0.16}{0}{90}
\psarc[linewidth=0.025]{-}(2,1){0.16}{90}{180}
\psarc[linewidth=0.025]{-}(2,0){0.16}{0}{90}
\psarc[linewidth=0.025]{-}(3,1){0.16}{90}{180}
\psarc[linewidth=0.025]{-}(3,0){0.16}{0}{90}
\psarc[linewidth=0.025]{-}(4,1){0.16}{90}{180}
\psarc[linewidth=0.025]{-}(4,0){0.16}{0}{90}
\psarc[linewidth=0.025]{-}(5,1){0.16}{90}{180}
\psarc[linewidth=0.025]{-}(5,0){0.16}{0}{90}
\psarc[linewidth=0.025]{-}(6,1){0.16}{90}{180}
\psarc[linewidth=0.025]{-}(6,0){0.16}{0}{90}
\psarc[linewidth=0.025]{-}(7,1){0.16}{90}{180}
\psarc[linewidth=0.025]{-}(7,0){0.16}{0}{90}
\psarc[linewidth=0.025]{-}(8,1){0.16}{90}{180}
\psarc[linewidth=1.5pt,linecolor=blue]{-}(1,1){0.5}{90}{-90}
\psarc[linewidth=1.5pt,linecolor=blue]{-}(8,1){0.5}{-90}{90}
\rput(0,0.15){
\psarc[linewidth=1.5pt,linecolor=blue]{-}(3,2){0.5}{0}{180}
\psarc[linewidth=1.5pt,linecolor=blue]{-}(6,2){0.5}{0}{180}
\psline[linewidth=1.5pt,linecolor=blue]{-}(1.5,2)(1.5,2.7)
\psline[linewidth=1.5pt,linecolor=blue]{-}(4.5,2)(4.5,2.7)
\psline[linewidth=1.5pt,linecolor=blue]{-}(7.5,2)(7.5,2.7)
}
\multiput(1,0)(1,0){7}{\psline[linewidth=1.5pt,linecolor=blue]{-}(0.5,2.00)(0.5,2.16)}
\rput(1.5,.5){$u$}
\rput(1.5,1.5){$u$}
\rput(2.5,0.5){$u$}
\rput(2.5,1.5){$u$}
\rput(3.5,0.5){$u$}
\rput(3.5,1.5){$u$}
\rput(4.5,.5){$u$}
\rput(4.5,1.5){$u$}
\rput(5.5,.5){$u$}
\rput(5.5,1.5){$u$}
\rput(6.5,.5){\scriptsize$u\!-\!\xi_2$}
\rput(6.5,1.5){\scriptsize$u\!+\!\xi_2$}
\rput(7.5,.5){\scriptsize$u\!-\!\xi_1$}
\rput(7.5,1.5){\scriptsize$u\!+\!\xi_1$}
\pspolygon[fillstyle=solid,fillcolor=pink](6.1,0)(7.9,0)(7.9,-0.3)(6.1,-0.3)(6.1,0)\rput(7.0,-0.15){$_{2}$}
\pspolygon[fillstyle=solid,fillcolor=pink](6.1,2)(7.9,2)(7.9,2.3)(6.1,2.3)(6.1,2)\rput(7.0,2.15){$_{2}$}
\psline[linecolor=red,linewidth=2pt,linestyle=dashed,dash=3pt 3pt]{-}(6,-0.1)(6,2.1)
\end{pspicture} \\[0.3cm]
\psset{unit=0.75}
\longleftrightarrow \quad  \
\begin{pspicture}[shift=-1.7](0.5,-0.7)(11.5,4.2)
\facegrid{(1,0)}{(11,2)}
\psarc[linewidth=0.025]{-}(1,0){0.16}{0}{90}
\psarc[linewidth=0.025]{-}(2,1){0.16}{90}{180}
\psarc[linewidth=0.025]{-}(2,0){0.16}{0}{90}
\psarc[linewidth=0.025]{-}(3,1){0.16}{90}{180}
\psarc[linewidth=0.025]{-}(3,0){0.16}{0}{90}
\psarc[linewidth=0.025]{-}(4,1){0.16}{90}{180}
\psarc[linewidth=0.025]{-}(4,0){0.16}{0}{90}
\psarc[linewidth=0.025]{-}(5,1){0.16}{90}{180}
\psarc[linewidth=0.025]{-}(5,0){0.16}{0}{90}
\psarc[linewidth=0.025]{-}(6,1){0.16}{90}{180}
\psarc[linewidth=0.025]{-}(6,0){0.16}{0}{90}
\psarc[linewidth=0.025]{-}(7,1){0.16}{90}{180}
\psarc[linewidth=0.025]{-}(7,0){0.16}{0}{90}
\psarc[linewidth=0.025]{-}(8,1){0.16}{90}{180}
\rput(8,0){\loopa}\rput(8,1){\loopb}
\rput(9,0){\loopa}\rput(9,1){\loopb}
\rput(10,0){\loopa}\rput(10,1){\loopb}
\psarc[linewidth=1.5pt,linecolor=blue]{-}(1,1){0.5}{90}{-90}
\psarc[linewidth=1.5pt,linecolor=blue]{-}(11,1){0.5}{-90}{90}
\rput(0,0.15){
\psarc[linewidth=1.5pt,linecolor=blue]{-}(3,2){0.5}{0}{180}
\psarc[linewidth=1.5pt,linecolor=blue]{-}(6,2){0.5}{0}{180}
\psarc[linewidth=1.5pt,linecolor=blue]{-}(8,2.0){0.5}{0}{180}
\psbezier[linewidth=1.5pt,linecolor=blue]{-}(4.5,2)(4.5,3.6)(9.5,3.6)(9.5,2)
\psbezier[linewidth=1.5pt,linecolor=blue]{-}(1.5,2)(1.5,4.4)(10.5,4.4)(10.5,2)
}
\multiput(1,0)(1,0){10}{\psline[linewidth=1.5pt,linecolor=blue]{-}(0.5,2.00)(0.5,2.16)}
\rput(1.5,.5){$u$}
\rput(1.5,1.5){$u$}
\rput(2.5,0.5){$u$}
\rput(2.5,1.5){$u$}
\rput(3.5,0.5){$u$}
\rput(3.5,1.5){$u$}
\rput(4.5,.5){$u$}
\rput(4.5,1.5){$u$}
\rput(5.5,.5){$u$}
\rput(5.5,1.5){$u$}
\rput(6.5,.5){\scriptsize$u\!-\!\xi_2$}
\rput(6.5,1.5){\scriptsize$u\!+\!\xi_2$}
\rput(7.5,.5){\scriptsize$u\!-\!\xi_1$}
\rput(7.5,1.5){\scriptsize$u\!+\!\xi_1$}
\pspolygon[fillstyle=solid,fillcolor=pink](6.1,0)(7.9,0)(7.9,-0.3)(6.1,-0.3)(6.1,0)\rput(7.0,-0.15){$_{2}$}
\pspolygon[fillstyle=solid,fillcolor=pink](6.1,2)(7.9,2)(7.9,2.3)(6.1,2.3)(6.1,2)\rput(7.0,2.15){$_{2}$}
\psline[linecolor=red,linewidth=2pt,linestyle=dashed,dash=3pt 3pt]{-}(6,-0.1)(6,2.1)
\psline[linecolor=red,linewidth=2pt,linestyle=dashed,dash=3pt 3pt]{-}(8,-0.1)(8,2.1)
\pspolygon[fillstyle=solid,fillcolor=pink](8.1,0)(10.9,0)(10.9,-0.3)(8.1,-0.3)(8.1,0)\rput(9.5,-0.15){$_{3}$}
\pspolygon[fillstyle=solid,fillcolor=pink](8.1,2)(10.9,2)(10.9,2.3)(8.1,2.3)(8.1,2)\rput(9.5,2.15){$_{3}$}
\rput(3.5,-0.65){\small(bulk)}
\rput(7.0,-0.65){\small($r$-type)}
\rput(9.5,-0.65){\small($s$-type)}
\end{pspicture} \ \  .
\end{gathered}
\end{equation}
The first diagram is an explicit evaluation of $\kac_{3,0}^3 \fuse \kac_{2,2}^0$ as defined above. The division of the bulk between $n_1=3$ and $n_2=2$ is indicated, in this diagram, by a thick vertical line.  We do not indicate this separation explicitly in the subsequent diagrams because the result only depends upon $n_1 = 3$ and $n_2 = 2$ through their sum, as mentioned above.
The second diagram gives 
the equivalent computation in $\kac_{5,2}^3$. The last diagram illustrates the result in terms of a transfer tangle with a single non-trivial boundary condition, interpreted as the fusion of an $r$-type and an $s$-type.
The limiting Virasoro module has $r$ given by the \emph{same} function of $p$ and $p'$ as in the example \eqref{eq:diagrtype}, while $s = d+1 = 4$.
The fusion rule 
\begin{equation} 
\kac_{n_1,0}^d \fuse \kac_{n_2,k}^0 = \kac_{n,k}^d, \qquad n_1 + n_2 = n,
\label{eq:latticersfusion}
\end{equation}
thus holds automatically with this lattice prescription. 
Checking that the same rule holds for the limiting Virasoro 
modules, see \eqref{FR:KxK}, not only confirms the consistency of this
lattice implementation of Virasoro fusion, but also provides evidence for our identification, in \cref{sec:scalinglimit}, of 
the scaling limits of lattice Kac modules with Virasoro Kac modules.

More complicated fusion products, such as the fusion $\kac_{n_1,k_1}^{d_1}\fuse\kac_{n_2,k_2}^{d_2}$
of two general lattice Kac modules, can also be implemented on the lattice. 
To do this, one includes seams on both sides of the bulk \cite{PRZ06,PR07,PRannecy}.
For instance,
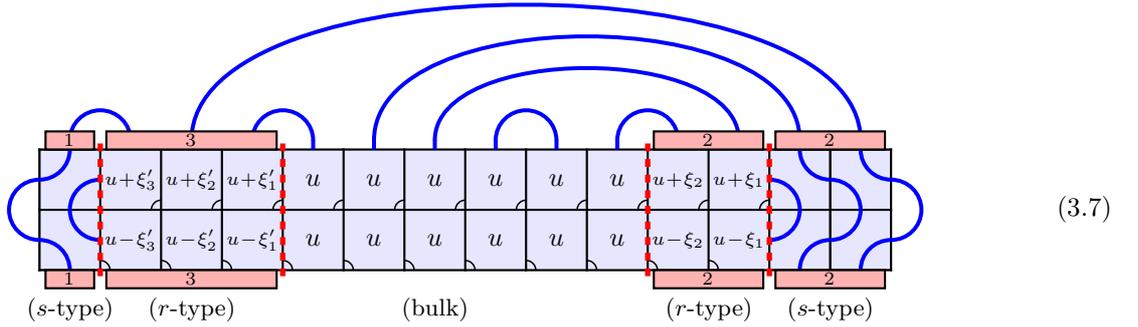
\begin{equation}
\psset{unit=0.8}
\begin{pspicture}[shift=-1.7](-3.5,-0.8)(11.5,4.5)
\facegrid{(-3,0)}{(11,2)}
\psarc[linewidth=0.025]{-}(-2,0){0.16}{0}{90}
\psarc[linewidth=0.025]{-}(-1,1){0.16}{90}{180}
\psarc[linewidth=0.025]{-}(-1,0){0.16}{0}{90}
\psarc[linewidth=0.025]{-}(0,1){0.16}{90}{180}
\psarc[linewidth=0.025]{-}(0,0){0.16}{0}{90}
\psarc[linewidth=0.025]{-}(1,1){0.16}{90}{180}
\psarc[linewidth=0.025]{-}(1,0){0.16}{0}{90}
\psarc[linewidth=0.025]{-}(2,1){0.16}{90}{180}
\psarc[linewidth=0.025]{-}(2,0){0.16}{0}{90}
\psarc[linewidth=0.025]{-}(3,1){0.16}{90}{180}
\psarc[linewidth=0.025]{-}(3,0){0.16}{0}{90}
\psarc[linewidth=0.025]{-}(4,1){0.16}{90}{180}
\psarc[linewidth=0.025]{-}(4,0){0.16}{0}{90}
\psarc[linewidth=0.025]{-}(5,1){0.16}{90}{180}
\psarc[linewidth=0.025]{-}(5,0){0.16}{0}{90}
\psarc[linewidth=0.025]{-}(6,1){0.16}{90}{180}
\psarc[linewidth=0.025]{-}(6,0){0.16}{0}{90}
\psarc[linewidth=0.025]{-}(7,1){0.16}{90}{180}
\psarc[linewidth=0.025]{-}(7,0){0.16}{0}{90}
\psarc[linewidth=0.025]{-}(8,1){0.16}{90}{180}
\psarc[linewidth=0.025]{-}(8,0){0.16}{0}{90}
\psarc[linewidth=0.025]{-}(9,1){0.16}{90}{180}
\rput(-3,0){\loopb}\rput(-3,1){\loopa}
\rput(9,0){\loopa}\rput(9,1){\loopb}
\rput(10,0){\loopa}\rput(10,1){\loopb}
\psarc[linewidth=1.5pt,linecolor=blue]{-}(-3,1){0.5}{90}{-90}
\psarc[linewidth=1.5pt,linecolor=blue]{-}(11,1){0.5}{-90}{90}
\rput(0,0.15){
\psarc[linewidth=1.5pt,linecolor=blue]{-}(-2,2){0.5}{0}{180}
\psarc[linewidth=1.5pt,linecolor=blue]{-}(1,2){0.5}{0}{180}
\psarc[linewidth=1.5pt,linecolor=blue]{-}(5,2){0.5}{0}{180}
\psarc[linewidth=1.5pt,linecolor=blue]{-}(7,2.0){0.5}{0}{180}
\psbezier[linewidth=1.5pt,linecolor=blue]{-}(3.5,2)(3.5,3.6)(8.5,3.6)(8.5,2)
\psbezier[linewidth=1.5pt,linecolor=blue]{-}(2.5,2)(2.5,4.3)(9.5,4.3)(9.5,2)
\psbezier[linewidth=1.5pt,linecolor=blue]{-}(-0.5,2)(-0.5,5.0)(10.5,5.0)(10.5,2)
}
\multiput(-3,0)(1,0){14}{\psline[linewidth=1.5pt,linecolor=blue]{-}(0.5,2.00)(0.5,2.16)}
\rput(-1.5,.5){\scriptsize$u\!-\!\xi'_3$}
\rput(-1.5,1.5){\scriptsize$u\!+\!\xi'_3$}
\rput(-0.5,.5){\scriptsize$u\!-\!\xi'_2$}
\rput(-0.5,1.5){\scriptsize$u\!+\!\xi'_2$}
\rput(0.5,.5){\scriptsize$u\!-\!\xi'_1$}
\rput(0.5,1.5){\scriptsize$u\!+\!\xi'_1$}
\rput(1.5,0.5){$u$}
\rput(1.5,1.5){$u$}
\rput(2.5,0.5){$u$}
\rput(2.5,1.5){$u$}
\rput(3.5,0.5){$u$}
\rput(3.5,1.5){$u$}
\rput(4.5,.5){$u$}
\rput(4.5,1.5){$u$}
\rput(5.5,.5){$u$}
\rput(5.5,1.5){$u$}
\rput(6.5,.5){$u$}
\rput(6.5,1.5){$u$}
\rput(7.5,.5){\scriptsize$u\!-\!\xi_2$}
\rput(7.5,1.5){\scriptsize$u\!+\!\xi_2$}
\rput(8.5,.5){\scriptsize$u\!-\!\xi_1$}
\rput(8.5,1.5){\scriptsize$u\!+\!\xi_1$}
\rput(1,0){\pspolygon[fillstyle=solid,fillcolor=pink](6.1,0)(7.9,0)(7.9,-0.3)(6.1,-0.3)(6.1,0)\rput(7.0,-0.15){$_{2}$}
\pspolygon[fillstyle=solid,fillcolor=pink](6.1,2)(7.9,2)(7.9,2.3)(6.1,2.3)(6.1,2)\rput(7.0,2.15){$_{2}$}}
\psline[linecolor=red,linewidth=2pt,linestyle=dashed,dash=3pt 3pt]{-}(-2,-0.1)(-2,2.1)
\psline[linecolor=red,linewidth=2pt,linestyle=dashed,dash=3pt 3pt]{-}(1,-0.1)(1,2.1)
\psline[linecolor=red,linewidth=2pt,linestyle=dashed,dash=3pt 3pt]{-}(7,-0.1)(7,2.1)
\psline[linecolor=red,linewidth=2pt,linestyle=dashed,dash=3pt 3pt]{-}(9,-0.1)(9,2.1)
\pspolygon[fillstyle=solid,fillcolor=pink](9.1,0)(10.9,0)(10.9,-0.3)(9.1,-0.3)(9.1,0)\rput(10,-0.15){$_{2}$}
\pspolygon[fillstyle=solid,fillcolor=pink](9.1,2)(10.9,2)(10.9,2.3)(9.1,2.3)(9.1,2)\rput(10,2.15){$_{2}$}
\rput(-10,0){\pspolygon[fillstyle=solid,fillcolor=pink](8.1,0)(10.9,0)(10.9,-0.3)(8.1,-0.3)(8.1,0)\rput(9.5,-0.15){$_{3}$}}
\rput(-10,2.3){\pspolygon[fillstyle=solid,fillcolor=pink](8.1,0)(10.9,0)(10.9,-0.3)(8.1,-0.3)(8.1,0)\rput(9.5,-0.15){$_{3}$}}
\pspolygon[fillstyle=solid,fillcolor=pink](-2.9,0)(-2.1,0)(-2.1,-0.3)(-2.9,-0.3)\rput(-2.5,-0.15){$_{1}$}
\pspolygon[fillstyle=solid,fillcolor=pink](-2.9,2)(-2.1,2)(-2.1,2.3)(-2.9,2.3)\rput(-2.5,2.15){$_{1}$}
\rput(-2.5,-0.65){\small($s$-type)}
\rput(-0.5,-0.65){\small($r$-type)}
\rput(3.5,-0.65){\small(bulk)}
\rput(8,-0.65){\small($r$-type)}
\rput(10,-0.65){\small($s$-type)}
\end{pspicture}
\label{eq:2bdyexample}
\end{equation}
is a typical example of the action in $\kac_{n_1,3}^{1} \fuse \kac_{n_2,2}^{2}$, with $n = n_1+n_2 = 6$,
of a transfer tangle with such boundary structures.
This description, with defects folded to the boundary, may also be understood by
quotienting out a trivial submodule. In this case, the transfer tangle is defined with a boundary triangle on each side and is described using a two-boundary version of the seam algebra, which for $\beta$ formal is 
initially constructed as a subalgebra of $\tl_{n_1+n_2+k_1+k_2}$.
The Temperley-Lieb modules associated to $\kac_{n_1,k_1}^{d_1} \fuse \kac_{n_2,k_2}^{d_2}$
separate the defects into two families, say I and II, respectively folded to the left and right in the example \eqref{eq:2bdyexample}, and are therefore not standard modules.
Indeed, the two projectors in the $s$-type boundaries do not ensure that the total number of defects is preserved, but instead impose that $d_{\text{I}} - d_{\text{II}}$ is constant. 
These $\tl_n$-modules are more complicated and also appear, perhaps unsurprisingly, 
in the work of Gainutdinov and Vasseur \cite{GaiLat11} on the fusion of Temperley-Lieb representations. The fusion products $\kac_{n_1,k_1}^{d_1} \fuse \kac_{n_2,k_2}^{d_2}$ are beyond the scope of our current analysis.

%
\subsection{Virasoro Kac modules}\label{sec:Kac}
%

\cref{sec:boundaryTLs,sec:latticeKacmod} described lattice loop models and their lattice Kac modules in rigorous algebraic terms. In this section, we turn to their scaling limits.

\subsubsection{Definition of Virasoro Kac modules}\label{sec:VKac}

As will be argued in \cref{sec:scalinglimit}, the Virasoro Kac modules are the scaling limits 
of the lattice-theoretic Kac modules.
These Virasoro modules may be interpreted as defining boundary sectors of the logarithmic minimal models. The definition of these Virasoro Kac modules requires the well known Feigin-Fuchs modules, reviewed in \cref{sec:FFm}.  In this section, we will freely use the terminology introduced there.

\begin{defn} \label{def:VKM}
The \emph{Virasoro Kac module} $\Kac{r,s}$, with $r,s \in \ZZ_+$, is the submodule 
of the \FFm{} $\FF{r,s}$ generated by the subsingular vectors of grade strictly less than $rs$.
\end{defn}

For any given $(p,p') \in \ZZ_+^2$, the submodule structure of $\Kac{r,s}$ depends on whether the Kac label $(r,s)$ corresponds to a corner, boundary or interior entry of the Kac table. The possible structures are displayed in \cref{fig:VirKacStructures}, with many explicit examples given in \cref{sec:Gramdata}. In all cases, $\Kac{r,s}$ has finitely many composition factors, each of which may be associated to a unique \ssv{} (more precisely, to a unique equivalence class of \ssvs{}).

If $(r,s)$ is a corner entry, then the structure of $\Kac{r,s}$ is represented by the 
\emph{islands} diagram of \cref{fig:VirKacStructures}. These Virasoro Kac modules are completely reducible, decomposing
as a direct sum of irreducible modules. In the special case where $r=p$ or $s = p'$, the Kac module
$\Kac{r,s}$ is actually irreducible and is therefore represented by the 
\emph{point} diagram. 

If $(r,s)$ is a boundary entry, then $\Kac{r,s}$ is indecomposable and its structure is given by the 
\emph{chain} diagram. With the conventions we use (see \cref{sec:FFm}), the corresponding \FFm{} $\FF{r,s}$, with $r,s>0$, always possesses a (non-singular) subsingular vector at grade $rs$, so our pictures of the chain-type Kac modules in \cref{fig:VirKacStructures} always have the lowest arrow pointing down. In particular, chain-type modules with two composition factors are highest-weight modules isomorphic to the quotient of $\Ver{\Delta_{r,s}}$ by $\Ver{\Delta_{r,s}+rs}$. 
This occurs for boundary entries $(r,s)$ precisely when $r=p$ and $s>p'$, or $r>p$ and $s=p'$.
On the other hand, if $r<p$ or $s<p'$, then the boundary $\Kac{r,s}$ is again point-type and irreducible.

If $(r,s)$ is an interior entry, then $\Kac{r,s}$ is indecomposable and is represented by the 
\emph{braid} diagram. Unlike the island and chain cases, braid-type Kac modules may contain a singular vector at a grade larger than $rs$. 
For $r<p$ or $s<p'$, the braid-type Kac modules have two composition factors and $\Kac{r,s}$ is a \hwm{} isomorphic to $\Ver{\Delta_{r,s}}/\Ver{\Delta_{r,s}+rs}$.

By inspection of \cref{fig:VirKacStructures,fig:VermaStructures}, it is easy to see that $\Kac{r,s}$ and $\Ver{\Delta_{r,s}}/\Ver{\Delta_{r,s}+rs}$ always share the same composition factors. They therefore share the same character,
\begin{equation} 
\chit_{r,s}(q) = q^{-c/24} \frac{q^{\Delta_{r,s}}(1-q^{rs})}{\prod_{j=1}^\infty (1-q^j)},
\label{eq:chirs}
\end{equation}
but are not isomorphic in general. 

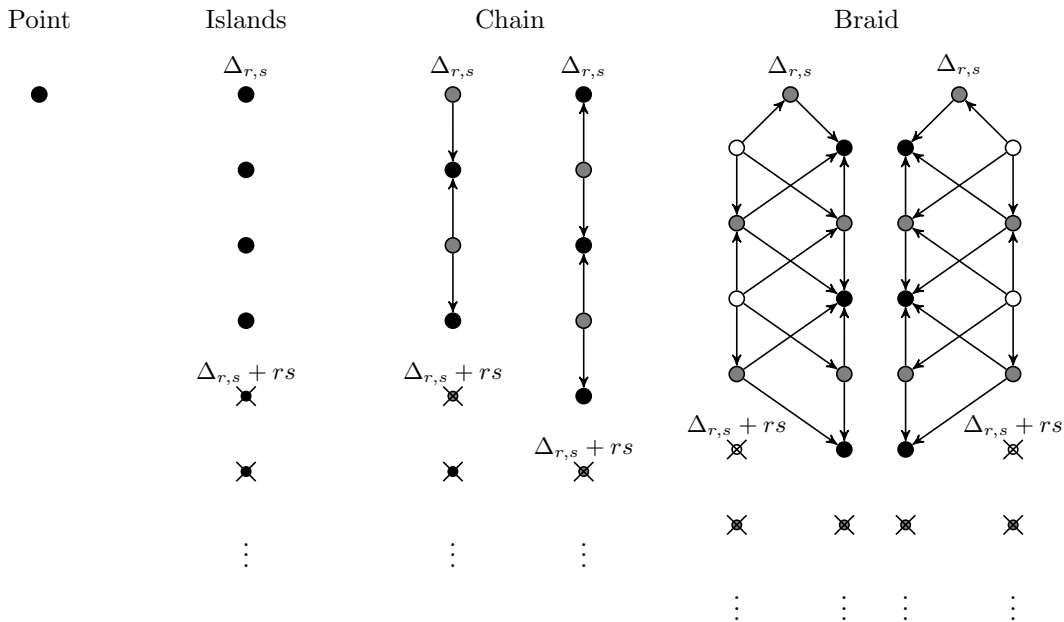
\begin{figure}
\begin{center}
\begin{tikzpicture}
  [->,node distance=1cm,>=stealth',semithick,scale=0.7,
   asoc/.style={circle,draw=black,fill=black,inner sep=2pt,minimum size=5pt},
   bsoc/.style={circle,draw=black,fill=gray,inner sep=2pt,minimum size=5pt},
   csoc/.style={circle,draw=black,fill=white,inner sep=2pt,minimum size=5pt},
   asocS/.style={circle,draw=black,fill=black,inner sep = 1.25pt,minimum size=3pt},
   bsocS/.style={circle,draw=black,fill=gray,inner sep = 1.25pt,minimum size=3pt},
   csocS/.style={circle,draw=black,fill=white,inner sep = 1.25pt,minimum size=3pt}
  ]
  \node[asoc] (1) [] {};
  \node[] (point) [above of =1] {Point};
  \node[asoc] (2) [right = 2.5cm of 1] {};
  \node[] (island) [above of =2] {Islands};
  \node[asoc] (2a) [below of =2] {};
  \node[asoc] (2b) [below of =2a] {};
  \node[asoc] (2c) [below of =2b] {};
  \node[asocS] (2d) [below of =2c] {};
  \node[-,cross out,minimum size=5pt,draw] at (2d) {};
  \node[asocS] (2e) [below of =2d] {};
  \node[-,cross out,minimum size=5pt,draw] at (2e) {};
  \node[inner sep = 2pt] (2f) [below of =2e] {$\vdots$};
  \node[] [above=-0.03cm of 2] {\small $\Delta_{r,s}$};
  \node[] [above=-0.05cm of 2d] {\small $\Delta_{r,s}+rs$};
  \node[bsoc] (3) [right = 2.5cm of 2] {};
  \node[] (chain) [right = 0.75cm of 3,above of =3] {Chain};
  \node[asoc] (3a) [below of =3] {};
  \node[bsoc] (3b) [below of =3a] {};
  \node[asoc] (3c) [below of =3b] {};
  \node[bsocS] (3d) [below of =3c] {};
  \node[-,cross out,minimum size=5pt,draw] at (3d) {};
  \node[asocS] (3e) [below of =3d] {};
  \node[-,cross out,minimum size=5pt,draw] at (3e) {};
  \node[inner sep = 2pt] (3f) [below of =3e] {$\vdots$};
  \node[] [above=-0.03cm of 3] {\small $\Delta_{r,s}$};
  \node[] [above=-0.05cm of 3d] {\small $\Delta_{r,s}+rs$};
  \path[] (3) edge node {} (3a)
          (3b) edge node {} (3a)
          (3b) edge node {} (3c);
  \node[asoc] (4) [right = 1.5cm of 3] {};
  \node[bsoc] (4a) [below of =4] {};
  \node[asoc] (4b) [below of =4a] {};
  \node[bsoc] (4c) [below of =4b] {};
  \node[asoc] (4d) [below of =4c] {};
  \node[bsocS] (4e) [below of =4d] {};
  \node[-,cross out,minimum size=5pt,draw] at (4e) {};
  \node[inner sep = 2pt] (4f) [below of =4e] {$\vdots$};
  \node[] [above=-0.03cm of 4] {\small $\Delta_{r,s}$};
  \node[] [above=-0.05cm of 4e] {\small $\Delta_{r,s}+rs$};
  \path[] (4a) edge node {} (4)
          (4a) edge node {} (4b)
          (4c) edge node {} (4b)
          (4c) edge node {} (4d);
  \node[bsoc] (5) [right = 2.5cm of 4] {};
  \node[] (braid) [right = 1cm of 5,above of =5] {Braid};
    \node[csoc] (5a) [below left of =5] {};
  \node[bsoc] (5b) [below of =5a] {};
  \node[csoc] (5c) [below of =5b] {};
  \node[bsoc] (5d) [below of =5c] {};
  \node[csocS] (5e) [below of =5d] {};
  \node[-,cross out,minimum size=5pt,draw] at (5e) {};
  \node[bsocS] (5f) [below of =5e] {};    
  \node[-,cross out,minimum size=5pt,draw] at (5f) {};
  \node[inner sep = 2pt] (5g) [below of =5f] {$\vdots$};
  \node[asoc] (5j) [below right of =5] {};
  \node[bsoc] (5k) [below of =5j] {};
  \node[asoc] (5l) [below of =5k] {};
  \node[bsoc] (5m) [below of =5l] {};
  \node[asoc] (5n) [below of =5m] {};
  \node[bsocS] (5o) [below of =5n] {};    
  \node[-,cross out,minimum size=5pt,draw] at (5o) {};
  \node[inner sep = 2pt] (5p) [below of =5o] {$\vdots$};
  \path[] (5) edge node {} (5j)
          (5k) edge node {} (5j)
          (5b) edge node {} (5j)
          (5k) edge node {} (5l)
          (5a) edge node {} (5k)
          (5c) edge node {} (5k)
          (5m) edge node {} (5l)
          (5b) edge node {} (5l)
          (5d) edge node {} (5l)
          (5m) edge node {} (5n)
          (5c) edge node {} (5m)
          (5d) edge node {} (5n)
          (5a) edge node {} (5)
          (5a) edge node {} (5b)
          (5c) edge node {} (5b)
          (5c) edge node {} (5d);
  \node[] [above=-0.03cm of 5] {\small $\Delta_{r,s}$};
  \node[] [above=-0.05cm of 5e] {\small $\Delta_{r,s}+rs$};
  \node[bsoc] (6) [right = 2cm of 5] {};
  \node[asoc] (6a) [below left of =6] {};
  \node[bsoc] (6b) [below of =6a] {};
  \node[asoc] (6c) [below of =6b] {};
  \node[bsoc] (6d) [below of =6c] {};
  \node[asoc] (6e) [below of =6d] {};
  \node[bsocS] (6f) [below of =6e] {};
  \node[-,cross out,minimum size=5pt,draw] at (6f) {};
  \node[inner sep = 2pt] (6g) [below of =6f] {$\vdots$};
  \node[csoc] (6j) [below right of =6] {};
  \node[bsoc] (6k) [below of =6j] {};
  \node[csoc] (6l) [below of =6k] {};
  \node[bsoc] (6m) [below of =6l] {};
  \node[csocS] (6n) [below of =6m] {};
  \node[-,cross out,minimum size=5pt,draw] at (6n) {};
  \node[bsocS] (6o) [below of =6n] {};
  \node[-,cross out,minimum size=5pt,draw] at (6o) {};
  \node[inner sep = 2pt] (6p) [below of =6o] {$\vdots$};
  \path[] (6) edge node {} (6a)
          (6b) edge node {} (6a)
          (6k) edge node {} (6a)
          (6b) edge node {} (6c)
          (6j) edge node {} (6b)
          (6l) edge node {} (6b)
          (6d) edge node {} (6c)
          (6k) edge node {} (6c)
          (6m) edge node {} (6c)
          (6d) edge node {} (6e)
          (6l) edge node {} (6d)
          (6m) edge node {} (6e)
          (6j) edge node {} (6)
          (6j) edge node {} (6k)
          (6l) edge node {} (6k)
          (6l) edge node {} (6m);
  \node[] [above=-0.03mm of 6] {\small $\Delta_{r,s}$};
  \node[] [above=-0.05cm of 6n] {\small $\Delta_{r,s}+rs$};            
\end{tikzpicture}
\end{center}
\caption{The structure of Virasoro Kac modules $\Kac{r,s}$. Black, grey and white circles represent states
in the first, second and third socle layers, respectively.  Black circles represent \svs{}, while the other colours represent 
vectors that are only subsingular in general.  The only exceptions to this are the top circles which always represent 
a \sv{}. The arrows connecting these circles indicate the action of the Virasoro algebra. 
Smaller, crossed out circles represent \ssvs{}
that appear in the Feigin-Fuchs module $\FF{r,s}$, but not in $\Kac{r,s}$. Note that the two braid diagrams that we have drawn are structurally identical.  As with \cref{fig:FeiginFuchsStructures}, this repetition serves to indicate that the corresponding Kac modules are not self-contragredient.} 
\label{fig:VirKacStructures}
\end{figure}

\subsubsection{Continuum scaling limits of lattice Kac modules}\label{sec:scalinglimit}

The lattice models that we are interested in are formulated as a family of models labelled by the system size $n$. Their evolution operators (transfer tangles and Hamiltonians) act on modules, the lattice Kac modules $\kac_{n,k}^d$, whose dimensions, for fixed $k$ and $d$, grow with $n$. For simplicity, let us denote by $\Mtl_n$ this family of modules.  As $n$ increases, one can define the infinite sequence of $i$-th eigenvalues $\Lambda_n^{i}$ of the matrices representing the chosen evolution operator, each one belonging to a different module $\Mtl_n$. This is ensured by specialising $\beta$, $u$ and $\xi$ to fixed real values --- 
the spectra of $\Dbk(u,\xi)$ and $\hamk$ are then real, according to \cref{conj1}. 
One can, for example, focus on the eigenvalues of the ground state, for each $n$, or the first excited state, the second excited state, or even the maximally excited state. Note that we allow sequences $\Lambda_n^i$ where the label $i=i(n)$
is a function of $n$, as is required for the maximally excited state.  The aim is to study the behaviour of such sequences of eigenvalues in the scaling limit where $n$ tends to infinity. We note that the lattice Kac modules with $n$ odd and even have different allowed defect numbers $d$, which results in differing conformal properties in general. 
The two parities are therefore treated as separate sequences.

For lattice models that are believed to become conformally invariant as $n \rightarrow \infty$, 
the (sequences of) eigenvalues of the evolution operators for the lowest 
excited states have particular $1/n$ expansions.
In this limit, the corresponding sequences of eigenstates are believed to become the states of a module $\Mtl$ over the 
Virasoro algebra. The difference between the conformal weight of each such
state and that of the conformal ground state is then finite. 
We call sequences of eigenstates with this property \emph{conformal sequences}. 
The eigenstates of such sequences are often referred to, in the literature, as \emph{finite excitations}.  
In contrast, the sequences of eigenvalues for highly excited states typically drift off to infinity 
(this will be defined formally below). They do not give rise to states in the Virasoro module and are referred to as
\emph{non-conformal sequences}. 
While we shall make this more precise shortly, this distinction between conformal and non-conformal sequences is the 
essence of the \emph{continuum scaling limit}.
Only the conformal sequences survive. 
We emphasise that when boundary seams of width $k$ are involved, in this limit, the bulk size $n$ is taken to $\infty$ while $k$ remains fixed. 

Our next goal is to explore the scaling limit of the lattice Kac modules,
examining the so-called finite-size corrections.
For the eigenvalues $D_0 \ge D_1 \ge D_2 \ge \cdots$ of $\Dbk(u,\xi)$, the aforementioned $1/n$ expansion takes the form \cite{BloCon86,AffUni86}
\begin{equation}
-\log D_i =  2n f_{bulk}(u) +  f_{bdy}(u,\xi) + \frac{2 \pi \sin(\frac{\pi u}{\lambda})}{n} \Big(\!-\frac c {24} + \Delta + j_i \Big) + O(n^{-2}).    
\label{eq:Dexpansion}
\end{equation} 
Here, $f_{bulk}(u)$ and $f_{bdy}(u,\xi)$ are the bulk and the boundary free energies (for which expressions for the 
logarithmic minimal models were given in \cite{PRZ06,PTC14}), 
$c$ is the central charge of the \cft{} describing the scaling limit, $\Delta$ is a conformal weight, and $j_i$ is an integer.  
The ground state, in the sector labelled by $\Delta$, corresponds to $i=0$ and $j_0 = 0$. 
For the logarithmic minimal model $\mathcal {LM}(p,p')$, the central charge and conformal weights are 
conjectured \cite{PRZ06} to be given by the familiar expressions
\begin{equation}
c = 1 - 6\frac{(p'-p)^2}{pp'}, \qquad \Delta = \Delta_{r,s} = \frac{(p'r-ps)^2-(p-p')^2}{4pp'} \quad \text{(\(r,s \in \mathbb Z_+\)).}
\label{eq:CDelta}
\end{equation} 
Evidence supporting this conjecture includes the numerical estimation of conformal weights and characters from lattice 
data \cite{PRZ06,PTC14} and the consistency of the lattice prescription for fusion with the conformal fusion rules \cite{RasFus07,RasFus07a,RasCla11}. 

The first non-trivial term in the $u$-expansion of $\Dbk(u,\xi)$
is the Hamiltonian $\hamk$ (up to a constant multiple of the identity). The asymptotic expansion for the eigenvalues $H_0 \le H_1 \le H_2 \le \cdots$ of the Hamiltonian takes the form
\begin{equation}\label{eq:Hj}
H_i = n\, h_{bulk} + h_{bdy}(\xi) + \frac{\pi v_s}n \Big(-\frac c {24} + \Delta + j_i \Big) + O(n^{-2}),
\end{equation}
where $h_{bulk}$ and $h_{bdy}$ are the Hamiltonian bulk and boundary free energies, and $v_s$ is the speed of 
sound defined below \eqref{eq:modes}. 
This is in accordance with $\hamk$ becoming the first integral of motion: 
\begin{equation}
\frac{n}{\pi v_s} \Big(\hamk - n\, h_{bulk} - h_{bdy} \Big) \xrightarrow{n \rightarrow \infty} \iom_1 = L_0 - \frac c{24}.
\label{eq:HL0}
\end{equation} 

These eigenvalue expansions can be used to give a more precise definition of conformal and non-conformal sequences. In what follows, we shall denote by $\Lambda_n^{0}$ 
the eigenvalues (for $-\log \Dbk(u,\xi)$ or $\hamk$) on $\kac_{n,k}^d$ (for fixed $k$ and $d$) corresponding to the 
(possibly degenerate) ground state. 
Let $\Lambda_n^i$ denote the eigenvalues of another sequence that we want to study. Then, we will say that this sequence is \emph{conformal} if 
\begin{equation}
\lim_{n \ra \infty} i(n) = \iota \quad \text{and} \quad
\lim_{n \rightarrow \infty} n\,(\Lambda_n^i - \Lambda_n^{0}) = 
\kappa_{\iota},
\qquad \text{for some \(
\kappa_{\iota} < \infty\).}
\label{eq:limitkappa}
\end{equation}
For finite excitations of the ground state, \eqref{eq:Dexpansion} and \eqref{eq:Hj} indicate that 
$\kappa_{\iota}/\kappa_1$ should equal the ratio $j_{\iota}/j_1$ of the grades of the excited states, where $\kappa_1$ and $j_1$ are associated to the (sequence of the) first excited state. 
\eqref{eq:Dexpansion} and \eqref{eq:Hj} also explain the factor of $n$ appearing in \eqref{eq:limitkappa}.
For \emph{non-conformal sequences}, however, at least one of the limits in \eqref{eq:limitkappa} 
diverges. 
If $\iota = \infty$ or $\kappa_{\iota} = \infty$, the sequence is said to \emph{drift off to infinity}. 
To avoid over-counting the conformal sequences of the Kac modules, we will identify all sequences, in the scaling limit, that have the same limiting value $\iota$, thinking of them as converging to the same conformal state.

\medskip

The word \emph{logarithmic} is used to characterise the models $\mathcal{LM}(p,p')$ because some of the correlation functions are believed to have logarithmic singularities in the scaling limit. This logarithmic dependence appears when the corresponding conformal field theory has representations on which the action of $L_0$ is non-diagonalisable. Logarithmic conformal field theories are non-unitary and necessarily involve reducible yet indecomposable Virasoro representations. Even though the appearance of such representations and the non-diagonalisability of $L_0$ are not equivalent, the former is often taken as an alternative defining characteristic of logarithmic conformal field theories.

A key observation first made by Pasquier and Saleur \cite{PS90} is that the indecomposable structures of certain 
Virasoro representations are often already present on the lattice 
in the representations of the diagrammatic algebra. 
Because $\Dbk(u,\xi)$ is initially defined as a tangle of (the specialised algebra) $\tl_{n+k}(\beta)$, it might seem 
surprising that the complicated Virasoro structures conjectured, for example 
in \cite{RasCla11}, are more intricate than those of the original \TL{} algebra (see \cref{app:TLrep}). Indeed, the 
standard modules over $\tl_n(\beta)$ never have more than two composition factors, while those over the Virasoro 
algebra that describe the scaling limit of the lattice Kac modules can have many more, see \cref{fig:VirKacStructures}. 

From our description of the lattice Kac modules $\kac_{n,k}^d$ in terms of the boundary seam algebras $\btl_{n,k}(\beta)$, it is expected that the rich indecomposable structures of their limiting Virasoro modules are inherited from similar structures for $\btl_{n,k}(\beta)$-modules. Because $\btl_{n,k}(\beta)$ is a quotient of the one-boundary \TL{} algebra $\tlone_n(\beta,U_{k-1},U_k)$, every $\btl_{n,k}(\beta)$-module is naturally a $\tlone_n(\beta,U_{k-1},U_k)$-module.  While this 
inclusion of modules does not necessarily preserve structure, it is worth noting that the standard $\tlone_n(\beta,U_{k-1},U_k)$-modules are known to admit more involved indecomposable structures than $\tl_n(\beta)$-modules \cite{MW00}. Without unravelling the full representation theory of $\btl_{n,k}(\beta)$, we will, in \cref{sec:Gramidentification,sec:Gramdata}, probe the structure of the scaling limits of the lattice Kac modules, as Virasoro modules, using the invariant bilinear form defined in \cref{sec:Gram}. This allows us to go beyond the standard character arguments, see \cref{sec:characters}, that have previously appeared in the literature 
and to arrive at \cref{TheConjecture} below.

Our analysis will use the Hamiltonian $\hamk$, instead of the full transfer tangle $\Dbk(u,\xi)$, for three reasons: (i) it 
allows us to ignore $u$ and reduce the parameter space to only $(\beta,\xi)$, (ii) $\hamk$ is believed to converge to 
$L_0 - c/24$ in the scaling limit and we do not use the additional information encoded in the higher
integrals of motion, and (iii) working with $\hamk$ only is less demanding computationally and allows us to reach larger 
system sizes and thus larger precision.  This last point is perhaps the most important one. 
One advantage of analysing logarithmic minimal models through transfer tangles, if one is not interested in the higher integrals of motion, is that their eigenvalues satisfy functional hierarchy equations \cite{MDPR14} which 
one can hope to solve analytically.
However, as our focus here is numerical, we will work with the Hamiltonians.

Recall from \eqref{eq:Ham} and \cref{sec:LatticeKac} that $\hamk$ is defined in terms of a function 
$s_0(\xi)s_{k+1}(\xi)$ whose singularities define the
\emph{exceptional points}, when $k>0$. This function is periodic in $\xi$ with period $\pi$, allowing us to restrict to the interval $(0,\pi]$. The singularities in this interval occur at $\xi = \pi$ and $\xi = \xi_{\text{exc}}$, where the latter satisfies
\begin{equation}
\xi_{\text{exc}}=-(k+1)\lambda \bmod{\pi}.
\label{eq:xis}
\end{equation}
For the specialisations \eqref{eq:rootsof1}, the singularity at $\xi = \xi_{\text{exc}}$ coincides with that at $\xi = \pi$ 
precisely when $k+1$ is a multiple of $p'$.  Otherwise, this singularity divides the interval $(0,\pi)$ into 
two \emph{regimes} which we denote by $A$ and $B$:
\begin{equation}
\text{regime}\ A :\;  \xi \in (0,\xi_{\text{exc}}), \qquad \text{regime}\ B:\; \xi \in (\xi_{\text{exc}}, \pi).
\label{eq:regimes}
\end{equation}
When $p'$ divides $k+1$, so $\xi_{\text{exc}} = \pi$, there is only a single regime:  regime $A$.
For $k=0$, the Hamiltonian does not depend on $\xi$. We will, for convenience, regard the case $k=0$ as corresponding to regime $A$.

The scaling limit of the lattice Kac modules is heavily affected by whether $\xi$ belongs to regime $A$ or $B$, but in general appears not to be influenced by varying $\xi$ within a given regime. Exceptions to this last statement have been found in regime $B$ and are discussed below as well as in \cref{sec:characters,sec:Gramidentification}. The following conjecture describes the behaviour in regime $A$. 
  
\begin{conj} \label{TheConjecture}
In regime $A$, the scaling limit of the lattice Kac module $\kac_{n,k}^d$ is the Virasoro Kac module $\Kac{r,s}$, with \begin{equation}r=\Big\lceil \frac{(k+1)p}{p'}\Big\rceil, \qquad s = d+1.
\label{eq:rs1}
\end{equation}
\end{conj}
\noindent The substantial new evidence that we present in favour of this conjecture constitutes one of the main results of this paper.

Conjectures for the relations between $r,s,k$ and $d$ have been formulated previously \cite{PRZ06,PRannecy,PRV12,PTC14}, but were only tying lattice Kac modules to the characters $\chit_{r,s}$. The evidence presented in these papers could not distinguish between the different possible structures allowed by $\chit_{r,s}$. \cref{TheConjecture} is stronger precisely because it identifies the predicted structure of the Virasoro module in the scaling limit. The structures of the Virasoro Kac modules are detailed in \cref{sec:VKac}.

The behaviour in regime $B$ is much less clear and our evidence is insufficient to present a clear conjecture. In many cases, the 
evidence that we have collected suggests that the scaling limit of the Kac module $\kac_{n,k}^d$ is the Virasoro Kac 
module $\Kac{r,s}$ with 
\begin{equation}
r=\Big\lceil \frac{(k+1)p}{p'}\Big\rceil-1, \qquad s = d+1.
\label{eq:rs2}
\end{equation}
This seems to be true for the principal series $\mathcal{LM}(m,m+1)$, as noticed in \cite{PRZ06}.

There are, however, exceptions to \eqref{eq:rs2} for $p' \neq p+1$ and $k \ge 1$. In one case, $\mathcal{LM}(1,4)$ with $k = 2$, the limiting Kac module appears to be the same in both regimes $A$ and $B$. 
In other cases, discussed in \cref{sec:characters}, the character at the point $\xi = \frac12(\pi + \xi_{\rm{exc}})$ 
appears to be $2\, \chit_{1,d+1}$. We believe that, for this value of $\xi$, the limiting
module is the direct sum of two copies of the Kac module $\Kac{1,d+1}$, whereas for other $\xi$-values in regime $B$, the scaling limit corresponds to a single copy. In another specific case, $\mathcal{LM}(1,5)$ with $k = 2$, the conformal status is uncertain, as the eigenvalues of the Hamiltonian do not seem to converge to a character in the scaling limit. Of course, this may simply be due to slow convergence, but could also indicate that the corresponding
boundary condition is not conformal.
We will discuss these observations further in \cref{sec:characters}. 

\subsection{Data and results}\label{sec:LatticeData}

\subsubsection{Character analysis}\label{sec:characters}

For the model of critical dense polymers, $(p,p') = (1,2)$, the characters $\chit_{r,s}$ associated to 
the lattice Kac modules $\kac_{n,k}^d$ have been obtained in \cite{PR07,PRV12} from 
exact analytic solutions, modulo certain conjectured selection rules. 
For $k=0$, these rules were proven in \cite{MD11}. 
This section describes how the limiting Virasoro characters for the general logarithmic minimal models can, in favourable circumstances, be guessed from the eigenvalues of $\hamk$ on $\kac_{n,k}^d$. Unless otherwise indicated, the value of $\xi$ is set to either $\tfrac{1}2\xi_{\text{exc}}$ or $\tfrac12(\pi + \xi_{\text{exc}})$, according to the regime of interest, see \eqref{eq:regimes}.

The character analysis splits into two parts. The first is to evaluate 
the leading power of $q$ in the character, $\Delta-c/24$, and verify its consistency with the conjectured central charge $c$ and conformal dimensions $\Delta_{r,s}$ in \eqref{eq:CDelta}.
Approximations for this leading power can be extracted from \eqref{eq:HL0} and the known theoretical values for the bulk and boundary free energies. In essence, the sequence
\begin{equation}
\frac{n}{\pi v_s} (H_0 - n\,h_{bulk} - h_{bdy})
\end{equation}
should converge to $\Delta_{r,s}-\frac{c}{24}$ as $n \rightarrow \infty$. The recent paper by Pearce, Tartaglia and Couvreur \cite{PTC14} presents considerable numerical data using this idea to estimate 
$\Delta_{r,s}-\frac{c}{24}$ in regime $A$, corroborating \eqref{eq:CDelta} and \eqref{eq:rs1}. 
The same analysis may be applied to regime $B$, despite some subtleties involved in the special case where two characters are produced.

We will instead focus on the second part of the character analysis. Here, the goal is to extract the integer coefficients of the subleading $q$-powers of the character. Upon dividing out by $q^{\Delta_{r,s}-\frac{c}{24}}$, the character 
$\chit_{r,s}(q)$ becomes the following function of $q$,
\begin{equation}
\hat \chit_{a}(q) = \frac{(1-q^{a})}{\prod_{j=1}^\infty (1-q^j)} \qquad \text{(\(a = rs\)),}
\end{equation}
whose formal series expansion involves only integer powers.
According to \cref{TheConjecture}, the eigenvalues of $\hamk$, acting on
$\kac_{n,k}^d$, should produce the character $\chit_{r,s}$, as $n \rightarrow \infty$, with $r$ and $s$ as given in \eqref{eq:rs1}. To test
this, we fix $k$ and $d$ and define sequences $\Lambda^i_n$ consisting of the $i$-th smallest eigenvalue of $\hamk$ on $\kac_{n,k}^d$, choosing the values of $n$ such that $n + k = d \bmod{2}$. We then compute the ratios 
\begin{equation}\label{eq:ratios}
R_n^i = \frac{\Lambda^i_n-\Lambda^0_n}{\Lambda^1_n - \Lambda^0_n}
\end{equation}
which, as discussed after \eqref{eq:limitkappa}, should converge to the ratio $j_i/j_1$ of grades of the corresponding conformal states as $n\rightarrow\infty$. This automatically sets $R_n^0$ and $R_n^1$ to $0$ and $1$. For finite $n$, $R_n^i$ will only approximate $j_i/j_1$, with increasing precision as $n$ grows. 
Our computer program can calculate the eigenvalues of $\hamk$ for $n+k \le 19$ or $20$ 
and the largest matrices that we considered had 
size $16\, 796 \times 16\, 796$. Seven examples are provided in \cref{tab:char} and were chosen to present both the successes and limits of this approach. The ratios are presented in the form of 
\emph{$\hat \chit_a$-approximations}, $\sum_{i} q^{R_n^i}$, thus facilitating a
direct comparison with $\hat \chit_a$. We remark that if $r$ and $s$ are both believed to be $1$, then the conformal grade of the first excited state should be $j_1= 2$ and $\hat \chit_1$ should have the form $1 + q^2 + q^3 + \dots$.  In this case, we will instead define the $\hat \chit_1$-approximation as $\sum_{i} q^{2R_n^i}$.

\begin{table}
\begin{alignat*}{2}
&\begin{array}{c|c|ll} \hline \multicolumn{2}{c}{} \\[-0.1cm]
\multirow{6}{*}{$(a)$}&(p,p') = (3,5) & n=14: & 1 + q^2 + q^{2.92} + q^{3.74}+ q^{3.88} + q^{4.42} + q^{4.68} + q^{4.93} + q^{5.24} + q^{5.35} + \cdots \\[0.1cm]
&k = 0 & n=16: & 1 + q^2 + q^{2.94} + q^{3.80} + q^{3.90} + q^{4.54} + q^{4.75} + q^{5.15} + q^{5.49} + q^{5.60} + \cdots \\[0.1cm]
&d = 0 & n=18: & 1 + q^2 + q^{2.95} + q^{3.84} + q^{3.92} + q^{4.63} + q^{4.80} + q^{5.32} + q^{5.59} + q^{5.72} + \cdots \\[0.1cm]
&{\rm reg.}\, A & n=20: & 1 + q^2 + q^{2.96} + q^{3.87} + q^{3.94} + q^{4.70} + q^{4.84} + q^{5.44} + q^{5.66} + q^{5.77} + \cdots \\[0.1cm]
&(r,s) \rightarrow (1,1) & n \rightarrow \infty: & \hat \chit_1 = 1 + q^2 + q^3 + 2 q^4 + 2 q^5 + 4 q^6 + \cdots 
\end{array} 
\\[0.2cm] 
&\begin{array}{c|c|ll} \hline \multicolumn{2}{c}{} \\[-0.1cm]
\multirow{6}{*}{$(b)$}&(p,p') = (1,3) & n=13: & 
1 + q + q^{2.05} + q^{2.96} + q^{3.15} + q^{3.85} + q^{4.15} + q^{4.31} + q^{4.57} + q^{4.78} + \cdots \\[0.1cm]
&k = 0 & n=15: & 
1 + q + q^{2.04} + q^{2.97} + q^{3.11} + q^{3.89} + q^{4.11} + q^{4.24} + q^{4.68} + q^{4.83} + \cdots \\[0.1cm]
&d = 1 & n=17: & 
1 + q  + q^{2.03} + q^{2.98} + q^{3.09} + q^{3.91} + q^{4.09} + q^{4.20} + q^{4.76} + q^{4.87} + \cdots \\[0.1cm]
&{\rm reg.}\, A & n=19: & 
1 + q + q^{2.02} + q^{2.98} + q^{3.07} + q^{3.93} + q^{4.07} + q^{4.16} + q^{4.81} + q^{4.90} + \cdots \\[0.1cm]
&(r,s) \rightarrow (1,2) & n \rightarrow \infty: & \hat \chit_2 = 1 + q + q^2 + 2 q^3 + 3 q^4 + 4 q^5 + 6 q^6 + \cdots 
\end{array}
\\[0.2cm]
&
\begin{array}{c|c|ll}\hline \multicolumn{2}{c}{} \\[-0.1cm]
\multirow{6}{*}{$(c)$}&(p,p') = (3,4) & n=11: &
1 + q + q^{1.68} + q^{1.88} + q^{2.44} + q^{2.63} + q^{2.80} + q^{3.08} + q^{3.21} + q^{3.27} + \cdots \\[0.1cm]
&k = 1 & n=13: & 
1 + q + q^{1.71} + q^{1.90} + q^{2.51} + q^{2.70} + q^{2.83} + q^{3.21} + q^{3.35} + q^{3.37} + \cdots \\[0.1cm]
&d = 2 & n=15: & 
1 + q + q^{1.74} + q^{1.92} + q^{2.56} + q^{2.75} + q^{2.85} + q^{3.30} + q^{3.42} + q^{3.49} + \cdots \\[0.1cm]
&{\rm reg.}\, A: \xi = \tfrac \pi 4 & n=17: & 
1 + q + q^{1.76} + q^{1.93} + q^{2.60} + q^{2.79} + q^{2.87} + q^{3.37} + q^{3.47} + q^{3.57} + \cdots \\[0.1cm]
&(r,s) \rightarrow (2,3) & n \rightarrow \infty: & \hat \chit_6 = 1 + q + 2 q^2 + 3 q^3 + 5 q^4 + 7 q^5 + 10 q^6 + \cdots 
\end{array}
\\[0.2cm]
&
\begin{array}{c|c|ll}\hline \multicolumn{2}{c}{} \\[-0.1cm]
\multirow{6}{*}{$(d)$}&(p,p') = (4,5) & n=12: & 1 + q + q^{1.90} + q^{2.01} + q^{2.65} + q^{2.90} + q^{3.22} + q^{3.57} + q^{3.64} + q^{3.86} + \cdots \\[0.1cm]
&k = 2 & n=14: & 
1 + q + q^{1.92} + q^{2.01} + q^{2.73} + q^{2.92} + q^{3.40} + q^{3.72} + q^{3.89} + q^{3.90} + \cdots \\[0.1cm]
&d = 0 & n=16: & 
1 + q  + q^{1.94} + q^{2.00} + q^{2.79} + q^{2.93} + q^{3.53} + q^{3.77} + q^{3.91} + q^{4.12} + \cdots \\[0.1cm]
&{\rm reg.}\, A: \xi = \tfrac \pi 5 & n=18: & 
1 + q + q^{1.95} + q^{2.00} + q^{2.83} + q^{2.94} + q^{3.62} + q^{3.81} + q^{3.92} + q^{4.29} + \cdots \\[0.1cm]
&(r,s) \rightarrow (3,1) & n \rightarrow \infty: & \hat \chit_3 = 1 + q + 2 q^2 + 2 q^3 + 4 q^4 + 5 q^5 + 8 q^6 + \cdots
\end{array}
\\[0.2cm] 
&
\begin{array}{c|c|ll}\hline \multicolumn{2}{c}{} \\[-0.1cm]
\multirow{6}{*}{$(e)$}&(p,p') = (2,3) & n=8: & 
1 + q + q^{1.79} + q^{2.10} + q^{2.31} + q^{2.48} + q^{2.67} + q^{2.85} + q^{3.20} + q^{3.30} + \cdots \\[0.1cm]
&k = 3 & n=10: &
1 + q + q^{1.86} + q^{2.09} + q^{2.52} + q^{2.74} + q^{2.92} + q^{2.93} + q^{3.08} + q^{3.40} + \cdots \\[0.1cm]
&d = 3 & n=12: & 
1 + q + q^{1.89} + q^{2.08} + q^{2.64} + q^{2.79} + q^{2.95} + q^{3.21} + q^{3.53} + q^{3.56} + \cdots \\[0.1cm]
&{\rm reg.}\, B: \xi = \tfrac {5\pi} 6 & n=14: & 
1 + q  + q^{1.92} + q^{2.07} + q^{2.72} + q^{2.83} + q^{2.97} + q^{3.39} + q^{3.62} + q^{3.76} + \cdots \\[0.1cm]
&(r,s) \rightarrow (2,4) & n \rightarrow \infty: & \hat \chit_{8} =  1 + q + 2 q^2 + 3 q^3 + 5 q^4 + 7 q^5 + 11 q^6 + \cdots 
\end{array}
\\[0.2cm] 
&
\begin{array}{c|c|ll}\hline \multicolumn{2}{c}{} \\[-0.1cm]
\multirow{6}{*}{$(f)$}&(p,p') = (1,4) & n=11: & 
1 + q + q^{2.25} + q^{2.42} + q^{2.42} + q^{3.12} + q^{3.32} + q^{3.32} + q^{3.68} + q^{4.24} + \cdots \\[0.1cm]
&k = 2 & n=13: & 
1 + q + q^{2.19} + q^{3.08} + q^{3.08} + q^{3.08} + q^{3.56} + q^{3.99} + q^{3.99} + q^{4.16} + \cdots \\[0.1cm]
&d = 1 & n=15: & 
1 + q  + q^{2.14} + q^{3.06} + q^{3.44} + q^{3.72} + q^{3.72} + q^{4.11} + q^{4.42} + q^{4.65} + \cdots \\[0.1cm]
&{\rm reg.}\,B: \xi = \tfrac {7\pi} 8 & n=17: & 
1 + q + q^{2.11} + q^{3.04} + q^{3.35} + q^{4.08} + q^{4.33} + q^{4.37} + q^{4.37} + q^{4.75} + \cdots \\[0.1cm]
&(r,s) \rightarrow (1,2) & n \rightarrow \infty: &  \hat \chit_2 = 1 + q + q^2 + 2 q^3 + 3 q^4 + 4 q^5 + 6 q^6 + \cdots 
\end{array}
\\[0.2cm] 
&
\begin{array}{c|c|ll}\hline \multicolumn{2}{c}{} \\[-0.1cm]
\multirow{6}{*}{$(g)$}&(p,p') = (1,5) & n=12: & 
1 + q + q^{1.00} + q^{2.19} + q^{3.04} + q^{3.05} + q^{3.71} + q^{4.30} + q^{4.50} + q^{4.51} + \cdots\\[0.1cm]
&k = 2 & n=14: & 
1 + q + q^{1.00} + q^{1.70} + q^{2.59} + q^{2.59} + q^{2.85} + q^{3.33} + q^{3.70} + q^{3.70} + \cdots \\[0.1cm]
&d = 0 & n=16: & 
1 + q + q^{1.00} + q^{1.38} + q^{2.29} + q^{2.30} + q^{2.30} + q^{2.70} + q^{3.17} + q^{3.17}+ \cdots  \\[0.1cm]
&{\rm reg.}\, B : \xi = \tfrac {4\pi} 5 & n=18 & 
1 + q + q^{1.00} + q^{1.16} + q^{1.90} + q^{2.10} + q^{2.10} + q^{2.27} + q^{2.75} + q^{2.81} + \cdots \\[0.1cm]
&(r,s) \rightarrow (?,?) & n \rightarrow \infty: &  \textrm{\scriptsize Data insufficient to determine a conformal character.}
\\
\multicolumn{2}{c}{}\\[-0.1cm] \hline
\end{array}
\end{alignat*}
\caption{Examples of explicit numerical eigenvalue analyses and comparisons with $\hat \chit_{rs}(q)$ candidates.}
\label{tab:char} 
\end{table}

\paragraph{Regime $\boldsymbol A$.} With the data produced by our computer program, the approximate $q$-series generally allows us to confidently guess the integer coefficients of the limiting character up to grade $4$, 
for $k=0,1$, though often only to lower grades for $k=2,3$. As the results of \cref{tab:char} show, the convergence is not particularly fast in general. This decrease in precision comes from our computational limitation, $n+k \le 19$ or $20$, and the fact that as $k$ grows, the boundary becomes a bigger fraction of the bulk and so 
the results are less representative of the scaling limit. Nevertheless, if we stay in regime $A$, then in each case studied, the first few coefficients that can be guessed from our data reproduce those of $\chit_{r,s}$, with $r$ and $s$ as in \cref{TheConjecture}. 

Up to grade $rs-1$, the integer coefficients in $\chit_{r,s}$ are identical to those of the Verma module characters.
As $k$ and $d$ grow, the value of $a=rs$ quickly becomes greater than $4$ and the character approximations do not allow one to discern whether the limiting character is that of a Kac module or a Verma module (or something else entirely). 
This occurs, for instance, in examples $(c)$ and $(e)$ in \cref{tab:char}. Both $\hat \chit_6$ and $\hat \chit_8$ have the form $1 + q + 2 q^2 + 3 q^3 + 5 q^4 + \cdots$, but our data is insufficient to determine 
the coefficients at grades $6$ and $8$, respectively. In these examples, 
we can arrive at the values of $r$ and $s$ predicted in \cref{TheConjecture} by assuming that $r$ is independent of 
$d$ and that $s$ is given by $d+1$ --- this has been the result in every other case that we have analysed.
The value of $r$ then follows from the $\hat\chit_a$-approximation of $\kac_{n,k}^0$ or $\kac_{n,k}^1$.  
We emphasise, though, that the character analysis of these examples cannot be regarded as independent evidence in support of \cref{TheConjecture}.

On the other hand, our low grade character analysis 
does allow us, in many cases, to witness the absence of the \ssv{} at grade $rs$ 
that distinguishes the characters of $\Kac{r,s}$ and $\Ver{r,s}$. These cases provide support for \cref{TheConjecture}.  However, to convincingly determine the composition factors of a 
Virasoro module, one would like to investigate not only the first missing subsingular vector, but also the next ones 
whose grades are typically larger than $4$. 
As our character analysis is numerical, we cannot exclude the possibility that the coefficients of the scaling limit character could differ from those of $\chit_{r,s}$, for sufficiently large grades. 
Strong evidence against this possibility comes from Kac module studies \cite{PR07,PRV12} for $\mathcal{LM}(1,2)$, where the eigenvalues and characters were computed exactly (and agree well with approximate numerical results). 

\paragraph{Regime $\boldsymbol B$.} 
In many cases, the limiting characters of the lattice Kac modules in regime $B$ also give rise to Kac characters, though these are generally different to 
their regime $A$ counterparts. We have observed, however, that the convergence appears to be slower, see for instance example $(f)$ in \cref{tab:char}. Moreover, as mentioned in \cref{sec:scalinglimit}, there are other cases in which the scaling limit of the lattice Kac module does not appear to converge to a single Virasoro Kac module.

Let us look at a particular example (not included in \cref{tab:char}): $(p,p') = (1,3)$ with $k=1$. 
At the value $\xi = \frac12(\pi + \xi_{\text{exc}}) = \tfrac{5 \pi}6$, the ground state eigenvalue is 
\emph{quasi-degenerate}: Even for small $n$, the first two eigenvalues are so close that the next excitation has a ratio $R_n^2 \ge 50$ that increases rapidly with $n$. We view this as suggesting that the ground state is degenerate in the scaling limit. Higher excitations are also all quasi-degenerate, with the degeneracy again of order $2$, and this appears to hold,
independent of $n$ and $d$. To perform the character analysis, we instead calculate the ratios $R_n^i$ by using the eigenvalue sequence of the second excitation, $R_n^i = (\Lambda^i_n-\Lambda^0_n)/(\Lambda^2_n - \Lambda^0_n)$. As we move away from $\xi = \tfrac{5 \pi}6$ to $\xi = \frac{5 \pi}6 + \epsilon$, the quasi-degeneracies are slowly lifted. For instance, setting $n=15, d=0$ for this example yields the following character approximations:
\begin{equation}
\begin{array}{ll}
\epsilon= 0:    & \quad 1 + q^{0.00} + q^2 + q^{2.00} + q^{3.07} + q^{3.07} + q^{3.88} + q^{3.88} + q^{4.18} + q^{4.18} + q^{4.74} + q^{4.74} + \dots  \\[0.1cm]
\epsilon= 0.01: & \quad 1+ q^{0.18} + q^{2} + q^{2.18} + q^{3.07} + q^{3.24} + q^{3.88} + q^{4.05} + q^{4.18} + q^{4.35} + q^{4.74} + q^{4.91} + \dots\\[0.1cm]
\epsilon= 0.02: & \quad 1+ q^{0.36} + q^{2} + q^{2.35} + q^{3.07} + q^{3.42} + q^{3.88} + q^{4.18} + q^{4.22} + q^{4.53} + q^{4.73} + q^{5.08} + \dots \\[0.1cm]
\epsilon= 0.04: & \quad 1+ q^{0.72} + q^{2} + q^{2.71} + q^{3.07} + q^{3.77} + q^{3.88} + q^{4.18} + q^{4.57} + q^{4.73} + q^{4.87} + q^{5.34} + \dots \\[0.1cm] 
\epsilon= 0.06: & \quad 1+ q^{1.09} + q^{2} + q^{3.07} + q^{3.07} + q^{3.88} + q^{4.12} + q^{4.19} + q^{4.73} + q^{4.93} + q^{5.23} + q^{5.34} + \dots\\[0.1cm]
\end{array}
\label{eq:2charexample}
\end{equation}
Note that the function $s_0(\xi)s_2(\xi)$ entering the definition of $\hamk$ is symmetric around $\xi = \tfrac{5\pi}6$, explaining why only $\epsilon \ge 0$ is considered here.

For $\epsilon = 0$, the approximations appear 
to be converging to $2\, \hat\chit_{1}$, while for $\epsilon > 0$, we seem to instead obtain $(1 + q^{f_n(\epsilon)}) \hat\chit_1$, where $f_n(\epsilon)$ is approximately linear in $\epsilon$, for small $\epsilon$.
From the data for $n=9,11,13$, it also appears that $f_n(\epsilon)$ increases with $n$, suggesting that the second copy of $\hat\chit_1$ corresponds to non-conformal sequences
that are drifting off to infinity, but very slowly.
In conclusion, we assert that the character is $2\,\chit_{1,1}$ for $\xi = \frac{5\pi}6$, but is instead $\chit_{1,1}$ for all other $\xi$ in regime $B$.

The example $(g)$, also belonging to regime $B$, is unclear too, though for slightly different reasons. 
The first excitation appears to be quasi-degenerate, implying that the limiting $q$-series is $1+ mq + \cdots$, for some $m\ge 2$, which never occurs for Kac characters.  In principle, the limiting series could represent a sum of Kac characters.  However,
it is unclear whether the next leading excitations have converging ratios. One could imagine that many of these higher excitations will also give eigenvalues $1$ in the scaling limit. 
A similar behaviour was also observed for $\mathcal{LM}(1,4)$ with
$d = 0$ and small $n$. In this case, the quasi-degeneracy of the first excitation is slowly lifted and the character is identified as $\chit_{1,1}$, with very slow convergence.
We therefore suspect that $(g)$ suffers from the same malady, but confirming this requires more
data than is presently available to us. 

\subsubsection{Virasoro module structure from invariant bilinear forms}\label{sec:Gramidentification}

Knowing the character of a module is, in general, insufficient to determine the structure of its \ssvs{}. As this section will show, for lattice Kac modules, it is possible to gain insight into the limiting Virasoro structures by using the invariant bilinear forms defined in \cref{sec:Gram}.

In the scaling limit, one could expect
that irreducible modules over the algebra $\btl_{n,k}(\beta)$ become irreducible Virasoro modules and, more generally, 
that any indecomposable yet reducible structures of the lattice modules are preserved.   
We will investigate this expectation below in a variety of explicit examples. In many of these, our analysis leads to an 
unambiguous prediction of the structure of the Virasoro module in the scaling limit, as asserted 
by \cref{TheConjecture}.  
When this is the case, the limiting structure of \cref{TheConjecture} may differ from that of the lattice Kac modules 
because it is possible for composition factors of the lattice modules to
drift off to infinity (they would correspond to non-conformal sequences).  
There are also examples for which 
our analysis does not lead to a prediction for the limiting Virasoro module because we do not know the complete 
structure of the lattice Kac modules.  
In such cases, we have confirmed that the partial information obtained is still consistent with \cref{TheConjecture}.
Further structural evidence supporting our conjecture in these cases may be obtained from explicit fusion computations in the continuum, see \cref{sec:NGKfusion}.

We detail this lattice analysis below, noting that it
relies on \cref{conj:highermodes} which posits the existence of Virasoro mode approximations $\Lmn$. 
The character $\chit_{r,s}$ determines which irreducible Virasoro modules appear as composition factors in the scaling limit of the lattice Kac module $\kac_{n,k}^d$. Because the $\Lmn$ are conjectured to be elements of $\btl_{n,k}(\beta)$, the maximal proper submodule of $\stan_{n,k}^d$ (the radical of $\gramprodk{\cdot}{\cdot}$) should be invariant under their action. In this way, we expect that the Virasoro module structure of the scaling limit may be (partially) explored using lattice technology.  More specifically, we may use the corresponding Gram matrix to determine whether, for finite $n$, the sequences that give rise
to the states in each composition factor appear in the radical or quotient of the corresponding standard $\btl_{n,k}(\beta)$-module. This can then be compared with the prediction of \cref{TheConjecture} 
which states that the embedding structure is that of $\Kac{r,s}$. 

The results of this analysis are presented in \cref{sec:Gramdata}, for each logarithmic model $\mathcal{LM}(p,p')$ with $2 \le p' \le 5$, in regimes $A$ and $B$, for $0 \le k \le 3$ and $0 \le d \le 4$. The chosen examples below describe how the analysis is performed, the conclusions that we draw and the difficulties we encounter.

\paragraph{Example (i): $\boldsymbol{\mathcal{LM}(2,3)\ {\rm with}\ k=0,\ d=2}$ in regime $\boldsymbol A$.}
The boundary seam algebra relevant in this example is $\btl_{n,0}(\beta) \simeq \tl_{n}(\beta)$ with $\beta = 1$. Its representation theory is known (see \cref{app:TLrep}). From \eqref{eq:GramdetTL}, the determinant of the Gram matrix $\grammat_{n,0}^2=\grammat_{n}^2$ is non-zero for all $n \in 2\mathbb Z_+$. The radical $\Rad_{n,0}^2 = \Rad_{n}^2$ is thus trivial and $\kac_{n,0}^2 = \stan_{n}^2\simeq \Itl_{n}^2$ is irreducible. In other words, all states may be viewed as belonging to the quotient of the Kac module (by its trivial radical).  We will therefore refer to these states as \emph{quotient states}.

From the character analysis, we find that as $n\rightarrow \infty$, the eigenvalues of $\hamk$ in $\kac_{n,0}^2$ seem to produce $\chit_{1,3}$, an irreducible character for $c=0$ with $\Delta = \frac1{3}$.  The bilinear form analysis is consistent with this irreducibility, but provides no further information, as the knowledge of an irreducible
character already determines the scaling limit of the lattice Kac module. Assuming this irreducibility, we represent the scaling limit as
\begin{equation}
\kac_{n,0}^2 \big|_{(p,p')=(2,3)}^{{\rm reg.}\,A}\  \simeq \ \Itl_{n}^2 \qquad \xrightarrow{n \rightarrow \infty} \qquad
\begin{pspicture}[shift=-0.475](-0.2,-0.6)(0.2,0.6)
\pscircle[fillstyle=solid,fillcolor=black](0,0){0.075}\rput(0,0.3){\scriptsize$\tfrac{1}{3}$}\rput(0,-0.3){\sq},
\end{pspicture}
\ , 
\end{equation}
mapping a sequence of irreducible $\tl_n(\beta)$-modules to an irreducible Virasoro one. The {\scshape (q)} appearing on the \rhs{} is the result of the analysis with the bilinear form, and indicates here that the sequences defining the states of the irreducible $\Delta = \frac1{3}$ module are formed from states in the quotient of $\kac_{n,0}^2$ by its radical. In what follows,
{(\scshape r)} will indicate that a given factor is in the radical and {\scshape (u)} that its status in the radical or quotient is unknown due to insufficient data.

\paragraph{Example (ii): $\boldsymbol{\mathcal{LM}(3,4)\ {\rm with}\ k=1,\ d=1}$ in regime $\boldsymbol A$.}
For $n \in 2 \ZZ_+$, $n \ge 4$, the Gram matrix $\grammat_{n,1}^1$ has a vanishing determinant at $\beta = \sqrt{2}$ and the corresponding standard module $\stan_{n,1}^1 = \kac_{n,1}^1$ is a reducible $\tl_{n+1}(\beta)$-module. Its Loewy diagram is of type $(b)$, see \cref{fig:Loewy}. 
The limiting character is $\chit_{2,2}$ indicating that the corresponding Virasoro module has two irreducible composition factors corresponding to
conformal highest weights $\Delta = \frac1{16}$ and $\Delta = \frac{33}{16}$. There are thus three possibilities for the structure: 
\begin{equation}
\begin{pspicture}[shift=-0.1](-0.2,-0.2)(0.95,0.4)
\pscircle[fillstyle=solid,fillcolor=black](0,0){0.075}\rput(0,0.3){\scriptsize$\tfrac1{16}$}
\rput(0.75,0){\pscircle[fillstyle=solid,fillcolor=black](0.0,0){0.075}\rput(0.0,0.3){\scriptsize$\tfrac{33}{16}$}}
\psline{->}(0.075,0)(0.675,0)
\end{pspicture},
\qquad
\begin{pspicture}[shift=-0.1](-0.2,-0.2)(0.95,0.4)
\pscircle[fillstyle=solid,fillcolor=black](0,0){0.075}\rput(0,0.3){\scriptsize$\tfrac1{16}$}
\rput(0.75,0){\pscircle[fillstyle=solid,fillcolor=black](0.0,0){0.075}\rput(0.0,0.3){\scriptsize$\tfrac{33}{16}$}}
\rput(0.375,0){$\oplus$}
\end{pspicture}
\qquad
\text{or}
\qquad
\begin{pspicture}[shift=-0.1](-0.2,-0.2)(0.95,0.4)
\pscircle[fillstyle=solid,fillcolor=black](0,0){0.075}\rput(0,0.3){\scriptsize$\tfrac1{16}$}
\rput(0.75,0){\pscircle[fillstyle=solid,fillcolor=black](0.0,0){0.075}\rput(0.0,0.3){\scriptsize$\tfrac{33}{16}$}}
\psline{<-}(0.075,0)(0.675,0)
\end{pspicture}\,.
\end{equation}

Applying $\grammat_{n,1}^1$ to the ground state of $\hamk$, acting on 
$\kac_{n,1}^1$, for small $n$, we find that the result is never zero, implying that the ground state is non-zero in the quotient. Pushing this analysis to sequences of higher excited states, we order the eigenvalues $\Lambda_n^i$ in an increasing fashion, determining those whose corresponding eigenvector is observed to belong to the radical.  We call these states the \emph{radical states}.  For example, we tabulate the numbers $\# i$ giving eigenvalues $\Lambda_n^i$ corresponding to radical states for a few values of $n$: 
\begin{equation}
\begin{array}{ll}
n = 10: & \# 4, \#7, \#11, \#12, \#14, \#17, \dots\\[0.1cm]
n = 12: & \# 4, \#7, \#10, \#12, \#15, \#18, \dots\\[0.1cm]
n = 14: & \# 4, \#7, \#10, \#12, \#16, \#18, \dots \\[0.1cm]
n = 16: & \# 4, \#7, \#10, \#11, \#15, \#17, \dots. 
\end{array}\label{eq:dataii}
\end{equation}
These numbers appear to converge as $n$ grows. Considering that $\hat \chit_4(q) 
= 1 + q + 2 q^2 + 3 q^3 + 4 q^4 + 6 q^5 + \dots$ separates into two irreducible contributions as $(1 + q + q^2 + 2q^3 + 2q^4 + 3 q^5+\dots) + q^2(1 + q + 2q^2 + 3 q^3 + \dots)$, the data \eqref{eq:dataii} suggests that the first radical state 
appears at grade $2$ and belongs to the $\Delta=\frac{33}{16}$ irreducible. In fact, this analysis allows us to identify 
each state as belonging to one of the two irreducible factors:
\begin{equation}
\begin{pspicture}[shift=-0.2](-1,-0.95)(1,0.25)   
\rput(0,0){$\hat \chit_4(q)= (1 + q + q^2 + 2q^3 + 2q^4 + 3 q^5 + \cdots) + q^2(1 + q + 2q^2 +3 q^3 + \cdots)$.}
\rput(0.228,0){\rput(0.02,0){\rput(-4.85,-0.45){\scriptsize$\#:$}
\rput(-4.45,-0.45){\scriptsize$1$}
\rput(-3.85,-0.45){\scriptsize$2$}
\rput(-3.22,-0.45){\scriptsize$3$}
\rput(-2.39,-0.45){\scriptsize$5$}
\rput(-2.39,-0.70){\scriptsize$6$}
\rput(-1.45,-0.45){\scriptsize$8$}
\rput(-1.45,-0.70){\scriptsize$9$}
\rput(-0.50,-0.45){\scriptsize$12$}
\rput(-0.50,-0.70){\scriptsize$13$}
\rput(-0.50,-0.95){\scriptsize$14$}
}
\rput(0.38,0){
\rput(1.435,-0.45){\scriptsize$4$}
\rput(2.035,-0.45){\scriptsize$7$}
\rput(2.78,-0.45){\scriptsize$10$}
\rput(2.78,-0.70){\scriptsize$11$}
\rput(3.715,-0.45){\scriptsize$15$}
\rput(3.715,-0.70){\scriptsize$17$}
\rput(3.715,-0.95){\scriptsize$?$}
}
}
\end{pspicture}
\end{equation}
Here, we have used the $n=16$ data. In this case, the
matching works well up to grade $4$; in other similar cases, it is sometimes consistent up to grade $5$ or $6$. We conclude that the $\Delta=\frac1{16}$ and $\Delta=\frac{33}{16}$ factors correspond to the quotients and radicals, respectively, of the Kac modules $\kac_{n,1}^1$ with $n \in 2 \ZZ_+$.

From \cref{conj:highermodes}, Virasoro mode approximations can map states from the quotient to the radical of $\btl_{n,1}(\beta)$, but not the other way around. This therefore rules out the case 
$\begin{pspicture}[shift=-0.1](-0.2,-0.2)(0.95,0.2)
\pscircle[fillstyle=solid,fillcolor=black](0,0){0.075}
\rput(0.75,0){\pscircle[fillstyle=solid,fillcolor=black](0.0,0){0.075}}
\psline{<-}(0.075,0)(0.675,0)
\end{pspicture}$.
The case
$\begin{pspicture}[shift=-0.1](-0.2,-0.2)(0.95,0.2)
\pscircle[fillstyle=solid,fillcolor=black](0,0){0.075}
\rput(0.75,0){\pscircle[fillstyle=solid,fillcolor=black](0.0,0){0.075}}
\psline{->}(0.075,0)(0.675,0)
\end{pspicture}
$
is therefore plausible, but more is needed to rule out
\begin{pspicture}[shift=-0.1](-0.2,-0.2)(0.95,0.2)
\pscircle[fillstyle=solid,fillcolor=black](0,0){0.075}
\rput(0.75,0){\pscircle[fillstyle=solid,fillcolor=black](0.0,0){0.075}}
\rput(0.375,0){$\oplus$}
\end{pspicture}.

\cref{conj:highermodes} asserts that the $\Lmn$ generate a subalgebra of $\btl_{n,k}(\beta)$, so in general one can 
only conclude that these \emph{can} map the quotient into the radical, but cannot ensure that they actually 
\emph{do}. In the present case however, $\btl_{n,1}(\beta) \simeq \tl_{n+1}(\beta)$, so we have an explicit realisation \eqref{eq:modes} of the $\Lmn$. This explicit realisation, though not unique,\footnote{For example, for $m\neq 0$, $\frac 1m[\Lmn,L_{0}^{\textrm{\tiny$(n)$}}]$ should also converge to $L_m$ in the scaling limit.} has the remarkable feature that $\Lmn + L_{-m}^{\textrm{\tiny$(n)$}}$ is $\sum_j e_j \cos(\frac{\pi m j} n)$, up to some multiplicative and additive constants. Each $e_j$ may thus be obtained as a linear combination of the $\Lmn + {L_{-m}^{\textrm{\tiny$(n)$}}}$. Here, the $\Lmn$ generate the full algebra $\tl_{n+1}(\beta)$, for each $n \in 2\, \mathbb Z_+$, so in this case they 
\emph{do} map the quotient into the radical. We thus find that
\begin{equation}
\kac_{n,1}^1 \big|_{(p,p')=(3,4)}^{{\rm reg.}\,A}\  \simeq
\begin{pspicture}[shift=-0.75](1,-0.75)(3,0.75)
\rput(0,-0.5){\rput(1.5,1){$\Itl_n^{1}$}
\rput(2.5,0){$\Itl_n^{5}$}
\psline[linewidth=.8pt,arrowsize=3pt 2]{->}(1.75,0.75)(2.25,0.25)}
\end{pspicture}
 \quad \xrightarrow{n \rightarrow \infty} \qquad
\begin{pspicture}[shift=-0.3](-0.2,-0.4)(0.95,0.4)
\pscircle[fillstyle=solid,fillcolor=black](0,0){0.075}\rput(0,0.3){\scriptsize$\tfrac1{16}$}\rput(0,-0.3){\sq}
\rput(0.75,0){\pscircle[fillstyle=solid,fillcolor=black](0.0,0){0.075}\rput(0.0,0.3){\scriptsize$\tfrac{33}{16}$}\rput(0.0,-0.3){\sr}}
\psline{->}(0.075,0)(0.675,0)
\end{pspicture} \ .
\end{equation}
The indecomposable structure of the standard module thus persists in the scaling limit. 

In similar cases, but with $k>1$, we do not know the actual expressions for the $\Lmn$, so we have no way of showing 
that the indecomposable structures of $\stan_{n,k}^d$ are preserved in the scaling limit. We therefore cannot go 
beyond \emph{plausibility} in extracting the conformal subsingular vector structure using lattice data. If one can find 
approximate Virasoro modes that generate $\btl_{n,k}$, then understanding its representation theory and knowing 
which composition factors drift off to infinity should provide 
even stronger evidence for predicting the limiting Virasoro structures.

Extra input can also be obtained from conformal field theory, assuming that fusion in the continuum does correspond to 
the lattice prescription for fusion described in \cref{sec:LatticeFusion}.
For example, for the $\mathcal{LM}(1,p')$ models, direct sum decompositions for lattice Kac modules with limiting 
characters $\chit_{1,s}$ were ruled out in \cite{RasCla11} on the basis of consistency with conjectured fusion rules in 
the scaling limit. Similar arguments should also allow one to rule out $\Kac{1,s}$ being
decomposable, for all $\mathcal{LM}(p,p')$, consistent with \cref{TheConjecture}.

\paragraph{Example (iii): $\boldsymbol{\mathcal{LM}(1,4)\ {\rm with}\ k=0,\ d=0}$ in regime $\boldsymbol A$.}
For $n \in 2 \ZZ_+$ and $n \ge 4$, $\det \grammat_{n,0}^0 = 0$ and the standard modules $\stan_{n,0}^0$ are again reducible yet indecomposable with two composition factors. Unlike the previous example, acting with the Gram matrix on the ground state gives zero. 
Applying $\grammat_{n,0}^0$ on the excited states, we find that the first quotient states (on which the action of the Gram matrix is non-zero) are
\begin{equation}
\begin{array}{ll}
n=10: &\#11, \#17, \#21, \dots \\[0.1cm]
n=12: &\#30, \#51, \#62, \dots \\[0.1cm]
n = 14: &\#106, \#138, \#163, \dots \\[0.1cm] 
n = 16: & \#304, \#457, \#536, \dots. 
\end{array}
\end{equation} 
The fact that these quotient states correspond to conformal weights that seem to be diverging, as $n$ increases, 
gives us our first example of an irreducible $\tl_n(\beta)$-module that drifts off to infinity as $n \rightarrow \infty$: 
the quotient state sequences appear to be non-conformal. 
The character $\chit_{1,1}$ obtained from the character analysis is irreducible for $c = -\frac{25}2$ and $h=0$, 
thus we conclude that
\begin{equation}
\kac_{n,0}^0 \big|_{(p,p')=(1,4)}^{{\rm reg.}\,A}\  \simeq
\begin{pspicture}[shift=-0.75](1,-0.75)(3,0.75)
\rput(0,-0.5){\rput(1.5,1){$\Itl_n^{0}$}
\rput(2.5,0){$\Itl_n^{6}$}
\psline[linewidth=.8pt,arrowsize=3pt 2]{->}(1.75,0.75)(2.25,0.25)}
\end{pspicture}
 \quad \xrightarrow{n \rightarrow \infty} \qquad
\begin{pspicture}[shift=-0.475](-0.2,-0.6)(0.2,0.6)
\pscircle[fillstyle=solid,fillcolor=black](0,0){0.075}\rput(0,0.3){\scriptsize$0$}\rput(0,-0.3){\sr}
\end{pspicture}
\ . 
\end{equation}
In this example, the indecomposable, reducible structure present at finite $n$ does not survive the scaling limit.

\paragraph{Example (iv): $\boldsymbol{\mathcal{LM}(2,5)\ {\rm with}\ k=1,\ d=1}$ in regime $\boldsymbol B$.}
This is one of the puzzling examples where the eigenvalues seem to produce two copies of the same Kac character, here $\hat\chit_{1,2}$.
We note that $\hat\chit_{1,2}$ splits as the sum of two irreducible characters with $\Delta = -\frac15$ and $\frac{14}5$.

The standard module $\stan_{n,1}^1$ is reducible, of type (b) in \cref{fig:Loewy}. At $\xi = \frac12(\pi+\xi_{\text{exc}}) = \frac{9 \pi}{10}$, the radical/quotient analysis shows that the first radical states are grouped by pairs and at positions 
that seem to converge, with
\begin{equation}
n=16: \ \#7, \#8, \#11, \#12, \#17, \#18, \#19, \#20, \#25, \#26, \dots.
\label{eq:radicalpositions(iv)}
\end{equation}
This is consistent, at least up to grade $5$, with the decomposition in terms of irreducible modules whose conformal weights differ by $3$:
\begin{equation}
\begin{pspicture}[shift=-1.2](0,-1.50)(0.1,0.05)
\rput(0.448,0){
\rput(0.91,-0.6){\scriptsize$\#:$}
\rput(1.58,-0.6){\scriptsize$1$}\rput(1.58,-0.85){\scriptsize$2$}
\rput(2.20,-0.6){\scriptsize$3$}\rput(2.20,-0.85){\scriptsize$4$}
\rput(2.81,-0.6){\scriptsize$5$}\rput(2.81,-0.85){\scriptsize$6$}
\rput(3.57,-0.6){\scriptsize$9$}\rput(3.57,-0.85){\scriptsize$10$}
\rput(4.48,-0.6){\scriptsize$13$}\rput(4.48,-0.85){\scriptsize$14$}\rput(4.48,-1.1){\scriptsize$15$}\rput(4.48,-1.35){\scriptsize$16$}
\rput(5.40,-0.6){\scriptsize$21$}\rput(5.40,-0.85){\scriptsize$22$}\rput(5.40,-1.1){\scriptsize$23$}\rput(5.40,-1.35){\scriptsize$24$}
\rput(7.84,-0.6){\scriptsize$7$}\rput(7.84,-0.85){\scriptsize$8$}
\rput(8.45,-0.6){\scriptsize$11$}\rput(8.45,-0.85){\scriptsize$12$}
\rput(9.13,-0.6){\scriptsize$17$}\rput(9.13,-0.85){\scriptsize$18$}\rput(9.13,-1.1){\scriptsize$19$}\rput(9.13,-1.35){\scriptsize$20$}
}
\end{pspicture}   
2 \hat\chit_2(q)= 2\,(1+ q + q^2 + q^3 + 2 q^4 + 2 q^5 + \dots) + 2 q^3(1+q + 2q^2 + \dots).
\label{eq:matchings(iv)}
\end{equation}
Because $k=1$, as in example (ii), the approximations $\Lmn$ generate $\btl_{1,k} \simeq \tl_{n+1}$ and the indecomposable structures are preserved in the scaling limit. The module structure at $\xi = \frac{9 \pi}{10}$ is the direct sum of two Virasoro Kac modules,\footnote{We can rule out the result being an indecomposable Virasoro module because it would then be a self-extension of the Kac module $\Kac{1,2}$.  Such a module cannot exist by the results of \cite[Sec.~7]{RidSta09}.}
\begin{equation}
\kac_{n,1}^1 \big|_{(p,p')=(2,5)}^{{\rm reg.}\,B}\  \simeq
\begin{pspicture}[shift=-0.75](1,-0.75)(3,0.75)
\rput(0,-0.5){\rput(1.5,1){$\Itl_n^{1}$}
\rput(2.5,0){$\Itl_n^{7}$}
\psline[linewidth=.8pt,arrowsize=3pt 2]{->}(1.75,0.75)(2.25,0.25)}
\end{pspicture} 
 \quad \xrightarrow{n \rightarrow \infty} \qquad
\begin{pspicture}[shift=-0.475](-0.2,-0.6)(0.90,0.6)
\pscircle[fillstyle=solid,fillcolor=black](0,0){0.075}\rput(0,0.3){\scriptsize$\!\!\!-\tfrac1{5}$}\rput(0,-0.3){\sq}
\rput(0.75,0){\pscircle[fillstyle=solid,fillcolor=black](0.0,0){0.075}\rput(0.0,0.3){\scriptsize$\tfrac{14}{5}$}\rput(0.0,-0.3){\sr}}
\psline{->}(0.075,0)(0.675,0)
\end{pspicture}
\ \oplus \ 
\begin{pspicture}[shift=-0.475](-0.15,-0.6)(0.95,0.6)
\pscircle[fillstyle=solid,fillcolor=black](0,0){0.075}\rput(0,0.3){\scriptsize$\!\!\!-\tfrac1{5}$}\rput(0,-0.3){\sq}
\rput(0.75,0){\pscircle[fillstyle=solid,fillcolor=black](0.0,0){0.075}\rput(0.0,0.3){\scriptsize$\tfrac{14}{5}$}\rput(0.0,-0.3){\sr}}
\psline{->}(0.075,0)(0.675,0)
\end{pspicture}\qquad \quad \text{($\xi =  \tfrac{9 \pi}{10}$).}
\end{equation}
Varying $\xi = \frac{9 \pi}{10}$ to $\xi = \frac{9 \pi}{10} + \epsilon$, with $\epsilon = 0.01, 0.02, \dots$, we find that the 
positions \eqref{eq:radicalpositions(iv)} and the matchings \eqref{eq:matchings(iv)} remain unchanged, though the prefactors of $2$ in \eqref{eq:matchings(iv)} become $\epsilon$-dependent. As in the example \eqref{eq:2charexample}, the character approximation seems to be 
$(1+ q^{f_n(\epsilon)})\hat \chit_2$,
with the states associated to $q^{f_n(\epsilon)}\hat \chit_2$ 
drifting off to infinity in the scaling limit. Here, the matchings \eqref{eq:matchings(iv)} are unchanged for small $\epsilon$ and the states associated to 
$q^{f_n(\epsilon)}\hat \chit_2$ are among the first radical states for $n=16$,
only because the corresponding sequences are drifting off to infinity very slowly.
Thus,
for the other $\xi$ in regime~$B$, 
\begin{equation}
\kac_{n,1}^1 \big|_{(p,p')=(2,5)}^{{\rm reg.}\,B}\  \simeq
\begin{pspicture}[shift=-0.75](1,-0.75)(3,0.75)
\rput(0,-0.5){\rput(1.5,1){$\Itl_n^{1}$}
\rput(2.5,0){$\Itl_n^{7}$}
\psline[linewidth=.8pt,arrowsize=3pt 2]{->}(1.75,0.75)(2.25,0.25)}
\end{pspicture} 
 \quad \xrightarrow{n \rightarrow \infty} \qquad
\begin{pspicture}[shift=-0.475](-0.2,-0.6)(0.90,0.6)
\pscircle[fillstyle=solid,fillcolor=black](0,0){0.075}\rput(0,0.3){\scriptsize$\!\!\!-\tfrac1{5}$}\rput(0,-0.3){\sq}
\rput(0.75,0){\pscircle[fillstyle=solid,fillcolor=black](0.0,0){0.075}\rput(0.0,0.3){\scriptsize$\tfrac{14}{5}$}\rput(0.0,-0.3){\sr}}
\psline{->}(0.075,0)(0.675,0)
\end{pspicture}
\qquad \quad \text{($\xi \in (\tfrac{4\pi}5,\pi)\setminus \{\tfrac{9 \pi}{10}\}$).}
\end{equation}
The point $\xi =  \tfrac{9 \pi}{10}$ is noteworthy because each irreducible composition factor of $\stan_{n,1}^1$ appears to split in two in the scaling limit.

\paragraph{Example (v): $\boldsymbol{\mathcal{LM}(1,2)\ {\rm with}\ k=3,\ d=3}$ in regime $\boldsymbol A$.}
For $k \ge 2$, the structures of the standard modules $\stan_{n,k}^d$ over $\btl_{n,k}(\beta)$ are not known. The radical $\Rad_{n,k}^d$ and quotient $\stan_{n,k}^d/\Rad_{n,k}^d$ are nevertheless well-defined subspaces (provided that $\gramprodk{\cdot}{\cdot}$ is itself well-defined), thus allowing us to carry on with the radical/quotient analysis.

Here, $\beta = 0$ and $\det \grammat_{n,3}^3 \neq 0$ for all $n \in 2 \ZZ_+$, 
so the full standard module belongs to the quotient.
The character $\chit_{2,4}$ obtained from the character analysis decomposes as the sum of the irreducible modules for $\Delta = -\frac18$ and $\Delta = \frac{15}8$. \cref{TheConjecture} then states that 
\begin{equation}
\kac_{n,3}^3 \big|_{(p,p')=(1,2)}^{{\rm reg.}\,A}  \quad \xrightarrow{n \rightarrow \infty} \qquad 
\begin{pspicture}[shift=-0.475](-0.2,-0.6)(0.95,0.6)
\pscircle[fillstyle=solid,fillcolor=black](0,0){0.075}\rput(0,0.3){\scriptsize$\!\!\!-\tfrac18$}\rput(0,-0.3){\sq}
\rput(0.75,0){\pscircle[fillstyle=solid,fillcolor=black](0.0,0){0.075}\rput(0.0,0.3){\scriptsize$\tfrac{15}8$}\rput(0.0,-0.3){\sq}}
\rput(0.375,0){$\oplus$}
\end{pspicture}\ .
\end{equation}
On the lattice side, this example is one where, as discussed below \eqref{eq:WallIdentity}, the module $\stan_{n,3}^3 = \kac_{n,3}^3$ is not necessarily irreducible, 
even though $\det \grammat_{n,3}^3 \neq 0$, because $\btl_{n,3}$ is not well defined diagrammatically at $\beta = 0$.
That each irreducible Virasoro module comes from the scaling limit of an irreducible $\btl_{n,k}(\beta)$-module is plausible here. Indeed, even though $\Rad_{n,3}^3$ is trivial, $\stan_{n,3}^3$ is not irreducible; it has a proper submodule generated by link states with $0$ or $1$ arcs going to the seam: Link states with $2$ or $3$ such arcs cannot be created by the action of $\btl_{n,k}(\beta)$ because of vanishing Chebyshev polynomials. 
This can also be traced back to the fact that for $\beta = 0$, $\btl_{n,3}(\beta) \simeq\btl_{n,1}(\beta)$, see \cref{sec:dimB}.

\paragraph{Example (vi): $\boldsymbol{\mathcal{LM}(1,2)\ {\rm with}\ k=3,\ d=2}$ in regime $\boldsymbol A$.}
In this example, the analysis of \cref{sec:characters} allows us to guess the character up to grade $3$. In this case however, the eigenvalues are known exactly and yield $\chit_{2,3}$ according to the conjectured selection rules \cite{PRV12}. This character is the sum of three irreducible characters with $\Delta = 0 ,1$ and $3$. The module $\Kac{2,3}$ has a non-trivial structure, with the first arrow pointing towards the $\Delta = 0$ factor. As we now show, this is reflected in the radical/quotient analysis in a non-trivial way.

The structure of $\stan_{n,3}^2$ is unknown, yet we know from \eqref{eq:Gnkd} that $\det \grammat_{n,3}^2$ is zero for $n \in 2 \mathbb Z_+ - 1$.  The ground state is found to be annihilated by $\grammat_{n,3}^2$, 
in all cases, and is thus in the radical. Moreover, the labels of the first sequences belonging to the radical appear to be converging quickly, with
\begin{equation}
n = 13:  \# 1, \# 4, \#6, \#7, \#9, \#10, \#11, \#15, \#16, \#17, \#18, \#21, \#22, \#23, \#24, \#25, \#26, \#28,\dots.
\end{equation}
Comparing these with the coefficients at each grade, we find that the radical and quotient states can be distributed among the irreducibles as
\begin{equation}
\begin{pspicture}[shift=-1.0](0,-1.30)(0.1,0.15)
\rput(0.41,0){
\rput(0.69,-0.6){\scriptsize$\#:$}
\rput(1.20,-0.6){\scriptsize$1$}
\rput(1.71,-0.6){\scriptsize$4$}
\rput(2.37,-0.6){\scriptsize$6$}
\rput(3.15,-0.6){\scriptsize$9$}
\rput(3.15,-0.85){\scriptsize$10$}
\rput(3.97,-0.6){\scriptsize$15$}
\rput(3.97,-0.85){\scriptsize$16$}
\rput(4.78,-0.6){\scriptsize$21$}
\rput(4.78,-0.85){\scriptsize$22$}
\rput(4.78,-1.10){\scriptsize$23$}
\rput(4.78,-1.35){\scriptsize$24$}
}
\rput(0.25,0){
\rput(6.855,-0.6){\scriptsize$2$}
\rput(7.36,-0.6){\scriptsize$3$}
\rput(7.89,-0.6){\scriptsize$5$}
\rput(8.62,-0.6){\scriptsize$8$}
\rput(8.62,-0.85){\scriptsize$12$}
\rput(9.45,-0.6){\scriptsize$13$}
\rput(9.45,-0.85){\scriptsize$14$}
\rput(9.45,-1.10){\scriptsize$19$}
\rput(10.26,-0.60){\scriptsize$20$}
\rput(10.26,-0.85){\scriptsize$27$}
\rput(10.26,-1.10){\scriptsize$?$}
\rput(10.26,-1.35){\scriptsize$?$}
}
\rput(12.54,-0.6){\scriptsize$7$}
\rput(13.02,-0.6){\scriptsize$11$}
\rput(13.67,-0.6){\scriptsize$17$}
\rput(13.67,-0.85){\scriptsize$18$}
\rput(14.47,-0.60){\scriptsize$25$}
\rput(14.47,-0.85){\scriptsize$26$}
\end{pspicture}   
\hat \chit_6(q) = (1 + q^2 + q^3 + 2q^4 + 2 q^5 + 4 q^6 + \dots) + q\,(1 + q + q^2 +2 q^3 + 3 q^4 + 4 q^5 + \dots) + q^3(1 + q + 2 q^2 + 2 q^3 + \dots).
\label{eq:matching(vi)}
\end{equation} 
In this case, the assignments of states to the $\Delta = 0$ and $\Delta = 3$ factors at a given grade are arbitrary when both have at least one contributing state. The conjectured structure and the radical/quotient are again consistent:
\begin{equation}
\kac_{n,3}^2 \big|_{(p,p')=(1,2)}^{{\rm reg.}\,A}  \quad \xrightarrow{n \rightarrow \infty} \qquad 
\begin{pspicture}[shift=-0.3](-0.2,-0.4)(1.70,0.4)
\pscircle[fillstyle=solid,fillcolor=black](0,0){0.075}\rput(0,0.3){\scriptsize$0$}\rput(0,-0.3){\sr}
\rput(0.75,0){\pscircle[fillstyle=solid,fillcolor=black](0.0,0){0.075}\rput(0.0,0.3){\scriptsize$1$}\rput(0.0,-0.3){\sq}}
\rput(1.50,0){\pscircle[fillstyle=solid,fillcolor=black](0.0,0){0.075}\rput(0.0,0.3){\scriptsize$3$}\rput(0.0,-0.3){\sr}}
\psline{<-}(0.075,0)(0.675,0)
\psline{->}(0.825,0)(1.425,0)
\end{pspicture}\ .
\end{equation}
The fact that the integers in \eqref{eq:matching(vi)} match up to grade $5$ is a great success of this analysis and provides convincing evidence to support \cref{TheConjecture}. Other examples  in \cref{tab:12,tab:13} with more than two composition factors are treated with the same method and have similar accuracy. In these cases, the radical/quotient status of the highest graded factors sometimes remains unknown 
as the analysis usually fails beyond grade $5$. This is then indicated by {\scshape (u)}.

\paragraph{Example (vii): $\boldsymbol{\mathcal{LM}(1,3)\ {\rm with}\ k=2,\ d=1}$ in regime $\boldsymbol A$.}
Here, even though $\btl_{n,2}(\beta)$ is well-defined both diagrammatically and algebraically,
the analysis with the bilinear form runs into a technical difficulty. For formal $\beta$, the entries of the Gram matrix 
$\grammat_{n,2}^1$, $n \in 2\mathbb Z_+-1$, 
all have an overall factor of $U_2 = \beta^2-1$. Upon specialising, we find that
\begin{equation}
\lim_{\beta \rightarrow -1} \gramprod{v}{w}^{\textrm{\tiny$(1)$}}\! = 0 \qquad \text{(\(v,w\in \stan_{n,2}^1\), \(n \in 2 \ZZ_+ - 1\)).} 
\end{equation} 
The entire standard module is then in the radical. This is due to vanishing Chebyshev polynomials produced by loops 
interacting with the Wenzl-Jones projector in the seam. For values $\beta = \beta_c$ where this occurs, 
a non-zero invariant bilinear form on $\kac_{n,k}^d$ is obtained by dividing out by $\beta -\beta_c$ and taking a limit: 
\begin{equation}\label{eq:underform}
\langle\!\langle v|w\rangle\!\rangle^{\textrm{\tiny$(k)$}}\!= \lim_{\beta \rightarrow \beta_{c}}\frac{\gramprodk{v}{w}}{\beta - \beta_{c}}.
\end{equation}
The determinant of the resulting renormalised Gram matrix is given by
\begin{equation}
\lim_{\beta \rightarrow \beta_c} \frac{\det \grammat_{n,k}^d}{(\beta-\beta_c)^{\dim \stan_{n,k}^d}}.
\end{equation} 
From \eqref{eq:Gnkd}, this determinant is non-zero
for $k=2$, $d=1$ and $\beta = -1$, for all $n$. The character obtained from the analysis of \cref{sec:characters} is 
$\chit_{1,2}$ and is irreducible for $c = -7$, $\Delta = -\frac14$. 
The rest of the analysis is identical to that of example (i): The determinant of the renormalised form is non-zero and
\begin{equation}
\kac_{n,2}^1\big|_{(p,p')=(1,3)}^{{\rm reg.}\,A} \qquad \xrightarrow{n \rightarrow \infty} \qquad
\begin{pspicture}[shift=-0.475](-0.2,-0.6)(0.2,0.6)
\pscircle[fillstyle=solid,fillcolor=black](0,0){0.075}\rput(0,0.3){\scriptsize$\!\!\!-\tfrac{1}{4}$}\rput(0,-0.3){\squ}
\end{pspicture} \ ,
\end{equation}
where {\scshape (\underline{q})} indicates that the analysis required the renormalised form \eqref{eq:underform}. The same convention is used in \cref{sec:Gramdata}. 

\paragraph{Example (viii): $\boldsymbol{\mathcal{LM}(1,3)\ {\rm with}\ k=3,\ d=4}$ in regimes $\boldsymbol A$ and $\boldsymbol B$.}
The difficulty in this case is that some matrix entries of $\grammat_{n,3}^4$ are singular at $\beta = -1$. This also happened for $\beta = 0$ in \eqref{eq:Gmat422}. At values $\beta = \beta_c$ where this occurs, a well-defined bilinear form is again obtained by a limiting procedure, but now 
\begin{equation}
\langle\!\langle v|w\rangle\!\rangle^{\textrm{\tiny$(k)$}}\!= \lim_{\beta \rightarrow \beta_{c}}\! (\beta - \beta_c)\, \gramprodk{v}{w}.
\label{eq:overform}
\end{equation}
Its determinant is given by
\begin{equation}
\lim_{\beta \rightarrow \beta_c} (\beta-\beta_c)^{\dim \stan_{n,k}^d} \, \det \grammat_{n,k}^d
\end{equation} 
and, in the current example, this evaluates to zero, for all $n$. Both the radical, with respect to this renormalised bilinear form, and the corresponding quotient of $\kac_{n,3}^4$ are non-trivial.

In regime $B$, the character is $\chit_{1,5}$ and splits as a direct sum of two irreducible ones. The rest of the analysis is identical to example (ii): The quotient and radical of $\kac_{n,3}^4$ each appear to survive the limit,
\begin{equation}
\kac_{n,3}^4 \big|_{(p,p')=(1,3)}^{{\rm reg.}\,B}  \quad \xrightarrow{n \rightarrow \infty} \qquad 
\begin{pspicture}[shift=-0.3](-0.2,-0.4)(0.75,0.4)
\pscircle[fillstyle=solid,fillcolor=black](0,0){0.075}\rput(0,0.3){\scriptsize$0$}\rput(0,-0.3){\sqo}
\rput(0.75,0){\pscircle[fillstyle=solid,fillcolor=black](0.0,0){0.075}\rput(0.0,0.3){\scriptsize$1$}\rput(0.0,-0.3){\sro}}
\psline{->}(0.075,0)(0.675,0)
\end{pspicture}\ \ ,
\end{equation}
with {\scshape (\textoverline{q})} and {\scshape (\textoverline{r})} now indicating that the bilinear form used is of the form 
\eqref{eq:overform}. 

In regime $A$, the character is $\chit_{2,5}$ and \cref{TheConjecture} states that the Virasoro module is $\Kac{2,5}$ and has thus three irreducible components. At the level of $\kac_{n,3}^4$, the quotient appears to drift off to infinity, as in example (iii), so the scaling limit is associated with the radical.  The conjectured limit is
\begin{equation}
\kac_{n,3}^4 \big|_{(p,p')=(1,3)}^{{\rm reg.}\,A}  \quad \xrightarrow{n \rightarrow \infty} \qquad 
\begin{pspicture}[shift=-0.3](-0.2,-0.4)(1.5,0.4)
\pscircle[fillstyle=solid,fillcolor=black](0,0){0.075}\rput(0,0.3){\scriptsize$\!\!\!-\tfrac14$}\rput(0,-0.3){\sro}
\rput(0.75,0){\pscircle[fillstyle=solid,fillcolor=black](0.0,0){0.075}\rput(0.0,0.3){\scriptsize$\tfrac74$}\rput(0.0,-0.3){\sro}}
\rput(1.50,0){\pscircle[fillstyle=solid,fillcolor=black](0.0,0){0.075}\rput(0.0,0.3){\scriptsize$\tfrac{15}4$}\rput(0.0,-0.3){\sro}}
\psline{<-}(0.075,0)(0.675,0)
\psline{->}(0.825,0)(1.425,0)
\end{pspicture}
\ \ .
\end{equation}
This does not confirm or contradict \cref{TheConjecture} because we do not know the substructure of the radical of $\kac_{n,3}^4$.  A refined analysis of its structure is therefore required in this case. Unfortunately, this analysis is beyond the scope of the paper.
The indicators {\scshape (\textoverline q), (\textoverline r)} and {\scshape (\textoverline u)}, 
used here and in \cref{sec:Gramdata}, indicate that the renormalised form \eqref{eq:overform} is used.

%
\section{Conformal analysis} \label{sec:CFTanalysis}
%

In this section, we enter the third and final phase of this work in which we directly establish results in the continuum using \cft{} methods, in order to confirm the lattice analysis and conjectures of the previous sections.  We will first discuss the modular transformation properties of the Virasoro characters and employ a continuous version \cite{CreLog13,RidVer14} of the well known Verlinde formula to determine (Grothendieck) fusion coefficients and, thereby, characters of fusion products.  In particular, we will confirm, at the level of characters, the fusion rules of the Virasoro Kac modules.  These turn out to be consistent with the results of the lattice prescription for fusion in \cref{sec:LatticeFusion} and the characters $\chit_{r,s}$ that are believed to arise in the scaling limit of the lattice Kac modules, see \cref{TheConjecture}.  In \cref{sec:NGKfusion}, we will use the Nahm-Gaberdiel-Kausch algorithm to perform explicit fusion computations that confirm these fusion rules at the level of the modules.  The evidence again supports \cref{TheConjecture}.

%
\subsection{Modular transformations and a Verlinde formula} \label{sec:Verlinde} 
%

For rational \cfts{}, one of the more efficient means to compute fusion rules is to determine the modular S-transforms of the characters and apply the Verlinde formula.  For \lcfts{}, any Verlinde-like formula obtained from the modular properties of characters cannot compute the fusion multiplicities themselves, because characters cannot distinguish between an indecomposable module and the direct sum of its composition factors.  Instead, one expects that such a Verlinde-like formula would only compute the structure constants of the Grothendieck ring of fusion, in which indecomposables are identified with the direct sum of their composition factors.\footnote{This argument also implicitly assumes that the characters of the irreducible modules are linearly independent.}
We will denote the image of a module $M$ in the Grothendieck ring by $\Gr{M}$.  The fusion product $\fuse$ therefore defines the Grothendieck fusion product $\Grfuse$ by
\begin{equation}
\Gr{M} \Grfuse \Gr{N} = \Gr{M \fuse N}.
\end{equation}
A Verlinde-like formula should then determine the composition factors, or equivalently the character, of a fusion product.

In this section, we will determine the modular S-transforms of certain Virasoro module characters and substitute the results into a Verlinde-like formula \eqref{eq:Verlinde} to obtain candidates for the Grothendieck fusion rules.  The formula proposed is, in fact, an extremely natural generalisation of the standard Verlinde formula for rational (bosonic) \cft{} from which it differs only in that a sum is replaced by an integral.  This difference reflects the fact that the spectrum of modules is discrete in the rational case, whereas it is continuous in the theory being considered here.  We note that the same continuous Verlinde formula has been applied to many other \lcfts{} \cite{CreRel11,CreMod12,BabTak12,CreMod13,RidMod14,RidBos14} and the results match the Grothendieck fusion rules (when known) 
perfectly.\footnote{Other, typically model-dependent, Verlinde-like formulae have been proposed for logarithmic
conformal field theories, see \cite{FucNon04,FloVer07,GabFro08,GaiRad09,PeaGro10,RasVer10,PeaCos11}.}
We therefore expect that this will remain true here and will check this with explicit fusion computations in the following section.

There are a few mathematical provisos to this expectation.  Technically, a Grothendieck ring of fusion will only exist if fusing with any given module defines an exact functor.  While this is not true in general, see \cite{GabFus09} for a counterexample, the modules for which it is true form a subring of the fusion ring \cite{KazTenIV94} (assuming that fusion defines a tensor structure on the physically relevant category of modules).  On the other hand, it is strongly believed that the modules that arise as boundary sectors of boundary \cfts{} do define exact functors.  We will therefore assume that the Virasoro Kac modules, defined in \cref{sec:VKac}, have this property.

\begin{conj}
\leavevmode
\begin{itemize}
\item Fusing with any given Kac module
$\Kac{r,s}$, with $r,s \in \ZZ_+$, defines an \emph{exact} endofunctor on the (non-abelian) category of Virasoro modules generated from the Kac modules by fusion.
\item The resulting product on the Grothendieck fusion ring generated by the Kac modules then coincides with that defined by the continuous Verlinde formula of \cref{eq:Verlinde}. 
\end{itemize}
\end{conj}
Having explicitly addressed these technicalities, we can now turn to one of the key results of this section, the confirmation of the lattice fusion product \eqref{eq:latticersfusion}.  Specifically, we deduce the Grothendieck fusion rule
\begin{equation} \label{FR:Kr1K1s}
\Gr{\Kac{r,1}} \Grfuse \Gr{\Kac{1,s}} = \Gr{\Kac{r,s}} \qquad \text{(\(r,s \in \ZZ_+\))},
\end{equation}
valid for all central charges, hence, in particular, for all $p$ and $p'$ defining the logarithmic minimal models.

\subsubsection{Modular transformations}

To study the modular transformation properties of the characters of the Kac modules $\Kac{r,s}$, we start with those of the \FFms{} $\FF{\lambda}$ (see \cref{sec:FFm} for background):
\begin{equation}
\fch{\FF{\lambda}}{\tau} = \frac{q^{\Delta_{\lambda} - c/24}}{\prod_{j=1}^{\infty} \brac{1-q^j}} = \frac{q^{\brac{\lambda - Q/2}^2 / 2}}{\func{\eta}{\tau}} \qquad \text{(\(q = \ee^{2 \pi \ii \tau}\)).}  
\end{equation}
Here, $Q$ is defined in \eqref{eq:DefQ}.  We remark immediately that this character formula cannot distinguish between $\FF{\lambda}$ and its contragredient dual $\FF{Q - \lambda}$.  When a \FFm{} is irreducible, it is self-contragredient, so the characters of the distinct irreducible \FFms{} are linearly independent, as required.  We could, as in \cite{CreLog13}, try to restore the linear independence of all these characters by incorporating the eigenvalue of the Heisenberg zero mode.  However, this turns out to be unnecessary for the application at hand because we will regard these characters as pertaining to modules over the Virasoro algebra, which does not contain this mode.  We do note, however, that the characters of the Kac modules $\Kac{r,s}$, with $r,s \in \ZZ_+$, are all linearly independent.

The modular S-transformation of the Feigin-Fuchs characters, where $\lambda$ is restricted to the range $[Q/2,\infty)$, is given by
\begin{subequations}
\begin{align}
\fch{\FF{\lambda}}{-1/\tau} = \int_{Q/2}^{\infty} \Smat{\FF{\lambda}}{\FF{\mu}} \fch{\FF{\mu}}{\tau} \: \dd \mu, \\
\Smat{\FF{\lambda}}{\FF{\mu}} = 2 \cos \sqbrac{2 \pi \brac{\lambda - Q/2} \brac{\mu - Q/2}}.
\end{align}
\end{subequations}
This follows from a straightforward computation involving a gaussian integral that converges when $\Im \tau > 0$, hence $\abs{q} < 1$.  However, the continuous nature of the spectrum leads us to expect that many of the quantities we subsequently calculate will be singular distributions.  As it can be confusing to allow endpoints to the integration domain, as in $\lambda \in [Q/2,\infty)$, 
when computing with these generalised functions, we will re-define the above S-transformation once and for all so that the integration range is open:
\begin{subequations} \label{eq:DefS}
\begin{align}
\fch{\FF{\lambda}}{-1/\tau} = \int_{\RR} \Smat{\FF{\lambda}}{\FF{\mu}} \fch{\FF{\mu}}{\tau} \: \dd \mu, \\
\Smat{\FF{\lambda}}{\FF{\mu}} = \cos \sqbrac{2 \pi \brac{\lambda - Q/2} \brac{\mu - Q/2}}. \label{S:FF}
\end{align}
\end{subequations}
Of course, we now need to remember that $\FF{\lambda}$ and $\FF{Q - \lambda}$ are identical as far as characters are concerned.  More precisely, we note the identification $\tGr{\FF{\lambda}} = \tGr{\FF{Q - \lambda}}$ in the Grothendieck ring.

We next show that the S-matrix defined by \eqref{S:FF} is symmetric, 
$\tSmat{\FF{\lambda}}{\FF{\mu}} = \tSmat{\FF{\mu}}{\FF{\lambda}}$, and unitary:
\begin{align}
\int_{\RR} \Smat{\FF{\lambda}}{\FF{\mu}} \Smat{\FF{\nu}}{\FF{\mu}}^* \: \dd \mu 
&= \frac{1}{2} \int_{\RR} \tbrac{\cos \sqbrac{2 \pi \brac{\lambda - \nu} \mu} + \cos \sqbrac{2 \pi \brac{\lambda + \nu - Q} \mu}} \: \dd \mu \notag \\
&= \frac{1}{2} \sqbrac{\func{\delta}{\nu = \lambda} + \func{\delta}{\nu = Q - \lambda}} = \func{\delta}{\nu = \lambda}.
\end{align}
In the last step, we have recalled that $\lambda$ and $Q - \lambda$ should be identified when they correspond to indices of \FFms{} because of the Grothendieck identity $\tGr{\FF{\lambda}} = \tGr{\FF{Q - \lambda}}$.  This calculation also shows, using symmetry and reality, that the S-matrix squares to the conjugation permutation, conjugation being trivial (at the level of characters) for Virasoro modules.  These three properties suggest that we may expect meaningful results from the Verlinde formula.

We remark that the setup described above fits in with the \emph{standard module formalism} proposed in \cite{CreLog13,RidVer14} for (logarithmic) \cfts{}.  Specifically, the standard modules of this formalism are the \FFms{} $\FF{\lambda}$ and the typical modules are the $\FF{\lambda}$ with $\lambda \neq \lambda_{r,s} = -\alpha' \brac{r-1} + \alpha \brac{s-1}$, for $r,s \in \ZZ$, see \eqref{eq:DefLambdaRS} and \eqref{eq:DefAlphas}.  
All modules corresponding to $\lambda = \lambda_{r,s}$ are therefore atypical.

In particular, the Kac modules $\Kac{r,s}$ are atypical.  We have defined $\Kac{r,s}$, for $r,s \in \ZZ_+$, as a certain submodule of the \FFm{} $\FF{r,s}$, see \cref{def:VKM}.
Inspection shows that the quotient $\FF{r,s} / \Kac{r,s}$ is not a \FFm{} in general, but that the \ssvs{} of the quotient, and hence its character, match those of $\FF{r,-s}$ (and $\FF{-r,s}$).  We may therefore write
\begin{equation} \label{ch:Kac}
\ch{\Kac{r,s}} = \ch{\FF{r,s}} - \ch{\FF{r,-s}} = \ch{\FF{r,s}} - \ch{\FF{-r,s}} \qquad \text{(\(r,s \in \ZZ_+\)).}
\end{equation}
Virasoro Kac modules $\Kac{r,s}$ with negative labels were not defined in \cref{sec:VKac}. Nevertheless, we remark that if we formally extend these character formulae to arbitrary $r,s \in \ZZ$, then we arrive at
\begin{equation} \label{eq:ExtKacVanish}
\ch{\Kac{r,-s}} = -\ch{\Kac{r,s}} = \ch{\Kac{-r,s}} \qquad \Ra \qquad 
\ch{\Kac{r,0}} = \ch{\Kac{0,s}} = 0, \quad \ch{\Kac{-r,-s}} = \ch{\Kac{r,s}}.
\end{equation}
This last identity is consistent with defining $\Kac{-r,-s}$ to be the contragredient of the Kac module $\Kac{r,s}$ 
(see \cite{BGT12} for another example where such ``generalised'' Kac modules are considered). 

Specialising to $\lambda = \lambda_{r,s}$, the S-matrix coefficients become
\begin{equation}
\Smat{\FF{r,s}}{\FF{\mu}} = \cos \sqbrac{2 \pi \brac{r \alpha' - s \alpha} \brac{\mu - Q/2}}.
\end{equation} 
It follows that the S-matrix entries for transforming the Kac module characters are given by
\begin{equation} \label{S:KF}
\Smat{\Kac{r,s}}{\FF{\mu}} = \Smat{\FF{r,s}}{\FF{\mu}} - \Smat{\FF{r,-s}}{\FF{\mu}} = 2 \sin \sqbrac{2 \pi r \alpha' \brac{\mu - Q/2}} \sin \sqbrac{2 \pi s \alpha \brac{\mu - Q/2}}.
\end{equation}
We remark that the S-transformation maps Kac characters to a linear combination of Feigin-Fuchs characters, not a linear combination of Kac characters.  The point here is that the Feigin-Fuchs characters 
are taken as the preferred topological basis of the space spanned by all the characters.  All computations are therefore performed in this basis.

\subsubsection{A Verlinde formula}

We will take the vacuum module to be the atypical Kac module $\Kac{1,1}$.  The obvious continuum analogue of the Verlinde formula is then
\begin{equation} \label{eq:Verlinde}
\Gr{M} \Grfuse \Gr{N} = \int_{\RR} \fuscoeff{M}{N}{\FF{\nu}} \Gr{\FF{\nu}} \: \dd \nu, \qquad 
\fuscoeff{M}{N}{\FF{\nu}} = \int_{\RR} \frac{\Smat{M}{\FF{\rho}} \Smat{N}{\FF{\rho}} \Smat{\FF{\nu}}{\FF{\rho}}^*}{\Smat{\Kac{1,1}}{\FF{\rho}}} \: \dd \rho.
\end{equation}
We emphasise that any Verlinde-like formula can only compute the Grothendieck fusion coefficients that describe the fusion product at the character level.  With this in mind, we note that the unitarity of the S-matrix ensures that the vacuum module is the unit of the Grothendieck fusion ring.

If we try to take both $M$ and $N$ to be \FFms{}, then we arrive at
\begin{equation}
\fuscoeff{\FF{\lambda}}{\FF{\mu}}{\FF{\nu}} = \int_{\RR} \frac{\cos \sqbrac{2 \pi \brac{\lambda - Q/2} \rho} \cos \sqbrac{2 \pi \brac{\mu - Q/2} \rho} \cos \sqbrac{2 \pi \brac{\nu - Q/2} \rho}}{2 \sin \sqbrac{2 \pi \alpha' \rho} \sin \sqbrac{2 \pi \alpha \rho}} \: \dd \rho,
\label{eq:fuscoeff}
\end{equation}
which is not easily interpreted.  Even after substituting \eqref{eq:fuscoeff} into $\tGr{\FF{\lambda}} \Grfuse \tGr{\FF{\mu}}$ in \eqref{eq:Verlinde} and performing the integration over $\nu$, the remaining integration over $\rho$ is still divergent.  This indicates that the character of $\FF{\lambda} \Grfuse \FF{\mu}$ is not defined --- the multiplicity of states at (at least) one conformal grade is infinite.  However, this is consistent with expectations because \FFms{} are not quasirational in the sense of Nahm \cite{NahQua94}, so their fusion products need not be finitely generated.  This Grothendieck fusion product cannot be computed from the Verlinde formula.

However, Kac modules are quasirational, so we turn to the Grothendieck fusion of $\Kac{1,2}$ with $\FF{\mu}$:
\begin{align}
\fuscoeff{\Kac{1,2}}{\FF{\mu}}{\FF{\nu}} &= 
\int_{\RR} \frac{\sin \sqbrac{4 \pi \alpha \rho} \cos \sqbrac{2 \pi \brac{\mu - Q/2} \rho} \cos \sqbrac{2 \pi \brac{\nu - Q/2} \rho}}{\sin \sqbrac{2 \pi \alpha \rho}} \: \dd \rho \notag \\
&= 2 \int_{\RR} \cos \sqbrac{2 \pi \alpha \rho} \cos \sqbrac{2 \pi \brac{\mu - Q/2} \rho} \cos \sqbrac{2 \pi \brac{\nu - Q/2} \rho} \: \dd \rho \notag \\
&= \frac{1}{2} \sqbrac{\func{\delta}{\nu = \mu - \alpha} + \func{\delta}{\nu = \mu + \alpha} + \func{\delta}{\nu = Q - \mu - \alpha} + \func{\delta}{\nu = Q - \mu + \alpha}} \notag \\
&= \func{\delta}{\nu = \mu - \alpha} + \func{\delta}{\nu = \mu + \alpha}.
\end{align}
Replacing $\Kac{1,2}$ by $\Kac{2,1}$ gives the same result, but with $\alpha$ replaced by $\alpha'$.  We therefore obtain the following Grothendieck fusion rules:
\begin{equation} \label{FR:K12FTyp}
\Gr{\Kac{1,2}} \Grfuse \Gr{\FF{\mu}} = \Gr{\FF{\mu - \alpha}} + \Gr{\FF{\mu + \alpha}}, \qquad 
\Gr{\Kac{2,1}} \Grfuse \Gr{\FF{\mu}} = \Gr{\FF{\mu - \alpha'}} + \Gr{\FF{\mu + \alpha'}}.
\end{equation}
In particular,
\begin{equation} \label{FR:K12Frs}
\Gr{\Kac{1,2}} \Grfuse \Gr{\FF{r,s}} = \Gr{\FF{r,s-1}} + \Gr{\FF{r,s+1}}, \qquad 
\Gr{\Kac{2,1}} \Grfuse \Gr{\FF{r,s}} = \Gr{\FF{r-1,s}} + \Gr{\FF{r+1,s}}.
\end{equation}
Combining this with \eqref{ch:Kac}, we effortlessly arrive at
\begin{equation} \label{FR:K12Krs}
\Gr{\Kac{1,2}} \Grfuse \Gr{\Kac{r,s}} = \Gr{\Kac{r,s-1}} + \Gr{\Kac{r,s+1}}, \qquad 
\Gr{\Kac{2,1}} \Grfuse \Gr{\Kac{r,s}} = \Gr{\Kac{r-1,s}} + \Gr{\Kac{r+1,s}},
\end{equation}
where we need to recall the symmetries of \eqref{eq:ExtKacVanish}, in particular that $\Gr{\Kac{r,0}} = \Gr{\Kac{0,s}} = 0$.  Associativity and distributivity then give the general Grothendieck fusion rules involving Kac modules:
\begin{subequations}
\begin{align}
\Gr{\Kac{r,s}} \Grfuse \Gr{\FF{\mu}} &= \sideset{}{'}\sum_{i=-\brac{r-1}}^{r-1} \ \sideset{}{'}\sum_{j=-\brac{s-1}}^{s-1} \Gr{\FF{\mu + i \alpha' + j \alpha}}, \label{FR:Krs1} \\
\Gr{\Kac{r,s}} \Grfuse \Gr{\Kac{r',s'}} &= \sideset{}{'}\sum_{r''=\abs{r-r'}+1}^{r+r'-1} \ \sideset{}{'}\sum_{s''=\abs{s-s'}+1}^{s+s'-1} \Gr{\Kac{r'',s''}}, \label{FR:Krs2}
\end{align}
\end{subequations}
where the primes on the sums indicate that the index increases in steps of $2$. 

The result \eqref{FR:Krs2} is consistent with the lattice-theoretic definition of the Kac modules with $r,s \in \ZZ_+$ and reduces to \eqref{FR:Kr1K1s} if $s=r'=1$.  These Grothendieck fusion rules have an obvious $\SLA{sl}{2}$ structure. As observed in \cite{PRZ06}, for $\Gr{\Kac{1,s}} \Grfuse \Gr{\Kac{1,s'}}$, this $\SLA{sl}{2}$ structure is automatically built-in from the lattice fusion prescription for $\kac_{n,0}^d \times \kac_{n,0}^{d'}$. We note that \eqref{FR:Krs2} is also consistent with the fusion rules for logarithmic minimal models conjectured in \cite{RasFus07} from lattice data (character analysis and Jordan block arguments). In the next section, we study the fusion products $\Kac{r,1} \fuse \Kac{1,s}$ using the combined power of the Verlinde formula \eqref{eq:Verlinde} and the Nahm-Gaberdiel Kausch fusion algorithm.

\subsection{Fusion of Virasoro Kac modules}\label{sec:NGKfusion}

The Verlinde formula gives the character of the fusion product of two Kac modules.  In favourable circumstances, this is sufficient to completely identify the fusion product, but in general more detailed information is required.  In particular, the Grothendieck fusion rules cannot distinguish between a (reducible but indecomposable) Kac module and its contragredient dual.  Neither can it distinguish between a direct sum of Kac modules and an indecomposable formed by gluing the Kac modules together.

To compute Kac module fusion rules explicitly, we will employ the algorithm of Nahm \cite{NahQua94} and Gaberdiel-Kausch \cite{GabInd96}, referring to the latter for a more detailed exposition.  For our purposes, it is enough to remark that this algorithm allows one to construct the fusion product $M \fuse N$, to any desired grade,\footnote{We assume from here on that the conformal dimensions of the states of both $M$ and $N$ are bounded below.} as a quotient of the usual tensor product (of complex vector spaces) $M \otimes N$ \cite{GabFus94}.  It does this by deriving coproduct formulae defining the action of the chiral algebra on this tensor product from the natural action of the algebra on the \opes{} of the theory.  Actually, locality lets one derive two seemingly different coproducts $\coprodsymb$ and $\altcoprodsymb$; their identification amounts to the quotienting that recovers the fusion product from the tensor product.  For the Virasoro algebra, the coproduct formulae may be distilled down to the following three master equations:
\begin{subequations} \label{eq:Master}
\begin{align}
\coproduct{L_n} &= \sum_{m=-1}^n \binom{n+1}{m+1} \brac{L_m \otimes \wun} + \brac{\wun \otimes L_n} & &\text{($n \ge -1$),} \label{eq:Master1} \\
\coproduct{L_{-n}} &= \sum_{m=-1}^{\infty} \binom{n+m-1}{n-2} \brac{-1}^{m+1} \brac{L_m \otimes \wun} + \brac{\wun \otimes L_{-n}} & &\text{($n \ge 2$),} \label{eq:Master2} \\
\brac{L_{-n} \otimes \wun} &= \sum_{m=n}^{\infty} \binom{m-2}{n-2} \coproduct{L_{-m}} + \brac{-1}^n \sum_{m=-1}^{\infty} \binom{n+m-1}{n-2} \brac{\wun \otimes L_m} & &\text{($n \ge 2$).} \label{eq:Master3}
\end{align}
\end{subequations}
The third formula is derived from the identification $\coprodsymb \sim \altcoprodsymb$, see \cite{GabInd96} for a precise statement, allowing us to omit all reference to $\altcoprodsymb$ entirely.

We remark that the infinite sum in \eqref{eq:Master2} and the second sum in \eqref{eq:Master3} are rendered finite if we assume that the conformal dimensions of the states of the modules being fused are bounded below.  This still leaves the first infinite sum in \eqref{eq:Master3}.  In practice, this sum is likewise truncated because one restricts attention to a collection of finite-dimensional quotients of the fusion product which have the property that $\coproduct{L_{-m}}$ acts as the zero operator, for $m$ sufficiently large, on each.  More precisely, let $U_k$ denote the subalgebra of the \uea{} of the Virasoro algebra generated by the set of monomials
\[
\set{L_{-n_1} \cdots L_{-n_r} \st n_1, \ldots, n_r > 0 ; \ n_1 + \cdots + n_r > k}
\]
and let $M^k$ denote the quotient $M / U_k M$.  Then, it was shown in \cite{NahQua94} for $k=0$ and \cite{GabInd96} for general $k$ that
\begin{equation} \label{eq:TruncFusProd}
\brac{M \fuse N}^k \equiv \frac{M \fuse N}{\func{U_k}{M \fuse N}} \subseteq M^{\text{ss}} \otimes N^k,
\end{equation}
where $M^{\text{ss}} = M / U_{\text{ss}} M$ is the \emph{special subspace} of $M$, obtained by quotienting by the action of the subalgebra 
\begin{equation}
U_{\text{ss}} = \vspn \set{L_{-n_1} \cdots L_{-n_r} \st n_1, \ldots, n_r \ge 2}.
\end{equation}
It is important to note that this result only identifies the (truncated) fusion product as a subspace of a potentially much larger tensor product.  In general, one must determine so-called \emph{spurious states} and set them to zero in the tensor product in order to recover the correct (truncated) fusion product.  These are non-trivial relations that one derives for the elements of $M^{\text{ss}} \otimes N^k$ from relations in $M$ and/or $N$.  Such relations typically arise because $M$ and/or $N$ is formed from a free module (for example, a Verma module) by setting a \sv{} to zero.  We will illustrate such spurious states in the examples of the next section.

\subsubsection{Explicit Examples}

We detail the use of the \NGK{} algorithm with two example computations of Kac module fusion rules of central charge $c=-2$.  In general, this algorithm is extremely computationally-intensive and the unambiguous identification of the fusion product is only possible, even with a computer implementation, for small values of $r$ and $s$.  However, having a working Verlinde formula identifies the character of the fusion product effortlessly and this information can be used to minimise the amount of explicit algorithmic computation needed.

\paragraph{Example 1:  $\Kac{2,1} \fuse \Kac{1,2}$.}
For this example, both $\Kac{2,1}$ and $\Kac{1,2}$ turn out to be irreducible (see \cref{fig:K21K12K22}) and their \hwss{} have respective conformal dimensions $1$ and $-\tfrac{1}{8}$, see \cref{fig:KacTables}.  The Grothendieck fusion rule \eqref{FR:Kr1K1s} tells us that the character of the fusion product is the character of $\Kac{2,2}$.  Since $\Kac{2,2}$ is likewise irreducible (\cref{fig:K21K12K22}), the Verlinde formula alone dictates that the fusion rule must be
\begin{equation} \label{FR:K21K12}
\Kac{2,1} \fuse \Kac{1,2} = \Kac{2,2}.
\end{equation}
In this case, the character of the fusion product uniquely specifies it as a module.
\begin{figure}
\begin{center}
\begin{tikzpicture}
  [->,node distance=1.1cm,>=stealth',semithick]
  \node[] (1) [] {$\FF{2,1}$};
  \node[] (1a) [below of =1] {$1$};
  \node[] (1b) [below of =1a] {$3$};
  \node[] (1c) [below of =1b] {$6$};
  \node[inner sep = 2pt] (1d) [below of =1c] {$\vdots$};
  \path[] (1b) edge node {} (1a)
          (1b) edge node {} (1c)
          (1d) edge node {} (1c);
  \node[] (2) [right = 2cm of 1] {$\FF{1,2}$};
  \node[] (2a) [below of =2] {$-\tfrac{1}{8}$};
  \node[] (2b) [below of =2a] {$\tfrac{15}{8}$};
  \node[] (2c) [below of =2b] {$\tfrac{63}{8}$};
  \node[inner sep = 2pt] (2d) [below of =2c] {$\vdots$};
  \path[] (2b) edge node {} (2a)
          (2b) edge node {} (2c)
          (2d) edge node {} (2c);
  \node[] (3) [right = 2cm of 2] {$\FF{2,2}$};
  \node[] (3a) [below of =3] {$\tfrac{3}{8}$};
  \node[] (3b) [below of =3a] {$\tfrac{35}{8}$};
  \node[] (3c) [below of =3b] {$\tfrac{99}{8}$};
  \node[inner sep = 2pt] (3d) [below of =3c] {$\vdots$};
  \path[] (3b) edge node {} (3a)
          (3b) edge node {} (3c)
          (3d) edge node {} (3c);
  \node[] (4) [right = 2cm of 3] {$\FF{1,3}$};
  \node[] (4a) [below of =4] {$0$};
  \node[] (4b) [below of =4a] {$1$};
  \node[] (4c) [below of =4b] {$3$};
  \node[inner sep = 2pt] (4d) [below of =4c] {$\vdots$};
  \path[] (4a) edge node {} (4b)
          (4c) edge node {} (4b)
          (4c) edge node {} (4d);
  \node[] (5) [right = 2cm of 4] {$\FF{2,3}$};
  \node[] (5a) [below of =5] {$0$};
  \node[] (5b) [below of =5a] {$1$};
  \node[] (5c) [below of =5b] {$3$};
  \node[inner sep = 2pt] (5d) [below of =5c] {$\vdots$};
  \path[] (5b) edge node {} (5a)
          (5b) edge node {} (5c)
          (5d) edge node {} (5c);
\end{tikzpicture}
\caption{The \ssv{} structure, for $c=-2$, of $\FF{2,1}$, $\FF{1,2}$, $\FF{2,2}$, $\FF{1,3}$ and $\FF{2,3}$.  Here, we label the \ssvs{} by their conformal dimensions.  As the corresponding Kac modules are the submodules generated by the \ssvs{} whose grades are less than $2$, $2$, $4$, $3$ and $6$, respectively, we see that $\Kac{2,1}$, $\Kac{1,2}$ and $\Kac{2,2}$ are all irreducible, whereas $\Kac{1,3}$ and $\Kac{2,3}$ correspond to the top two and three \ssvs{} in $\FF{1,3}$ and $\FF{2,3}$, respectively.} \label{fig:K21K12K22}
\end{center}
\end{figure}

To confirm this using the \NGK{} algorithm requires some work.  First, let $\ket{v}$ and $\ket{w}$ be the \hwss{} of $\Kac{2,1}$ and $\Kac{1,2}$, respectively. These states satisfy the relations
\begin{equation} \label{eq:Ex1Rels}
(L_{-1}^2 - 2 L_{-2}) \ket{v} = 0, \qquad (L_{-1}^2 - \tfrac{1}{2} L_{-2}) \ket{w} = 0
\end{equation}
that arise from setting \svs{} to zero in the corresponding Verma modules. \cref{eq:TruncFusProd} now ensures that
\begin{equation}
\brac{\Kac{2,1} \fuse \Kac{1,2}}^0 \subseteq \Kac{2,1}^{\text{ss}} \otimes \Kac{1,2}^0 = \vspn \set{\ket{v} \otimes \ket{w}, L_{-1} \ket{v} \otimes \ket{w}},
\end{equation}
since $L_{-1}^2 \ket{v} = 2 L_{-2} \ket{v} \in U_{\text{ss}} \Kac{2,1}$.  

However, \eqref{FR:K21K12} shows that the fusion product to depth $0$ is only one-dimensional (we actually only need the character of the product for this).  We therefore need to determine a spurious state. For this, we combine the relations \eqref{eq:Ex1Rels} with the coproduct formula \eqref{eq:Master1} for $n=-1$:
\begin{equation}
\coproduct{L_{-1}} = L_{-1} \otimes \wun + \wun \otimes L_{-1} \qquad \Ra \qquad 
\coproduct{L_{-1}^2} = L_{-1}^2 \otimes \wun + 2 \: L_{-1} \otimes L_{-1} + \wun \otimes L_{-1}^2.
\end{equation}
Because we are computing to depth $0$, all states in the image of $\coproduct{L_{-1}}$ and $\coproduct{L_{-1}^2}$ have been set to $0$.  For finding spurious states, the first relation in \eqref{eq:Ex1Rels} is only useful because it determines the special subspace $\Kac{2,1}^{\text{ss}}$. The second, however, may lead to something non-trivial.

Combining the above coproducts with the relations \eqref{eq:Ex1Rels}, we arrive at
\begin{equation} \label{eq:K21K12FindSpur}
0 = \coproduct{L_{-1}^2} \brac{\ket{v} \otimes \ket{w}} = 2 \: L_{-2} \ket{v} \otimes \ket{w} + 2 \: L_{-1} \ket{v} \otimes L_{-1} \ket{w} + \tfrac{1}{2} \ket{v} \otimes L_{-2} \ket{w}.
\end{equation}
We now use the master formulae \eqref{eq:Master} to deal with the three terms of \eqref{eq:K21K12FindSpur}:
\begin{subequations}
\begin{align}
0 &= \coproduct{L_{-1}} \brac{L_{-1} \ket{v} \otimes \ket{w}} = L_{-1}^2 \ket{v} \otimes \ket{w} + L_{-1} \ket{v} \otimes L_{-1} \ket{w} \notag \\
&= 2 \: L_{-2} \ket{v} \otimes \ket{w} + L_{-1} \ket{v} \otimes L_{-1} \ket{w}, \\
0 &= \coproduct{L_{-2}} \brac{\ket{v} \otimes \ket{w}} = L_{-1} \ket{v} \otimes \ket{w} - L_0 \ket{v} \otimes \ket{w} + \ket{v} \otimes L_{-2} \ket{w} \notag \\
&= L_{-1} \ket{v} \otimes \ket{w} - \ket{v} \otimes \ket{w} + \ket{v} \otimes L_{-2} \ket{w}, \\
L_{-2} \ket{v} \otimes \ket{w} &= \ket{v} \otimes L_{-1} \ket{w} + \ket{v} \otimes L_0 \ket{w} = -L_{-1} \ket{v} \otimes \ket{w} - \tfrac{1}{8} \ket{v} \otimes \ket{w}.
\end{align}
\end{subequations}
The last equality uses $\coproduct{L_{-1}} \brac{\ket{v} \otimes \ket{w}} = 0$ again.  Substituting back into \eqref{eq:K21K12FindSpur}, we finally obtain
\begin{equation} \label{eq:K21K12Spur}
0 = \tfrac{3}{2} L_{-1} \ket{v} \otimes \ket{w} + \tfrac{3}{4} \ket{v} \otimes \ket{w} \qquad \Ra \qquad 
L_{-1} \ket{v} \otimes \ket{w} = -\tfrac{1}{2} \ket{v} \otimes \ket{w}.
\end{equation}
This is the spurious state (relation) that we sought.

One could repeat this exercise, starting from $0 = \coproduct{L_{-1}^2} \brac{L_{-1} \ket{v} \otimes \ket{w}}$ for example, but no new linearly independent spurious states will be found.  This is the virtue of using a Verlinde formula to determine the character of the fusion product in advance.\footnote{We remark that one of the drawbacks of the \NGK{} fusion algorithm is that it does not provide any means to determine when one has found a complete basis of spurious states.  In some cases, though not all, this information can be extracted from the character of the fusion product. For instance, the character does not suffice in the next example.}  It therefore only remains to compute the action of $L_0$ on this truncated fusion product (the Virasoro modes $L_n$, with $n \neq 0$, necessarily act as the zero operator because the truncation is to depth $0$).  Since $\set{\ket{v} \otimes \ket{w}}$ is a basis for the depth $0$ fusion product, we find that
\begin{equation}
\coproduct{L_0} \brac{\ket{v} \otimes \ket{w}} = L_{-1} \ket{v} \otimes \ket{w} + L_0 \ket{v} \otimes \ket{w} + \ket{v} \otimes L_0 \ket{w} = \tfrac{3}{8} \ket{v} \otimes \ket{w},
\end{equation}
using \eqref{eq:Master1} and \eqref{eq:K21K12Spur}, and thereby identify $\ket{v} \otimes \ket{w}$ with the \hws{} of $\Kac{2,2}$ (at depth $0$).

However, computing to depth $0$ only determines that the fusion product is a \hwm{} generated by a \hws{} of conformal dimension $\tfrac{3}{8}$.  To prove that this is indeed the irreducible module $\Kac{2,2}$, we need to compute to depth $4$ and verify that the fusion product has only four linearly independent states of conformal dimension $\tfrac{35}{8}$ instead of five.   This is indeed possible, but it requires finding five linearly independent spurious states which reduce the dimension of the truncated fusion product from $16$ to $11$.  Our implementation of the fusion algorithm in \textsc{Maple} achieved this in around twenty minutes. Clearly, it is much more efficient to utilise the results of the Verlinde formula in this case.

\paragraph{Example 2:  $\Kac{2,1} \fuse \Kac{1,3}$.}
In this example, we will combine the information from the Verlinde formula with explicit fusion computations so as to determine the result as efficiently as possible.  The arguments detailed here are typical of those used to arrive at the results summarised in \cref{sec:NGKResults}. 
This fusion product was also analysed in \cite{RasCla11}.
We note that a pure \NGK{} analysis of this fusion product requires computing to depth $5$, which is quite laborious even with a computer.  As we shall see, by first obtaining the character of the product from the Verlinde formula, we may completely identify the structure from a depth $1$ \NGK{} calculation, a significant improvement.
In \cite{RasCla11}, this character was merely inferred from lattice considerations.

The Grothendieck fusion rule \eqref{FR:Kr1K1s} of $\Kac{2,1}$ with $\Kac{1,3}$ states that the character of the result is that of $\Kac{2,3}$.  However, this Kac module is reducible (\cref{fig:K21K12K22}) with precisely three composition factors corresponding to \ssvs{} with conformal dimensions $0$, $1$ and $3$.  There are therefore \emph{nine} candidate structures consistent with this character:\footnote{The structures $0 \lra 3$ and $3 \lra 0$ do not appear here.  The first would correspond to a \hwm{}, hence to a quotient of the chain type Verma module $\Ver{0}$, but it is easy to check that no such quotient exists.  The non-existence of the second structure follows by considering contragredient duals.}
\[
\begin{tikzpicture}[->,node distance=1cm,>=stealth',semithick]
  \node[] (0) [] {$0$};
  \node[] (1) [below of =0] {$1$};
  \node[] (3) [below of =1] {$3$};
  \path[] (0) -- (1) node[midway] {$\oplus$};
  \path[] (3) -- (1) node[midway] {$\oplus$};
\end{tikzpicture}
\qquad
\begin{tikzpicture}[->,node distance=1cm,>=stealth',semithick]
  \node[] (0) [] {$0$};
  \node[] (1) [below of =0] {$1$};
  \node[] (3) [below of =1] {$3$};
  \draw[] (0) -- (1);
  \path[] (3) -- (1) node[midway] {$\oplus$};
\end{tikzpicture}
\qquad
\begin{tikzpicture}[->,node distance=1cm,>=stealth',semithick]
  \node[] (0) [] {$0$};
  \node[] (1) [below of =0] {$1$};
  \node[] (3) [below of =1] {$3$};
  \draw[] (1) -- (0);
  \path[] (3) -- (1) node[midway] {$\oplus$};
\end{tikzpicture}
\qquad
\begin{tikzpicture}[->,node distance=1cm,>=stealth',semithick]
  \node[] (0) [] {$0$};
  \node[] (1) [below of =0] {$1$};
  \node[] (3) [below of =1] {$3$};
  \draw[] (1) -- (3);
  \path[] (0) -- (1) node[midway] {$\oplus$};
\end{tikzpicture}
\qquad
\begin{tikzpicture}[->,node distance=1cm,>=stealth',semithick]
  \node[] (0) [] {$0$};
  \node[] (1) [below of =0] {$1$};
  \node[] (3) [below of =1] {$3$};
  \draw[] (3) -- (1);
  \path[] (0) -- (1) node[midway] {$\oplus$};
\end{tikzpicture}
\qquad
\begin{tikzpicture}[->,node distance=1cm,>=stealth',semithick]
  \node[] (0) [] {$0$};
  \node[] (1) [below of =0] {$1$};
  \node[] (3) [below of =1] {$3$};
  \draw[] (0) -- (1);
  \draw[] (1) -- (3);
\end{tikzpicture}
\qquad
\begin{tikzpicture}[->,node distance=1cm,>=stealth',semithick]
  \node[] (0) [] {$0$};
  \node[] (1) [below of =0] {$1$};
  \node[] (3) [below of =1] {$3$};
  \draw[] (0) -- (1);
  \draw[] (3) -- (1);
\end{tikzpicture}
\qquad
\begin{tikzpicture}[->,node distance=1cm,>=stealth',semithick]
  \node[] (0) [] {$0$};
  \node[] (1) [below of =0] {$1$};
  \node[] (3) [below of =1] {$3$};
  \draw[] (1) -- (0);
  \draw[] (1) -- (3);
\end{tikzpicture}
\qquad
\begin{tikzpicture}[->,node distance=1cm,>=stealth',semithick]
  \node[] (0) [] {$0$};
  \node[] (1) [below of =0] {$1$};
  \node[] (3) [below of =1] {$3$};
  \draw[] (1) -- (0);
  \draw[] (3) -- (1);
  \node[]     [right = 0mm of 3] {$\vphantom{3}$.};
\end{tikzpicture}
\]
We will denote the \ssvs{} of $\Kac{2,1} \fuse \Kac{1,3}$ by $\ket{s_0}$, $\ket{s_1}$ and $\ket{s_3}$ for convenience, the subscript indicating the dimension.

To investigate the structure of $\Kac{2,1} \fuse \Kac{1,3}$, we apply the \NGK{} algorithm to depth $0$.  Both modules turn out to be highest weight, so we let $\ket{v}$ and $\ket{w}$ denote their respective \hwss{}.  The depth $0$ fusion product is contained within $\vspn \set{\ket{v} \otimes \ket{w}, L_{-1} \ket{v} \otimes \ket{w}}$ which has dimension $2$.  As the nine structures above correspond to depth $0$ dimensions of $3$, $2$, $3$, $2$, $3$, $1$, $2$, $2$ and $3$, respectively, this immediately rules out the first, third, fifth and ninth structures.

We therefore search for a spurious state using the relation $L_{-1}^3 \ket{w} = 2 L_{-2} L_{-1} \ket{w}$ in $\Kac{1,3}$, but find none.  This implies that the fusion product is two-dimensional to depth $0$, ruling out the sixth structure above (in which $\Kac{2,1} \fuse \Kac{1,3}$ is a \hwm{}). Thus, one of $\ket{s_0}$, $\ket{s_1}$ or $\ket{s_3}$ must be in the image of $U_0$. It cannot be $\ket{s_0}$, because this state cannot be obtained from another state by acting with negative modes, so the remaining possibilities are that either $\ket{s_1}$ or $\ket{s_3}$, but not both, are in the image of $U_0$.  To determine which, we simply compute $\coproduct{L_0}$:
\begin{subequations}
\begin{align}
\coproduct{L_0} \brac{\ket{v} \otimes \ket{w}} &= L_{-1} \ket{v} \otimes \ket{w} + \ket{v} \otimes \ket{w}, \\
\coproduct{L_0} \brac{L_{-1} \ket{v} \otimes \ket{w}} &= L_{-1}^2 \ket{v} \otimes \ket{w} + 2 \: L_{-1} \ket{v} \otimes \ket{w} = 2 \: L_{-2} \ket{v} \otimes \ket{w} + 2 \: L_{-1} \ket{v} \otimes \ket{w} \notag \\
&= 2 \: \ket{v} \otimes L_{-1} \ket{w} + 2 \: \ket{v} \otimes L_0 \ket{w} + 2 \: L_{-1} \ket{v} \otimes \ket{w} = 0 \\
\Ra \qquad \coproduct{L_0} &= 
\begin{pmatrix}
1 & 0 \\
1 & 0
\end{pmatrix}
.
\end{align}
\end{subequations}
The eigenvalues of $L_0$ on the depth $0$ fusion product are $0$ and $1$, so it is $\ket{s_3}$ which is in the image of $U_0$ and therefore missing from the depth $0$ analysis.  This rules out the second and the seventh structure above.

To complete the structure of $\Kac{2,1} \fuse \Kac{1,3}$, it only remains to decide between the fourth and eighth structures.  These are distinguished by whether the fusion product is indecomposable or not.  We therefore test if $\ket{s_0}$ can be obtained from $\ket{s_1}$ by acting with $\coproduct{L_1}$, requiring the \NGK{} algorithm to depth $1$.  The truncated fusion product is then contained within $\vspn \set{\ket{v} \otimes \ket{w}, L_{-1} \ket{v} \otimes \ket{w}, \ket{v} \otimes L_{-1} \ket{w}, L_{-1} \ket{v} \otimes L_{-1} \ket{w}}$ and searching with our \textsc{Maple} implementation turns up precisely one spurious state, resulting in the relation
\begin{equation}
L_{-1} \ket{v} \otimes L_{-1} \ket{w} = -2 \: \ket{v} \otimes L_{-1} \ket{w}.
\end{equation}
Computing the action of $L_{-1}$, $L_0$ and $L_1$, we obtain
\begin{equation}
\coproduct{L_{-1}} = 
\begin{pmatrix}
0 & 0 & 0 \\
0 & 0 & 0 \\
0 & 1 & 0
\end{pmatrix}
, \qquad 
\coproduct{L_0} = 
\begin{pmatrix}
0 & 0 & 0 \\
0 & 1 & 0 \\
0 & 0 & 2
\end{pmatrix}
, \qquad 
\coproduct{L_1} = 
\begin{pmatrix}
0 & -1 & 0 \\
0 & 0 & 2 \\
0 & -3 & 6
\end{pmatrix}
,
\end{equation}
with respect to the (ordered) eigenbasis
\begin{equation}
\set{\ket{v} \otimes L_{-1} \ket{w}, -\ket{v} \otimes \ket{w} + 2 \: \ket{v} \otimes L_{-1} \ket{w} + L_{-1} \ket{v} \otimes \ket{w}, -L_{-1} \ket{v} \otimes \ket{w} - \ket{v} \otimes L_{-1} \ket{w}}
\end{equation}
of $\coproduct{L_0}$, which we can identify with $\set{\ket{s_0}, \ket{s_1}, \coproduct{L_{-1}} \ket{s_1}}$.

The matrix that we have computed for $\coproduct{L_1}$ does not look correct as it claims that acting on $\ket{s_1}$ gives a linear combination of $\ket{s_0}$ and a dimension $2$ state.  This is a universal issue with computing the action of positive modes resulting from the fact that $\coproduct{L_1}$ should be regarded as mapping the depth $1$ fusion product to its depth $0$ counterpart.\footnote{Similarly, $\coproduct{L_{-1}}$ maps the depth $0$ fusion product to its depth $1$ counterpart.  However, the former is naturally a subspace of the latter, so one can (correctly) extend this to a map between the depth $1$ spaces by adding extra vectors which map to zero.}  Since the depth $0$ counterpart was spanned by $\ket{s_0}$ and $\ket{s_1}$, the easiest fix is to simply ignore any contribution from the dimension $2$ state to $\coproduct{L_1}$ (this issue is discussed at length in, for example, \cite{GabInd96,EbeVir06}). With this fix, we compute that
\begin{equation}
\coproduct{L_1} \ket{s_1} = -\ket{s_0},
\end{equation}
hence that $\Kac{2,1} \fuse \Kac{1,3}$ is an indecomposable module with the eighth
structure above:
\begin{center}
\begin{tikzpicture}[->,node distance=1cm,>=stealth',semithick]
  \node[] (5) [] {$\Kac{2,1} \fuse \Kac{1,3}$:};
  \node[] (5b) [right = 1cm of 5] {$1$}; 
  \node[] (5a) [above of =5b] {$0$};
  \node[] (5c) [below of =5b] {$3$};
  \path[] (5b) edge node {} (5a)
          (5b) edge node {} (5c);
  \node[]      [right = 0mm of 5c] {$\vphantom{3}$.};
\end{tikzpicture}
\end{center}
Comparing with \cref{fig:K21K12K22}, we conclude that $\Kac{2,1} \fuse \Kac{1,3} = \Kac{2,3}$.

\subsubsection{Results} \label{sec:NGKResults}

In this section, we summarise the results of the further explicit computations which we performed with the aid of a \textsc{Maple} implementation of the \NGK{} fusion algorithm for the Virasoro algebra.  Several such summaries have previously appeared in the literature, see \cite{GabInd96,EbeVir06,RidPer07,RidLog07,RidPer08,GabFus09,RasCla11} for example.  However, these works concentrated, to a large degree, on constructing the staggered modules \cite{RohRed96,RidSta09} (and their generalisations) that are responsible for the logarithmic nature of the logarithmic minimal models.  Here, our purpose is to verify the following conjecture which naturally extends the Grothendieck fusion rule \eqref{FR:Kr1K1s}. The data supporting it, presented below, constitutes strong evidence that the lattice prescription for fusion discussed in \cref{sec:LatticeFusion} is correct.
\begin{conj} \label{conj:VirFusion}
The Virasoro Kac modules satisfy the fusion rule
\begin{equation} \label{FR:KxK}
\Kac{r,1} \fuse \Kac{1,s} = \Kac{r,s} \qquad \textup{(\(r,s \in \ZZ_+\)).} 
\end{equation}
\end{conj}
\noindent We recall that the Kac module $\Kac{r,s}$ is defined to be the submodule of the \FFm{} $\FF{r,s}$ that is generated by the \ssvs{} of grades less than $rs$.

In case $r$ or $s$ is $1$, the fusion rule \eqref{FR:KxK} follows from the fact that $\Kac{1,1}$ is the vacuum module.  If $\Kac{r,s}$ happens to be irreducible, then the fusion rule may be deduced as a corollary of its Grothendieck counterpart \eqref{FR:Krs2}.  In all other cases, we will use Verlinde formula methods to aid with the identification of the fusion product, minimising the amount of explicit calculation required.  We mention that \eqref{FR:KxK} was (partially) verified using this fusion algorithm in \cite{RasCla11} for a few cases with $p=1$.  What follows is a significant extension of these verifications that provides solid evidence for the fusion rules \eqref{FR:KxK} or, equivalently, for the definition that we have adopted for the Kac modules $\Kac{r,s}$.

We mention that, for arbitrary central charges $c \in \RR$, we have also confirmed the fusion rule \eqref{FR:KxK} for a variety of small values of $r$ and $s$.\footnote{When $c \in \mathbb R$ does not have the form given in \eqref{eq:ParByt}, the structures of the Virasoro Kac modules are considerably simpler and are thus relatively easy to analyse.  We will therefore not discuss these central charges here.}
However, in order to facilitate comparison with the lattice results, and for brevity's sake, we will restrict our summary to the central charges of five of the logarithmic minimal models, specifically those with $(p,p') = (1,2)$, $(1,3)$, $(2,3)$, $(2,5)$ and $(3,4)$.  As we shall see, the computation required to fully analyse the interesting fusion products increases very quickly with $p$ and $p'$, effectively limiting the useful results to these models.

For each of these values of $p$ and $p'$, the results of fusing $\Kac{r,1}$ with $\Kac{1,s}$ are tabulated in \cref{tab:NGK} for various values of $r$ and $s$.  In each case where we have been able to identify the fusion product, the result confirms \eqref{FR:KxK}.  The entries of each table are to be interpreted in the following manner:
\begin{itemize}
\item A dash ``$-$'' indicates that no \NGK{} fusion computation is required because the Verlinde formula indicates that the fusion product is irreducible.
\item A number $d$ indicates the depth to which the fusion algorithm must compute in order to identify the fusion product, given that we know the character of the fusion product.
\item A number $d$ may be followed by another $d'$ in parentheses which indicates that while a complete identification 
of the fusion product requires a depth $d'$ computation, it is actually sufficient to only compute to grade $d$.  We will 
discuss the reasons for this shortly with the aid of an example.
\end{itemize}
Some entries in \cref{tab:NGK} are left blank.  These correspond to $r$ and $s$ for which fusion computations were regarded as too difficult.  A light blue background for the table entries indicates that the computation was attempted, but was aborted due to either memory or time constraints (we tended to abort after three or four days of continuous runtime).  A blue background indicates that the fusion computations were successfully performed to depth $d$.

{
\renewcommand{\arraystretch}{1.1} 
\begin{table}
\begin{center}
\begin{tikzpicture}
\node (t12) at (0,0) {\setlength{\extrarowheight}{2pt}
\begin{tabular}{C|CCCCCCC}
\rs & 2 & 3 & 4 & 5 & 6 & 7 & 8 \\
\hline
  2 & - & \IKL 1 & \IKL 2 & \IKL 3 & \IKL 4 & \IKL 5 & \BKL 6 \\
  3 & - & \IKL 2 & \IKL 4 & \IKL 5 &      8 &        &        \\
  4 & - & \IKL 3 & \BKL 6 &        &        &        &        \\
  5 & - & \IKL 4 &      8 &        &        &        &        \\
  6 & - & \BKL 5 &        &        &        &        & 
\end{tabular}
};
\node (n12) [below = 0.2cm of t12] {$(p,p') = (1,2)$};
\node (t13) at (9,0) {\setlength{\extrarowheight}{2pt}
\begin{tabular}{C|CCCCCCC}
\rs & 2 & 3 & 4 & 5 & 6 & 7 & 8 \\
\hline
  2 & - & - & \IKL 1 & \IKL 2 & \IKL 3 & \IKL 4 & \BKL 5 \\
  3 & - & - & \IKL 2 & \IKL 4 & \BKL 6 &      7 &        \\
  4 & - & - & \IKL 3 & \BKL 6 &      9 &        &        \\
  5 & - & - & \IKL 4 &      8 &        &        &        \\
  6 & - & - & \BKL 5 &        &        &        &
\end{tabular}
};
\node (n12) [below = 0.2cm of t13] {$(p,p') = (1,3)$};
\node (t23) at (0,-5) {\setlength{\extrarowheight}{2pt}
\begin{tabular}{C|CCCCCCC}
\rs & 2 & 3 & 4 & 5 & 6 & 7 & 8 \\
\hline
  2 & -      & -      & \IKL 0       & \IKL 0      &      - & \IKL 0       & \IKL 0  \\
  3 & \IKL 0 & \IKL 0 & \IKL 1\ (7)  & \IKL 2\ (7) & \IKL 3 & \BKL 4\ (14) & 5\ (13) \\
  4 & -      & -      & \IKL 2       & \BKL 4      &      6 &              &         \\
  5 & \IKL 0 & \IKL 0 & \IKL 3\ (13) & 6\ (14)     &        &              &         \\
  6 & -      & -      & \BKL 4       &             &        &              &
\end{tabular}
};
\node (n23) [below = 0.2cm of t23] {$(p,p') = (2,3)$};
\node (t25) at (9,-5) {\setlength{\extrarowheight}{2pt}
\begin{tabular}{C|CCCCCCC}
\rs & 2 & 3 & 4 & 5 & 6 & 7 & 8 \\
\hline
  2 & -      & -      & -      & -      & \IKL 0       & \IKL 0       & \IKL 0       \\
  3 & \IKL 0 & \IKL 0 & \IKL 0 & \IKL 0 & \IKL 1\ (13) & \IKL 2\ (11) & \BKL 3\ (11) \\
  4 & -      & -      & -      & -      & \IKL 2       & \BKL 4       &      6       \\
  5 & \IKL 0 & \IKL 0 & \IKL 0 & \IKL 0 & \BKL 3\ (23) &              &              \\
  6 & -      & -      & -      & -      &      4       &              &
\end{tabular}
};
\node (n25) [below = 0.2cm of t25] {$(p,p') = (2,5)$};
\node (t34) at (4.5,-10) {\setlength{\extrarowheight}{2pt}
\begin{tabular}{C|CCCCCC}
\rs & 2 & 3 & 4 & 5 & 6 & 7 \\
\hline
  2 & \IKL 0 & \IKL 0 & -      & \IKL 0       & \IKL 0       & \IKL 0       \\
  3 & -      & -      & -      & \IKL 0       & \IKL 0       & \IKL 0       \\
  4 & \IKL 0 & \IKL 0 & \IKL 0 & \IKL 1\ (13) & \IKL 2\ (14) & \BKL 3\ (17) \\
  5 & \IKL 0 & \IKL 0 & \IKL 0 & \IKL 2\ (17) & \BKL 4\ (14) &              \\
  6 & -      & -      & -      & \BKL 3       &              &
\end{tabular}
};
\node (n34) [below = 0.2cm of t34] {$(p,p') = (3,4)$};
\end{tikzpicture}
\caption{Tables for various $p$ and $p'$ indicating, by a blue background, the values of $r$ and $s$ for which we have been able to explicitly confirm 
the fusion rule $\Kac{r,1} \fuse \Kac{1,s} = \Kac{r,s}$ by combining the Verlinde formula with the \NGK{} fusion algorithm and theorems concerning possible module structures.  Values with a light blue background correspond to cases in which only partial confirmation was achieved.  The entries in the tables describe the depth to which the fusion algorithm would need to compute.} \label{tab:NGK}
\end{center}
\end{table}
}

We illustrate the meaning of the entries of \cref{tab:NGK} of the form $d \ (d')$ with an example.  Take $(p,p') = (2,3)$, so the central charge is $c=0$, and $(r,s) = (3,4)$.  The corresponding entry is $1 \ (7)$.  The Verlinde formula says that the composition factors of the fusion product $\Kac{3,1} \fuse \Kac{1,4}$ are the irreducibles $\Irr{\Delta}$, each appearing with multiplicity $1$, where $\Delta \in \set{0,1,2,5,7,15}$ (these are, of course, the composition factors of $\Kac{3,4}$).  Fusing to depth $0$ results in two linearly independent states of conformal dimensions $0$ and $1$.  These are therefore the \ssvs{} corresponding to the composition factors $\Irr{0}$ and $\Irr{1}$, respectively. It follows that $\Irr{2}$, $\Irr{5}$, $\Irr{7}$ and $\Irr{15}$ are all descended from 
$\Irr{0}$ and/or $\Irr{1}$.  As $\Irr{2}$ cannot be descended from $\Irr{1}$, it must be descended from $\Irr{0}$.  As $\Irr{1}$ is not descended from $\Irr{0}$ (it appears at depth $0$), it now follows from Verma module considerations that $\Irr{5}$, $\Irr{7}$ and $\Irr{15}$ cannot be descended from $\Irr{2}$.\footnote{Any \hw{} (sub)module of weight $\Delta$ must be realisable as a quotient of the Verma module $\Ver{\Delta}$. In the example, $\Irr{5}$ cannot be descended from $\Irr{2}$: If it were, the structure of the submodule generated by $\Irr{0}$ would be $0 \ 
\begin{pspicture}[shift=0](0,-0.1)(0.5,-0.1)
\psline{->}(0,0)(0.5,0)
\end{pspicture}\
2 \ 
\begin{pspicture}[shift=0](0,-0.1)(0.5,-0.1)
\psline{->}(0,0)(0.5,0)
\end{pspicture}\
5\ 
\begin{pspicture}[shift=0](0,-0.1)(0.5,-0.1)
\psline{->}(0,0)(0.5,0)
\end{pspicture}\
\cdots$ which cannot be obtained as a quotient of the braid type Verma module $\Ver{0}$,
see \cref{fig:VermaStructures}.}  Thus, $\Irr{5}$, $\Irr{7}$ and $\Irr{15}$ must all be descended from $\Irr{1}$. 

It is possible that the fusion product is decomposable; if so, it can only decompose as the direct sum of a \hwm{} generated from $\Irr{0}$ and another generated from $\Irr{1}$.  By computing to depth $1$, we rule out this possibility, arriving at the following (partial) structure:
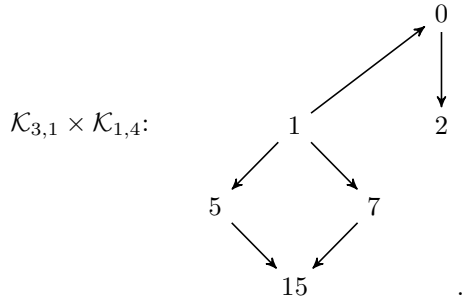
\begin{center}
\begin{tikzpicture}
  [->,>=stealth',semithick]
  \node[] (nom) [] {$\Kac{3,1} \fuse \Kac{1,4}$:};
  \node[] (lb) [right = 3.5cm of nom] {$2$};
  \node[] (la) [above = 1cm of lb] {$0$};
  \node[] (ra) [left = 1.5cm of lb] {$1$};
  \node[] (rb) [below = 7mm of ra] {};
  \node[] (rc) [left = 7mm of rb] {$5$};
  \node[] (rd) [right = 7mm of rb] {$7$};
  \node[] (re) [below = 7mm of rb] {$15$};
  \draw[] (la) -- (lb);
  \draw[] (ra) -- (la);
  \draw[] (ra) -- (rd);
  \draw[] (ra) -- (rc);
  \draw[] (rc) -- (re);
  \draw[] (rd) -- (re);
  \node[]      [right = 17mm of re] {$\vphantom{5}$.};
\end{tikzpicture}
\end{center}
The arrow pointing upwards from $1$ to $0$ indicates that the \ssv{} corresponding to $\Irr{1}$ generates a submodule that includes the \ssv{} corresponding to $\Irr{0}$ (this is the conclusion of the depth $1$ computation).  It only remains to check if the \ssvs{} for $\Irr{5}$, $\Irr{7}$ and $\Irr{15}$ generate submodules containing those of $\Irr{0}$ or $\Irr{2}$; diagrammatically, this asks us to add any further upwards pointing arrows.  Checking this explicitly using the fusion algorithm is infeasible for the foreseeable future (one would need to compute to depth $15$).  However, the Projection Lemma\footnote{The hypotheses of this lemma assume that the module being considered is staggered, meaning in particular that the Virasoro zero mode $L_0$ acts on it non-diagonalisably.  This is not the case for the fusion product considered here, but non-diagonalisability turns out to be irrelevant to the lemma's proof.} for staggered modules \cite[Lem.~5.1]{RidSta09} rules out the arrows from $5$ or $7$ to $0$ and from $15$ to $0$ or $2$.  We are thus left with potential arrows from $5$ and $7$ to $2$.

Another general conclusion of \cite{RidSta09} is that these remaining arrows are \emph{almost always} present.  If one is not, $5$ to $2$ say, then there would be a \ssv{} of conformal dimension $5$ which is actually \emph{singular} in $\Kac{3,1} \fuse \Kac{1,4}$.  The work of \cite{RidSta09} demonstrates that such \ssvs{} are only singular when the data defining the module belongs to a subspace of codimension at least $1$.  In this sense, these \ssvs{} are almost never singular, hence the arrows we are discussing are almost always present.

In this case, we can easily confirm the presence of these arrows, 
following \cite{RidLog07}.  Let $\ket{s_{\Delta}}$ denote a choice of \ssv{} in $\Kac{3,1} \fuse \Kac{1,4}$ of conformal 
dimension $\Delta$.  We may assume that $\ket{s_0} = L_1 \ket{s_1}$ and $\ket{s_2} = L_{-2} \ket{s_0}$.  
As $L_{-1} \ket{s_0} = 0$ and $\ket{s_5}$ must become singular upon quotienting by the submodule generated by 
$\ket{s_0}$, the most general form for $\ket{s_5}$ is
\begin{equation}
\ket{s_5} = (L_{-1}^4 - \tfrac{20}{3} L_{-2} L_{-1}^2 + 4 L_{-2}^2 + 4 L_{-3} L_{-1} - 4 L_{-4}) \ket{s_1} + (a L_{-3} L_{-2} + b L_{-5}) \ket{s_0} \qquad \text{(\(a,b \in \CC\)).}
\end{equation}
It is easy, though somewhat tedious, to check now that
\begin{equation}
L_1 \ket{s_5} = \sqbrac{4 (a+1) L_{-2}^2 + 2 (3b-2) L_{-4}} \ket{s_0}, \qquad 
L_2 \ket{s_5} = (5a+7b+12) L_{-3} \ket{s_0},
\end{equation}
which do not vanish simultaneously for any $a,b \in \CC$.  It follows that there is no singular choice for $\ket{s_5}$ in 
$\Kac{3,1} \fuse \Kac{1,4}$, 
explicitly verifying that the arrow from $5$ to $2$ is present.  A similar calculation verifies that the arrow from $7$ to $2$ is also present.  We note that these calculations are purely representation-theoretic and do not require the explicit construction afforded by the fusion algorithm.

We have therefore verified that the structure of the fusion product matches that of $\Kac{3,4}$:
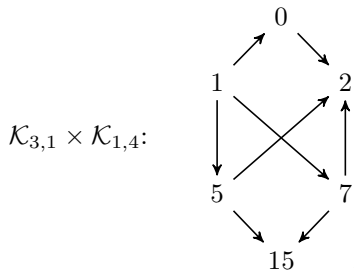
\begin{center}
\begin{tikzpicture}[->,>=stealth',semithick]
  \node[] (nom) [] {$\Kac{3,1} \fuse \Kac{1,4}$:};
  \node[] (grr) [right = 1.5cm of nom] {};
  \node[] (emp) [above = 0.5cm of grr] {};
  \node[] (c) [left = 0.5cm of emp] {$1$};
  \node[] (b) [right = 0.5cm of emp] {$2$};
  \node[] (a) [above = 0.5cm of emp] {$0$};
  \node[] (d) [below = 1cm of b] {$7$};
  \node[] (e) [below = 1cm of c] {$5$};
  \node[] (f) [below = 2cm of emp] {$15$};
  \draw[] (c) -- (a);
  \draw[] (a) -- (b);
  \draw[] (c) -- (d);
  \draw[] (c) -- (e);
  \draw[] (d) -- (f);
  \draw[] (e) -- (f);
  \draw[] (d) -- (b);
  \draw[] (e) -- (b);
  \node[]     [right = 0.5cm of f] {$\vphantom{7}$.};
\end{tikzpicture}
\end{center}
The conclusion that $\Kac{3,1} \fuse \Kac{1,4} = \Kac{3,4}$ is based on fusion computations to depth $1$, combined with the explicit checks for \svs{} above to verify the presence of arrows from $5$ and $7$ to $2$.  Confirming these arrows directly with the fusion algorithm would require computing to depth $7$ which is well beyond our current capabilities (we were however able to verify the arrow from $5$ to $2$ in this fashion).  This is the meaning of the corresponding entry $1 \ (7)$ in \cref{tab:NGK}:  We can be almost sure of the result if we compute to depth $1$, but to be completely sure, we would have to compute to depth $7$.  As this is, in every case, too deep for our implementation of the fusion algorithm, we instead try to verify the arrows directly using \svs{}.  This latter approach succeeded for each entry of \cref{tab:NGK} that is shaded blue.  Even checking for \svs{} of grade $17$ required less than five minutes with our \textsc{Maple} implementation.

%
\section{Conclusion} \label{sec:Conc}
%

In this paper, we have resolved two outstanding issues concerning the Kac modules introduced in \cite{PRZ06}
almost ten years ago. Although these representations are fundamental to logarithmic minimal models, their precise module structures were previously unknown in general.  This was true on the lattice as well as in the continuum. 
The two issues that we have resolved are then the precise identification of the lattice Kac modules and that of their continuum limits, the Virasoro Kac modules.

To achieve the first, we have introduced the appropriate algebraic framework for the lattice analysis in terms of
quotients of the one-boundary Temperley-Lieb algebras. 
We call these quotients the \emph{boundary seam algebras}.
In this framework, the prescription used in \cite{PRZ06} to study transfer tangles with seams is 
recognised as producing the standard modules over these boundary seam algebras. 
These modules then define the lattice Kac modules. 
For a given sequence of lattice Kac modules, where the bulk lattice size $n$ increases in steps of $2$, 
the first few integer coefficients of the limiting character were extracted, 
for small system sizes, allowing us to guess its Kac labels $r$ and $s$.
The corresponding invariant bilinear forms on the lattice Kac modules were constructed 
and their Gram determinants were computed and used to partially determine
the structures of the standard modules. 

From these lattice results, we inferred the structures of the Virasoro Kac modules arising in the 
scaling limit and found that they correspond to certain finitely generated submodules of Feigin-Fuchs modules. 
This conjecture was subsequently confirmed by two independent conformal field theory analyses.
The characters of the Virasoro Kac modules and the lattice prescription for fusion
were found to be in complete agreement with the results of a Verlinde-like formula, while the precise 
Virasoro Kac module structures were verified in many examples using the Nahm-Gaberdiel-Kausch fusion algorithm.

As indicated, our results for the structure of the Virasoro Kac modules follow from a combination of three approaches:
the character analysis (which combines lattice data with a Verlinde-like formula), the invariant bilinear forms on the 
lattice Kac modules, and the Nahm-Gaberdiel-Kausch fusion algorithm. 
Separately, these approaches all have their limitations, but taken together,
they led us to propose \cref{TheConjecture} for the scaling limit of the lattice Kac modules. 
Indeed, the fact that the results from all three approaches agree perfectly 
makes a very strong case for the validity of this conjecture.

It is noteworthy that all three approaches have the potential to be strengthened. 
Regarding the characters, the transfer tangles with vacuum boundary conditions of the logarithmic minimal
models have been shown to satisfy functional relations in \cite{MDPR14}. We believe that this extends trivially when a seam is added. Similar (albeit simpler) relations were solved analytically for rational models, for instance for the tricritical hard squares model \cite{ObrAna97}, where partition functions were computed. 
There is hope that this can be generalised and used to calculate the limiting characters
of the lattice Kac modules 
analytically. For the invariant bilinear form analysis, although we have obtained partial results, including the
computation of the Gram determinants, 
the full representation theory of the boundary seam algebra $\btl_{n,k}$ has not yet been determined. 
Unravelling the structure of the radicals and quotients of the standard modules would yield extra insight 
that will facilitate the identification of the limiting Virasoro module structures. This identification also requires understanding which composition factors drift off to infinity.
Finally, the fact that Virasoro Kac modules are
realised as submodules of \FFms{} suggests strongly that a complete verification of the fusion rule \eqref{FR:KxK} may be attainable using correlation functions and free field methods.

This paper leaves some questions unanswered and opens several avenues for further work, in particular following up on the representation theory of the boundary seam algebras, as noted above. 
We have gained some insight into regime $B$ for the boundary parameter $\xi$, 
but our analysis remains incomplete, even in the cases that we have examined.
Admittedly, regime $B$ is quite poorly understood at the moment. 

Of particular interest is the extension of our analysis of Kac module fusion to the general case
$\Kac{r,s} \fuse \Kac{r',s'}$. At the lattice level, it is likely that one can describe this using diagrammatic algebras
defined in different ways. A first natural suggestion is that they are quotients of the 
\emph{two}-boundary Temperley-Lieb algebra \cite{MNdeGB04,deGN09}.
This corresponds to implementing the Kac boundary triangles on both sides of the lattice.
Alternatively, as illustrated in \cite{PR07}, one can implement the fusion procedure on 
\emph{one} side of the lattice only, by placing the two seams side-by-side. In this way, 
a fusion product is encoded in a new one-sided boundary condition.
The corresponding algebraic framework is expected to involve
new quotients of the one-boundary Temperley-Lieb algebra.
In either scenario, one should expect a rich representation theory of these quotient algebras because
Virasoro modules upon which $L_0$ acts with higher rank Jordan blocks are believed to appear at the conformal field theory level \cite{EbeVir06}.

We also remark that, just as the Kac boundary conditions did
before this work, the recently introduced Robin boundary conditions \cite{PRT14}
lack a proper algebraic definition. They yield well-defined realisations
of the corresponding transfer matrices, but in general do not result in representations of the one-boundary Temperley-Lieb algebras used in their diagrammatic lattice construction.
We nevertheless expect that quotients of these algebras will provide the appropriate algebraic framework for the 
description of Robin boundary conditions, much akin to the situation for the Kac boundary conditions addressed
in the present paper.

Finally, another obvious direction to explore is the generalisation to other loop models, in particular when one fuses 
$2 \times 2$ blocks of elementary face operators in the logarithmic minimal models.  In this case, the Virasoro 
structures of the scaling limit are expected to be replaced by $N=1$ superconformal structures, at least under some 
circumstances.  A numerical lattice-theoretic study providing evidence for this expectation has recently 
appeared \cite{PeaLog14} and preliminary evidence from fusing 
the superconformal analogues of the Virasoro Kac 
modules will be detailed in \cite{CanFus15}.  It would be interesting to understand the correct algebraic formalism for 
describing these lattice models and to generalise to $m \times m$ fused blocks.

\subsection*{Acknowledgements}

AMD was supported by the National Sciences and Engineering Research Council of Canada and by the Belgian Interuniversity Attraction Poles Program P7/18 through the network DYGEST (Dynamics, Geometry and Statistical Physics). 
JR was supported by the Australian Research Council under the Future Fellowship scheme, project number FT100100774.  DR's research is supported by an Australian Research Council Discovery Project DP1093910.  The authors thank Michael Canagasabey, Azat Gainutdinov and Yvan Saint-Aubin for discussions, Steven Flores for pointing out a link between \cref{prop:BlockDiagGram} and theta nets, and
an anonymous referee for useful comments.

\appendix

%
\section*{Appendices}
%

%
\section{Temperley-Lieb representation theory}\label{app:TLrep}
%

This appendix reviews the finite-dimensional
representation theory of the algebra $\tl_n(\beta)$, including the structures of the irreducible, standard and principal indecomposable modules. 
The presentation follows \cite{RidSta12}.

\paragraph{Generalities.} 

The representation theory of $\tl_n$ is semisimple when $\beta$ is a formal parameter, meaning that all its 
(finite-dimensional) 
modules decompose as a direct sum of irreducible modules.  When $\beta$ is specialised to a complex number, $\beta = q + q^{-1} \in \CC$, the representation theory strongly depends upon whether or not $q$ is a root of unity. In the generic case, when $q$ is not a root of unity, $\tl_n(\beta)$ is again semisimple; when $q$ is a root of unity, the semisimplicity of $\tl_n(\beta)$ is not guaranteed. We will discuss this in detail below.

$\tl_n(\beta)$ is a finite-dimensional complex associative algebra, so it admits a finite number of inequivalent irreducible modules $\Itl_n^d$ and the same finite number of \emph{principal indecomposable modules} $\Ptl_n^d$.  The latter are precisely the modules that appear when writing the left-regular module (where $\tl_n(\beta)$ acts on itself by left-multiplication) as a direct sum of indecomposable modules. Each $\Ptl_n^d$ is therefore 
a projective module, meaning that it cannot be realised as a quotient of any indecomposable module except itself.  When $\tl_n$ is not semisimple, some of the $\Ptl_n^d$ will be reducible yet indecomposable.  We shall discuss their submodule structures shortly.

For almost every specialisation $\beta \in \CC$, the index $d$ of $\Itl_n^d$ and $\Ptl_n^d$ takes integer values from $0$ to $n$, with $d=n \bmod{2}$.  There are then $\lfloor \frac{n+2}{2} \rfloor$ inequivalent irreducible $\tl_n(\beta)$-modules and the same number of principal indecomposable $\tl_n(\beta)$-modules.  The only exception to this rule occurs when $\beta = 0$ and $n$ is even.  In this case, $d=0$ must be excluded from the allowed set of values for $d$; alternatively, one may set $\Itl_n^0 = \Ptl_n^0 = \set{0}$ when $\beta = 0$.

In general, any module $\Mtl$ over a finite-dimensional associative algebra admits a 
\emph{composition series}. 
Such a series consists of a set of submodules $\Mtl_i, i = 1, \dots k$, organised into a filtration,
\begin{equation}
0 = \Mtl_0 \subset \Mtl_1 \subset \Mtl_2 \subset \dots \subset \Mtl_k = \Mtl,
\end{equation}
in such a way that each 
\emph{composition factor} $\Mtl_i/\Mtl_{i-1}$, $i = 1, \dots, k$, is irreducible. The composition series of a given module is not unique, but its composition factors are, up to permutation, by the Jordan-H\"older theorem. In particular, the number of composition factors of $\Mtl$ does not depend upon the choice of composition series.

A $\tl_n(\beta)$-module $\Mtl$ is alternatively described by its 
\emph{Loewy diagram}. This diagram consists of vertices occupied by the (irreducible) composition factors of $\Mtl$ and connected by a collection of arrows that, roughly speaking, indicate the action of $\tl_n(\beta)$. \cref{fig:Loewy} includes all the different Loewy diagrams of the standard and principal indecomposable modules of $\tl_n(\beta)$. A composition factor with no arrow pointing away from it is an irreducible submodule --- the action of $\tl_n(\beta)$ leaves this subspace invariant. More generally, a (not necessarily irreducible) submodule is also indicated by any collection of composition factors whose outwards pointing arrows only point towards another factor in this collection.  Conversely, if a collection of factors possesses an outwards pointing arrow that does not point to a factor in the collection, then the corresponding subspace is not a submodule --- the arrow indicates that one can leave this subspace by acting with $\tl_n(\beta)$. 
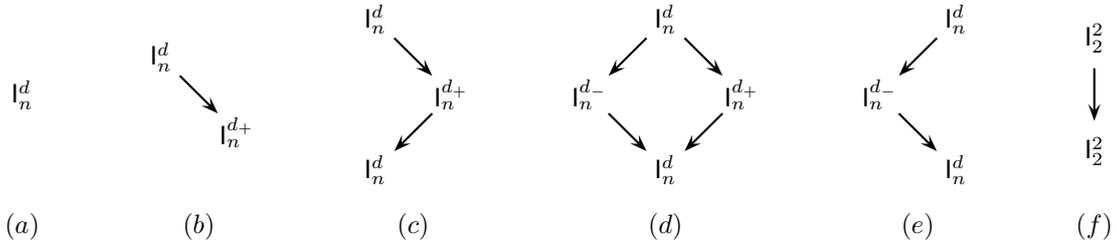
\begin{figure}
\begin{center}
\begin{pspicture}[shift=-0.9](1.5,-1.5)(2.5,1)
\rput(2,0){$\Itl_n^{d}$}
\rput(2,-1.75){$(a)$}
\end{pspicture} \qquad
\begin{pspicture}[shift=-0.9](1,-1.5)(3,1)
\rput(0,-0.5){\rput(1.5,1){$\Itl_n^{d}$}
\rput(2.5,0){$\Itl_n^{d_+}$}
\psline[linewidth=.8pt,arrowsize=3pt 2]{->}(1.75,0.75)(2.25,0.25)}
\rput(2,-1.75){$(b)$}
\end{pspicture} \qquad
\begin{pspicture}[shift=-0.9](1,-1.5)(3,1)
\rput(1.5,1){$\Itl_n^{d}$}
\rput(1.5,-1){$\Itl_n^{d}$}
\rput(2.5,0){$\Itl_n^{d_+}$}
\psline[linewidth=.8pt,arrowsize=3pt 2]{->}(1.75,0.75)(2.25,0.25)
\psline[linewidth=.8pt,arrowsize=3pt 2]{->}(2.25,-0.25)(1.75,-0.75)
\rput(2,-1.75){$(c)$}
\end{pspicture}
\qquad 
\begin{pspicture}[shift=-0.9](0,-1.5)(3,1)
\rput(0.5,0){$\Itl_n^{d_-}$}
\rput(1.5,1){$\Itl_n^{d}$}
\rput(1.5,-1){$\Itl_n^{d}$}
\rput(2.5,0){$\Itl_n^{d_+}$}
\psline[linewidth=.8pt,arrowsize=3pt 2]{->}(1.25,0.75)(0.75,0.25)
\psline[linewidth=.8pt,arrowsize=3pt 2]{->}(1.75,0.75)(2.25,0.25)
\psline[linewidth=.8pt,arrowsize=3pt 2]{->}(0.75,-0.25)(1.25,-0.75)
\psline[linewidth=.8pt,arrowsize=3pt 2]{->}(2.25,-0.25)(1.75,-0.75)
\rput(1.5,-1.75){$(d)$}
\end{pspicture}
\qquad
\begin{pspicture}[shift=-0.9](0,-1.5)(2,1)
\rput(0.5,0){$\Itl_n^{d_-}$}
\rput(1.5,1){$\Itl_n^{d}$}
\rput(1.5,-1){$\Itl_n^{d}$}
\psline[linewidth=.8pt,arrowsize=3pt 2]{->}(1.25,0.75)(0.75,0.25)
\psline[linewidth=.8pt,arrowsize=3pt 2]{->}(0.75,-0.25)(1.25,-0.75)
\rput(1,-1.75){$(e)$}
\end{pspicture}
\qquad
\begin{pspicture}[shift=-0.9](1,-1.5)(2,1)
\rput(1.5,0.75){$\Itl_2^{2}$}
\rput(1.5,-0.75){$\Itl_2^{2}$}
\psline[linewidth=.8pt,arrowsize=3pt 2]{->}(1.5,0.35)(1.5,-0.35)
\rput(1.5,-1.75){$(f)$}
\end{pspicture}
\end{center}
\caption{The six types of Loewy diagrams that arise
for standard and principal indecomposable modules of $\tl_n(\beta)$.}
\label{fig:Loewy}
\end{figure} 

\paragraph{Semisimple cases.} Let $\beta = q + q^{-1} \in \CC$ and, if $q$ is a root of unity, let $\ell$ be the smallest positive integer 
satisfying $q^{2 \ell} = 1$. The Temperley-Lieb algebra is semisimple in four cases: 
(i) $q$ is not a root of unity, 
(ii) $q$ is a root of unity with $\ell = 1$, so $q=\pm1$ and $\beta=\pm2$, 
(iii) $q$ is a root of unity with $\ell = 2$, so $q=\pm\ii$ and $\beta=0$, and $n$ is odd, and 
(iv) $q$ is a root of unity with $\ell > 2$ and $n < \ell$. 
In these cases, the semisimplicity of $\tl_n(\beta)$ implies that the irreducible and principal indecomposable modules coincide:  $\Itl_n^d = \Ptl_n^d$, for all $d$. The standard modules, built from the standard action on link states in \cref{sec:TLa}, realise all of these: $\Itl_n^d = \stan_n^d = \Ptl_n^d$. These modules are all described by the Loewy diagram of type $(a)$ (see \cref{fig:Loewy}). 

\paragraph{Non-semisimple cases.} Non-semisimplicity occurs in the following cases: 
(i) $q$ is a root of unity with $\ell = 2$, so $q=\pm\ii$ and $\beta=0$, and $n$ is even, and 
(ii) $q$ is a root of unity with $\ell > 2$ and $n \ge \ell$. 
Under these conditions, there exist representations of $\tl_n(\beta)$ that are reducible yet indecomposable, including some of the standard and principal indecomposable modules. To describe their structures, 
it is convenient to fix a little nomenclature. 

Let us fix $n$ and $\ell$ and denote by $\pi(n)$ the set of integers $\{d\,|\,  0 \le d \le n, \ d = n \bmod{2}\}$. 
We now define an $\ell$-dependent partition of $\pi(n)$. An integer $d \in \pi(n)$ is said to be 
\emph{critical} if $d+1 = 0 \bmod{\ell}$. Each critical integer is taken to form a one-element part of the partition of $\pi(n)$. The other integers are grouped in 
\emph{non-critical orbits}:
\begin{equation}
O_a = \{d \in \pi(n)\,| \,d+1 = \pm a \bmod{2\ell} \} \qquad 0<a<\ell , \quad a+1 = n \bmod{2}.
\end{equation} 
The non-critical orbits form the other parts of the partition of $\pi(n)$. For instance, for $n = 22$ and $\ell = 7$, $\pi(22)$ partitions as $\{\{6\},\{20\}, O_1, O_3, O_5\} $ with $O_1 = \{0,12,14\}$, $O_3 = \{2,10,16\}$ and $O_5=\{4,8,18,22\}$. It is useful to assume that the integers in a non-critical orbit are given in increasing order. For $d$ non-critical, we respectively denote by $d_-$ and $d_+$ the integers appearing immediately before and after $d$ in its orbit, whenever such integers exist.

We now use this construction to give the rules that dictate the structures of the irreducible, standard and projective modules when $\ell>2$ and $n\ge\ell$. 
The case $\ell=2$ with $n$ even is special and discussed at the end. We note that in every case, the irreducible module $\Itl_n^d$ is a quotient of the standard module $\stan_n^d$, which in turn is realised as a quotient of the projective module $\Ptl_n^d$.

First, if $d$ is critical, then $\Itl_n^d = \stan_n^d = \Ptl_n^d$ is irreducible and projective, so the Loewy diagram is of type $(a)$. Similarly, if $d$ is the largest integer in its non-critical orbit, then $\stan_n^d$ is also irreducible, 
$\stan_n^d = \Itl_n^d$, hence its  Loewy diagram is of type $(a)$. Moreover, $\Ptl_n^d$ then has three composition factors and its Loewy diagram is of type $(e)$.  On the other hand, if $d$ is the smallest integer in its non-critical orbit, then $\stan_n^d = \Ptl_n^d$ has two composition factors and the Loewy diagram is of type $(b)$. Finally, if $d$ is neither the smallest, nor the largest, integer in its critical orbit, then $\Itl_n^d \neq \stan_n^d \neq \Ptl_n^d$.  $\stan_n^d$ again has two composition factors with Loewy diagram of type $(b)$, but $\Ptl_n^d$ now has four composition factors and its Loewy diagram is of type $(d)$.

These rules degenerate for the case $\ell = 2$ with $n$ even. In this case, the partition of $\pi(n)$ takes the form of a single non-critical orbit, $O_1 = \set{0, 2, \ldots, n}$. As in the above prescription, the standard module corresponding to the largest integer of $O_1$ is irreducible, 
$\stan_n^n = \Itl_n^n$, and thus of type $(a)$. For $0<d<n$, the modules $\stan_n^d$ are reducible and their Loewy diagram is of type $(b)$ with $d_+ = d+2$. The smallest integer of $O_1$ is where things differ: $\stan_n^0$ is irreducible and, exceptionally, is isomorphic to $\Itl_n^2$, the irreducible quotient of $\stan_n^2$. There are therefore only $\frac n2$ non-isomorphic irreducible modules, with a full set given by the irreducible quotient modules of $\stan_n^d$ for $d = 2, 4, \dots, n$.  We remark that it follows that $\stan_2^2 = \Itl_2^2 \simeq \stan_2^0$, 
which is the only isomorphism between standard $\tl_n(\beta)$-modules, for any $\beta \in \CC$, with different labels $d$.

For the principal indecomposables, we first restrict to $n\ge4$.  $\Ptl_n^n$ is then described by the Loewy diagram of type $(e)$, with $d = n$ and $d_- = n-2$, and $\Ptl_n^d$ for $2<d<n$ by the diagram of type $(d)$, with $d_\pm = d\pm2$. The case $d = 2$  
is special as $\Ptl_n^{2}$ is represented by the Loewy diagram of type $(c)$ with $d_+ = 4$. This is the only occurrence of the type $(c)$ Loewy diagram.  The case where $n=2$ degenerates even further as there is now a single principal indecomposable module $\Ptl_2^2$ and its Loewy diagram is of type $(f)$.  This is also the only time this Loewy diagram occurs.

%
\section{Presenting the one-boundary \TL{} algebras} \label{app:TLone} 
%

In this appendix, we prove that the diagrammatic one-boundary \TL{} algebra $\tlone_n$, introduced in \cref{sec:TLone}, is actually isomorphic to its algebraic counterpart defined by the generators $e_j$, $j=1, 2, \ldots, n$, and the relations \eqref{eq:defTL} and \eqref{eq:defTL1}.  Demonstrating such isomorphisms, hence that the relations of the algebraic definition form a complete set, is a very
subtle business in general.  For the \TL{} algebra itself, Kauffman indicated the first direct proof in \cite{KauInv90}, at the level of an example, referring to Jones' pre-diagrammatic work \cite{JonInd83} where many of the arguments had already appeared. A full proof may be found in \cite{RidSta12}.  We are not aware of any proofs in the analogous case of the one-boundary \TL{} algebras, though the equivalence is widely recognised and frequently used.  Motivated by the need to establish the completeness of the relations of the boundary seam algebras, see \cref{prop:BTLComplete}, we 
provide a proof here.

For the rest of this appendix, we shall distinguish the diagrammatic and algebraic definitions of the one-boundary \TL{} algebra, denoting the former by $\diatl_n$ and the latter by $\algtl_n$.  The map defined by \eqref{eq:iso} and \eqref{eq:ison} on the $\algtl_n$ generators extends to all words $\set{e_{j_1} \cdots e_{j_m}}$ using the $\diatl_n$ product defined in \cref{sec:boundaryTLs}; this defines a map 
\begin{equation}
 \psi \colon \algtl_n \to \diatl_n.
\end{equation}
Our first task is to establish that $\psi$ is surjective.  Because the connectivities satisfy the algebraic relations \eqref{eq:defTL} and \eqref{eq:defTL1}, this map then
extends to a surjective homomorphism 
and we conclude that $\diatl_n$ is a quotient of $\algtl_n$.  To prove that $\psi$ is actually an isomorphism, it is enough to show that $\dim \algtl_n \le \dim \diatl_n$.  This, in turn, may be demonstrated by finding a spanning set for the words of $\algtl_n$ whose cardinality is equal to the number \eqref{eq:diatlcounting}  
of connectivities, these forming a basis of $\diatl_n$ by definition.

\begin{Proposition} \label{prop:Surjective}
The map $\psi$ is surjective.
\end{Proposition}
\begin{proof}
Given an arbitrary connectivity $D \in \diatl_n$, we must construct a word $w \in \algtl_n$ such that $\psi(w) = D$.  This construction proceeds in three steps. The first step notes which nodes of the top and bottom edges are connected to the right boundary in $D$ and constructs a word $w'$ such that $D' = \psi(w')$ is a connectivity which also has these top and bottom nodes connected to the right boundary. The connections of the remaining nodes of $w'$ differ,
in general, from those of $w$. Any bottom (top) node to the right of a node connected to the boundary must be connected to another bottom (top) node, provided that it is not itself connected to the boundary.  In $D'$, each of these nodes connects to the node 
immediately to the left or right; we call such connections \emph{simple arcs}.  The second step constructs from $w'$ a new word $w''$ whose corresponding diagram $D'' = \psi(w'')$ is obtained from $D'$ by adding simple arcs to the left of all nodes connected to the boundary in such a way that the arcs of $D''$ that connect a top to a bottom node precisely match the corresponding arcs in $D$. It may happen that this step is trivial and $w'' = w'$. Finally, the third step converts simple arcs into nests of arcs as required to arrive at $w$ and $D$.  The construction of $w'$ involves the generator $e_n$, whereas the subsequent steps only require the $e_j$, with $j<n$, and are precisely the steps needed to prove the corresponding surjectivity result for the \TL{} algebra $\tl_n$, detailed in \cite[Sec.~2]{RidSta12}.  We will therefore only discuss the first step here.

In $D$, the number of connections to the boundary is necessarily even.  We will order these boundary connections, starting from the rightmost on the bottom edge (if one exists), proceeding leftwards along the bottom edge, then taking the leftmost on the top edge and proceeding rightwards along the top.  Suppose that the first boundary connection is node $j$ on the bottom edge.  Then, we construct an $\algtl_n$ word, and its image under $\psi$, by displacing the bottom boundary connection of the generator $e_n$ from node $n$ to node $j$ using a 
\emph{snake}, constructed by left- and right-multiplying by products of generators whose indices increase in steps of two:
\begin{equation}
\psi \colon (e_{j+1} e_{j+3} \cdots e_{n-1}) e_n (e_j e_{j+2} \cdots e_{n-2}) \mapsto \ 
\begin{pspicture}[shift=-1.05](0.0,-1.15)(2.4,1.05)
\pspolygon[fillstyle=solid,fillcolor=lightlightblue](0.0,1.05)(2.4,1.05)(2.4,0.35)(0.0,0.35)
\rput(1.4,0.7){\small$...$}
\rput(0.4,0.7){\small$...$}
\psline[linecolor=blue,linewidth=1.5pt]{-}(0.2,0.35)(0.2,1.05)
\psline[linecolor=blue,linewidth=1.5pt]{-}(0.6,0.35)(0.6,1.05)
\psline[linecolor=blue,linewidth=1.5pt]{-}(2.2,0.35)(2.2,1.05)
\psarc[linecolor=blue,linewidth=1.5pt]{-}(1.9,0.35){0.1}{0}{180}
\psarc[linecolor=blue,linewidth=1.5pt]{-}(1.5,0.35){0.1}{0}{90}
\psarc[linecolor=blue,linewidth=1.5pt]{-}(1.3,0.35){0.1}{90}{180}
\psarc[linecolor=blue,linewidth=1.5pt]{-}(0.9,0.35){0.1}{0}{180}
\psarc[linecolor=blue,linewidth=1.5pt]{-}(1.9,1.05){0.1}{180}{0}
\psarc[linecolor=blue,linewidth=1.5pt]{-}(1.5,1.05){0.1}{270}{0}
\psarc[linecolor=blue,linewidth=1.5pt]{-}(1.3,1.05){0.1}{180}{270}
\psarc[linecolor=blue,linewidth=1.5pt]{-}(0.9,1.05){0.1}{180}{0}
\pspolygon[fillstyle=solid,fillcolor=lightlightblue](0.0,-0.35)(2.4,-0.35)(2.4,0.35)(0.0,0.35)
\rput(1.4,0.0){\small$...$}
\rput(0.4,0.0){\small$...$}
\psline[linecolor=blue,linewidth=1.5pt]{-}(0.2,0.35)(0.2,-0.35)
\psline[linecolor=blue,linewidth=1.5pt]{-}(0.6,0.35)(0.6,-0.35)
\psline[linecolor=blue,linewidth=1.5pt]{-}(0.8,0.35)(0.8,-0.35)
\psline[linecolor=blue,linewidth=1.5pt]{-}(1.0,0.35)(1.0,-0.35)
\psline[linecolor=blue,linewidth=1.5pt]{-}(1.2,0.35)(1.2,-0.35)
\psline[linecolor=blue,linewidth=1.5pt]{-}(1.6,0.35)(1.6,-0.35)
\psline[linecolor=blue,linewidth=1.5pt]{-}(1.8,0.35)(1.8,-0.35)
\psline[linecolor=blue,linewidth=1.5pt]{-}(2.0,0.35)(2.0,-0.35)
\psarc[linecolor=blue,linewidth=1.5pt]{-}(2.4,-0.35){0.2}{90}{180}
\psarc[linecolor=blue,linewidth=1.5pt]{-}(2.4,0.35){0.2}{180}{270}
\pspolygon[fillstyle=solid,fillcolor=lightlightblue](0.0,-0.35)(2.4,-0.35)(2.4,-1.05)(0.0,-1.05)
\rput(1.4,-0.7){\small$...$}
\rput(0.4,-0.7){\small$...$}
\psline[linecolor=blue,linewidth=1.5pt]{-}(0.2,-1.05)(0.2,-0.35)
\psline[linecolor=blue,linewidth=1.5pt]{-}(0.6,-1.05)(0.6,-0.35)
\psline[linecolor=blue,linewidth=1.5pt]{-}(0.8,-1.05)(0.8,-0.35)
\psarc[linecolor=blue,linewidth=1.5pt]{-}(2.1,-0.35){0.1}{180}{0}
\psarc[linecolor=blue,linewidth=1.5pt]{-}(1.7,-0.35){0.1}{180}{0}
\psarc[linecolor=blue,linewidth=1.5pt]{-}(1.1,-0.35){0.1}{180}{0}
\psarc[linecolor=blue,linewidth=1.5pt]{-}(2.1,-1.05){0.1}{0}{180}
\psarc[linecolor=blue,linewidth=1.5pt]{-}(1.7,-1.05){0.1}{0}{180}
\psarc[linecolor=blue,linewidth=1.5pt]{-}(1.1,-1.05){0.1}{0}{180}
\rput(0.8,-1.25){$_j$}
\rput(2.2,-1.25){$_n$}
\end{pspicture}
\ = \ 
\begin{pspicture}[shift=-0.35](0.0,-0.45)(2.4,0.35)
\pspolygon[fillstyle=solid,fillcolor=lightlightblue](0.0,-0.35)(2.4,-0.35)(2.4,0.35)(0.0,0.35)
\rput(1.4,0.0){\small$...$}
\rput(0.4,0.0){\small$...$}
\psline[linecolor=blue,linewidth=1.5pt]{-}(0.2,0.35)(0.2,-0.35)
\psline[linecolor=blue,linewidth=1.5pt]{-}(0.6,0.35)(0.6,-0.35)
\psarc[linecolor=blue,linewidth=1.5pt]{-}(1.9,0.35){0.1}{180}{0}
\psarc[linecolor=blue,linewidth=1.5pt]{-}(1.5,0.35){0.1}{270}{0}
\psarc[linecolor=blue,linewidth=1.5pt]{-}(1.3,0.35){0.1}{180}{270}
\psarc[linecolor=blue,linewidth=1.5pt]{-}(0.9,0.35){0.1}{180}{0}
\psarc[linecolor=blue,linewidth=1.5pt]{-}(2.4,0.35){0.2}{180}{270}
\psarc[linecolor=blue,linewidth=1.5pt]{-}(2.1,-0.35){0.1}{0}{180}
\psarc[linecolor=blue,linewidth=1.5pt]{-}(1.7,-0.35){0.1}{0}{180}
\psarc[linecolor=blue,linewidth=1.5pt]{-}(1.1,-0.35){0.1}{0}{180}
\psarc[linecolor=blue,linewidth=1.5pt]{-}(1.0,-0.35){0.2}{90}{180}
\psline[linecolor=blue,linewidth=1.5pt]{-}(1.0,-0.15)(1.3,-0.15)
\psline[linecolor=blue,linewidth=1.5pt]{-}(1.5,-0.15)(2.4,-0.15)
\rput(0.8,-0.55){$_j$}
\rput(2.2,-0.55){$_n$}
\end{pspicture}
\ .
\end{equation}
This works because $j$ and $n$ must have the same parity.  Note that node $n$ on the top edge is still connected to the boundary.  This boundary connection can also be moved, using another snake, to the node $j'$ that has the second boundary connection of $D$.  If $j'$ is on the bottom edge, its parity will be opposite that of $j$ and $n$.  If $j'$ lies on the top edge, its parity matches that of $j$ and $n$.  Either way, a snake may be constructed as before. The products of generators that are needed to form this snake need not terminate with $e_{n-1}$ or $e_{n-2}$.  We note that it may happen that such a product is not required, in which case the standard convention applies that an empty product gives the identity $I$.

We proceed in this fashion, moving the boundary arcs to the desired places using snakes, until the boundary arcs of $D$ have been constructed.  After moving two arcs in this manner, the boundary arcs of the original generator $e_n$ are both used; the remedy is to right-multiply by another $e_n$ and repeat.  The snakes clearly leave simple arcs to the right of the boundary connection nodes, so the result is the word $w' \in \algtl_n$ that is the input for the second and last steps of the construction.  These steps are detailed in \cite{RidSta12}, completing the proof.
\end{proof}

We illustrate this construction for a connectivity with $n=8$ and four nodes connected to the boundary:
\begin{equation}
D = \ 
\begin{pspicture}[shift=-0.6](0.0,-0.7)(1.8,0.7)
\pspolygon[fillstyle=solid,fillcolor=lightlightblue](0.0,-0.7)(1.8,-0.7)(1.8,0.7)(0.0,0.7)
\psarc[linecolor=blue,linewidth=1.5pt]{-}(0.7,-0.7){0.1}{0}{180}
\psarc[linecolor=blue,linewidth=1.5pt]{-}(1.5,-0.7){0.1}{0}{180}
\psarc[linecolor=blue,linewidth=1.5pt]{-}(0.3,0.7){0.1}{180}{0}
\psarc[linecolor=blue,linewidth=1.5pt]{-}(1.3,0.7){0.1}{180}{0}
\psarc[linecolor=blue,linewidth=1.5pt]{-}(1.3,0.7){0.3}{180}{0}
\psarc[linecolor=blue,linewidth=1.5pt]{-}(1.5,-0.7){0.3}{90}{180} 
\psline[linecolor=blue,linewidth=1.5pt]{-}(1.5,-0.4)(1.8,-0.4)
\psarc[linecolor=blue,linewidth=1.5pt]{-}(1.5,-0.7){0.5}{90}{180} 
\psline[linecolor=blue,linewidth=1.5pt]{-}(1.5,-0.2)(1.8,-0.2)
\psarc[linecolor=blue,linewidth=1.5pt]{-}(1.1,-0.7){0.7}{90}{180} 
\psline[linecolor=blue,linewidth=1.5pt]{-}(1.1,0.0)(1.8,0.0)
\psline[linecolor=blue,linewidth=1.5pt]{-}(0.2,-0.7)(0.2,-0.2)
\psarc[linecolor=blue,linewidth=1.5pt]{-}(0.4,-0.2){0.2}{90}{180} 
\psarc[linecolor=blue,linewidth=1.5pt]{-}(0.4,0.2){0.2}{270}{0} 
\psline[linecolor=blue,linewidth=1.5pt]{-}(0.6,0.2)(0.6,0.7)
\psarc[linecolor=blue,linewidth=1.5pt]{-}(1.3,0.7){0.5}{180}{270}
\psline[linecolor=blue,linewidth=1.5pt]{-}(1.3,0.2)(1.8,0.2)
\end{pspicture}
\ \in \diatl_8.
\end{equation}
The ordering of the boundary connection nodes is $6$, $5$, $2$ (bottom), then $4$ (top).  We start with $e_8$ and use snakes as in the proof to arrive at 
\begin{center} 
{\setlength{\tabcolsep}{0.03\textwidth}
\begin{tabular}{CCCC}
&
\begin{pspicture}[shift=-0.8](0.0,-0.9)(1.8,0.9)
\pspolygon[fillstyle=solid,fillcolor=lightlightblue](0.0,-0.9)(1.8,-0.9)(1.8,-0.3)(0.0,-0.3)
\pspolygon[fillstyle=solid,fillcolor=lightlightblue](0.0,-0.3)(1.8,-0.3)(1.8,+0.3)(0.0,+0.3)
\pspolygon[fillstyle=solid,fillcolor=lightlightblue](0.0,+0.3)(1.8,+0.3)(1.8,+0.9)(0.0,+0.9)
\psline[linecolor=blue,linewidth=1.5pt]{-}(0.2,-0.9)(0.2,0.9)
\psline[linecolor=blue,linewidth=1.5pt]{-}(0.4,-0.9)(0.4,0.9)
\psline[linecolor=blue,linewidth=1.5pt]{-}(0.6,-0.9)(0.6,0.9)
\psline[linecolor=blue,linewidth=1.5pt]{-}(0.8,-0.9)(0.8,0.9)
\psline[linecolor=blue,linewidth=1.5pt]{-}(1.0,-0.9)(1.0,0.9)
\psline[linecolor=blue,linewidth=1.5pt]{-}(1.2,-0.9)(1.2,0.3)
\psline[linecolor=blue,linewidth=1.5pt]{-}(1.4,-0.3)(1.4,0.3)
\psline[linecolor=blue,linewidth=1.5pt]{-}(1.6,0.3)(1.6,0.9)
\psarc[linecolor=blue,linewidth=1.5pt]{-}(1.3,0.9){0.1}{180}{0}
\psarc[linecolor=blue,linewidth=1.5pt]{-}(1.3,0.3){0.1}{0}{180}
\psarc[linecolor=blue,linewidth=1.5pt]{-}(1.5,-0.3){0.1}{180}{0}
\psarc[linecolor=blue,linewidth=1.5pt]{-}(1.5,-0.9){0.1}{0}{180}
\psarc[linecolor=blue,linewidth=1.5pt]{-}(1.8,0.3){0.2}{180}{270}
\psarc[linecolor=blue,linewidth=1.5pt]{-}(1.8,-0.3){0.2}{90}{180}
\end{pspicture}
\ = \ 
\begin{pspicture}[shift=-0.3](0.0,-0.4)(1.8,0.4)
\pspolygon[fillstyle=solid,fillcolor=lightlightblue](0.0,-0.4)(1.8,-0.4)(1.8,+0.4)(0.0,+0.4)
\psline[linecolor=blue,linewidth=1.5pt]{-}(0.2,-0.4)(0.2,0.4)
\psline[linecolor=blue,linewidth=1.5pt]{-}(0.4,-0.4)(0.4,0.4)
\psline[linecolor=blue,linewidth=1.5pt]{-}(0.6,-0.4)(0.6,0.4)
\psline[linecolor=blue,linewidth=1.5pt]{-}(0.8,-0.4)(0.8,0.4)
\psline[linecolor=blue,linewidth=1.5pt]{-}(1.0,-0.4)(1.0,0.4)
\psarc[linecolor=blue,linewidth=1.5pt]{-}(1.3,0.4){0.1}{180}{0}
\psarc[linecolor=blue,linewidth=1.5pt]{-}(1.5,-0.4){0.1}{0}{180}
\psarc[linecolor=blue,linewidth=1.5pt]{-}(1.8,0.4){0.2}{180}{270}
\psarc[linecolor=blue,linewidth=1.5pt]{-}(1.5,-0.4){0.3}{90}{180}
\psline[linecolor=blue,linewidth=1.5pt]{-}(1.5,-0.1)(1.8,-0.1)
\end{pspicture}
& \lra &
\begin{pspicture}[shift=-0.6](0.0,-0.7)(1.8,0.7)
\pspolygon[fillstyle=solid,fillcolor=lightlightblue](0.0,-0.7)(1.8,-0.7)(1.8,+0.1)(0.0,+0.1)
\pspolygon[fillstyle=solid,fillcolor=lightlightblue](0.0,+0.1)(1.8,+0.1)(1.8,+0.7)(0.0,+0.7)
\psline[linecolor=blue,linewidth=1.5pt]{-}(0.2,-0.7)(0.2,0.7)
\psline[linecolor=blue,linewidth=1.5pt]{-}(0.4,-0.7)(0.4,0.7)
\psline[linecolor=blue,linewidth=1.5pt]{-}(0.6,-0.7)(0.6,0.7)
\psline[linecolor=blue,linewidth=1.5pt]{-}(0.8,-0.7)(0.8,0.7)
\psline[linecolor=blue,linewidth=1.5pt]{-}(1.0,-0.7)(1.0,0.1)
\psarc[linecolor=blue,linewidth=1.5pt]{-}(1.1,0.7){0.1}{180}{0}
\psarc[linecolor=blue,linewidth=1.5pt]{-}(1.5,0.7){0.1}{180}{0}
\psarc[linecolor=blue,linewidth=1.5pt]{-}(1.1,0.1){0.1}{0}{180}
\psarc[linecolor=blue,linewidth=1.5pt]{-}(1.5,0.1){0.1}{0}{180}
\psarc[linecolor=blue,linewidth=1.5pt]{-}(1.3,0.1){0.1}{180}{0}
\psarc[linecolor=blue,linewidth=1.5pt]{-}(1.8,0.1){0.2}{180}{270}
\psarc[linecolor=blue,linewidth=1.5pt]{-}(1.5,-0.7){0.1}{0}{180}
\psarc[linecolor=blue,linewidth=1.5pt]{-}(1.5,-0.7){0.3}{90}{180}
\psline[linecolor=blue,linewidth=1.5pt]{-}(1.5,-0.4)(1.8,-0.4)
\end{pspicture}
\ = \ 
\begin{pspicture}[shift=-0.4](0.0,-0.5)(1.8,0.5)
\pspolygon[fillstyle=solid,fillcolor=lightlightblue](0.0,-0.5)(1.8,-0.5)(1.8,+0.5)(0.0,+0.5)
\psline[linecolor=blue,linewidth=1.5pt]{-}(0.2,-0.5)(0.2,0.5)
\psline[linecolor=blue,linewidth=1.5pt]{-}(0.4,-0.5)(0.4,0.5)
\psline[linecolor=blue,linewidth=1.5pt]{-}(0.6,-0.5)(0.6,0.5)
\psline[linecolor=blue,linewidth=1.5pt]{-}(0.8,-0.5)(0.8,0.5)
\psarc[linecolor=blue,linewidth=1.5pt]{-}(1.1,0.5){0.1}{180}{0}
\psarc[linecolor=blue,linewidth=1.5pt]{-}(1.5,0.5){0.1}{180}{0}
\psarc[linecolor=blue,linewidth=1.5pt]{-}(1.5,-0.5){0.1}{0}{180}
\psarc[linecolor=blue,linewidth=1.5pt]{-}(1.5,-0.5){0.3}{90}{180}
\psarc[linecolor=blue,linewidth=1.5pt]{-}(1.5,-0.5){0.5}{90}{180}
\psline[linecolor=blue,linewidth=1.5pt]{-}(1.5,-0.2)(1.8,-0.2)
\psline[linecolor=blue,linewidth=1.5pt]{-}(1.5,0.0)(1.8,0.0)
\end{pspicture}
\\[0.8cm]
& (e_7) (e_8) (e_6) & & (I) e_7 e_8 e_6 (e_5 e_7) \\[5mm]
\lra & 
\begin{pspicture}[shift=-1.3](0.0,-1.4)(1.8,1.4)
\pspolygon[fillstyle=solid,fillcolor=lightlightblue](0.0,-1.4)(1.8,-1.4)(1.8,-0.8)(0.0,-0.8)
\pspolygon[fillstyle=solid,fillcolor=lightlightblue](0.0,-0.8)(1.8,-0.8)(1.8,+0.2)(0.0,+0.2)
\pspolygon[fillstyle=solid,fillcolor=lightlightblue](0.0,+0.2)(1.8,+0.2)(1.8,+0.8)(0.0,+0.8)
\pspolygon[fillstyle=solid,fillcolor=lightlightblue](0.0,+0.8)(1.8,+0.8)(1.8,+1.4)(0.0,+1.4)
\psline[linecolor=blue,linewidth=1.5pt]{-}(0.2,-1.4)(0.2,1.4)
\psline[linecolor=blue,linewidth=1.5pt]{-}(0.4,-1.4)(0.4,0.8)
\psline[linecolor=blue,linewidth=1.5pt]{-}(0.6,-0.8)(0.6,0.8)
\psline[linecolor=blue,linewidth=1.5pt]{-}(0.8,-0.8)(0.8,0.8)
\psline[linecolor=blue,linewidth=1.5pt]{-}(1.0,-1.4)(1.0,-0.8)
\psline[linecolor=blue,linewidth=1.5pt]{-}(1.2,-1.4)(1.2,-0.8)
\psline[linecolor=blue,linewidth=1.5pt]{-}(1.4,-1.4)(1.4,-0.8)
\psline[linecolor=blue,linewidth=1.5pt]{-}(1.6,-1.4)(1.6,-0.8)
\psline[linecolor=blue,linewidth=1.5pt]{-}(1.0,0.2)(1.0,0.8)
\psline[linecolor=blue,linewidth=1.5pt]{-}(1.2,0.2)(1.2,0.8)
\psline[linecolor=blue,linewidth=1.5pt]{-}(1.4,0.2)(1.4,0.8)
\psline[linecolor=blue,linewidth=1.5pt]{-}(1.6,0.8)(1.6,1.4)
\psarc[linecolor=blue,linewidth=1.5pt]{-}(0.5,1.4){0.1}{180}{0}
\psarc[linecolor=blue,linewidth=1.5pt]{-}(0.9,1.4){0.1}{180}{0}
\psarc[linecolor=blue,linewidth=1.5pt]{-}(1.3,1.4){0.1}{180}{0}
\psarc[linecolor=blue,linewidth=1.5pt]{-}(0.5,0.8){0.1}{0}{180}
\psarc[linecolor=blue,linewidth=1.5pt]{-}(0.9,0.8){0.1}{0}{180}
\psarc[linecolor=blue,linewidth=1.5pt]{-}(1.3,0.8){0.1}{0}{180}
\psarc[linecolor=blue,linewidth=1.5pt]{-}(1.1,0.2){0.1}{180}{0}
\psarc[linecolor=blue,linewidth=1.5pt]{-}(1.5,0.2){0.1}{180}{0}
\psarc[linecolor=blue,linewidth=1.5pt]{-}(0.7,-0.8){0.1}{180}{0}
\psarc[linecolor=blue,linewidth=1.5pt]{-}(1.5,-0.8){0.1}{0}{180}
\psarc[linecolor=blue,linewidth=1.5pt]{-}(0.7,-1.4){0.1}{0}{180}
\psarc[linecolor=blue,linewidth=1.5pt]{-}(1.8,0.8){0.2}{180}{270}
\psarc[linecolor=blue,linewidth=1.5pt]{-}(1.8,0.2){0.2}{90}{180}
\psarc[linecolor=blue,linewidth=1.5pt]{-}(1.5,-0.8){0.3}{90}{180}
\psarc[linecolor=blue,linewidth=1.5pt]{-}(1.5,-0.8){0.5}{90}{180}
\psline[linecolor=blue,linewidth=1.5pt]{-}(1.5,-0.5)(1.8,-0.5)
\psline[linecolor=blue,linewidth=1.5pt]{-}(1.5,-0.3)(1.8,-0.3)
\end{pspicture}
\ = \ 
\begin{pspicture}[shift=-0.5](0.0,-0.6)(1.8,0.6)
\pspolygon[fillstyle=solid,fillcolor=lightlightblue](0.0,-0.6)(1.8,-0.6)(1.8,+0.6)(0.0,+0.6)
\psline[linecolor=blue,linewidth=1.5pt]{-}(0.2,-0.6)(0.2,0.6)
\psarc[linecolor=blue,linewidth=1.5pt]{-}(0.5,0.6){0.1}{180}{0}
\psarc[linecolor=blue,linewidth=1.5pt]{-}(0.9,0.6){0.1}{180}{0}
\psarc[linecolor=blue,linewidth=1.5pt]{-}(1.3,0.6){0.1}{180}{0}
\psarc[linecolor=blue,linewidth=1.5pt]{-}(0.7,-0.6){0.1}{0}{180}
\psarc[linecolor=blue,linewidth=1.5pt]{-}(1.5,-0.6){0.1}{0}{180}
\psarc[linecolor=blue,linewidth=1.5pt]{-}(1.8,0.6){0.2}{180}{270}
\psarc[linecolor=blue,linewidth=1.5pt]{-}(1.5,-0.6){0.3}{90}{180}
\psarc[linecolor=blue,linewidth=1.5pt]{-}(1.5,-0.6){0.5}{90}{180}
\psarc[linecolor=blue,linewidth=1.5pt]{-}(1.1,-0.6){0.7}{90}{180}
\psline[linecolor=blue,linewidth=1.5pt]{-}(1.5,-0.3)(1.8,-0.3)
\psline[linecolor=blue,linewidth=1.5pt]{-}(1.5,-0.1)(1.8,-0.1)
\psline[linecolor=blue,linewidth=1.5pt]{-}(1.1,0.1)(1.8,0.1)
\end{pspicture}
& \lra &
\begin{pspicture}[shift=-0.8](0.0,-0.9)(1.8,0.9)
\pspolygon[fillstyle=solid,fillcolor=lightlightblue](0.0,-0.9)(1.8,-0.9)(1.8,+0.3)(0.0,+0.3)
\pspolygon[fillstyle=solid,fillcolor=lightlightblue](0.0,+0.3)(1.8,+0.3)(1.8,+0.9)(0.0,+0.9)
\psline[linecolor=blue,linewidth=1.5pt]{-}(0.2,-0.9)(0.2,0.9)
\psline[linecolor=blue,linewidth=1.5pt]{-}(0.4,0.3)(0.4,0.9)
\psline[linecolor=blue,linewidth=1.5pt]{-}(0.6,0.3)(0.6,0.9)
\psline[linecolor=blue,linewidth=1.5pt]{-}(0.8,0.3)(0.8,0.9)
\psarc[linecolor=blue,linewidth=1.5pt]{-}(1.1,0.9){0.1}{180}{0}
\psarc[linecolor=blue,linewidth=1.5pt]{-}(1.5,0.9){0.1}{180}{0}
\psarc[linecolor=blue,linewidth=1.5pt]{-}(1.1,0.3){0.1}{0}{180}
\psarc[linecolor=blue,linewidth=1.5pt]{-}(1.5,0.3){0.1}{0}{180}
\psarc[linecolor=blue,linewidth=1.5pt]{-}(0.5,0.3){0.1}{180}{0}
\psarc[linecolor=blue,linewidth=1.5pt]{-}(0.9,0.3){0.1}{180}{0}
\psarc[linecolor=blue,linewidth=1.5pt]{-}(1.3,0.3){0.1}{180}{0}
\psarc[linecolor=blue,linewidth=1.5pt]{-}(0.7,-0.9){0.1}{0}{180}
\psarc[linecolor=blue,linewidth=1.5pt]{-}(1.5,-0.9){0.1}{0}{180}
\psarc[linecolor=blue,linewidth=1.5pt]{-}(1.8,0.3){0.2}{180}{270}
\psarc[linecolor=blue,linewidth=1.5pt]{-}(1.5,-0.9){0.3}{90}{180}
\psarc[linecolor=blue,linewidth=1.5pt]{-}(1.5,-0.9){0.5}{90}{180}
\psarc[linecolor=blue,linewidth=1.5pt]{-}(1.1,-0.9){0.7}{90}{180}
\psline[linecolor=blue,linewidth=1.5pt]{-}(1.5,-0.6)(1.8,-0.6)
\psline[linecolor=blue,linewidth=1.5pt]{-}(1.5,-0.4)(1.8,-0.4)
\psline[linecolor=blue,linewidth=1.5pt]{-}(1.1,-0.2)(1.8,-0.2)
\end{pspicture}
\ = \ 
\begin{pspicture}[shift=-0.6](0.0,-0.7)(1.8,0.7)
\pspolygon[fillstyle=solid,fillcolor=lightlightblue](0.0,-0.7)(1.8,-0.7)(1.8,+0.7)(0.0,+0.7)
\psline[linecolor=blue,linewidth=1.5pt]{-}(0.2,-0.7)(0.2,0.7)
\psarc[linecolor=blue,linewidth=1.5pt]{-}(0.5,0.7){0.1}{180}{0}
\psarc[linecolor=blue,linewidth=1.5pt]{-}(1.1,0.7){0.1}{180}{0}
\psarc[linecolor=blue,linewidth=1.5pt]{-}(1.5,0.7){0.1}{180}{0}
\psarc[linecolor=blue,linewidth=1.5pt]{-}(0.7,-0.7){0.1}{0}{180}
\psarc[linecolor=blue,linewidth=1.5pt]{-}(1.5,-0.7){0.1}{0}{180}
\psarc[linecolor=blue,linewidth=1.5pt]{-}(1.5,-0.7){0.3}{90}{180}
\psarc[linecolor=blue,linewidth=1.5pt]{-}(1.5,-0.7){0.5}{90}{180}
\psarc[linecolor=blue,linewidth=1.5pt]{-}(1.1,-0.7){0.7}{90}{180}
\psarc[linecolor=blue,linewidth=1.5pt]{-}(1.3,0.7){0.5}{180}{270}
\psline[linecolor=blue,linewidth=1.5pt]{-}(1.5,-0.4)(1.8,-0.4)
\psline[linecolor=blue,linewidth=1.5pt]{-}(1.5,-0.2)(1.8,-0.2)
\psline[linecolor=blue,linewidth=1.5pt]{-}(1.1,0.0)(1.8,0.0)
\psline[linecolor=blue,linewidth=1.5pt]{-}(1.3,0.2)(1.8,0.2)
\end{pspicture}
\\[1.3cm]
& (e_3) e_7 e_8 e_6 e_5 e_7 (e_8) (e_2 e_4 e_6) & & (I) e_3 e_7 e_8 e_6 e_5 e_7 e_8 e_2 e_4 e_6 (e_5 e_7) = w'.
\end{tabular}
}
\end{center}
The word $w'$ is the result of the first step of the construction.  To obtain a word $w \in \algtl_n$ satisfying $\psi(w) = D$, we move the leftmost arc connecting node $1$ (bottom) to $1$ (top) so that it connects $1$ (bottom) to $3$ (top), by right-multiplying by $e_1$, and then convert the two consecutive simple arcs at top-right into two nested arcs, by right-multiplying by $e_6$:
\begin{equation}
\psi \colon w = e_3 e_7 e_8 e_6 e_5 e_7 e_8 e_2 e_4 e_6 e_5 e_7 e_1 e_6 \mapsto D.
\end{equation}

Having proven the surjectivity of $\psi$, hence that $\diatl_n$ is a quotient of $\algtl_n$, we now introduce a spanning set of words for $\algtl_n$.  Following Jones \cite{JonInd83}, we will say that a word $e_{j_1} \cdots e_{j_m}$ in $\algtl_n$ 
is \emph{reduced} if it is not equal to any other word formed from fewer than $m$ generators multiplied by 
powers of $\beta$, $\beta_1$ and $\beta_2$. Obviously, every word is proportional to a reduced word, so the reduced words span $\algtl_n$.  We remark that the empty word $I$ is reduced.
\begin{Lemma} \label{lem:TLoneReduced}
Any reduced word $w$ in the $\algtl_n$ generators may be written in the form
\begin{subequations}\label{eq:TLoneJonesform} 
\begin{equation}
w = w' (e_n e_{n-1} \cdots e_{\ell_1}) (e_n e_{n-1} \cdots e_{\ell_2}) \cdots (e_n e_{n-1} \cdots e_{\ell_s}),
\end{equation}
where w' is a reduced word for the $\tl_n$ subalgebra generated by the $e_j$ with $j<n$,
\begin{equation}\label{eq:TLJonesform}
w' = (e_{j_1} e_{j_1-1} \cdots e_{k_1}) (e_{j_2} e_{j_2-1} \cdots e_{k_2}) \cdots (e_{j_r} e_{j_r-1} \cdots e_{k_r}),
\end{equation}
and the indices satisfy
\begin{equation}
 j_1 < j_2 < \cdots < j_r < n,\qquad k_i \le j_i\quad \text{(for all $i$),}\qquad k_1 < k_2 < \ldots < k_r < \ell_1 < \ell_2 < \cdots < \ell_s \le n.
\end{equation}
\end{subequations}
\end{Lemma}
\begin{proof}
If $e_n$ does not appear in the reduced word $w$, then $w$ is a reduced word for the \TL{} subalgebra and so $w$ may be written \cite{JonInd83} in the form $w'$ given above.  We 
will therefore assume that $e_n$ does appear in $w$ so that we may write $w = w' e_n \cdots$, where $w'$ is reduced for the $\tl_n$ subalgebra.

If there is no generator $e_j$ to the right of this $e_n$, then we are done.  If there is such an $e_j$, we cannot have $j=n$ because $w$ is reduced.  On the other hand, if $j<n-1$, then we may commute $e_j$ to the left and absorb it into $w'$.  Thus, we may assume that $j=n-1$.  Repeating, we find that $w$ has the form $w' (e_n e_{n-1} \cdots e_{\ell_1}) e_m \cdots$, for some $\ell_1, m$ satisfying $m \neq \ell_1 - 1$.  If $w'$ had the form $\cdots e_{k_r}$, as in \eqref{eq:TLJonesform}, with $k_r \ge \ell_1$, then we could commute the $e_{k_r}$ to the right and use $e_{k_r} e_{k_r+1} e_{k_r} = e_{k_r}$, contradicting $w$ being reduced.  Thus, we have $k_r < \ell_1$.

If $m < \ell_1 - 1$, then we could commute the $e_m$ to the left and absorb it into $w'$.  If $m = \ell_1$, then $w$ is not reduced.  Finally, if $\ell_1 < m < n$, then we may commute the $e_m$ to the left and use $e_m e_{m-1} e_m = e_m$, again contradicting $w$ being reduced.  The only remaining option is $m=n$ because $e_n e_{n-1} e_n$ cannot be simplified in $\algtl_n$.  Thus, $e_m$ forms the leftmost end of another descending chain of the form $e_n e_{n-1} \cdots e_{\ell_2}$ and the same arguments as above prove that $\ell_2 > \ell_1$.  Repeating, we arrive at the form \eqref{eq:TLoneJonesform}, completing the proof.
\end{proof}

It should now be clear that $\algtl_n$ is finite-dimensional. 
We therefore proceed to count the number of connectivities in $\diatl_n$ and the number of reduced words in $\algtl_n$.  The former number is $\dim \diatl_n$ and the latter gives an upper bound, $b$ say, for $\dim \algtl_n$.  \cref{prop:Surjective} shows that $\dim \diatl_n \le \dim \algtl_n \le b$, so if the counting gives $\dim \diatl_n = b$, then we will have $\dim \diatl_n = \dim \algtl_n$, hence $\diatl_n \simeq \algtl_n$.

The counting of connectivities of $\diatl_n$ follows from a bijection to the link states of $\diatl_{2n}$
with zero defects. To construct this bijection,
we shall cut the connectivities in half horizontally and regard them as link states for $\diatl_{2n}$. These have two types of defects: 
\emph{vertical} defects and \emph{boundary} defects. 
Their construction is almost identical to the construction of link states for $\tl_n$.  A $\diatl_n$ link state is a diagram made of loop segments that live above a horizontal line with $n$ marked nodes. These are either connected pairwise, occupied by a vertical defect or connected to the right boundary by
a boundary defect. We denote the set of link states on $n$ nodes, with $d$ vertical defects and $b$ boundary defects,
by $\linksb_n^{d,b}$. The full set of link states with $d$ vertical defects is denoted by $\linksb_n^d = \bigcup_b\, \linksb_n^{d,b}$. The numbers $n, d$ and $b$ must satisfy the constraint $d+b = n \bmod{2}$, but $d$ by itself can be odd or even. For example, here are all the sets of link states with
$n = 4$:
\begin{equation}
\begin{gathered}
\linksb_4^4 = \Big\{
\psset{unit=0.575}
\begin{pspicture}[shift=-0.10](-.1,0)(1.7,0.5)
\psline{-}(0,0)(1.6,0)
\psline[linecolor=blue,linewidth=1.25pt]{-}(0.2,0)(0.2,0.6)
\psline[linecolor=blue,linewidth=1.25pt]{-}(0.6,0)(0.6,0.6)
\psline[linecolor=blue,linewidth=1.25pt]{-}(1.0,0)(1.0,0.6)
\psline[linecolor=blue,linewidth=1.25pt]{-}(1.4,0)(1.4,0.6)
\end{pspicture} \Big\}, \quad
\linksb_4^3 = \Big\{
\begin{pspicture}[shift=-0.10](-.1,0)(1.7,0.5)
\psline{-}(0,0)(1.6,0)
\psline[linecolor=blue,linewidth=1.25pt]{-}(0.2,0)(0.2,0.6)
\psline[linecolor=blue,linewidth=1.25pt]{-}(0.6,0)(0.6,0.6)
\psline[linecolor=blue,linewidth=1.25pt]{-}(1.0,0)(1.0,0.6)
\psarc[linecolor=blue,linewidth=1.25pt]{-}(1.6,0){0.2}{90}{180}
\end{pspicture} \Big\},\quad
\linksb_4^2 = \Big\{
\begin{pspicture}[shift=-0.10](-.1,0)(1.7,0.5)
\psline{-}(0,0)(1.6,0)
\psarc[linecolor=blue,linewidth=1.25pt]{-}(0.4,0){0.2}{0}{180}
\psline[linecolor=blue,linewidth=1.25pt]{-}(1.0,0)(1.0,0.6)
\psline[linecolor=blue,linewidth=1.25pt]{-}(1.4,0)(1.4,0.6)
\end{pspicture},
\begin{pspicture}[shift=-0.10](-.1,0)(1.7,0.5)
\psline{-}(0,0)(1.6,0)
\psarc[linecolor=blue,linewidth=1.25pt]{-}(0.8,0){0.2}{0}{180}
\psline[linecolor=blue,linewidth=1.25pt]{-}(0.2,0)(0.2,0.6)
\psline[linecolor=blue,linewidth=1.25pt]{-}(1.4,0)(1.4,0.6)
\end{pspicture}, 
\begin{pspicture}[shift=-0.10](-.1,0)(1.7,0.5)
\psline{-}(0,0)(1.6,0)
\psarc[linecolor=blue,linewidth=1.25pt]{-}(1.2,0){0.2}{0}{180}
\psline[linecolor=blue,linewidth=1.25pt]{-}(0.2,0)(0.2,0.6)
\psline[linecolor=blue,linewidth=1.25pt]{-}(0.6,0)(0.6,0.6)
\end{pspicture},
\begin{pspicture}[shift=-0.10](-.1,0)(1.7,0.5)
\psline{-}(0,0)(1.6,0)
\psline[linecolor=blue,linewidth=1.25pt]{-}(0.2,0)(0.2,0.6)
\psline[linecolor=blue,linewidth=1.25pt]{-}(0.6,0)(0.6,0.6)
\psbezier[linecolor=blue,linewidth=1.25pt]{-}(1.0,0)(1.0,0.3)(1.1,0.5)(1.6,0.5)
\psarc[linecolor=blue,linewidth=1.25pt]{-}(1.6,0){0.2}{90}{180}
\end{pspicture}\Big\},\\[0.3cm]
\linksb_4^1 = \Big\{
\psset{unit=0.575}
\begin{pspicture}[shift=-0.10](-.1,0)(1.7,0.5)
\psline{-}(0,0)(1.6,0)
\psline[linecolor=blue,linewidth=1.25pt]{-}(0.2,0)(0.2,0.6)
\psarc[linecolor=blue,linewidth=1.25pt]{-}(0.8,0){0.2}{0}{180}
\psarc[linecolor=blue,linewidth=1.25pt]{-}(1.6,0){0.2}{90}{180}
\end{pspicture},
\begin{pspicture}[shift=-0.10](-.1,0)(1.7,0.5)
\psline{-}(0,0)(1.6,0)
\psline[linecolor=blue,linewidth=1.25pt]{-}(0.2,0)(0.2,0.6)
\psarc[linecolor=blue,linewidth=1.25pt]{-}(1.2,0){0.2}{0}{180}
\psbezier[linecolor=blue,linewidth=1.25pt]{-}(0.6,0)(0.6,0.3)(0.7,0.5)(1.6,0.5)
\end{pspicture},
\begin{pspicture}[shift=-0.10](-.1,0)(1.7,0.5)
\psline{-}(0,0)(1.6,0)
\psline[linecolor=blue,linewidth=1.25pt]{-}(1.0,0)(1.0,0.6)
\psarc[linecolor=blue,linewidth=1.25pt]{-}(0.4,0){0.2}{0}{180}
\psarc[linecolor=blue,linewidth=1.25pt]{-}(1.6,0){0.2}{90}{180}
\end{pspicture}, 
\begin{pspicture}[shift=-0.10](-.1,0)(1.7,0.5)
\psline{-}(0,0)(1.6,0)
\psline[linecolor=blue,linewidth=1.25pt]{-}(0.2,0)(0.2,0.6)
\psbezier[linecolor=blue,linewidth=1.25pt]{-}(0.6,0)(0.6,0.5)(0.7,0.8)(1.6,0.8)
\psbezier[linecolor=blue,linewidth=1.25pt]{-}(1.0,0)(1.0,0.3)(1.1,0.5)(1.6,0.5)
\psarc[linecolor=blue,linewidth=1.25pt]{-}(1.6,0){0.2}{90}{180}
\end{pspicture}
 \Big\}, \quad
\linksb_4^0 = \Big\{
\begin{pspicture}[shift=-0.10](-.1,0)(1.7,0.5)
\psline{-}(0,0)(1.6,0)
\psarc[linecolor=blue,linewidth=1.25pt]{-}(0.4,0){0.2}{0}{180}
\psarc[linecolor=blue,linewidth=1.25pt]{-}(1.2,0){0.2}{0}{180}
\end{pspicture}, 
\begin{pspicture}[shift=-0.10](-.1,0)(1.7,0.5)
\psline{-}(0,0)(1.6,0)
\psarc[linecolor=blue,linewidth=1.25pt]{-}(0.8,0){0.2}{0}{180}
\psbezier[linecolor=blue,linewidth=1.25pt]{-}(0.2,0)(0.2,0.6)(1.4,0.6)(1.4,0)
\end{pspicture}, 
\begin{pspicture}[shift=-0.10](-.1,0)(1.7,0.5)
\psline{-}(0,0)(1.6,0)
\psarc[linecolor=blue,linewidth=1.25pt]{-}(0.4,0){0.2}{0}{180}
\psbezier[linecolor=blue,linewidth=1.25pt]{-}(1.0,0)(1.0,0.3)(1.1,0.5)(1.6,0.5)
\psarc[linecolor=blue,linewidth=1.25pt]{-}(1.6,0){0.2}{90}{180}
\end{pspicture},
\begin{pspicture}[shift=-0.10](-.1,0)(1.7,0.5)
\psline{-}(0,0)(1.6,0)
\psarc[linecolor=blue,linewidth=1.25pt]{-}(0.8,0){0.2}{0}{180}
\psbezier[linecolor=blue,linewidth=1.25pt]{-}(0.2,0)(0.2,0.3)(0.3,0.5)(1.6,0.5)
\psarc[linecolor=blue,linewidth=1.25pt]{-}(1.6,0){0.2}{90}{180}
\end{pspicture}, 
\begin{pspicture}[shift=-0.10](-.1,0)(1.7,0.5)
\psline{-}(0,0)(1.6,0)
\psarc[linecolor=blue,linewidth=1.25pt]{-}(1.2,0){0.2}{0}{180}
\psbezier[linecolor=blue,linewidth=1.25pt]{-}(0.6,0)(0.6,0.3)(0.7,0.5)(1.6,0.5)
\psbezier[linecolor=blue,linewidth=1.25pt]{-}(0.2,0)(0.2,0.5)(0.3,0.8)(1.6,0.8)
\end{pspicture}, 
\begin{pspicture}[shift=-0.10](-.1,0)(1.7,0.5)
\psline{-}(0,0)(1.6,0)
\psarc[linecolor=blue,linewidth=1.25pt]{-}(1.6,0){0.2}{90}{180}
\psbezier[linecolor=blue,linewidth=1.25pt]{-}(1.0,0)(1.0,0.3)(1.1,0.5)(1.6,0.5)
\psbezier[linecolor=blue,linewidth=1.25pt]{-}(0.6,0)(0.6,0.5)(0.7,0.8)(1.6,0.8)
\psbezier[linecolor=blue,linewidth=1.25pt]{-}(0.2,0)(0.2,0.9)(0.7,1.1)(1.6,1.1)
\end{pspicture}
\Big\}.
\label{eq:Bs}
\end{gathered}
\end{equation}
The cardinalities of $\linksb_n^{d,b}$ and $\linksb_n^{d}$ are easy to compute. A link state in $\linksb_n^{d,b}$ has exactly $d+b$ defects --- the first $d$ are bulk defects and the last $b$ are boundary ones --- and can be mapped one-to-one to a \TL{} link state in $\links_n^{d+b}$ by forgetting this distinction between boundary and bulk defects. It then follows that
\begin{equation}
\abs{\linksb_n^{d,b}} = \abs{\links_n^{d+b}} = \binom{n}{\frac{n-d-b}{2}} - \binom{n}{\frac{n-d-b-2}{2}}, \qquad
\abs{\linksb_n^d} = 
\sum_{\substack{0 \le b \le n-d \\ b=n-d \bmod{2}}} 
\abs{\linksb_n^{d,b}}= \binom{n}{\big\lfloor \frac{n-d}{2} \big\rfloor}.
\label{eq:cards}
\end{equation}

$\diatl_n$ connectivities can be mapped bijectively onto the $\diatl_{2n}$ link states with 
no vertical defects.  This is achieved by cutting the connectivity horizontally and rotating the top edge of the box so that it lies to the left of the bottom edge, while maintaining the arc connections. For example, 
\begin{equation}\label{eq:cutandrotated}
\psset{unit=0.7} 
\begin{pspicture}[shift=-0.75](-0.4,-0.0)(2.0,1.7)
\pspolygon[fillstyle=solid,fillcolor=lightlightblue](-0.4,0)(-0.4,1.7)(2.00,1.7)(2.00,0)
\psarc[linewidth=1.0pt,linecolor=blue]{-}(0.4,0){0.2}{0}{180}
\psbezier[linewidth=1.0pt,linecolor=blue]{-}(-0.2,0)(-0.2,0.6)(1.0,0.6)(1.0,0)
\psbezier[linewidth=1.0pt,linecolor=blue]{-}(-0.2,1.7)(-0.2,0.85)(1.4,0.85)(1.4,0)
\psline[linewidth=1.0pt,linecolor=blue]{-}(1.8,0)(1.8,0.2)
\psline[linewidth=1.0pt,linecolor=blue]{-}(1.4,1.5)(1.4,1.7)
\psline[linewidth=1.0pt,linecolor=blue]{-}(1.8,1.5)(1.8,1.7)
\psarc[linewidth=1.0pt,linecolor=blue]{-}(2.0,0.2){0.2}{0}{180}
\psarc[linewidth=1.0pt,linecolor=blue]{-}(2.0,1.5){0.2}{180}{0}
\psarc[linewidth=1.0pt,linecolor=blue]{-}(0.8,1.7){0.2}{180}{0}
\psbezier[linewidth=1.0pt,linecolor=blue]{-}(1.4,1.5)(1.4,0.9)(2.6,0.9)(2.6,1.5)
\psbezier[linewidth=1.0pt,linecolor=blue]{-}(0.2,1.7)(0.2,0.95)(2.6,0.95)(2.6,0.2)
\pspolygon[fillstyle=solid,fillcolor=white,linecolor=white](2.035,0)(2.035,1.7)(2.8,1.7)(2.8,0)
\end{pspicture}
\quad \rightarrow \quad
\psset{unit=0.575cm}
\begin{pspicture}[shift=-0.65](-.1,0)(4.9,1.5)
\psline{-}(0,0)(4.8,0)
\psline{-}(2.4,-0.1)(2.4,0.1)
\psarc[linecolor=blue,linewidth=1.0pt]{-}(1.2,0){0.2}{0}{180}
\psarc[linecolor=blue,linewidth=1.0pt]{-}(3.2,0){0.2}{0}{180}
\psbezier[linecolor=blue,linewidth=1.0pt]{-}(2.6,0)(2.6,0.6)(3.8,0.6)(3.8,0)
\psbezier[linecolor=blue,linewidth=1.0pt]{-}(2.2,0)(2.2,1.0)(4.2,1.0)(4.2,0)
\psbezier[linecolor=blue,linewidth=1.0pt]{-}(1.8,0)(1.8,0.95)(2.3,1.1)(4.8,1.1)
\psbezier[linecolor=blue,linewidth=1.0pt]{-}(0.6,0)(0.6,1.05)(1.3,1.4)(4.8,1.4)
\psbezier[linecolor=blue,linewidth=1.0pt]{-}(0.2,0)(0.2,1.15)(1.0,1.7)(4.8,1.7)
\psarc[linecolor=blue,linewidth=1.0pt]{-}(4.8,0){0.2}{90}{180}
\end{pspicture}
\ .
\end{equation}
This rotation is obviously invertible, so it follows that the dimension of $\diatl_n$ is simply given by
\begin{equation}\label{eq:diatlcounting}
\dim \diatl_n = \abs{\linksb_{2n}^0} = \binom{2n}{n}.
\end{equation}

We now proceed to count the distinct reduced words in $\algtl_n$. For the original Temperley-Lieb algebra, 
Jones noted \cite{JonInd83} that this followed from a bijection between 
reduced words and north-east subdiagonal walks in $\mathbb Z^2$ that start at $(0,0)$ and end at $(n,n)$, crediting 
Wilf with this construction. The empty word $I$
is mapped to the walk $(0,0) \rightarrow (n,0) \rightarrow (n,n)$, while the other words, written in their reduced 
forms \eqref{eq:TLJonesform}, are mapped to
\begin{equation}
(0,0) \rightarrow (j_1,0) \rightarrow (j_1,k_1) \rightarrow (j_2,k_1) \rightarrow (j_2,k_2) \rightarrow \dots \rightarrow (j_r, k_{r-1}) \rightarrow (j_r,k_r) \rightarrow (n,k_r) \rightarrow (n,n).
\label{eq:TLwalk}
\end{equation}
Note that the individual steps of these walks always go north or east and never venture above the diagonal, 
as required. The pairs in \eqref{eq:TLwalk} are called the 
\emph{stops} of the walk and will be depicted in $\mathbb Z^2$ by black dots.
For example, 
\begin{equation}
(e_3 e_2 e_1)(e_6 e_5 e_4 e_3)(e_7 e_6)(e_8) \in \tl_{10} \quad \rightarrow \quad
\psset{unit=0.25}
\begin{pspicture}[shift=-5.1](-0.5,-0.5)(10.5,10.5)
\multiput(0,0)(0,1){11}{\psline[linestyle=dashed,dash=2pt 1pt,linecolor=gray, linewidth=0.5pt]{-}(-0.5,0)(10.5,0)}
\multiput(0,0)(1,0){11}{\psline[linestyle=dashed,dash=2pt 1pt,linecolor=gray, linewidth=0.5pt]{-}(0,-0.5)(0,10.5)}
\pscircle[fillstyle=solid,fillcolor=black](0,0){0.20}
\pscircle[fillstyle=solid,fillcolor=black](3,0){0.20}
\pscircle[fillstyle=solid,fillcolor=black](3,1){0.20}
\pscircle[fillstyle=solid,fillcolor=black](6,1){0.20}
\pscircle[fillstyle=solid,fillcolor=black](6,3){0.20}
\pscircle[fillstyle=solid,fillcolor=black](7,3){0.20}
\pscircle[fillstyle=solid,fillcolor=black](7,6){0.20}
\pscircle[fillstyle=solid,fillcolor=black](8,6){0.20}
\pscircle[fillstyle=solid,fillcolor=black](8,8){0.20}
\pscircle[fillstyle=solid,fillcolor=black](10,8){0.20}
\pscircle[fillstyle=solid,fillcolor=black](10,10){0.20}
\psline{->}(0,0)(3,0)(3,1)(6,1)(6,3)(7,3)(7,6)(8,6)(8,8)(10,8)(10,9.9)
\psline[linestyle=dashed,dash=1.75pt 1.5pt]{-}(0,0)(10,10)
\rput(-1.1,-0.75){\scriptsize$_{(0,0)}$}
\rput(11.7,10.75){\scriptsize$_{(10,10)}$}
\end{pspicture}
\ .
\end{equation}

For a given walk, it is straightforward to reproduce the reduced word, so the map is indeed bijective.
The actual counting proceeds by defining intermediate walks from $(0,0)$ to $(m,p)$ that are otherwise subject to the same constraints. Their cardinality, 
$a_{m,p}$, satisfies an easily solved recursion relation \cite{RidSta12} that gives
\begin{equation}
a_{m,p} = \binom{m+p}{p} - \binom{m+p}{p-1}.
\end{equation}
Clearly, $a_{n,n}$ yields the number of reduced $\tl_n$ words.

The reduced words \eqref{eq:TLJonesform} of $\algtl_n$ are also in bijection with a family of subdiagonal north-east 
walks, starting at $(0,0)$ but now ending at $(n+1,n+1)$. The empty word is now
mapped to the walk $(0,0) \rightarrow (n+1,0) \rightarrow (n+1,n+1)$ and the other words in $\tl_n \subset \algtl_n$, in their reduced form \eqref{eq:TLJonesform}, are mapped to
\begin{equation}
(0,0) \rightarrow (j_1,0) \rightarrow (j_1,k_1) \rightarrow (j_2,k_1) \rightarrow \dots \rightarrow (j_r, k_{r-1}) \rightarrow (j_r,k_r) \rightarrow (n\!+\!1,k_r) \rightarrow (n+1,n+1).
\end{equation}
The remaining words --- those in the reduced form \eqref{eq:TLoneJonesform} with $e_n$ appearing --- are mapped to
\begin{alignat}{2}
(0,0) \rightarrow (j_1,0) \rightarrow  (j_1,k_1) &\rightarrow (j_2,k_1) \rightarrow \dots \rightarrow (j_r, k_{r-1}) \rightarrow (j_r,k_r) \rightarrow (n,k_r) \nonumber\\
 &\rightarrow (n,\ell_1) \rightarrow (n,\ell_2) \rightarrow \dots \rightarrow (n,\ell_s) \rightarrow (n+1,\ell_s) \rightarrow (n+1,n+1).
\end{alignat}
Again the inverse map is easily constructed. A key difference with the walks \eqref{eq:TLwalk} is that there can now be more than two stops
in the $n$-th column, corresponding to the fact that $e_n$ may appear more than once in a reduced word. Moreover, these additional stops in the $n$-th column do not lead to a change of direction. We also represent these new stops by black dots, for instance:
\begin{equation}
(e_2 e_1)(e_5 e_4 e_3 e_2)(e_6 e_5 e_4)(e_8 e_7 e_6 e_5)(e_8 e_7)(e_8) \in \algtl_{8} \quad \rightarrow \quad
\psset{unit=0.25}
\begin{pspicture}[shift=-4.6](-0.5,-0.5)(9.5,9.5)
\multiput(0,0)(0,1){10}{\psline[linestyle=dashed,dash=2pt 1pt,linecolor=gray, linewidth=0.5pt]{-}(-0.5,0)(9.5,0)}
\multiput(0,0)(1,0){10}{\psline[linestyle=dashed,dash=2pt 1pt,linecolor=gray, linewidth=0.5pt]{-}(0,-0.5)(0,9.5)}
\pscircle[fillstyle=solid,fillcolor=black](0,0){0.20}
\pscircle[fillstyle=solid,fillcolor=black](2,0){0.20}
\pscircle[fillstyle=solid,fillcolor=black](2,1){0.20}
\pscircle[fillstyle=solid,fillcolor=black](5,1){0.20}
\pscircle[fillstyle=solid,fillcolor=black](5,2){0.20}
\pscircle[fillstyle=solid,fillcolor=black](6,2){0.20}
\pscircle[fillstyle=solid,fillcolor=black](6,4){0.20}
\pscircle[fillstyle=solid,fillcolor=black](8,4){0.20}
\pscircle[fillstyle=solid,fillcolor=black](8,5){0.20}
\pscircle[fillstyle=solid,fillcolor=black](8,7){0.20}
\pscircle[fillstyle=solid,fillcolor=black](8,8){0.20}
\pscircle[fillstyle=solid,fillcolor=black](9,8){0.20}
\pscircle[fillstyle=solid,fillcolor=black](9,9){0.20}
\psline{->}(0,0)(2,0)(2,1)(5,1)(5,2)(6,2)(6,4)(8,4)(8,5)(8,7)(8,8)(9,8)(9,8.9)
\psline[linestyle=dashed,dash=1.75pt 1.5pt]{-}(0,0)(9,9)
\rput(-1.1,-0.75){\scriptsize$_{(0,0)}$}
\rput(10.7,9.75){\scriptsize$_{(9,9)}$}
\end{pspicture}
\ .
\end{equation}
To count these walks, we again consider intermediate walks from $(0,0)$ to $(m,p)$ where extra stops are allowed, but only in the $n$-th column. For $m<n$, these intermediate walks are identical to those considered for the algebra $\tl_n$. Their number is therefore given by $a_{m,p}$. 

For $m = n$, we refine the descriptions of these walks by considering those ending at $(n,p)$ that contain exactly $s$ stops on the $n$-th column. We denote their counting by $b_{p}^{(s)}$ (with $s \in 1, \dots, p+1$). For example, if
$n=6$, the intermediate walks
\begin{equation}
\psset{unit=0.25}
\begin{pspicture}[shift=-4.6](-0.5,-0.75)(7.5,7.5)
\multiput(0,0)(0,1){8}{\psline[linestyle=dashed,dash=2pt 1pt,linecolor=gray, linewidth=0.5pt]{-}(-0.5,0)(7.5,0)}
\multiput(0,0)(1,0){8}{\psline[linestyle=dashed,dash=2pt 1pt,linecolor=gray, linewidth=0.5pt]{-}(0,-0.5)(0,7.5)}
\pscircle[fillstyle=solid,fillcolor=black](0,0){0.20}
\pscircle[fillstyle=solid,fillcolor=black](2,0){0.20}
\pscircle[fillstyle=solid,fillcolor=black](2,1){0.20}
\pscircle[fillstyle=solid,fillcolor=black](3,1){0.20}
\pscircle[fillstyle=solid,fillcolor=black](3,3){0.20}
\pscircle[fillstyle=solid,fillcolor=black](4,3){0.20}
\pscircle[fillstyle=solid,fillcolor=black](4,4){0.20}
\pscircle[fillstyle=solid,fillcolor=black](6,4){0.20}
\psline{->}(0,0)(2,0)(2,1)(3,1)(3,3)(4,3)(4,4)(5.9,4)
\psline[linestyle=dashed,dash=1.75pt 1.5pt]{-}(0,0)(7,7)
\rput(-1.1,-0.75){\scriptsize$_{(0,0)}$}
\rput(8.4,7.75){\scriptsize$_{(7,7)}$}
\end{pspicture}
\ , \qquad \qquad
\begin{pspicture}[shift=-4.6](-0.5,-0.75)(7.5,7.5)
\multiput(0,0)(0,1){8}{\psline[linestyle=dashed,dash=2pt 1pt,linecolor=gray, linewidth=0.5pt]{-}(-0.5,0)(7.5,0)}
\multiput(0,0)(1,0){8}{\psline[linestyle=dashed,dash=2pt 1pt,linecolor=gray, linewidth=0.5pt]{-}(0,-0.5)(0,7.5)}
\pscircle[fillstyle=solid,fillcolor=black](0,0){0.20}
\pscircle[fillstyle=solid,fillcolor=black](1,0){0.20}
\pscircle[fillstyle=solid,fillcolor=black](1,1){0.20}
\pscircle[fillstyle=solid,fillcolor=black](4,1){0.20}
\pscircle[fillstyle=solid,fillcolor=black](4,3){0.20}
\pscircle[fillstyle=solid,fillcolor=black](6,3){0.20}
\pscircle[fillstyle=solid,fillcolor=black](6,4.9){0.20}
\psline{->}(0,0)(1,0)(1,1)(4,1)(4,3)(6,3)(6,4.9)
\psline[linestyle=dashed,dash=1.75pt 1.5pt]{-}(0,0)(7,7)
\rput(-1.1,-0.75){\scriptsize$_{(0,0)}$}
\rput(8.4,7.75){\scriptsize$_{(7,7)}$}
\end{pspicture}
\ , \qquad \qquad
\begin{pspicture}[shift=-4.6](-0.5,-0.75)(7.5,7.5)
\multiput(0,0)(0,1){8}{\psline[linestyle=dashed,dash=2pt 1pt,linecolor=gray, linewidth=0.5pt]{-}(-0.5,0)(7.5,0)}
\multiput(0,0)(1,0){8}{\psline[linestyle=dashed,dash=2pt 1pt,linecolor=gray, linewidth=0.5pt]{-}(0,-0.5)(0,7.5)}
\pscircle[fillstyle=solid,fillcolor=black](0,0){0.20}
\pscircle[fillstyle=solid,fillcolor=black](3,0){0.20}
\pscircle[fillstyle=solid,fillcolor=black](3,1){0.20}
\pscircle[fillstyle=solid,fillcolor=black](6,1){0.20}
\pscircle[fillstyle=solid,fillcolor=black](6,1){0.20}
\pscircle[fillstyle=solid,fillcolor=black](6,2){0.20}
\pscircle[fillstyle=solid,fillcolor=black](6,4){0.20}
\pscircle[fillstyle=solid,fillcolor=black](6,6){0.20}
\psline{->}(0,0)(3,0)(3,1)(6,1)(6,5.9)
\psline[linestyle=dashed,dash=1.75pt 1.5pt]{-}(0,0)(7,7)
\rput(-1.1,-0.75){\scriptsize$_{(0,0)}$}
\rput(8.4,7.75){\scriptsize$_{(7,7)}$}
\end{pspicture}
\end{equation}
respectively have $s = 1, 2$ and $4$. The countings are easily seen to satisfy the relations
\begin{equation}
b_{0}^{(s)} = \delta_{s,1}, \qquad b_p^{(1)} = a_{n-1,p}, \qquad b_p^{(s)} = b_{p-1}^{(s-1)} + b_{p-1}^{(s)} \quad (2 \le s \le p+1), 
\end{equation}
which determine them completely:
\begin{equation}
b_p^{(s)} = \binom{n+p-1}{p-s+1} - \binom{n+p-1}{p-s}.
\end{equation}

Finally, let us denote by $c_p$ the number of intermediate walks ending at $(n+1,p)$. These satisfy the relations
\begin{equation}
c_0 = 1, \qquad c_p = c_{p-1} + \sum_{s = 1}^{p+1} b_p^{(s)} \quad (1 \le p \le n), \qquad c_{n+1} = c_n,
\end{equation}
from which we find
\begin{equation}
c_p = \begin{cases}
\displaystyle\binom{n+p}{p} & 
0 \le p \le n,\\[0.3cm]
\displaystyle\binom{2n}{n} & 
p = n+1.
\end{cases}
\end{equation}
The number of distinct reduced $\algtl_n$ words is then given by $c_{n+1}$, which equals the number \eqref{eq:diatlcounting} of connectivities of $\diatl_n$. We have thus proved the following proposition:
\begin{Proposition}
The algebras $\algtl_n$ and $\diatl_n$ are isomorphic. 
\end{Proposition}

We conclude this section by further refining the above counting exercises, by considering the walks ending at $(n+1,p)$ for which each intermediate walk has at most $t$ stops in the $n$-th column. (Here, the range of $t$ is $\{1, \dots, p+1 \}$ for $0\le p \le n$ and $\{1, \dots, n+1 \}$ for $p = n+1$.) The numbers $c_p^{(t)}$ (with $c_p = c_p^{(p+1)}$, for $0\le p \le n$, and $c_{n+1} = c_{n+1}^{(n+1)}$) of these walks satisfy the relations
\begin{equation}
c_0^{(t)} = 1, \qquad c_p^{(t)} = c_{p-1}^{(t)} + \sum_{s = 1}^{t} b_p^{(s)} \quad (1 \le p \le n), \qquad c_{n+1}^{(t)} = c_n^{(t)},
\end{equation}
which leads to
\begin{equation} \label{eq:FinestCounting}
c_p^{(t)} = \begin{cases} 
\displaystyle\binom{n+p}{p} - \binom{n+p}{p-t} & 
0 \le p \le n,\\[3mm]
\displaystyle\binom{2n}{n} - \binom{2n}{n-t} & 
p = n+1.
\end{cases}
\end{equation}
This result will be used to count the dimensions of the seam algebras $\btl_{n,k}$ in \cref{app:SeamsBetaFormal}.

%
\section{Boundary seam algebra proofs} \label{app:Seams}
%

In this appendix, we present the proofs of certain results pertaining to the boundary seam algebras $\btl_{n,k}$ that are 
used in the body of the paper. These are divided into proofs for which $\beta$ is a formal parameter and proofs for 
which $\beta$ is specialised to a complex number.

Here, we will distinguish between the formal
boundary seam algebras according to whether they are defined diagrammatically or algebraically. The 
diagrammatic algebras need not admit a well-defined specialisation at all $\beta \in \CC$ whereas their algebraic 
counterparts may always be specialised, see \cref{app:SeamsBetaC}.
The (formal) diagrammatic algebra, $\diabtl_{n,k}$, is defined to be the linear span of all products of the diagrams of 
$\Ik$ and the $\Ekj{k}{j}$, $j = 1, \dots, n$, see \eqref{eq:Ik} and \eqref{eq:Ek}. Here, by product, we mean the 
diagrammatic one defined by restriction from $\tl_{n+k}$. Its algebraic counterpart, 
$\algbtl_{n,k}$, is defined to be the linear span of all formal words in the generators $\Ik$ and the 
$\Ekj{k}{j}$, $j = 1, \dots, n$, subject to the relations \eqref{eq:newBTL} and, if $n>k$, \eqref{eq:ClosureRelation}. 
One result of \cref{app:SeamsBetaFormal} is that over a complex function field, the two algebras are isomorphic.

\subsection{Proofs for $\beta$ formal} \label{app:SeamsBetaFormal}

It is useful to first give an alternative characterisation of these algebras in terms of the projections (idempotents) $\Ik$, more specifically, in terms of the subalgebras
\begin{equation} \label{eq:defatl}
\atl_{n,k} = \Ik \tl_{n+k} \Ik = \set{\Ik a \Ik \, ; \: a \in \tl_{n+k}} \subseteq \tl_{n+k}.
\end{equation}
The subalgebra
$\atl_{n,k}$ is a unital associative algebra in its own right, the unit being $\Ik = \Ik I \Ik$. It is spanned by diagrams made of connectivities in $\tl_{n+k}$ sandwiched between two projectors $P_k$ acting on the nodes $n\!+\!1, \dots, n\!+\!k$. If $a$ can be written as $e_j a'$ or $a' e_j$, for $j \in \{n\!+\!1, \dots, n\!+\!k\!-\!1\}$, and some $a' \in \tl_{n+k}$, then $\Ik a\, \Ik$ is zero. For $\Ik a\, \Ik$ to be non-zero, $a$ must therefore have no arc connecting two neighbouring nodes in the range $n\!+\!1, \dots, n\!+\!k$ of both its bottom and top edges.  We will refer to the nodes in the range $1, \ldots, n$ as \emph{bulk nodes} and those in the range $n+1, \ldots, n+k$ as \emph{boundary nodes}.

\begin{Proposition}
Over a complex function field, the algebras $\atl_{n,k}$ and $\diabtl_{n,k}$ are isomorphic. 
\label{sec:aandbiso}
\end{Proposition}
\begin{proof}
It is clear that the unit $\Ik$ and the generators $\Ekj{k}{j}$ of $\diabtl_{n,k}$ are elements of $\atl_{n,k}$, hence that $\diabtl_{n,k}$ is a subalgebra of $\atl_{n,k}$.  It follows that the tangles $\Yk_t \in \diabtl_{n,k}$, defined in \eqref{eq:Yt}, also belong to $\atl_{n,k}$.  We claim that these tangles, along with the $\Ekj{k}{j}$ with $j<n$, generate $\atl_{n,k}$.  Clearly, the proposition will be proved once this claim is established.

A basis of the subalgebra
$\atl_{n,k}$ is given by the $\Ik a \Ik$, where $a$ is a connectivity of $\tl_{n+k}$ in which there are no arcs connecting boundary nodes along the bottom or the top edges.  We let $t$ denote the number of arcs connecting a top bulk node to a boundary node (top or bottom) and $t'$ be the number of arcs connecting a bottom bulk node to a boundary node.  Clearly, $t+t'$ must be even and bounded above by $2k$.  Here is an example with $n=6$, $k=4$, $t=3$ and $t'=1$:
\begin{equation} \label{eq:ExOfAnkBasisElt}
I^{_{(4)}}_{\phantom{j}} a I^{_{(4)}}_{\phantom{j}} = \ 
\psset{unit=1} 
\begin{pspicture}[shift=-0.75](-0.4,-0.0)(3.65,1.7)
\pspolygon[fillstyle=solid,fillcolor=lightlightblue,linewidth=0pt,linecolor=white](-0.4,0)(-0.4,1.7)(3.65,1.7)(3.65,0)
\pspolygon(-0.4,0)(-0.4,1.7)(3.65,1.7)(3.65,0)
\psarc[linewidth=1.0pt,linecolor=blue]{-}(0.4,0){0.2}{0}{180}
\psbezier[linewidth=1.0pt,linecolor=blue]{-}(-0.2,0)(-0.2,0.6)(1.0,0.6)(1.0,0)
\psbezier[linewidth=1.0pt,linecolor=blue]{-}(-0.2,1.7)(-0.2,0.85)(1.4,0.85)(1.4,0)
\psline[linewidth=1.0pt,linecolor=blue]{-}(1.8,0)(1.8,0.2)
\psline[linewidth=1.0pt,linecolor=blue]{-}(2.2,0)(2.2,0.2)
\psline[linewidth=1.0pt,linecolor=blue]{-}(2.6,0)(2.6,0.2)
\psline[linewidth=1.0pt,linecolor=blue]{-}(1.4,1.5)(1.4,1.7)
\psline[linewidth=1.0pt,linecolor=blue]{-}(1.8,1.5)(1.8,1.7)
\psline[linewidth=1.0pt,linecolor=blue]{-}(2.2,1.5)(2.2,1.7)
\psline[linewidth=1.0pt,linecolor=blue]{-}(2.6,1.5)(2.6,1.7)
\psarc[linewidth=1.0pt,linecolor=blue]{-}(2.0,0.2){0.2}{0}{180}
\psarc[linewidth=1.0pt,linecolor=blue]{-}(2.0,1.5){0.2}{180}{0}
\psarc[linewidth=1.0pt,linecolor=blue]{-}(0.8,1.7){0.2}{180}{0}
\psbezier[linewidth=1.0pt,linecolor=blue]{-}(1.4,1.5)(1.4,0.9)(2.6,0.9)(2.6,1.5)
\psbezier[linewidth=1.0pt,linecolor=blue]{-}(0.2,1.7)(0.2,0.95)(2.6,0.95)(2.6,0.2)
\psline[linewidth=1.0pt,linecolor=blue]{-}(3.0,0)(3.0,1.7)
\psline[linewidth=1.0pt,linecolor=blue]{-}(3.4,0)(3.4,1.7)
\pspolygon[fillstyle=solid,fillcolor=pink](2.05,0.1)(2.05,0.2)(3.55,0.2)(3.55,0.1)
\pspolygon[fillstyle=solid,fillcolor=pink](2.05,1.5)(2.05,1.6)(3.55,1.6)(3.55,1.5)
\end{pspicture}
\ .
\end{equation}
If $t \ge t'$, then there will be $\frac{1}{2} (t+t')$ arcs connecting a top bulk node to a top boundary node, $t'$ arcs connecting a bottom bulk node to a bottom boundary node and $\frac{1}{2} (t-t')$ arcs connecting a top bulk node to a bottom boundary node.  The boundary nodes in each case will be labelled by $n+1, \ldots, n + \frac{1}{2} (t+t')$ (top-top), $n+1, \ldots, n + t'$ (bottom-bottom) and $n+t'+1, \ldots, n + \frac{1}{2} (t+t')$ (top-bottom).  When $t \le t'$, the situation is similar --- swap $t$ with $t'$ as well as ``top'' with ``bottom'' in the previous description.

The $\tl_{n+k}$-connectivities $a$ are then uniquely determined by the arcs connecting the bulk nodes to one another and the numbers $t$ and $t'$.  Cutting away the boundary nodes from the connectivity $a$ then results in a $\tlone_n$-connectivity $\Phi(a)$ in which there are $t$ arcs from the top edge to the right boundary and $t'$ arcs from the bottom to the right boundary.  In the example \eqref{eq:ExOfAnkBasisElt}, the cutting is illustrated thusly:
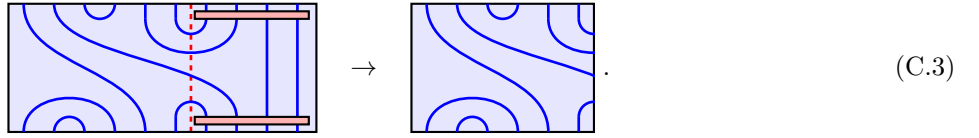
\begin{equation}
\psset{unit=1} 
\begin{pspicture}[shift=-0.75](-0.4,-0.0)(3.65,1.7)
\pspolygon[fillstyle=solid,fillcolor=lightlightblue,linewidth=0pt,linecolor=white](-0.4,0)(-0.4,1.7)(3.65,1.7)(3.65,0)
\psline[linewidth=1pt,linecolor=red,linestyle=dashed,dash=2pt 2pt](2.0,0)(2.0,1.7)
\pspolygon(-0.4,0)(-0.4,1.7)(3.65,1.7)(3.65,0)
\psarc[linewidth=1.0pt,linecolor=blue]{-}(0.4,0){0.2}{0}{180}
\psbezier[linewidth=1.0pt,linecolor=blue]{-}(-0.2,0)(-0.2,0.6)(1.0,0.6)(1.0,0)
\psbezier[linewidth=1.0pt,linecolor=blue]{-}(-0.2,1.7)(-0.2,0.85)(1.4,0.85)(1.4,0)
\psline[linewidth=1.0pt,linecolor=blue]{-}(1.8,0)(1.8,0.2)
\psline[linewidth=1.0pt,linecolor=blue]{-}(2.2,0)(2.2,0.2)
\psline[linewidth=1.0pt,linecolor=blue]{-}(2.6,0)(2.6,0.2)
\psline[linewidth=1.0pt,linecolor=blue]{-}(1.4,1.5)(1.4,1.7)
\psline[linewidth=1.0pt,linecolor=blue]{-}(1.8,1.5)(1.8,1.7)
\psline[linewidth=1.0pt,linecolor=blue]{-}(2.2,1.5)(2.2,1.7)
\psline[linewidth=1.0pt,linecolor=blue]{-}(2.6,1.5)(2.6,1.7)
\psarc[linewidth=1.0pt,linecolor=blue]{-}(2.0,0.2){0.2}{0}{180}
\psarc[linewidth=1.0pt,linecolor=blue]{-}(2.0,1.5){0.2}{180}{0}
\psarc[linewidth=1.0pt,linecolor=blue]{-}(0.8,1.7){0.2}{180}{0}
\psbezier[linewidth=1.0pt,linecolor=blue]{-}(1.4,1.5)(1.4,0.9)(2.6,0.9)(2.6,1.5)
\psbezier[linewidth=1.0pt,linecolor=blue]{-}(0.2,1.7)(0.2,0.95)(2.6,0.95)(2.6,0.2)
\psline[linewidth=1.0pt,linecolor=blue]{-}(3.0,0)(3.0,1.7)
\psline[linewidth=1.0pt,linecolor=blue]{-}(3.4,0)(3.4,1.7)
\pspolygon[fillstyle=solid,fillcolor=pink](2.05,0.1)(2.05,0.2)(3.55,0.2)(3.55,0.1)
\pspolygon[fillstyle=solid,fillcolor=pink](2.05,1.5)(2.05,1.6)(3.55,1.6)(3.55,1.5)
\end{pspicture}
\quad \rightarrow \quad
\begin{pspicture}[shift=-0.75](-0.4,-0.0)(2.0,1.7)
\pspolygon[fillstyle=solid,fillcolor=lightlightblue](-0.4,0)(-0.4,1.7)(2.00,1.7)(2.00,0)
\psarc[linewidth=1.0pt,linecolor=blue]{-}(0.4,0){0.2}{0}{180}
\psbezier[linewidth=1.0pt,linecolor=blue]{-}(-0.2,0)(-0.2,0.6)(1.0,0.6)(1.0,0)
\psbezier[linewidth=1.0pt,linecolor=blue]{-}(-0.2,1.7)(-0.2,0.85)(1.4,0.85)(1.4,0)
\psline[linewidth=1.0pt,linecolor=blue]{-}(1.8,0)(1.8,0.2)
\psline[linewidth=1.0pt,linecolor=blue]{-}(1.4,1.5)(1.4,1.7)
\psline[linewidth=1.0pt,linecolor=blue]{-}(1.8,1.5)(1.8,1.7)
\psarc[linewidth=1.0pt,linecolor=blue]{-}(2.0,0.2){0.2}{0}{180}
\psarc[linewidth=1.0pt,linecolor=blue]{-}(2.0,1.5){0.2}{180}{0}
\psarc[linewidth=1.0pt,linecolor=blue]{-}(0.8,1.7){0.2}{180}{0}
\psbezier[linewidth=1.0pt,linecolor=blue]{-}(1.4,1.5)(1.4,0.9)(2.6,0.9)(2.6,1.5)
\psbezier[linewidth=1.0pt,linecolor=blue]{-}(0.2,1.7)(0.2,0.95)(2.6,0.95)(2.6,0.2)
\pspolygon[fillstyle=solid,fillcolor=white,linecolor=white](2.028,0)(2.028,1.7)(2.8,1.7)(2.8,0)
\end{pspicture}
\ . \label{eq:ExOfAnkBasisElt2}
\end{equation}
This establishes a bijection $\Phi$ between the connectivities $a$ that parametrise a basis of $\atl_{n,k}$ and the connectivities of $\tlone_n$ that have at most $2k$ arcs connecting to the right boundary.

We now cut the $\tlone_n$-connectivity in half horizontally, converting any arcs that connected the top and bottom edges into defects.  The two halves may be regarded as link states for $\tl_n$, upon reflecting the top half to adopt the customary orientation, if we also regard the arcs that connected to the right boundary as defects.  The top link state therefore has at least $t$ defects and the bottom link state at least $t'$.  To illustrate,
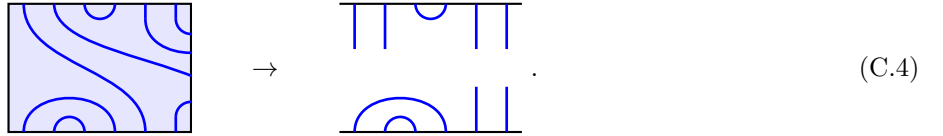
\begin{equation}
\begin{pspicture}[shift=-0.75](-0.4,-0.0)(2.0,1.7)
\pspolygon[fillstyle=solid,fillcolor=lightlightblue](-0.4,0)(-0.4,1.7)(2.00,1.7)(2.00,0)
\psarc[linewidth=1.0pt,linecolor=blue]{-}(0.4,0){0.2}{0}{180}
\psbezier[linewidth=1.0pt,linecolor=blue]{-}(-0.2,0)(-0.2,0.6)(1.0,0.6)(1.0,0)
\psbezier[linewidth=1.0pt,linecolor=blue]{-}(-0.2,1.7)(-0.2,0.85)(1.4,0.85)(1.4,0)
\psline[linewidth=1.0pt,linecolor=blue]{-}(1.8,0)(1.8,0.2)
\psline[linewidth=1.0pt,linecolor=blue]{-}(1.4,1.5)(1.4,1.7)
\psline[linewidth=1.0pt,linecolor=blue]{-}(1.8,1.5)(1.8,1.7)
\psarc[linewidth=1.0pt,linecolor=blue]{-}(2.0,0.2){0.2}{0}{180}
\psarc[linewidth=1.0pt,linecolor=blue]{-}(2.0,1.5){0.2}{180}{0}
\psarc[linewidth=1.0pt,linecolor=blue]{-}(0.8,1.7){0.2}{180}{0}
\psbezier[linewidth=1.0pt,linecolor=blue]{-}(1.4,1.5)(1.4,0.9)(2.6,0.9)(2.6,1.5)
\psbezier[linewidth=1.0pt,linecolor=blue]{-}(0.2,1.7)(0.2,0.95)(2.6,0.95)(2.6,0.2)
\pspolygon[fillstyle=solid,fillcolor=white,linecolor=white](2.028,0)(2.028,1.7)(2.8,1.7)(2.8,0)
\end{pspicture}
\qquad \rightarrow \qquad 
\begin{pspicture}[shift=-0.75](-0.4,-0.0)(2.0,1.7)
\psline{-}(-0.4,0)(2.00,0)
\psline{-}(2.00,1.7)(-0.4,1.7)
\psarc[linewidth=1.0pt,linecolor=blue]{-}(0.4,0){0.2}{0}{180}
\psbezier[linewidth=1.0pt,linecolor=blue]{-}(-0.2,0)(-0.2,0.6)(1.0,0.6)(1.0,0)
\psline[linewidth=1.0pt,linecolor=blue]{-}(1.4,0)(1.4,0.6)
\psline[linewidth=1.0pt,linecolor=blue]{-}(1.8,0)(1.8,0.6)
\psline[linewidth=1.0pt,linecolor=blue]{-}(-0.2,1.1)(-0.2,1.7)
\psline[linewidth=1.0pt,linecolor=blue]{-}(0.2,1.1)(0.2,1.7)
\psline[linewidth=1.0pt,linecolor=blue]{-}(1.4,1.1)(1.4,1.7)
\psline[linewidth=1.0pt,linecolor=blue]{-}(1.8,1.1)(1.8,1.7)
\psarc[linewidth=1.0pt,linecolor=blue]{-}(0.8,1.7){0.2}{180}{0}
\pspolygon[fillstyle=solid,fillcolor=white,linecolor=white](2.028,0)(2.028,1.7)(2.8,1.7)(2.8,0)
\end{pspicture} \ .
\end{equation}
Again, this cutting can be inverted to uniquely reconstitute the $\tlone_n$-connectivity, provided that we remember the numbers $t$ and $t'$.

As the standard $\tl_n$-modules are irreducible over a complex function field, acting with $\tl_n$ on any link state with $d$ defects lets us produce all the link states with $d$ defects.  One can therefore also obtain all $d$ defect link states from any $d$ defect link state using the natural action of $\tl_n$ on link states, where we do not set the result to zero if the number of defects has decreased. For instance, in the natural representation, $e_2$ acting on 
$\,
\psset{unit=0.45}
\begin{pspicture}(0.0,0.0)(1.2,0.7)
\psline{-}(0,0)(1.2,0)
\psline[linewidth=1.0pt,linecolor=blue]{-}(0.2,0)(0.2,0.6)
\psline[linewidth=1.0pt,linecolor=blue]{-}(0.6,0)(0.6,0.6)
\psline[linewidth=1.0pt,linecolor=blue]{-}(1.0,0)(1.0,0.6)
\end{pspicture}\,
$
 gives 
 $\,
\psset{unit=0.45}
\begin{pspicture}(0.0,0.0)(1.2,0.7)
\psline{-}(0,0)(1.2,0)
\psline[linewidth=1.0pt,linecolor=blue]{-}(0.2,0)(0.2,0.6)
\psarc[linewidth=1.0pt,linecolor=blue]{-}(0.8,0){0.2}{0}{180}
\end{pspicture}\,
$.
This is relevant because left- or right-multiplying a $\tlone_n$ connectivity by $\tl_n \subset \tlone_n$ amounts to the natural $\tl_n$-action, up to factors of $\beta_1$ and $\beta_2$, on the bottom or top link states, respectively, that result from cutting  the connectivity. More importantly, it also corresponds to left- or right-multiplication by the subalgebra $\langle \Ik e_j; j = 1, \dots, n-1 \rangle \simeq \tl_n$ of $\atl_{n,k}$ when the boundary nodes are reinserted.

This shows that a set of generators of $\atl_{n,k}$ is obtained by choosing one basis element $\Ik a_{t,t'} \Ik \in \atl_{n,k}$, for each $t$ and $t'$, where $a_{t,t'}$ has $t$ top bulk to boundary and $t'$ bottom bulk to boundary arcs.  When $t=t'$, we may choose $\Ik a_{t,t} \Ik = \Yk_t$.  When $t > t'$, we may obtain such a basis element from $\Yk_{(t+t')/2}$ by left-multiplying by the $\tl_n$-subalgebra to convert bulk to boundary arcs from the bottom into an arc that ties a bulk node from the top to a boundary node in the bottom. The case $t<t'$ is handled similarly.  The $\Ekj{k}{j}$, with $j \neq n$, and the $\Yk_t$ are therefore generators of $\atl_{n,k}$, completing the proof.
\end{proof}

The bijection $\Phi$ in this proof is also useful to determine the dimension of the boundary seam algebras.  Indeed, when $k \ge n$, $\Phi$ maps a basis of $\atl_{n,k}$ to a basis of $\tlone_n$, where we recall from \eqref{eq:betas} that $\beta_1$ and $\beta_2$ are identified with $U_{k-1}$ and $U_k$, respectively. Since $\atl_{n,k} \cong \diabtl_{n,k}$ is a quotient of $\tlone_n$, we obtain the following result.

\begin{cor} \label{cor:BTL=TLone}
If $k \ge n$, then $\diabtl_{n,k} \cong \tlone_n$ over a complex function field, so $\dim \diabtl_{n,k} = \displaystyle \binom{2n}{n}$.
\end{cor}

When $k<n$, this argument instead shows
that the dimension of $\atl_{n,k}$ is given by the number of $\tlone_n$-connectivities with at most $2k$ arcs connecting to the right boundary.  We shall again cut these connectivities in half horizontally, but rather than treat the results as pairs of $\tl_n$ link states, useful in the above proof because $\tl_n$ is naturally a subalgebra of $\atl_n$, we shall instead regard each result as a (single) link state for $\tlone_{2n}$.  Similar manipulations were already performed in \cref{app:TLone} to count the number of connectivities in $\diatl_n$.
\begin{Proposition}\label{sec:dimA}
Over a complex function field, the dimension of the algebra $\diabtl_{n,k} \simeq \atl_{n,k}$, for $k<n$, is
\begin{equation}
\dim \diabtl_{n,k} = \binom{2n}{n}- \binom{2n}{n-k-1}.
\label{eq:dimatl}
\end{equation}
\end{Proposition}
\begin{proof}
The proof uses the bijection $\Phi$ to map the basis tangles of $\atl_{n,k}$ onto the connectivities of $\tlone_{n}$ with at most $2k$ boundary arcs.  The latter connectivities are then mapped bijectively onto the $\tlone_{2n}$ link states with $2n$ nodes and no vertical defects.  This is achieved by the same process, described above \eqref{eq:cutandrotated}, that was used to count $\tlone_n$ connectivities:
One cuts the connectivity horizontally and rotates its top edge so that it lies to the left of the bottom edge. For the $\tlone_n$ connectivity appearing in \eqref{eq:ExOfAnkBasisElt2}, for example, the relevant rotation is precisely that displayed in \eqref{eq:cutandrotated}.
This rotation is again invertible and it follows that the dimension of $\atl_{n,k}$ may be written as a sum over link state cardinalities:
\begin{equation}
\dim \atl_{n,k} = \sideset{}{'}\sum_{b = 0}^{2k} \abs{\linksb_{2n}^{0,b}}.
\end{equation}
Here, the primed summation indicates that the index increases in steps of two.  This is straightforwardly simplified to \eqref{eq:dimatl} using \eqref{eq:cards}.
\end{proof}

It follows immediately from \cref{cor:BTL=TLone} that the $\algbtl_{n,k}$ relations \eqref{eq:newBTL} inherited from $\tlone_n$ are a complete set of relations for $\diabtl_{n,k}$ when $k \ge n$, so that $\algbtl_{n,k} \simeq \diabtl_{n,k}$ in this case. When $k<n$, we can now prove that a complete set is obtained by adding the closure relation \eqref{eq:ClosureRelation}.

\begin{Proposition} \label{prop:BTLComplete}
For $k<n$, the algebra $\algbtl_{n,k}$, defined by \eqref{eq:newBTL} and \eqref{eq:ClosureRelation}, is isomorphic to $\diabtl_{n,k}$.
\end{Proposition}
\begin{proof}
The proof consists of demonstrating that imposing the closure relation \eqref{eq:ClosureRelation} allows one to refine the spanning set of $\tlone_n$, given by the reduced words of \cref{lem:TLoneReduced}, to a spanning set of $\algbtl_{n,k}$ whose cardinality matches $\dim \diabtl_{n,k}$, as given in \cref{sec:dimA}.  This will yield the inequality $\dim \algbtl_{n,k} \le \dim \diabtl_{n,k}$.  Because $\diabtl_{n,k}$ is already known to be a quotient of $\algbtl_{n,k}$ (the diagrammatic algebra could satisfy further algebraic relations, in principle), the reverse inequality is trivial and the desired isomorphism follows.

Recall from \eqref{eq:surjectiveh} that there is a surjective homomorphism $\mathfrak{h}$ from $\tlone_n$ onto $\algbtl_{n,k}$.  For $k<n$, \eqref{eq:ClosureRelation} may be reinterpreted as the statement that a certain linear combination $w$ of words in the $\tlone_n$ generators $e_j$ belongs to the kernel of $\mathfrak{h}$:  $\mathfrak{h}(w) = 0$.  By referring to \eqref{eq:Yrec}, one finds that $w$ has the form
\begin{equation}
w = (e_n e_{n-1} \cdots e_{n-k}) (e_n e_{n-1} \cdots e_{n-k+1}) \cdots (e_n e_{n-1}) (e_n) + v,
\end{equation}
where $v$ is a linear combination of words whose lengths are strictly less than that of the displayed word $w-v$.  Note that $w-v \in \tlone_n$ is already in the reduced form guaranteed by \cref{lem:TLoneReduced}, so it cannot be a linear combination of any of the shorter words appearing in $v$. In particular, the coefficient of the reduced word $w-v$ in $w$ is non-zero.  Note also that $e_n$ appears exactly $k+1$ times in $w-v$.

As kernels are two-sided ideals, we may right-multiply $w$ by $\tlone_n$ generators and still end up with an element which vanishes when mapped into $\algbtl_{n,k}$.  In particular, we may right-multiply by $e_{n-k-1} e_{n-k-2} \cdots e_{n - \ell_1}$, then by $e_{n-k} e_{n-k-1} \cdots e_{n - \ell_2}$, and so on, so as to transform $w-v$ into
\begin{equation}
(e_n e_{n-1} \cdots e_{\ell_1}) (e_n e_{n-1} \cdots e_{\ell_2}) \cdots (e_n e_{n-1} \cdots e_{\ell_{k+1}}).
\end{equation}
If $\ell_1 < \ell_2 < \cdots < \ell_{k+1}$, then this transformed word is still reduced, so it is longer than any word in the transformed version of $v$ and so cannot be cancelled.  By left-multiplying appropriately, we can now convert this transformed version of $w-v$ into any of the reduced words \eqref{eq:TLoneJonesform} of $\tlone_n$ that have at least $k+1$ occurrences of $e_n$.  Again, the result cannot be cancelled by any of the identically transformed words in $v$.

It follows that every reduced word of $\tlone_n$, with $e_n$ occurring at least $k+1$ times, corresponds to a non-trivial relation in $\algbtl_{n,k}$.  A spanning set for $\algbtl_{n,k}$ is then obtained from the reduced words \eqref{eq:TLoneJonesform} by discarding those in which $e_n$ occurs at least $k+1$ times and then applying $\mathfrak{h}$.  We can count the number of words that remains:  The cardinality of this spanning set of $\algbtl_{n,k}$ is precisely the number of subdiagonal north-east walks from $(0,0)$ to $(n+1,n+1)$ that may have at most two stops in every column but the $n$-th, where up to $k+1$ stops are allowed.  This number was computed in \eqref{eq:FinestCounting} and is given by
\begin{equation}
c_{n+1}^{(k+1)} = \binom{2n}{n} - \binom{2n}{n-k-1} = \dim \diabtl_{n,k},
\end{equation}
completing the proof.
\end{proof}

\subsection{Proofs for $\beta \in \CC$} \label{app:SeamsBetaC} 

To specialise these results to $\beta = q+q^{-1}\in \mathbb C$, we first note that for $q$ generic, the \WJ{} projectors are well-defined and the proofs of \cref{sec:aandbiso}, \cref{cor:BTL=TLone}, \cref{sec:dimA} and \cref{prop:BTLComplete} 
carry through. It follows in this case that the specialised algebras $\atl_{n,k}(\beta)$ and $\btl_{n,k}(\beta)$ ($\simeq \algbtl_{n,k}(\beta) \simeq \diabtl_{n,k}(\beta)$) are still isomorphic and that
\begin{equation}
\dim \btl_{n,k}(\beta) = \binom{2n}{n}- \binom{2n}{n-k-1} \qquad \text{(\(q\) generic).}
\end{equation}
For $q$ a root of unity, this is no longer true, in general. If $\ell$ denotes the smallest positive integer satisfying $q^{2 \ell} = 1$, then the Chebyshev polynomial $U_{\ell - 1}$ is zero, making some of the \WJ{} projectors $P_k$ singular for $k \ge \ell$. The specialised diagrammatic algebra $\diabtl_{n,k}(\beta)$ is no longer well-defined, for $k \ge \ell$, as the tangles $\Ik$ and $\Ekj{k}{j}$ would be linear combinations of $\tl_{n+k}(\beta)$ connectivities with divergent coefficients. In stark contrast, the algebraic definition of $\algbtl_{n,k}(\beta)$ remains well-defined as the boundary \TL{} relations \eqref{eq:newBTL} and the closure relation \eqref{eq:ClosureRelation} are non-singular for all $\beta \in \CC$.

However, there is a subtlety to this presentation in terms of generators and relations.  When $q^{2\ell} = 1$, the 
Chebyshev polynomials $U_{\ell m-1}$ vanish for each $m \in \ZZ_{\ge 0}$. As $\btl_{n,0} \simeq \tl_n$, it is natural to 
define the specialised algebra $\btl_{n,0}(\beta)$ to be $\tl_n(\beta)$.  
If $k>0$, then we proceed as follows:  Let $k'$ be the smallest positive integer satisfying $k'=k \bmod{\ell}$. Assuming 
that $k'<n$, it then follows that the recursive construction of the tangle $\Yk_{k'+1}$ fails because its coefficient 
$U_{k-k'-1}$, in \eqref{eq:Yrec}, is zero when evaluated at $t = k'$.\footnote{This recursive construction, starting from 
$\Yk_1 = \Ekj{k}{n}$, explains why we do not allow $k' = 0$ here.}  Instead, we obtain a closure relation similar to \eqref{eq:ClosureRelation}:
\begin{equation} \label{eq:ClosureRelation'}
\Big[\prod_{j=0}^{k'} \Ekj{k}{n-j}\Big] \Yk_{k'} = \sum_{i=0}^{k'-1} (-1)^i U_{k-1-i} \Big[ \prod_{j=i+2}^{k'} \Ekj{k}{n-j} \Big] \Yk_{k'} \qquad \text{(\(n>k'\)).}
\end{equation}
In fact, this is precisely \eqref{eq:ClosureRelation} when $0 < k \le \ell$, that is, $k' = k$.

When $k > \ell$ (and $n>k'$), this closure relation is stronger
than the generic one \eqref{eq:ClosureRelation} --- it generates linear dependences involving monomials in the $\Ekj{k}{j}$ of shorter lengths than those generated by its generic counterpart.
The dimension of the specialised algebra $\btl_{n,k}(\beta)$ is thus strictly smaller than that given by \cref{sec:dimA}.  We therefore define the specialised boundary seam algebra $\btl_{n,k}(\beta)$, when $k>0$ and $q$ is a root of unity, to be the complex associative algebra with unit $\Ik$, generators $\Ekj{k}{j}$, $j=1, \ldots, n$, and relations \eqref{eq:newBTL}, supplemented by \eqref{eq:ClosureRelation'} if $n>k'$. 

For $n \le k'$, this amounts to the identification
\begin{equation}
\btl_{n,k}(\beta) = \tlone_n(\beta,U_{k-1},U_k) \qquad \text{(\(k \neq 0\), \(n \le k'\))} 
\end{equation}
and, as an immediate consequence, 
\begin{equation}
\dim \btl_{n,k}(\beta) = \binom{2n}{n} \qquad \text{(\(k\neq 0\), \(n \le k'\)).}
\end{equation}

Because we are primarily interested in comparing with the scaling limit, it is the case $n >k$\:($\ge k'$) which is the most important. For $n> k'$, we have the following results. 
\begin{Proposition}\label{sec:dimB} 
Let $q$ be a root of unity and let $\ell$ be the smallest positive integer satisfying $q^{2\ell}=1$.  If $k>0$ and $k'<n$ is the smallest positive integer satisfying $k'=k \bmod{\ell}$, then $\btl_{n,k}(\beta) \simeq \btl_{n,k'}(\beta)$.
\end{Proposition}
\begin{proof}
Since $q^{2 \ell} = 1$, we have $q = \ee^{\ii \pi a / \ell}$ for some $a \in \ZZ$.  Let $m$ be the non-negative integer satisfying $k-k' = m\ell$.  Then, the Chebyshev polynomials satisfy the periodicity property $U_{k'} = U_{k-m \ell} = (-1)^{am} U_k$.  

For $k > 0$ and $n>k'$, the closure relation \eqref{eq:ClosureRelation'} inspires us to define a surjective homomorphism $\psi \colon \btl_{n,k}(\beta) \to \btl_{n,k'}(\beta)$ of unital associative algebras:
\begin{equation} \label{eq:DefPsi}
\psi(\Ik) = I^{_{(k')}}_{\phantom{j}}, \qquad 
\psi(\Ekj{k}{j}) = 
\begin{cases}
\Ekj{k'}{j} & \text{if \(j \neq n\),} \\
(-1)^{am} \Ekj{k'}{n} & \text{if \(j=n\).}
\end{cases}
\end{equation}
We only need to show that this map preserves the relations \eqref{eq:newBTL} and \eqref{eq:ClosureRelation}.  Of the former, only those involving $\Ekj{k}{n}$ require checking. 
For example,
$(\Ekj{k}{n})^2 = U_k \Ekj{k}{n}$ is checked as follows:
\begin{equation}
\psi \bigl( (\Ekj{k}{n})^2 \bigr) = (\Ekj{k'}{n})^2 = U_{k'} \Ekj{k'}{n} = (-1)^{am} U_k \Ekj{k'}{n} = \psi (U_k \Ekj{k}{n}).
\end{equation}
Similarly, one may check from \eqref{eq:Yrec} that $\psi (\Yk_t) = (-1)^{am} Y^{_{(k')}}_t$, for $t \le k'$, and that $\psi$ maps the relation \eqref{eq:ClosureRelation'} onto the closure relation of $\btl_{n,k'}(\beta)$.

This shows that $\btl_{n,k}(\beta)$ is a quotient of $\btl_{n,k'}(\beta)$.  To conclude that these algebras are isomorphic, we note that the inverse map, obtained by swapping $k$ and $k'$ in \eqref{eq:DefPsi}, is likewise well-defined.  
\end{proof}
For $1\le k' < \ell$, the algebra $\btl_{n,k'}(\beta)$ is well-defined both diagrammatically and algebraically. It is therefore the specialisation of $\btl_{n,k'}$ at $\beta = 2 \cos \frac{\pi a}{\ell} \in \CC$,
hence the dimension of $\btl_{n,k}(\beta) \simeq \btl_{n,k'}(\beta)$ is given by  
\begin{equation}
\dim \btl_{n,k}(\beta) = \binom{2n}{n}- \binom{2n}{n-k'-1}, \qquad \text{(\(k \neq 0 \bmod \ell\), \(n>k'\)).}
\label{eq:dimbtlspecial}
\end{equation}
In fact, this dimension formula also covers the case $n \le k'$, if we understand that the second binomial coefficient above is then $0$.
We also note that taking $k'=0$ in \eqref{eq:dimbtlspecial} gives the correct dimension of $\btl_{n,0} \simeq \tl_n(\beta)$.

The only missing case is $k'=\ell$, that is $k = m \ell$, $m\in \mathbb Z_{+}$. In this case, \cref{sec:dimB} asserts that $\btl_{n,k}(\beta) \simeq \btl_{n,\ell}(\beta)$, but not that it is isomorphic to $\btl_{n,0}(\beta) \simeq \tl_{n}(\beta)$. This is because the generator $\Ekj{k'}{n}$, required to define $\psi$ in \eqref{eq:DefPsi}, is not defined in $\tl_n(\beta)$. Because $\btl_{n,\ell}$ is not defined diagrammatically, further analysis is required to determine the dimension of $\btl_{n,k}$ in this case.  This will not be addressed in this paper.

%
\section{Boundary seam module proofs} \label{app:btlstan}
%

This section provides
proofs of statements made in \cref{sec:reps,sec:Gram} about standard representations of the boundary seam algebras $\btl_{n,k}$ and the determinants of the Gram matrices of their invariant bilinear forms.

\subsection{Properties of standard modules}
\label{app:standardproperties}

Recall that a basis of the standard $\btl_{n,k}$-module $\stan_{n,k}^d$ is given by (the equivalence classes of) 
the $\tl_{n+k}$ link states, with $d$ defects, for which no two of the $k$ boundary 
nodes are linked together.  This basis was denoted by $\links_{n,k}^d$ in \cref{sec:reps} and the boundary 
nodes were indicated diagrammatically in pink. We study the restriction of 
$\stan_{n,k}^d$ to the $\tl_n$ subalgebra $\langle \Ik, \Ekj{k}{j}
; j = 1, \dots, n-1\rangle \subseteq \btl_{n,k}$.
Under the restricted action, seen from the bulk, the arcs connecting to the boundary nodes may be
viewed as a second class of defects.  This may be formalised by taking a link state in $\links_{n,k}^d$, erasing the 
boundary nodes and any arcs connecting to them, and then inserting a new type of defect, which we shall draw as a 
wavy pink line, at any bulk node that was originally connected to the boundary. For example,
\begin{equation}
\psset{unit=0.8}
\begin{pspicture}[shift=-0.15](0,-0.1)(2.4,0.7)
\psline{-}(0,0)(1.6,0)
\psline[linecolor=purple,linewidth=2.0pt]{-}(1.6,0)(2.4,0)
\psarc[linecolor=blue,linewidth=1.5pt]{-}(1.2,0){0.2}{0}{180}
\psline[linecolor=blue,linewidth=1.5pt]{-}(0.2,0)(0.2,0.6)
\psline[linecolor=blue,linewidth=1.5pt]{-}(2.2,0)(2.2,0.6)
\psbezier[linecolor=blue,linewidth=1.5pt]{-}(0.6,0)(0.6,0.5)(1.8,0.5)(1.8,0)
\end{pspicture}
\quad \rightarrow \quad
\psset{unit=0.8}
\begin{pspicture}[shift=-0.15](0,-0.1)(1.6,0.7)
\psline{-}(0,0)(1.6,0)
\psarc[linecolor=blue,linewidth=1.5pt]{-}(1.2,0){0.2}{0}{180}
\psline[linecolor=blue,linewidth=1.5pt]{-}(0.2,0)(0.2,0.6)
\rput(0.6,0){\wobbly}
\end{pspicture}
\ .
\end{equation}

If we ignore the difference between the two types of defect in these diagrams, then they may be identified with $\tl_n$ link states.  Suppose that we start with a $\btl_{n,k}$ link state with $d$ defects.  By considering the number $d'$ of bulk defects and $d''$ of boundary defects, separately, it follows that the number of defects in the resulting $\tl_n$ link state is
\begin{equation} \label{eq:ConstrainBulkBdryDefects}
e = d'+k-d'' = d+k-2d'' = 2d'+k-d,
\end{equation}
which is bounded below by $\abs{k-d}$ and above by $\min (k+d,n)$.  Moreover, it is easy to reconstruct the $\btl_{n,k}$ link state from the $\tl_n$ link state, given $k$, $d$ and $e$, by solving for $d''$ and linking the rightmost $k-d''$ bulk 
nodes to the leftmost boundary nodes (the remaining boundary nodes are defects).  This proves that we have
a bijective map between the $\btl_{n,k}$ and $\tl_n$ link state bases:
\begin{equation} 
\links_{n,k}^d \quad\longleftrightarrow
\sideset{}{'} \bigcup_{e = |k-d|}^{\min (k+d,n)} \links_n^e.
\label{eq:bij}
\end{equation}
The prime indicates that $e$ increases in steps of two.
We illustrate this map for $(n,k,d) = (4,2,2)$:
\begin{equation} 
\links_{4,2}^2 = \left\{\ \ 
\begin{matrix}
\psset{unit=0.8}
\begin{pspicture}[shift=-0.15](0,-0.1)(2.4,0.7)
\psline{-}(0,0)(1.6,0)
\psline[linecolor=purple,linewidth=2.0pt]{-}(1.6,0)(2.4,0)
\psarc[linecolor=blue,linewidth=1.5pt]{-}(0.4,0){0.2}{0}{180}
\psarc[linecolor=blue,linewidth=1.5pt]{-}(1.2,0){0.2}{0}{180}
\psline[linecolor=blue,linewidth=1.5pt]{-}(1.8,0)(1.8,0.6)
\psline[linecolor=blue,linewidth=1.5pt]{-}(2.2,0)(2.2,0.6)
\end{pspicture}
\\[0.2cm]
\psset{unit=0.8}
\begin{pspicture}[shift=-0.15](0,-0.1)(2.4,0.7)
\psline{-}(0,0)(1.6,0)
\psline[linecolor=purple,linewidth=2.0pt]{-}(1.6,0)(2.4,0)
\psarc[linecolor=blue,linewidth=1.5pt]{-}(0.8,0){0.2}{0}{180}
\psline[linecolor=blue,linewidth=1.5pt]{-}(1.8,0)(1.8,0.6)
\psline[linecolor=blue,linewidth=1.5pt]{-}(2.2,0)(2.2,0.6)
\psbezier[linecolor=blue,linewidth=1.5pt]{-}(0.2,0)(0.2,0.5)(1.4,0.5)(1.4,0)
\end{pspicture} 
\\[0.2cm]
\psset{unit=0.8}
\begin{pspicture}[shift=-0.15](0,-0.1)(2.4,0.7)
\psline{-}(0,0)(1.6,0)
\psline[linecolor=purple,linewidth=2.0pt]{-}(1.6,0)(2.4,0)
\psarc[linecolor=blue,linewidth=1.5pt]{-}(1.2,0){0.2}{0}{180}
\psline[linecolor=blue,linewidth=1.5pt]{-}(0.2,0)(0.2,0.6)
\psline[linecolor=blue,linewidth=1.5pt]{-}(2.2,0)(2.2,0.6)
\psbezier[linecolor=blue,linewidth=1.5pt]{-}(0.6,0)(0.6,0.5)(1.8,0.5)(1.8,0)
\end{pspicture} 
\\[0.2cm]
\psset{unit=0.8}
\begin{pspicture}[shift=-0.15](0,-0.1)(2.4,0.7)
\psline{-}(0,0)(1.6,0)
\psline[linecolor=purple,linewidth=2.0pt]{-}(1.6,0)(2.4,0)
\psarc[linecolor=blue,linewidth=1.5pt]{-}(0.8,0){0.2}{0}{180}
\psarc[linecolor=blue,linewidth=1.5pt]{-}(1.6,0){0.2}{0}{180}
\psline[linecolor=blue,linewidth=1.5pt]{-}(0.2,0)(0.2,0.6)
\psline[linecolor=blue,linewidth=1.5pt]{-}(2.2,0)(2.2,0.6)
\end{pspicture} 
\\[0.2cm]
\psset{unit=0.8}
\begin{pspicture}[shift=-0.15](0,-0.1)(2.4,0.7)
\psline{-}(0,0)(1.6,0)
\psline[linecolor=purple,linewidth=2.0pt]{-}(1.6,0)(2.4,0)
\psarc[linecolor=blue,linewidth=1.5pt]{-}(0.4,0){0.2}{0}{180}
\psarc[linecolor=blue,linewidth=1.5pt]{-}(1.6,0){0.2}{0}{180}
\psline[linecolor=blue,linewidth=1.5pt]{-}(1.0,0)(1.0,0.6)
\psline[linecolor=blue,linewidth=1.5pt]{-}(2.2,0)(2.2,0.6)
\end{pspicture} 
\\[0.2cm]
\psset{unit=0.8}
\begin{pspicture}[shift=-0.15](0,-0.1)(2.4,0.7)
\psline{-}(0,0)(1.6,0)
\psline[linecolor=purple,linewidth=2.0pt]{-}(1.6,0)(2.4,0)
\psarc[linecolor=blue,linewidth=1.5pt]{-}(1.6,0){0.2}{0}{180}
\psline[linecolor=blue,linewidth=1.5pt]{-}(0.2,0)(0.2,0.6)
\psline[linecolor=blue,linewidth=1.5pt]{-}(0.6,0)(0.6,0.6)
\psbezier[linecolor=blue,linewidth=1.5pt]{-}(1.0,0)(1.0,0.5)(2.2,0.5)(2.2,0)
\end{pspicture} 
\end{matrix}\right.
\qquad \longleftrightarrow \qquad 
\begin{matrix}
\psset{unit=0.8}
\begin{pspicture}[shift=-0.15](0,-0.1)(1.6,0.7)
\psline{-}(0,0)(1.6,0)
\psarc[linecolor=blue,linewidth=1.5pt]{-}(0.4,0){0.2}{0}{180}
\psarc[linecolor=blue,linewidth=1.5pt]{-}(1.2,0){0.2}{0}{180}
\end{pspicture} & 
\multirow{3}{*}{$\left.
\psset{unit=0.8}
\begin{pspicture}[shift=-0.15](0,-0.1)(0,0.84)
\end{pspicture}
\right\} \lra \  
\links_4^0$}
\\[0.2cm]
\psset{unit=0.8}
\begin{pspicture}[shift=-0.15](0,-0.1)(1.6,0.7)
\psline{-}(0,0)(1.6,0)
\psarc[linecolor=blue,linewidth=1.5pt]{-}(0.8,0){0.2}{0}{180}
\psbezier[linecolor=blue,linewidth=1.5pt]{-}(0.2,0)(0.2,0.5)(1.4,0.5)(1.4,0)
\end{pspicture}
\\[0.2cm]
\psset{unit=0.8}
\begin{pspicture}[shift=-0.15](0,-0.1)(1.6,0.7)
\psline{-}(0,0)(1.6,0)
\psarc[linecolor=blue,linewidth=1.5pt]{-}(1.2,0){0.2}{0}{180}
\psline[linecolor=blue,linewidth=1.5pt]{-}(0.2,0)(0.2,0.6)
\rput(0.6,0){\wobbly}
\end{pspicture}
& 
\multirow{2}{*}{$\left.
\psset{unit=0.8}
\begin{pspicture}[shift=-0.15](0,-0.1)(0,1.54)
\end{pspicture}
\right\} \lra \ 
\links_4^2$}
\\[0.2cm]
\psset{unit=0.8}
\begin{pspicture}[shift=-0.15](0,-0.1)(1.6,0.7)
\psline{-}(0,0)(1.6,0)
\psarc[linecolor=blue,linewidth=1.5pt]{-}(0.8,0){0.2}{0}{180}
\psline[linecolor=blue,linewidth=1.5pt]{-}(0.2,0)(0.2,0.6)
\rput(1.4,0){\wobbly}
\end{pspicture}
\\[0.2cm]
\psset{unit=0.8}
\begin{pspicture}[shift=-0.15](0,-0.1)(1.6,0.7)
\psline{-}(0,0)(1.6,0)
\psarc[linecolor=blue,linewidth=1.5pt]{-}(0.4,0){0.2}{0}{180}
\psline[linecolor=blue,linewidth=1.5pt]{-}(1.0,0)(1.0,0.6)
\rput(1.4,0){\wobbly}
\end{pspicture}
\\[0.2cm]
\psset{unit=0.8}
\begin{pspicture}[shift=-0.15](0,-0.1)(1.6,0.7)
\psline{-}(0,0)(1.6,0)
\psline[linecolor=blue,linewidth=1.5pt]{-}(0.2,0)(0.2,0.6)
\psline[linecolor=blue,linewidth=1.5pt]{-}(0.6,0)(0.6,0.6)
\rput(1.0,0){\wobbly}
\rput(1.4,0){\wobbly}
\end{pspicture} & 
\left.
\psset{unit=0.8}
\begin{pspicture}[shift=-0.15](0,-0.0)(0,0.5)
\rput(0.10,0.40){$\Big \} \mspace{5mu} \lra \  
\links_4^4$.}
\end{pspicture}
\right.
\end{matrix}
\label{eq:B422}
\end{equation}
Combining this bijection with the cardinalities of the bases of $\tl_n$ link states, given in \eqref{eq:dimV}, 
we have proven the following proposition. 

\begin{Proposition} 
The dimension of the standard module of $\btl_{n,k}$ with $d$ defects is
\begin{equation} 
\dim \stan_{n,k}^d = \binom{n}{\frac{n+k-d}{2}} - \binom{n}{\frac{n-k-d-2}{2}}.
\end{equation}
\label{sec:dimVnkd}
\end{Proposition}

We remark that the action of $\btl_{n,k}$ on the standard module $\stan_{n,k}^d$ does not restrict to the standard action of $\tl_n$ upon cutting link states.  This is why we did not express the bijection \eqref{eq:bij} as a decomposition of the restricted module.  The standard rule, whereby one sets the result to zero when the number of defects decreases, is modified so that connecting 
defects of the same kind gives zero, but connecting defects of different kinds need not. For example, $e_1$ acting on 
$
\psset{unit=0.54}
\begin{pspicture}[shift=-0.10](0,-0.1)(1.6,0.7)
\psline{-}(0,0)(1.6,0)
\psline[linecolor=blue,linewidth=1.5pt]{-}(0.2,0)(0.2,0.6)
\psline[linecolor=blue,linewidth=1.5pt]{-}(0.6,0)(0.6,0.6)
\rput(1.0,0){\wobbly}
\rput(1.4,0){\wobbly}
\end{pspicture}
$\,
gives zero, but acting with $e_2$
instead gives 
$
\psset{unit=0.54}
\begin{pspicture}[shift=-0.10](0,-0.1)(1.6,0.7)
\psline{-}(0,0)(1.6,0)
\psarc[linecolor=blue,linewidth=1.5pt]{-}(0.8,0){0.2}{0}{180}
\psline[linecolor=blue,linewidth=1.5pt]{-}(0.2,0)(0.2,0.6)
\rput(1.4,0){\wobbly}
\end{pspicture}
$\,.

Our next goal is to show that the standard representations $\rho_{n,k}^d$ are well-defined for all $\beta \in \mathbb C$. This is the content of the following proposition.

\begin{Proposition}
The matrix $\rho_{n,k}^d(a)$ has non-divergent entries
for all $a \in \btl_{n,k}$, so specialising the representation $\rho_{n,k}^d$ of the formal boundary seam algebra $\btl_{n,k}$ to any $\beta \in \CC$ defines a representation of the specialised algebra $\btl_{n,k}(\beta)$. 
\label{sec:nonsingular}
\end{Proposition}
\begin{proof}
It is enough to show that $\rho_{n,k}^d(a)$ is non-divergent
when $a$ is the unit $\Ik$ or a generator $\Ekj{k}{j}$ of $\btl_{n,k}$.  Recall that $\Ik \in \btl_{n,k} \subseteq \tl_{n+k}$ is the identity element of $\tl_n$ glued to the \WJ{} projector $P_k =
\begin{pspicture}[shift=-0.05](-0.03,0.00)(1.07,0.3)
\pspolygon[fillstyle=solid,fillcolor=pink](0,0)(1,0)(1,0.3)(0,0.3)(0,0)\rput(0.5,0.15){$_k$} 
\end{pspicture}
$, acting on the boundary nodes $n+1, \ldots, n+k$.  The $\tl_{n+k}$ representation $\rho_{n+k}^d$ 
will therefore have singularities when specialised to
certain $\beta \in \CC$, if $k>1$. 
However, we may expand $P_k$ as the identity plus a linear combination of non-trivial 
monomials in the \TL{} generators $e_j$, where $j=n+1, \ldots, n+k-1$.  Acting with such a monomial on an arbitrary 
$\tl_{n+k}$ link state will result in a link state with two linked boundary 
nodes or in zero.  Either way, the result is set to zero in the quotient 
\eqref{eq:Vdecomp} defining the standard $\btl_{n,k}$ representation 
$\rho_{n,k}^d$.  Thus, $\rho_{n,k}^d(\Ik)$ is the identity matrix,
which is clearly non-divergent upon specialising to $\btl_{n,k}(\beta)$. 

Similarly, when $j<n$,
\begin{equation}
\rho_{n+k}^d(\Ekj{k}{j}) = \rho_{n+k}^d(\Ik) \rho_{n+k}^d(e_j) \rho_{n+k}^d(\Ik) = \rho_{n+k}^d(\Ik) \rho_{n+k}^d(e_j)
\end{equation}
has singularities because of the projector in $\rho_{n+k}^d(\Ik)$. However, $\rho_{n+k}^d(e_j)$ has no singularities, 
so the above argument shows that $\rho_{n,k}^d(\Ekj{k}{j})$ is non-divergent
in the quotient, hence may be specialised, for all $j<n$.

It remains to consider the generator $\Ekj{k}{n} = \Ik e_n \Ik$.  For this, we employ \eqref{eq:otherWJrelation} to expand $\Ekj{k}{n}$ as
\begin{align}
\Ekj{k}{n} &=
\psset{unit=0.5}
U_{k-1} \ 
\begin{pspicture}[shift=-1.30](0.4,-.8)(6.1,2.0)
\psarc[linewidth=1.5pt,linecolor=blue]{-}(1,1.5){0.5}{-180}{0}
\psarc[linewidth=1.5pt,linecolor=blue]{-}(1,0){0.5}{0}{180}
\psline[linewidth=1.5pt,linecolor=blue]{-}(0.5,-0.8)(0.5,0)
\psline[linewidth=1.5pt,linecolor=blue]{-}(0.5,1.5)(0.5,2.3)
\psline[linewidth=1.5pt,linecolor=blue]{-}(1.5,-0.8)(1.5,-0.3)
\psline[linewidth=1.5pt,linecolor=blue]{-}(1.5,1.8)(1.5,2.3)
\psline[linewidth=1.5pt,linecolor=blue]{-}(2.5,-0.8)(2.5,2.3)
\psline[linewidth=1.5pt,linecolor=blue]{-}(3.5,-0.8)(3.5,2.3)
\psline[linewidth=1.5pt,linecolor=blue]{-}(4.5,-0.8)(4.5,2.3)
\psline[linewidth=1.5pt,linecolor=blue]{-}(5.5,-0.8)(5.5,2.3)
\pspolygon[fillstyle=solid,fillcolor=pink](1.1,0)(5.9,0)(5.9,-0.5)(1.1,-0.5)(1.1,0)
\rput(3.5,-0.25){$_{k}$}
\pspolygon[fillstyle=solid,fillcolor=pink](1.1,1.5)(5.9,1.5)(5.9,2)(1.1,2)(1.1,1.5)
\rput(3.5,1.75){$_{k}$}
\end{pspicture}
= U_{k-1} \ 
\begin{pspicture}[shift=-1.30](0.4,-.8)(6.1,2.0)
\psarc[linewidth=1.5pt,linecolor=blue]{-}(1,1.5){0.5}{-180}{0}
\psarc[linewidth=1.5pt,linecolor=blue]{-}(1,0){0.5}{0}{180}
\psline[linewidth=1.5pt,linecolor=blue]{-}(0.5,-0.8)(0.5,0)
\psline[linewidth=1.5pt,linecolor=blue]{-}(0.5,1.5)(0.5,2.3)
\psline[linewidth=1.5pt,linecolor=blue]{-}(1.5,-0.8)(1.5,0)
\psline[linewidth=1.5pt,linecolor=blue]{-}(1.5,1.5)(1.5,2.3)
\psline[linewidth=1.5pt,linecolor=blue]{-}(2.5,-0.8)(2.5,2.3)
\psline[linewidth=1.5pt,linecolor=blue]{-}(3.5,-0.8)(3.5,2.3)
\psline[linewidth=1.5pt,linecolor=blue]{-}(4.5,-0.8)(4.5,2.3)
\psline[linewidth=1.5pt,linecolor=blue]{-}(5.5,-0.8)(5.5,2.3)
\pspolygon[fillstyle=solid,fillcolor=pink](1.1,0)(5.9,0)(5.9,-0.5)(1.1,-0.5)(1.1,0)
\rput(3.5,-0.25){$_{k}$}
\end{pspicture}  - U_{k-2} \ 
\begin{pspicture}[shift=-1.30](0.4,-.8)(6.1,2.0)
\psarc[linewidth=1.5pt,linecolor=blue]{-}(2,1.5){0.5}{-180}{0}
\psarc[linewidth=1.5pt,linecolor=blue]{-}(1,0){0.5}{0}{180}
\psline[linewidth=1.5pt,linecolor=blue]{-}(0.5,-0.8)(0.5,0)
\psline[linewidth=1.5pt,linecolor=blue]{-}(0.5,1.5)(0.5,2.3)
\psline[linewidth=1.5pt,linecolor=blue]{-}(1.5,-0.8)(1.5,0)
\psline[linewidth=1.5pt,linecolor=blue]{-}(1.5,1.5)(1.5,2.3)
\psline[linewidth=1.5pt,linecolor=blue]{-}(2.5,-0.8)(2.5,0)
\psline[linewidth=1.5pt,linecolor=blue]{-}(2.5,1.5)(2.5,2.3)
\psbezier[linewidth=1.5pt,linecolor=blue]{-}(2.5,0)(2.5,0.75)(0.5,0.75)(0.5,1.5)
\psline[linewidth=1.5pt,linecolor=blue]{-}(3.5,-0.8)(3.5,2.3)
\psline[linewidth=1.5pt,linecolor=blue]{-}(4.5,-0.8)(4.5,2.3)
\psline[linewidth=1.5pt,linecolor=blue]{-}(5.5,-0.8)(5.5,2.3)
\pspolygon[fillstyle=solid,fillcolor=pink](1.1,0)(5.9,0)(5.9,-0.5)(1.1,-0.5)(1.1,0)
\rput(3.5,-0.25){$_{k}$}
\end{pspicture}
\notag \\[0.5cm]
&\mspace{250mu} + U_{k-3} \ 
\psset{unit=0.5}
\begin{pspicture}[shift=-1.30](0.4,-.8)(6.1,2.0)
\psarc[linewidth=1.5pt,linecolor=blue]{-}(3,1.5){0.5}{-180}{0}
\psarc[linewidth=1.5pt,linecolor=blue]{-}(1,0){0.5}{0}{180}
\psline[linewidth=1.5pt,linecolor=blue]{-}(0.5,-0.8)(0.5,0)
\psline[linewidth=1.5pt,linecolor=blue]{-}(0.5,1.5)(0.5,2.3)
\psline[linewidth=1.5pt,linecolor=blue]{-}(1.5,-0.8)(1.5,0)
\psline[linewidth=1.5pt,linecolor=blue]{-}(1.5,1.5)(1.5,2.3)
\psline[linewidth=1.5pt,linecolor=blue]{-}(2.5,-0.8)(2.5,0)
\psline[linewidth=1.5pt,linecolor=blue]{-}(2.5,1.5)(2.5,2.3)
\psline[linewidth=1.5pt,linecolor=blue]{-}(3.5,-0.8)(3.5,0)
\psline[linewidth=1.5pt,linecolor=blue]{-}(3.5,1.5)(3.5,2.3)
\psbezier[linewidth=1.5pt,linecolor=blue]{-}(2.5,0)(2.5,0.75)(0.5,0.75)(0.5,1.5)
\psbezier[linewidth=1.5pt,linecolor=blue]{-}(3.5,0)(3.5,0.75)(1.5,0.75)(1.5,1.5)
\psline[linewidth=1.5pt,linecolor=blue]{-}(4.5,-0.8)(4.5,2.3)
\psline[linewidth=1.5pt,linecolor=blue]{-}(5.5,-0.8)(5.5,2.3)
\pspolygon[fillstyle=solid,fillcolor=pink](1.1,0)(5.9,0)(5.9,-0.5)(1.1,-0.5)(1.1,0)
\rput(3.5,-0.25){$_{k}$}
\end{pspicture}
\ - \dots + (-1)^{k+1} U_{0}\
\begin{pspicture}[shift=-1.30](0.4,-.8)(6.1,2.0)
\psarc[linewidth=1.5pt,linecolor=blue]{-}(5,1.5){0.5}{-180}{0}
\psarc[linewidth=1.5pt,linecolor=blue]{-}(1,0){0.5}{0}{180}
\psline[linewidth=1.5pt,linecolor=blue]{-}(0.5,-0.8)(0.5,0)
\psline[linewidth=1.5pt,linecolor=blue]{-}(0.5,1.5)(0.5,2.3)
\psline[linewidth=1.5pt,linecolor=blue]{-}(1.5,-0.8)(1.5,0)
\psline[linewidth=1.5pt,linecolor=blue]{-}(1.5,1.5)(1.5,2.3)
\psline[linewidth=1.5pt,linecolor=blue]{-}(2.5,-0.8)(2.5,0)
\psline[linewidth=1.5pt,linecolor=blue]{-}(2.5,1.5)(2.5,2.3)
\psline[linewidth=1.5pt,linecolor=blue]{-}(3.5,-0.8)(3.5,0)
\psline[linewidth=1.5pt,linecolor=blue]{-}(3.5,1.5)(3.5,2.3)
\psline[linewidth=1.5pt,linecolor=blue]{-}(4.5,-0.8)(4.5,0)
\psline[linewidth=1.5pt,linecolor=blue]{-}(4.5,1.5)(4.5,2.3)
\psline[linewidth=1.5pt,linecolor=blue]{-}(5.5,-0.8)(5.5,0)
\psline[linewidth=1.5pt,linecolor=blue]{-}(5.5,1.5)(5.5,2.3)
\psbezier[linewidth=1.5pt,linecolor=blue]{-}(2.5,0)(2.5,0.75)(0.5,0.75)(0.5,1.5)
\psbezier[linewidth=1.5pt,linecolor=blue]{-}(3.5,0)(3.5,0.75)(1.5,0.75)(1.5,1.5)
\psbezier[linewidth=1.5pt,linecolor=blue]{-}(4.5,0)(4.5,0.75)(2.5,0.75)(2.5,1.5)
\psbezier[linewidth=1.5pt,linecolor=blue]{-}(5.5,0)(5.5,0.75)(3.5,0.75)(3.5,1.5)
\pspolygon[fillstyle=solid,fillcolor=pink](1.1,0)(5.9,0)(5.9,-0.5)(1.1,-0.5)(1.1,0)
\rput(3.5,-0.25){$_{k}$}
\end{pspicture}
\notag \\
&= \sum_{j=0}^{k-1} (-1)^j U_{k-1-j} \Ik e_n e_{n+1} \cdots e_{n+j}.
\end{align}
Since $\rho_{n+k}^d(e_j)$ is non-divergent,
for each $j$, the only singularities are again in the \WJ{} projector and 
these are again removed in the quotient. Thus, $\rho_{n,k}^d(\Ekj{k}{n})$ is non-divergent,
so it defines a representation of the specialised algebra $\btl_{n,k}(\beta)$, for all $\beta \in \CC$.
\end{proof}

\subsection{Gram determinant} \label{app:Gram}

The invariant bilinear form $\gramprodk{\cdot}{\cdot}$ for the standard modules $\stan_{n,k}^d$ of $\btl_{n,k}$ was introduced in \cref{sec:Gram} using the inclusion of $\btl_{n,k}$ in $\tl_{n+k}$ and the definition of $\stan_{n,k}^d$ as a quotient of the $\tl_{n+k}$ standard module $\stan_{n+k}^d$. Indeed, if $w_1, w_2 \in \links_{n,k}^d$ are $\tl_{n+k}$ link states with $d$ defects and no two boundary nodes linked together, then their 
product $\gramprodk{w_1}{w_2}$ (or, rather, the product of their images in $\stan_{n,k}^d$) is the value of the $\tl_{n+k}$ product $\gramprod{w_1}{\Ik w_2}$. 
Diagrammatically, this is equivalent to sandwiching $\Ik$ between $w_1$ and the reflection of $w_2$. This product may alternatively be interpreted using the bijective map \eqref{eq:bij} between the link states of $\links_{n,k}^d$ and those of $\tl_n$, where the latter are modified so as to admit two different kinds of defect:
boundary and bulk.
The invariant bilinear form on $\stan_{n,k}^d$ is then viewed as acting on these modified $\tl_n$ link states (and their linear combinations): 
\begin{equation}
\gramprodk{\cdot}{\cdot} \colon \brac{\sideset{}{'}\bigoplus_{e=\abs{k-d}}^{\min (k+d,n)} \vspn\links_n^e} \times \brac{\sideset{}{'}\bigoplus_{e=\abs{k-d}}^{\min (k+d,n)} \vspn\links_n^e} \to \mathbb C,
\end{equation}
where a prime on $\bigoplus$ indicates that the indices increase by steps of two.
In this case, for $w_1 \in \links_n^e$ and $w_2 \in \links_n^{e'}$, the value of $\gramprodk{w_1}{w_2}$ may
be non-zero even if $e \neq e'$.
In terms of boundary and bulk defects, the fundamental rules are given by
\begin{equation}
\psset{unit=0.7}
\begin{pspicture}[shift=-0.45](0,-0.6)(0.8,0.6)
\psline(0,0)(0.8,0)
\psarc[linecolor=blue,linewidth=1.5pt]{-}(0.4,0){0.2}{0}{360}
\end{pspicture} \, = \beta, \qquad
\begin{pspicture}[shift=-0.45](0,-0.6)(0.4,0.6)
\psline(0,0)(0.4,0)
\psline[linecolor=blue,linewidth=1.5pt]{-}(0.2,0)(0.2,0.6)
\psline[linecolor=blue,linewidth=1.5pt]{-}(0.2,0)(0.2,-0.6)
\end{pspicture} \, = 1, \qquad
\begin{pspicture}[shift=-0.45](0,-0.6)(0.8,0.6)
\psline(0,0)(0.8,0)
\psline[linecolor=blue,linewidth=1.5pt]{-}(0.2,0)(0.2,0.6)
\psline[linecolor=blue,linewidth=1.5pt]{-}(0.6,0)(0.6,0.6)
\psarc[linecolor=blue,linewidth=1.5pt]{-}(0.4,0){0.2}{180}{360}
\end{pspicture} \, = \,
\begin{pspicture}[shift=-0.45](0,-0.6)(0.8,0.6)
\psline(0,0)(0.8,0)
\psline[linecolor=blue,linewidth=1.5pt]{-}(0.2,0)(0.2,-0.6)
\psline[linecolor=blue,linewidth=1.5pt]{-}(0.6,0)(0.6,-0.6)
\psarc[linecolor=blue,linewidth=1.5pt]{-}(0.4,0){0.2}{0}{180}
\end{pspicture} \, = \,
\begin{pspicture}[shift=-0.45](0,-0.6)(0.8,0.6)
\psline(0,0)(0.8,0)
\rput(0.2,0){\wobbly}
\rput(0.6,0){\wobbly}
\psarc[linecolor=blue,linewidth=1.5pt]{-}(0.4,0){0.2}{180}{360}
\end{pspicture} \, = \,
\begin{pspicture}[shift=-0.45](0,-0.6)(0.8,0.6)
\psline(0,0)(0.8,0)
\rput{180}(0.2,0){\wobbly}
\rput{180}(0.6,0){\wobbly}
\psarc[linecolor=blue,linewidth=1.5pt]{-}(0.4,0){0.2}{0}{180}
\end{pspicture} \, = 0, \qquad
\begin{pspicture}[shift=-0.45](0,-0.6)(0.8,0.6)
\psline(0,0)(0.8,0)
\psline[linecolor=blue,linewidth=1.5pt]{-}(0.2,0)(0.2,-0.6)
\rput{180}(0.6,0){\wobbly}
\psarc[linecolor=blue,linewidth=1.5pt]{-}(0.4,0){0.2}{0}{180}
\end{pspicture} \, = \,
\begin{pspicture}[shift=-0.45](0,-0.6)(0.8,0.6)
\psline(0,0)(0.8,0)
\psline[linecolor=blue,linewidth=1.5pt]{-}(0.2,0)(0.2,0.6)
\rput(0.6,0){\wobbly}
\psarc[linecolor=blue,linewidth=1.5pt]{-}(0.4,0){0.2}{180}{360}
\end{pspicture}
\, =1, \qquad
\begin{pspicture}[shift=-1.10](0,-1.2)(2.0,0.6)
\psline(0,0)(2.0,0)
\rput(0.2,0){\wobbly}\rput{180}(0.2,0){\wobbly}
\rput(0.6,0){\wobbly}\rput{180}(0.6,0){\wobbly}
\rput(1.0,0){\wobbly}\rput{180}(1.0,0){\wobbly}
\rput(1.4,0.3){...}\rput(1.4,-0.3){...}
\rput(1.8,0){\wobbly}\rput{180}(1.8,0){\wobbly}
\rput(1.0,-1.0){$\underbrace{\ \hspace{1.1cm}\ }_j$}
\end{pspicture} \, = \frac{U_k}{U_{k-j}}.
\label{eq:krule}
\end{equation}
These rules follow directly from the definition of the bilinear form and from the properties \eqref{eq:WJprop2} of the Wenzl-Jones projectors. We note that 
$\psset{unit=0.4}
\begin{pspicture}[shift=-0.5](0,-0.7)(0.4,0.7)
\psline(0,0)(0.4,0)
\rput(0.2,0){\wobblysmall}
\psline[linecolor=blue,linewidth=1.5pt]{-}(0.2,0)(0.2,-0.6)
\end{pspicture}$ and 
$\psset{unit=0.4}
\begin{pspicture}[shift=-0.5](0,-0.7)(0.4,0.7)
\psline(0,0)(0.4,0)
\psline[linecolor=blue,linewidth=1.5pt]{-}(0.2,0)(0.2,0.6)
\rput{180}(0.2,0){\wobblysmall}
\end{pspicture}$
only appear in cases where the result is already zero, so no extra rule is needed to account for these cases.

The next proposition shows that for $\beta$ formal, a change of basis recasts the Gram matrix in a block-diagonal form. We then use this to compute the Gram determinant explicitly. Upon specialisation, the change of basis 
used in the proof fails at certain roots of unity. If the Gram matrix is well-defined at a given root of unity $\beta = \beta_c$, its determinant is obtained using continuity by taking the limit of the generic expression as $\beta$ tends to $\beta_c$. If the bilinear form is identically zero or undefined at $\beta_c$, then the determinant of the renormalised form $\langle\!\langle v|w\rangle\!\rangle^{\textrm{\tiny$(k)$}}\!$, see \eqref{eq:underform} and \eqref{eq:overform}, is obtained by first multiplying by an appropriate power of $(\beta-\beta_c)$ before taking the limit.

Before proving these results, it is useful to see an explicit example.  Take $(n,k,d) = (4,2,2)$. Working with the modified link states $w$ on the right-hand side of \eqref{eq:B422}, obtained from the bijection \eqref{eq:bij}, the elements $\mathcal U(w)$ of the new basis are diagrammatically given by including a \WJ{} projector that only acts on the defects:
\begin{equation}
\psset{unit=0.8}
\begin{pspicture}[shift=-0.15](0,-0.1)(1.6,0.7)
\psline{-}(0,0)(1.6,0)
\psarc[linecolor=blue,linewidth=1.5pt]{-}(0.4,0){0.2}{0}{180}
\psarc[linecolor=blue,linewidth=1.5pt]{-}(1.2,0){0.2}{0}{180}
\end{pspicture} \qquad 
\begin{pspicture}[shift=-0.15](0,-0.1)(1.6,0.7)
\psline{-}(0,0)(1.6,0)
\psarc[linecolor=blue,linewidth=1.5pt]{-}(0.8,0){0.2}{0}{180}
\psbezier[linecolor=blue,linewidth=1.5pt]{-}(0.2,0)(0.2,0.5)(1.4,0.5)(1.4,0)
\end{pspicture} \qquad
\begin{pspicture}[shift=-0.15](0,-0.1)(1.6,1.2)
\psline{-}(0,0)(1.6,0)
\psarc[linecolor=blue,linewidth=1.5pt]{-}(1.2,0){0.2}{0}{180}
\psline[linecolor=blue,linewidth=1.5pt]{-}(0.2,0)(0.2,0.3)
\psline[linecolor=blue,linewidth=1.5pt]{-}(0.6,0)(0.6,0.3)
\psline[linecolor=blue,linewidth=1.5pt]{-}(0.2,0.6)(0.2,1.2)
\rput(0.6,0.6){\wobbly}
\pspolygon[fillstyle=solid,fillcolor=pink](0.1,0.3)(0.7,0.3)(0.7,0.6)(0.1,0.6)\rput(0.4,0.45){\small$_{2}$}
\end{pspicture} \qquad
\begin{pspicture}[shift=-0.15](0,-0.1)(1.6,1.2)
\psline{-}(0,0)(1.6,0)
\psarc[linecolor=blue,linewidth=1.5pt]{-}(0.8,0){0.2}{0}{180}
\psline[linecolor=blue,linewidth=1.5pt]{-}(0.2,0.6)(0.2,1.2)
\rput(1.4,0.6){\wobbly}
\psline[linecolor=blue,linewidth=1.5pt]{-}(0.2,0)(0.2,0.3)
\psline[linecolor=blue,linewidth=1.5pt]{-}(1.4,0)(1.4,0.3)
\pspolygon[fillstyle=solid,fillcolor=pink](0.1,0.3)(1.5,0.3)(1.5,0.6)(0.1,0.6)\rput(0.8,0.45){\small$_{2}$}
\end{pspicture} \qquad
\begin{pspicture}[shift=-0.15](0,-0.1)(1.6,1.2)
\psline{-}(0,0)(1.6,0)
\psarc[linecolor=blue,linewidth=1.5pt]{-}(0.4,0){0.2}{0}{180}
\psline[linecolor=blue,linewidth=1.5pt]{-}(1.0,0.6)(1.0,1.2)
\rput(1.4,0.6){\wobbly}
\pspolygon[fillstyle=solid,fillcolor=pink](0.9,0.3)(1.5,0.3)(1.5,0.6)(0.9,0.6)\rput(1.2,0.45){\small$_{2}$}
\psline[linecolor=blue,linewidth=1.5pt]{-}(1.0,0)(1.0,0.3)
\psline[linecolor=blue,linewidth=1.5pt]{-}(1.4,0)(1.4,0.3)
\end{pspicture} \qquad
\begin{pspicture}[shift=-0.15](0,-0.1)(1.6,1.2)
\psline{-}(0,0)(1.6,0)
\psline[linecolor=blue,linewidth=1.5pt]{-}(0.2,0.6)(0.2,1.2)
\psline[linecolor=blue,linewidth=1.5pt]{-}(0.6,0.6)(0.6,1.2)
\rput(1.0,0.6){\wobbly}
\rput(1.4,0.6){\wobbly}
\psline[linecolor=blue,linewidth=1.5pt]{-}(0.2,0)(0.2,0.3)
\psline[linecolor=blue,linewidth=1.5pt]{-}(0.6,0)(0.6,0.3)
\psline[linecolor=blue,linewidth=1.5pt]{-}(1.0,0)(1.0,0.3)
\psline[linecolor=blue,linewidth=1.5pt]{-}(1.4,0)(1.4,0.3)
\pspolygon[fillstyle=solid,fillcolor=pink](0.1,0.3)(1.5,0.3)(1.5,0.6)(0.1,0.6)\rput(0.8,0.45){\small$_{4}$}
\end{pspicture}
\ .
\label{eq:UB422}
\end{equation}
The first two states already appeared in $\links_{4,2}^2$, while the last four are linear combinations of link states. For example,
\begin{equation}
\begin{gathered}
\psset{unit=0.8}
\begin{pspicture}[shift=-0.15](0,-0.1)(1.6,1.2)
\psline{-}(0,0)(1.6,0)
\psarc[linecolor=blue,linewidth=1.5pt]{-}(1.2,0){0.2}{0}{180}
\psline[linecolor=blue,linewidth=1.5pt]{-}(0.2,0)(0.2,0.3)
\psline[linecolor=blue,linewidth=1.5pt]{-}(0.6,0)(0.6,0.3)
\psline[linecolor=blue,linewidth=1.5pt]{-}(0.2,0.6)(0.2,1.2)
\rput(0.6,0.6){\wobbly}
\pspolygon[fillstyle=solid,fillcolor=pink](0.1,0.3)(0.7,0.3)(0.7,0.6)(0.1,0.6)\rput(0.4,0.45){\small$_{2}$}
\end{pspicture} \ = \  
\begin{pspicture}[shift=-0.15](0,-0.1)(1.6,0.7)
\psline{-}(0,0)(1.6,0)
\psarc[linecolor=blue,linewidth=1.5pt]{-}(1.2,0){0.2}{0}{180}
\psline[linecolor=blue,linewidth=1.5pt]{-}(0.2,0)(0.2,0.6)
\rput(0.6,0){\wobbly}
\end{pspicture}
\ - \frac 1{U_1} \ 
\begin{pspicture}[shift=-0.15](0,-0.1)(1.6,0.7)
\psline{-}(0,0)(1.6,0)
\psarc[linecolor=blue,linewidth=1.5pt]{-}(0.4,0){0.2}{0}{180}
\psarc[linecolor=blue,linewidth=1.5pt]{-}(1.2,0){0.2}{0}{180}
\end{pspicture}
\, , \qquad
\begin{pspicture}[shift=-0.15](0,-0.1)(1.6,1.2)
\psline{-}(0,0)(1.6,0)
\psarc[linecolor=blue,linewidth=1.5pt]{-}(0.8,0){0.2}{0}{180}
\psline[linecolor=blue,linewidth=1.5pt]{-}(0.2,0.6)(0.2,1.2)
\rput(1.4,0.6){\wobbly}
\psline[linecolor=blue,linewidth=1.5pt]{-}(0.2,0)(0.2,0.3)
\psline[linecolor=blue,linewidth=1.5pt]{-}(1.4,0)(1.4,0.3)
\pspolygon[fillstyle=solid,fillcolor=pink](0.1,0.3)(1.5,0.3)(1.5,0.6)(0.1,0.6)\rput(0.8,0.45){\small$_{2}$}
\end{pspicture}
\ = \ 
\begin{pspicture}[shift=-0.15](0,-0.1)(1.6,0.7)
\psline{-}(0,0)(1.6,0)
\psarc[linecolor=blue,linewidth=1.5pt]{-}(0.8,0){0.2}{0}{180}
\psline[linecolor=blue,linewidth=1.5pt]{-}(0.2,0)(0.2,0.6)
\rput(1.4,0){\wobbly}
\end{pspicture}
\ - \frac 1{U_1} \ 
\begin{pspicture}[shift=-0.15](0,-0.1)(1.6,0.7)
\psline{-}(0,0)(1.6,0)
\psarc[linecolor=blue,linewidth=1.5pt]{-}(0.8,0){0.2}{0}{180}
\psbezier[linecolor=blue,linewidth=1.5pt]{-}(0.2,0)(0.2,0.6)(1.4,0.6)(1.4,0)
\end{pspicture}
\ , \qquad \\
\psset{unit=0.7}
\begin{pspicture}[shift=-0.15](0,-0.1)(1.6,1.2)
\psline{-}(0,0)(1.6,0)
\psline[linecolor=blue,linewidth=1.5pt]{-}(0.2,0.6)(0.2,1.2)
\psline[linecolor=blue,linewidth=1.5pt]{-}(0.6,0.6)(0.6,1.2)
\rput(1.0,0.6){\wobbly}
\rput(1.4,0.6){\wobbly}
\psline[linecolor=blue,linewidth=1.5pt]{-}(0.2,0)(0.2,0.3)
\psline[linecolor=blue,linewidth=1.5pt]{-}(0.6,0)(0.6,0.3)
\psline[linecolor=blue,linewidth=1.5pt]{-}(1.0,0)(1.0,0.3)
\psline[linecolor=blue,linewidth=1.5pt]{-}(1.4,0)(1.4,0.3)
\pspolygon[fillstyle=solid,fillcolor=pink](0.1,0.3)(1.5,0.3)(1.5,0.6)(0.1,0.6)\rput(0.8,0.45){\small$_{4}$}
\end{pspicture} \ = \
\begin{pspicture}[shift=-0.15](0,-0.1)(1.6,0.7)
\psline{-}(0,0)(1.6,0)
\psline[linecolor=blue,linewidth=1.5pt]{-}(0.2,0)(0.2,0.6)
\psline[linecolor=blue,linewidth=1.5pt]{-}(0.6,0)(0.6,0.6)
\rput(1.0,0){\wobbly}
\rput(1.4,0){\wobbly}
\end{pspicture} \ + \frac{U_1}{U_3} \Big(\ 
\begin{pspicture}[shift=-0.15](0,-0.1)(1.6,0.7)
\psline{-}(0,0)(1.6,0)
\psarc[linecolor=blue,linewidth=1.5pt]{-}(0.4,0){0.2}{0}{180}
\psline[linecolor=blue,linewidth=1.5pt]{-}(1.0,0)(1.0,0.6)
\rput(1.4,0){\wobbly}
\end{pspicture} 
\ - U_1 \ 
\begin{pspicture}[shift=-0.15](0,-0.1)(1.6,0.7)
\psline{-}(0,0)(1.6,0)
\psarc[linecolor=blue,linewidth=1.5pt]{-}(0.8,0){0.2}{0}{180}
\psline[linecolor=blue,linewidth=1.5pt]{-}(0.2,0)(0.2,0.6)
\rput(1.4,0){\wobbly}
\end{pspicture}
 \ + \ 
\begin{pspicture}[shift=-0.15](0,-0.1)(1.6,0.7)
\psline{-}(0,0)(1.6,0)
\psarc[linecolor=blue,linewidth=1.5pt]{-}(1.2,0){0.2}{0}{180}
\psline[linecolor=blue,linewidth=1.5pt]{-}(0.2,0)(0.2,0.6)
\rput(0.6,0){\wobbly}
\end{pspicture}
 \ \Big) + \frac{U_1}{U_2 U_3} \Big( \ 
\begin{pspicture}[shift=-0.15](0,-0.1)(1.6,0.7)
\psline{-}(0,0)(1.6,0)
\psarc[linecolor=blue,linewidth=1.5pt]{-}(0.8,0){0.2}{0}{180}
\psbezier[linecolor=blue,linewidth=1.5pt]{-}(0.2,0)(0.2,0.6)(1.4,0.6)(1.4,0)
\end{pspicture} \ - U_1 \ 
\begin{pspicture}[shift=-0.15](0,-0.1)(1.6,0.7)
\psline{-}(0,0)(1.6,0)
\psarc[linecolor=blue,linewidth=1.5pt]{-}(0.4,0){0.2}{0}{180}
\psarc[linecolor=blue,linewidth=1.5pt]{-}(1.2,0){0.2}{0}{180}
\end{pspicture}
\ \Big).
\end{gathered}
\end{equation}
If the elements $\mathcal U(w)$ are ordered so that the number of links is weakly decreasing, as in \eqref{eq:UB422}, then the matrix of this change of basis becomes upper-triangular, with 
ones on the diagonal. Its determinant is therefore $1$. Moreover, we find that \eqref{eq:Gmat422} becomes block-diagonal in this basis:
\begin{equation}
\mathcal U^T \grammat_{4,2}^2 \, \mathcal U = 
\begin{pmatrix} 
1 \cdot \begin{pmatrix} (U_1)^2 & U_1 \\ U_1 & (U_1)^2\end{pmatrix} & & \\
& {\displaystyle\frac{U_3}{(U_1)^2}}\cdot \begin{pmatrix} U_1 & 1 & 0 \\ 1 & U_1 & 1 \\ 0 & 1 & U_1\end{pmatrix} &  \\
& & \displaystyle\frac {U_4}{U_2}\cdot  \big(\,1\,\big)
\end{pmatrix}
,
\end{equation}
where we use the same symbol, $\mathcal U$, for the map and its matrix realisation.
Up to the constant prefactors, the three blocks may be recognised as the $\tl_{n}$ Gram matrices $\grammat_{4,0}^0$, $\grammat_{4,0}^2$ and $\grammat_{4,0}^4$.
\begin{Proposition}\label{prop:BlockDiagGram}
For $\beta$ formal, there exists a (change of basis) matrix
$\mathcal U$, with $\det \, \mathcal U = 1$, for which
\begin{equation}
\mathcal U^{T} \grammat_{n,k}^d \mathcal U = \sideset{}{'}\bigoplus_{e=\abs{k-d}}^{\min(k+d,n)} A_k^{\frac{1}{2} (e+d-k), \frac{1}{2} (e-d+k)} \grammat_{n,0}^e,
\end{equation}
where the prime indicates that $e$ increases in steps of two and the $A_k^{i,j}$ are functions that multiply each block, given by
\begin{equation} \label{eq:DefAConstants}
A_{k}^{i,j} = \frac{U_{k+i,j-1}}{U_{i+j-1,j-1}\,U_{k-1,j-1}}, \qquad U_{k,r} = \frac{U_k!}{U_{k-r-1}!\,U_r!}, \qquad U_k! = \prod_{i=1}^{k}U_i,
\end{equation}
with the conventions $U_{-1}! = U_0! = 1$.
\label{sec:gramdecomp}
\end{Proposition}
\begin{proof}
As in the example above, we use the bijection \eqref{eq:bij} and order the modified $\tl_n$ link states so that the number of links weakly decreases.
The link states with $|k-d|$ defects appear unchanged in the new basis, while those in $\vspn \links _n^e$ with $e>|k-d|$ are replaced by linear combinations.
If the modified link state $w$ has $e$ defects, we define $\mathcal{U} (w)$ to be the linear combination of link states obtained by letting the \WJ{} projector $P_e$ act on the defects from above.
Since $P_e$ is the identity plus terms which will increase the number of links, we conclude that the 
matrix $\mathcal U$ is upper-triangular with ones on the diagonal, so 
its determinant is $1$.

In this new basis, the Gram matrix $\mathcal U^T \grammat_{n,k}^d\, \mathcal U$ is block diagonal. Indeed, if $w$ and $w'$ don't have the same number of defects, 
then in the diagram describing the product of (the equivalence classes of) $\mathcal{U} (w)$ and $\mathcal{U} (w')$, two defects must be linked.  This could lead to a non-zero result because the 
rules \eqref{eq:krule} allow defects of different types to be linked.  However, this link will then touch a projector at two nodes, so the result must vanish. If $w$ and $w'$ have the same number of defects,
then they each have the same number $i$ of bulk defects and the same number $j$ of boundary defects, by \eqref{eq:ConstrainBulkBdryDefects}.  A similar argument as before now shows that we have
\begin{equation}
\gramprodk{\,\mathcal U(w)\,}{\,\mathcal U(w')} = A_k^{i,j} \gramprod{w}{w'}^{_{(0)}}_{\phantom i}\qquad \text{(\(w,w' \in \sideset{}{'}\bigcup_{e=\abs{k-d}}^{\min (k+d,n)}\hspace{-0.2cm} \links_n^e\)),}
\end{equation}
where the indices of $\sideset{}{'}\bigcup$ increase by steps of two and the constants $A_k^{i,j}$ are defined by
\begin{equation}
A_k^{i,j} = 
\psset{unit=0.8}
\, \, 
\begin{pspicture}[shift=-0.8](0,-1.1)(2.8,1.1)
\rput(0.6,1.1){$\overbrace{\ \quad \ }^i$}
\rput(2.0,1.1){$\overbrace{\ \qquad \ }^j$}
\psline[linecolor=blue,linewidth=1.5pt]{-}(0.2,-1)(0.2,0.6)
\rput(0.6,.3){...}\rput(0.6,-.7){...}
\psline[linecolor=blue,linewidth=1.5pt]{-}(1.0,-1)(1.0,0.6)
\rput(1.4,0){\wobbly}
\rput(1.8,0){\wobbly}
\rput(2.2,.3){...}
\rput(2.6,0){\wobbly}
\rput{180}(1.4,-0.4){\wobbly}
\rput{180}(1.8,-0.4){\wobbly}
\rput(2.2,-.7){...}
\rput{180}(2.6,-0.4){\wobbly}
\pspolygon[fillstyle=solid,fillcolor=pink](0,0)(2.8,0)(2.8,-0.4)(0,-0.4)(0,0)\rput(1.4,-0.2){$_{i+j}$}
\end{pspicture}
\ .
\end{equation}
To evaluate the $A_k^{i,j}$, we use \eqref{eq:otherWJrelation} and its adjoint 
(obtained by reversing the order of multiplication), along with the rules 
\eqref{eq:krule}. This leads to the following relations which determine the constants uniquely:
\begin{equation}
A_k^{i,j} = A_k^{i-1,j} - \frac{(U_{j-1})^2}{U_{i+j-1}U_{i+j-2}}A_k^{i-1,j-1}, \qquad A_{k}^{i,0} = 1, \qquad A_{k}^{0,j} = \frac{U_{k}}{U_{k-j}}.
\end{equation}
The solution is \eqref{eq:DefAConstants} which completes the proof.
\end{proof}

Let us note that the coefficients $A_k^{i,j}$ appearing in \cref{prop:BlockDiagGram} may be nicely written in terms of $q$-binomials (recall that $\beta = q+q^{-1}$)
\begin{equation} \left[\begin{matrix}m\\ n\end{matrix}\right]_q = \frac{[m]_q!}{[m-n]_q![n]_q!},\qquad [n]_q! = \prod_{k=1}^n[k]_q, \qquad [n]_q = \frac{q^n - q^{-n}}{q-q^{-1}}, \qquad [0]_q! \equiv 1,
\end{equation}
namely
\begin{equation}
A_k^{i,j} = \frac{\left[\begin{matrix}k+i+1\\ j\end{matrix}\right]_q}{\left[\begin{matrix}i+j\\ j\end{matrix}\right]_q\left[\begin{matrix}k\\ j\end{matrix}\right]_q}.
\end{equation}
They may also be expressed in terms of the theta nets of Kauffman and Lins \cite{KauTem94}:
\begin{equation}
A_k^{i,j} = \frac{\text{Net}(i,j,k-j)}{\text{Net}(i,0,k-j)} = \frac{\theta(i-j+k,i+j,k)}{\theta(i-j+k,i,k-j)}.
\end{equation}
We recall that theta nets are built from \WJ{} projectors, suggesting that this expression might have an elegant derivation in the setting of \TL{} skein theory.

\begin{Proposition}
The determinant of the invariant bilinear form $\gramprodk{\cdot}{\cdot}$, in the basis $\links_{n,k}^d$, is given by
\begin{equation}
\det \grammat_{n,k}^d = \prod_{i=1}^{\lfloor \frac k2 \rfloor} \Big(\frac{U_{i-1}}{U_{k-i}}\Big)^{\dim \stan^d_{n,k-2i}} 
\prod_{j=1}^{\frac{n+k-d}2} \Big(\frac{U_{d+j}}{U_{j-1}}\Big)^{\dim \stan^{d+2j}_{n,k}}. 
\label{eq:Gramfinalform}
\end{equation}
\label{sec:Gramdet}
\end{Proposition}
\begin{proof} From \cref{sec:gramdecomp}, the Gram determinants satisfy the relation
\begin{equation}
\det \grammat_{n,k}^d = \sideset{}{'} \prod_{e = \abs{k-d}}^{\min(k+d,n)} \Big( \det \grammat_{n,0}^e \cdot \big(A_k^{\frac{1}{2} (e+d-k), \frac{1}{2} (e-d+k)} \big)^{\dim \stan_{n,0}^e} \Big).
\label{eq:prodofgram}
\end{equation}
Closed expressions for $\det \grammat_{n,0}^e$ are known, see \eqref{eq:GramdetTL}, so \eqref{eq:prodofgram} 
already constitutes a closed form for $\grammat_{n,k}^d$. 
After a lengthy but straightforward computation
using \eqref{eq:dimVnkd} and its consequent recursion relation
\begin{equation}
\dim \stan_{n,k}^d = \dim \stan_{n-1,k}^{d-1} + \dim \stan_{n-1,k}^{d+1}, \qquad  \dim \stan_{n,k}^{-1} = 0,
\label{eq:stanwreck}
\end{equation}
we find that \eqref{eq:prodofgram} translates into a recursion relation for the Gram determinants:
\begin{equation}
\det \grammat_{n,k}^d = \Big(\frac{U_{d+1}}{U_d}\Big)^{\dim \stan_{n-1,k}^{d+1}}\det \grammat_{n-1,k}^{d-1} \det \grammat_{n-1,k}^{d+1}, \qquad \det \grammat_{n,k}^{n+k} = 1,  \qquad \det \grammat_{n,k}^{-1} = 1. 
\label{eq:Gramwreck}
\end{equation}
We note that the conventions $\dim \stan_{n,k}^{-1} = 0$ and $\det \grammat_{n,k}^{-1} = 1$ ensure that the recursion relations for $\dim \stan_{n,k}^d$ and $\det \grammat_{n,k}^d$ also apply to the case $d=0$.
For $k = 0$, the relation \eqref{eq:Gramwreck}
was computed in \cite{RidSta12} to calculate $\det \grammat_{n,0}^d$. The final form \eqref{eq:Gramfinalform} is 
seen to satisfy both the recursion relation \eqref{eq:Gramwreck} and the two boundary conditions. 
\end{proof}

%
\section{Virasoro representation theory} \label{sec:Back}
%

In this appendix, we review the structure of certain classes of modules over the Virasoro algebra. Further details may be found in \cite{KacBom88,IohRep11}.  The Virasoro algebra is the infinite-dimensional complex Lie algebra spanned by the modes $L_n$, where $n \in \ZZ$, and $C$, subject to the commutation relations
\begin{equation} \label{eq:VirComm}
\comm{L_m}{L_n} = \brac{m-n} L_{m+n} + \frac{m^3-m}{12} \delta_{m+n=0} C.
\end{equation}
The element $C$ is central and will be assumed to act on all modules as a given multiple $c$ of the identity operator called the central charge.  Formally, we are therefore considering modules over the quotient of the \uea{} of the Virasoro algebra by the two-sided ideal generated by $C - c \: \wun$.

\subsection{Highest-weight modules}

A \hws{} $\ket{\Delta}$ for the Virasoro algebra satisfies
\begin{equation}
L_n \ket{\Delta} = 0 \quad \text{for \(n>0\);} \qquad L_0 \ket{\Delta} = \Delta \ket{\Delta}.
\end{equation}
Such states are characterised by their conformal dimension ($L_0$-eigenvalue) $\Delta$ and central charge $c$.  
We generally omit explicit reference to the latter in notation, regarding it as fixed for the model under consideration.  
As usual, a \hws{} $\ket{\Delta}$ generates a Verma module $\Ver{\Delta}$ through the free action of the $L_n$ with 
$n<0$ and this module has a unique irreducible quotient that we will denote by $\Irr{\Delta}$.

With the standard parametrisation 
\begin{equation}
c = 13 - 6 \brac{t+t^{-1}}, \qquad 
\Delta_{r,s} = \frac{r^2-1}{4} t^{-1} - \frac{rs-1}{2} + \frac{s^2-1}{4} t \qquad \text{(\(t \in \CC \setminus \set{0}\)),}
\end{equation}
the Verma module $\Ver{r,s} \equiv \Ver{\Delta_{r,s}}$ may be shown to be reducible precisely when $r$ and $s$ are positive integers.  If $t$ is rational, then this parametrisation may be written in the form
\begin{equation} \label{eq:ParByt}
t = \frac{p}{p'}, \qquad 
c = 1 - \frac{6 \brac{p'-p}^2}{pp'}, \qquad 
\Delta_{r,s} = \frac{\brac{p'r-ps}^2 - \brac{p'-p}^2}{4pp'},
\end{equation}
where one customarily takes $\gcd \set{p, p'} = 1$.  The Virasoro minimal models $\MinMod{p}{p'}$ are built, for $p,p' \ge 2$, from the irreducible \hwms{} $\Irr{r,s} \equiv \Irr{\Delta_{r,s}}$ with $1 \le r \le p-1$ and $1 \le s \le p'-1$.  Our focus is on the logarithmic counterparts of these minimal models, so we will restrict ourselves to $t=p/p'$ rational and positive, taking $p' \ge p > 0$ (and hence $c \le 1$).  However, we shall not insist that $p,p' \ge 2$.

It turns out that the submodules of a Virasoro Verma module are always generated by \svs{} and that the maximal dimension of the space of \svs{} of any given conformal dimension is one.  The submodule structure of a Verma module $\Ver{\Delta}$ then reduces to determining its \svs{}.  The possible \sv{} structures for the Verma modules $\Ver{\Delta}$, with $t \in \QQ_+$, 
are illustrated diagrammatically in \cref{fig:VermaStructures}.  As mentioned above, if $\Delta \neq \Delta_{r,s}$ for any integers $r$ and $s$, then $\Ver{\Delta}$ is irreducible.\footnote{In general, this is only true for positive integers $r$ and $s$.  However, $t \in \QQ_+$ implies the symmetry $\Delta_{r,s} = \Delta_{r+p,s+p'}$ which shows that the set of $\Delta_{r,s}$ with $r,s \in \ZZ_+$ coincides with the set of $\Delta_{r,s}$ with $r,s \in \ZZ$.}  This is represented by the point case in the figure.  If $\Delta = \Delta_{r,s}$, where $r$ is a multiple of $p$ or $s$ is a multiple of $p'$, then the \sv{} structure of $\Ver{r,s}$ is represented by the infinite chain of \cref{fig:VermaStructures}.  For all other choices of $r,s \in \ZZ_+$, the structure of $\Ver{r,s}$ is represented by a braided pattern.  In both the chain and braid cases, a \sv{} is always present at grade $rs$, though this need not be the lowest (positive) grade at which a \sv{} appears.\footnote{If $M$ is a Virasoro module for which the conformal dimensions of the states are bounded below, we define the grade of a (generalised) $L_0$-eigenvector in $M$ to be the difference between its conformal dimension and the minimal conformal dimension among states of $M$.}
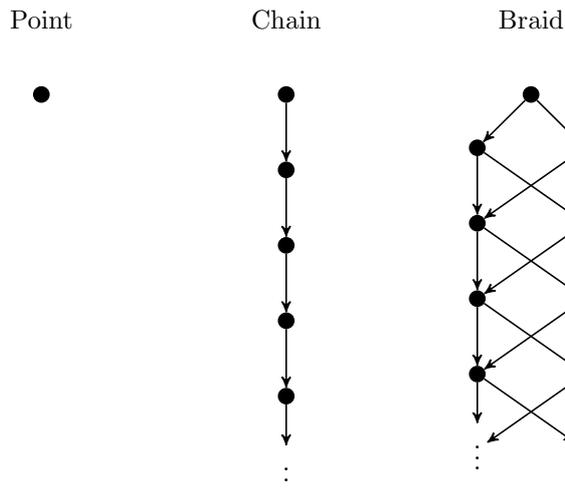
\begin{figure}
\begin{center}
\begin{tikzpicture}
  [->,node distance=1cm,>=stealth',semithick,scale=0.7,
   sv/.style={circle,draw=black,fill=black,inner sep = 2pt,minimum size=5pt}
  ]
  \node[sv] (1) [] {};
  \node[] (point) [above of =1] {Point};
  \node[sv] (3) [right = 3cm of 1] {};
  \node[] (chain) [above of =3] {Chain};
  \node[sv] (3a) [below of =3] {};
  \node[sv] (3b) [below of =3a] {};
  \node[sv] (3c) [below of =3b] {};
  \node[sv] (3d) [below of =3c] {};
  \node[inner sep = 2pt] (3e) [below of =3d] {$\vdots$};
  \path[] (3) edge node {} (3a)
          (3a) edge node {} (3b)
          (3b) edge node {} (3c)
          (3c) edge node {} (3d)
          (3d) edge node {} (3e);
  \node[sv] (5) [right = 3cm of 3] {};
  \node[] (braid) [above of =5] {Braid};
  \node[sv] (5a) [below left of =5] {};
  \node[sv] (5b) [below of =5a] {};
  \node[sv] (5c) [below of =5b] {};
  \node[sv] (5d) [below of =5c] {};
  \node[inner sep = 2pt] (5e) [below of =5d] {$\vdots$};
  \node[sv] (5j) [below right of =5] {};
  \node[sv] (5k) [below of =5j] {};
  \node[sv] (5l) [below of =5k] {};
  \node[sv] (5m) [below of =5l] {};
  \node[inner sep = 2pt] (5n) [below of =5m] {$\vdots$};
  \path[] (5) edge node {} (5a)
          (5) edge node {} (5j)
          (5a) edge node {} (5b)
          (5a) edge node {} (5k)
          (5b) edge node {} (5c)
          (5b) edge node {} (5l)
          (5c) edge node {} (5d)
          (5c) edge node {} (5m)
          (5d) edge node {} (5e)
          (5d) edge node {} (5n)
          (5j) edge node {} (5b)
          (5j) edge node {} (5k)
          (5k) edge node {} (5c)
          (5k) edge node {} (5l)	
          (5l) edge node {} (5d)
          (5l) edge node {} (5m)
          (5m) edge node {} (5e)
          (5m) edge node {} (5n);
\end{tikzpicture}
\end{center}
\caption{The \sv{} structure, marked by black circles, of Virasoro Verma modules for $t \in \QQ_+$.  
Arrows from one \sv{} to another indicate that the latter may be obtained from the former by acting with Virasoro modes.  The chain and braid structures have infinitely many \svs{}. The conformal dimensions of the \svs{} increase as one moves down.} \label{fig:VermaStructures}
\end{figure}

Given a central charge $c$, the conformal dimensions $\Delta_{r,s}$ signalling reducible Verma modules are conveniently summarised in an extended Kac table.  To make contact with the above structural distinctions and those of the next subsection, we will partition the extended Kac table into three subsets as follows:
\begin{itemize}
\item If $p$ divides $r$ and $p'$ divides $s$, then we say that $(r,s)$ is of \emph{corner} type in the extended Kac table.
\item If $p$ divides $r$ or $p'$ divides $s$, but not both, then $(r,s)$ is said to be of \emph{boundary} type.
\item If $p$ does not divide $r$ and $p'$ does not divide $s$, then $(r,s)$ is said to be of \emph{interior} type.
\end{itemize}
Summarising, corner and boundary type Verma modules have singular vectors arranged in chains whereas interior type Verma modules have a braided pattern of singular vectors.  We remark that if $p=1$ or $p'=1$, then there are no interior entries in the extended Kac table, and if $p=p'=1$, then there will be no boundary entries either.  We illustrate this with the extended Kac tables corresponding to $(p,p') = (1,2)$ and $(2,3)$ in \cref{fig:KacTables}.

{
\renewcommand{\arraystretch}{1.1} 
\begin{figure}
\begin{center}
\begin{tikzpicture} 
\node (Kac2)
{
\setlength{\extrarowheight}{4pt}
\begin{tabular}{|C|C|C|C|C|C|C|C|C|C|C|C|C}
\hline
\BKL 0 & -\frac{1}{8} & \BKL 0 & \frac{3}{8} & \BKL 1 & \frac{15}{8} & \BKL 3 & \frac{35}{8} & \BKL 6 & \frac{63}{8} & \BKL 10 & \frac{99}{8} & \BKL \cdots \\[1mm]
\hline
\BKL 1 & \frac{3}{8} & \BKL 0 & -\frac{1}{8} & \BKL 0 & \frac{3}{8} & \BKL 1 & \frac{15}{8} & \BKL 3 & \frac{35}{8} & \BKL 6 & \frac{63}{8} & \BKL \cdots \\[1mm]
\hline
\BKL 3 & \frac{15}{8} & \BKL 1 & \frac{3}{8} & \BKL 0 & -\frac{1}{8} & \BKL 0 & \frac{3}{8} & \BKL 1 & \frac{15}{8} & \BKL 3 & \frac{35}{8} & \BKL \cdots \\[1mm]
\hline
\BKL 6 & \frac{35}{8} & \BKL 3 & \frac{15}{8} & \BKL 1 & \frac{3}{8} & \BKL 0 & -\frac{1}{8} & \BKL 0 & \frac{3}{8} & \BKL 1 & \frac{15}{8} & \BKL \cdots \\[1mm]
\hline
\BKL 10 & \frac{63}{8} & \BKL 6 & \frac{35}{8} & \BKL 3 & \frac{15}{8} & \BKL 1 & \frac{3}{8} & \BKL 0 & -\frac{1}{8} & \BKL 0 & \frac{3}{8} & \BKL \cdots \\[1mm]
\hline
\BKL 15 & \frac{99}{8} & \BKL 10 & \frac{63}{8} & \BKL 6 & \frac{35}{8} & \BKL 3 & \frac{15}{8} & \BKL 1 & \frac{3}{8} & \BKL 0 & -\frac{1}{8} & \BKL \cdots \\[1mm]
\hline
\BKL \vdots & \vdots & \BKL \vdots & \vdots & \BKL \vdots & \vdots & \BKL \vdots & \vdots & \BKL \vdots & \vdots & \BKL \vdots & \vdots & \BKL \ddots
\end{tabular}
};

\node [below=3mm of Kac2] {$t = \dfrac{1}{2}, \qquad c = -2.$};

\node (Kac3) at (0,0) [below=20mm of Kac2]
{
\setlength{\extrarowheight}{4pt}
\begin{tabular}{|CC|C|CC|C|CC|C|CC|C|C}
\hline
\IKL 0 & \IKL 0 & \BKL \frac{1}{3} & \IKL 1 & \IKL 2 & \BKL \frac{10}{3} & \IKL 5 & \IKL 7 & \BKL \frac{28}{3} & \IKL 12 & \IKL 15 & \BKL \frac{55}{3} & \IKL \cdots \\[1mm]
\hline
\BKL \frac{5}{8} & \BKL \frac{1}{8} & -\frac{1}{24} & \BKL \frac{1}{8} & \BKL \frac{5}{8} & \frac{35}{24} & \BKL \frac{21}{8} & \BKL \frac{33}{8} & \frac{143}{24} & \BKL \frac{65}{8} & \BKL \frac{85}{8} & \frac{323}{24} & \BKL \cdots \\[1mm]
\hline
\IKL 2 & \IKL 1 & \BKL \frac{1}{3} & \IKL 0 & \IKL 0 & \BKL \frac{1}{3} & \IKL 1 & \IKL 2 & \BKL \frac{10}{3} & \IKL 5 & \IKL 7 & \BKL \frac{28}{3} & \IKL \cdots \\[1mm]
\hline
\BKL \frac{33}{8} & \BKL \frac{21}{8} & \frac{35}{24} & \BKL \frac{5}{8} & \BKL \frac{1}{8} & -\frac{1}{24} & \BKL \frac{1}{8} & \BKL \frac{5}{8} & \frac{35}{24} & \BKL \frac{21}{8} & \BKL \frac{33}{8} & \frac{143}{24} & \BKL \cdots \\[1mm]
\hline
\IKL 7 & \IKL 5 & \BKL \frac{10}{3} & \IKL 2 & \IKL 1 & \BKL \frac{1}{3} & \IKL 0 & \IKL 0 & \BKL \frac{1}{3} & \IKL 1 & \IKL 2 & \BKL \frac{10}{3} & \IKL \cdots \\[1mm]
\hline
\BKL \frac{85}{8} & \BKL \frac{65}{8} & \frac{143}{24} & \BKL \frac{33}{8} & \BKL \frac{21}{8} & \frac{35}{24} & \BKL \frac{5}{8} & \BKL \frac{1}{8} & -\frac{1}{24} & \BKL \frac{1}{8} & \BKL \frac{5}{8} & \frac{35}{24} & \BKL \cdots \\[1mm]
\hline
\IKL \vdots & \IKL \vdots & \BKL \vdots & \IKL \vdots & \IKL \vdots & \BKL \vdots & \IKL \vdots & \IKL \vdots & \BKL \vdots & \IKL \vdots & \IKL \vdots & \BKL \vdots & \IKL \ddots
\end{tabular}
};
\node [below=3mm of Kac3] {$t = \dfrac{2}{3}, \qquad c = 0.$};
\end{tikzpicture}
\caption{A part of the extended Kac table for $c=-2$ ($p=1$, $p'=2$) and $c=0$ ($p=2$, $p'=3$). The rows of the table 
are labelled by $r = 1, 2, 3, \ldots$ and the columns by $s = 1, 2, 3, \ldots$\,. Interior entries
are shaded dark blue, boundary entries are shaded light blue, while corner entries are white.} \label{fig:KacTables}
\end{center}
\end{figure}
}

\subsection{Feigin-Fuchs modules}\label{sec:FFm}

An important class of modules that are not \hw{} are the \FFms{} that arise in
the Coulomb gas free field realisation of the Virasoro algebra.  Their structures were first studied in \cite{FeiSke82} and 
we refer to \cite{IohRep11} for details omitted here.  A general, but briefer, summary of \FFm{} structures appears 
in \cite[App.~A]{RidMod14}.

In this realisation, one starts with the infinite-dimensional complex Lie algebra called the Heisenberg algebra.  This algebra is spanned by modes $a_n$, where $n \in \ZZ$, and a central element $\wun$ which is generally assumed to act on all modules as the identity.  We therefore write the commutation relations as
\begin{equation}
\comm{a_m}{a_n} = m \delta_{m+n=0} \wun.
\end{equation}
The (\uea{} of the) Heisenberg algebra carries a one-parameter family of copies of the (\uea{} of the) Virasoro algebra:
\begin{equation}
L_n = \frac{1}{2} \sum_{j \in \ZZ} a_j a_{n-j} - \frac{1}{2} \brac{n+1} Q a_n \quad \text{for \(n \neq 0\);} \qquad 
L_0 = \frac{1}{2} a_0^2 + \sum_{j=1}^{\infty} a_{-j} a_j - \frac{1}{2} Q a_0.
\end{equation}
The parameter $Q$ then determines the central charge $c=1-3Q^2$ of the corresponding Virasoro algebra.  We choose $Q$ so as to reproduce the central charge of \eqref{eq:ParByt}, namely
\begin{equation} \label{eq:DefQ}
Q = \sqrt{\frac{2p'}{p \vphantom{p'}}} - \sqrt{\frac{2p \vphantom{p'}}{p'}} = \sqrt{\frac{2 \vphantom{p'}}{pp'}} \brac{p'-p}.
\end{equation}

Heisenberg \hwss{} $\ket{\lambda}$ are eigenvectors of $a_0$, parametrised by their eigenvalues $\lambda$:
\begin{equation}
a_n \ket{\lambda} = 0 \quad \text{for \(n>0\);} \qquad 
a_0 \ket{\lambda} = \lambda \ket{\lambda}.
\end{equation}
The conformal dimension of $\ket{\lambda}$ is then
\begin{equation}
\Delta_{\lambda} = \frac{1}{2} \lambda \brac{\lambda - Q} = \frac{\brac{\lambda - Q/2}^2 - Q^2/4}{2} = \frac{2pp' \brac{\lambda - Q/2}^2 - \brac{p'-p}^2}{4pp'}.
\end{equation}
Comparing with \eqref{eq:ParByt}, we see that $\Delta_{\lambda}$ will coincide with the $\Delta_{r,s}$ of the extended Kac table when
\begin{equation} \label{eq:DefLambdaRS}
\lambda = \lambda_{r,s} \equiv -\alpha' \brac{r-1} + \alpha \brac{s-1},
\end{equation}
where we introduce
\begin{equation} \label{eq:DefAlphas}
\alpha = \sqrt{\frac{p \vphantom{p'}}{2p'}}, \qquad \alpha' = \sqrt{\frac{p'}{2p}}
\end{equation}
for convenience.  We note the symmetries
\begin{equation} \label{eq:FFSymm}
\lambda_{r+p,s} = \lambda_{r,s} - \sqrt{\frac{pp'}{2}}, \quad 
\lambda_{p,s+p'} = \lambda_{r,s} + \sqrt{\frac{pp'}{2}} \qquad \Ra \qquad
\lambda_{r+p,s+p'} = \lambda_{r,s}.
\end{equation}
Since $\Delta_{\lambda} = \Delta_{Q - \lambda}$, it follows that there are two distinct $a_0$-eigenvalues giving rise to the same conformal dimension (unless $\lambda = \lambda_{0,0} = Q / 2$).  In particular, $\lambda_{r,s}$ and $\lambda_{-r,-s}$ give rise to the same conformal dimension.

The Heisenberg Verma module $\FF{\lambda}$ generated from $\ket{\lambda}$ is often referred to as a Fock space.  When regarding the $\FF{\lambda}$ as Virasoro modules, we will refer to them as \FFms{}.  In contrast to Virasoro Verma modules, Fock spaces over the Heisenberg algebra are always irreducible.  However, their structure as Virasoro modules is more interesting --- see \cref{fig:FeiginFuchsStructures}.  In particular, their Virasoro submodules need not be generated by \svs{}, but rather one needs \ssvs{} as well.  These are vectors belonging to a given module which become singular in an appropriate quotient module.  The space of \ssvs{} at any grade has dimension greater than $1$ in general.  However, dimensions greater than $1$ arise solely because there are many states which are in the kernel of the projection onto the appropriate quotient.  We will often speak of the \ssv{} at a given grade, understanding that it is only specified up to this kernel (and normalisation).

\begin{figure}
\begin{center}
\begin{tikzpicture}
  [->,node distance=1cm,>=stealth',semithick,scale=0.7,
   asoc/.style={circle,draw=black,fill=black,inner sep = 2pt,minimum size=5pt},
   bsoc/.style={circle,draw=black,fill=gray,inner sep = 2pt,minimum size=5pt},
   csoc/.style={circle,draw=black,fill=white,inner sep = 2pt,minimum size=5pt}
  ]
  \node[asoc] (1) [] {};
  \node[] (point) [above of =1] {Point};
  \node[asoc] (2) [right = 2.5cm of 1] {};
  \node[] (island) [above of =2] {Islands};
  \node[asoc] (2a) [below of =2] {};
  \node[asoc] (2b) [below of =2a] {};
  \node[asoc] (2c) [below of =2b] {};
  \node[inner sep = 2pt] (2d) [below of =2c] {$\vdots$};
  \node[bsoc] (3) [right = 2.5cm of 2] {};
  \node[] (chain) [right = 0.75cm of 3,above of =3] {Chain};
  \node[asoc] (3a) [below of =3] {};
  \node[bsoc] (3b) [below of =3a] {};
  \node[asoc] (3c) [below of =3b] {};
  \node[bsoc] (3d) [below of =3c] {};
  \node[inner sep = 2pt] (3e) [below of =3d] {$\vdots$};
  \path[] (3) edge node {} (3a)
          (3b) edge node {} (3a)
          (3b) edge node {} (3c)
          (3d) edge node {} (3c)
          (3d) edge node {} (3e);
  \node[asoc] (4) [right = 1.5cm of 3] {};
  \node[bsoc] (4a) [below of =4] {};
  \node[asoc] (4b) [below of =4a] {};
  \node[bsoc] (4c) [below of =4b] {};
  \node[asoc] (4d) [below of =4c] {};
  \node[inner sep = 2pt] (4e) [below of =4d] {$\vdots$};
  \path[] (4a) edge node {} (4)
          (4a) edge node {} (4b)
          (4c) edge node {} (4b)
          (4c) edge node {} (4d)
          (4e) edge node {} (4d);
  \node[bsoc] (5) [right = 2.5cm of 4] {};
    \node[] (braid) [right = 1cm of 5,above of =5] {Braid};
  \node[csoc] (5a) [below left of =5] {};
  \node[bsoc] (5b) [below of =5a] {};
  \node[csoc] (5c) [below of =5b] {};
  \node[bsoc] (5d) [below of =5c] {};
  \node[inner sep = 2pt] (5e) [below of =5d] {$\vdots$};
  \node[asoc] (5j) [below right of =5] {};
  \node[bsoc] (5k) [below of =5j] {};
  \node[asoc] (5l) [below of =5k] {};
  \node[bsoc] (5m) [below of =5l] {};
  \node[inner sep = 2pt] (5n) [below of =5m] {$\vdots$};
  \path[] (5) edge node {} (5j)
          (5k) edge node {} (5j)
          (5b) edge node {} (5j)
          (5k) edge node {} (5l)
          (5a) edge node {} (5k)
          (5c) edge node {} (5k)
          (5m) edge node {} (5l)
          (5b) edge node {} (5l)
          (5d) edge node {} (5l)
          (5m) edge node {} (5n)
          (5c) edge node {} (5m)
          (5e) edge node {} (5m)
          (5d) edge node {} (5n)
          (5a) edge node {} (5)
          (5a) edge node {} (5b)
          (5c) edge node {} (5b)
          (5c) edge node {} (5d)
          (5e) edge node {} (5d);
  \node[bsoc] (6) [right = 2cm of 5] {};
  \node[asoc] (6a) [below left of =6] {};
  \node[bsoc] (6b) [below of =6a] {};
  \node[asoc] (6c) [below of =6b] {};
  \node[bsoc] (6d) [below of =6c] {};
  \node[inner sep = 2pt] (6e) [below of =6d] {$\vdots$};
  \node[csoc] (6j) [below right of =6] {};
  \node[bsoc] (6k) [below of =6j] {};
  \node[csoc] (6l) [below of =6k] {};
  \node[bsoc] (6m) [below of =6l] {};
  \node[inner sep = 2pt] (6n) [below of =6m] {$\vdots$};
  \path[] (6) edge node {} (6a)
          (6b) edge node {} (6a)
          (6k) edge node {} (6a)
          (6b) edge node {} (6c)
          (6j) edge node {} (6b)
          (6l) edge node {} (6b)
          (6d) edge node {} (6c)
          (6k) edge node {} (6c)
          (6m) edge node {} (6c)
          (6d) edge node {} (6e)
          (6l) edge node {} (6d)
          (6n) edge node {} (6d)
          (6m) edge node {} (6e)
          (6j) edge node {} (6)
          (6j) edge node {} (6k)
          (6l) edge node {} (6k)
          (6l) edge node {} (6m)
          (6n) edge node {} (6m);	
\end{tikzpicture}
\end{center}
\caption{The \ssv{} structure of Virasoro \FFms{} for $t \in \QQ_+$. Black circles represent the socle of the module, grey the second socle layer, and white the third.  Each of the black circles corresponds to a \sv{} and the top circle (of lowest conformal dimension) is always singular, but the rest are only subsingular.  Arrows from one \ssv{} to another indicate that the latter may be obtained from the former via the Virasoro action. Note that the two braid diagrams that we have drawn are structurally identical.  This repetition serves to indicate that the corresponding \FFms{} are not self-contragredient. 
The conformal dimensions of the \ssvs{} are chosen to increase as one moves down and to the right.}
\label{fig:FeiginFuchsStructures}
\end{figure}

The \ssvs{} indicated in \cref{fig:FeiginFuchsStructures} have been sorted into three classes.  Those represented by black circles are actually singular.  They generate the socle of the \FFm{}, the socle being defined as the maximal completely reducible submodule.  Quotienting the \FFm{} by its socle corresponds, diagrammatically, to removing the black circles and any arrows pointing to them.  If the quotient is non-trivial (the chain and braid cases), then the socle of the quotient corresponds to the grey circles.  Finally, if quotienting this quotient by its socle is non-trivial (the braid case), then the result corresponds to the white circles (there will be no arrows left).  The final non-trivial quotient is the maximal semisimple quotient, also called the head of the \FFm{}.\footnote{We remark that the socle of the chain- and braid-type Virasoro Verma modules with $c \le 1$ is the zero submodule. This concept is therefore not particularly useful for \hwms{} of this type.}

As with Virasoro Verma modules, Feigin-Fuchs modules are generically irreducible (the point case).  More precisely, $\FF{\lambda}$ is irreducible as a Virasoro module if and only if $\lambda \neq \lambda_{r,s}$ for any $r,s \in \ZZ$.  Suppose then that $\lambda = \lambda_{r,s}$ for some $r,s \in \ZZ$ (and that $t \in \QQ_+$).  We may summarise the remaining possibilities for $\FF{r,s} \equiv \FF{\lambda_{r,s}}$ as follows:
\begin{itemize}
\item If $(r,s)$ is of corner type, then the structure of $\FF{r,s}$ is represented by the islands diagram in \cref{fig:FeiginFuchsStructures}.
\item If $(r,s)$ is of boundary type, then the structure of $\FF{r,s}$ is represented by one of the chain diagrams.
\item If $(r,s)$ is of interior type, then the structure of $\FF{r,s}$ is represented by one of the braid diagrams.
\end{itemize}
For corner type $(r,s)$, the \FFm{} thus decomposes as a direct sum of irreducible Virasoro modules.  This is the only case in which $\FF{\lambda}$ is decomposable.  We remark that the two different diagrams for the chain and braid cases arise because the \FFm{} is then not isomorphic to its contragredient dual.  In particular, $\FF{r,s}$ and $\FF{-r,-s}$ are contragredient to one another.

For $r,s \in \ZZ_+$, the grades of the \ssvs{} in the Feigin-Fuchs module $\FF{r,s}$ precisely match the grades of the \svs{} in the Virasoro Verma module $\Ver{r,s}$.  In particular, $\FF{r,s}$ has a \ssv{} at grade $rs$.  Inspection of \cref{fig:FeiginFuchsStructures} shows that this \ssv{} must be either associated to the socle or the head of 
$\FF{\lambda_{r,s}}$.  With the parametrisation chosen above, it turns out that it is always associated to the head for 
$r,s \in \ZZ_+$.  The grade $rs$ \ssv{} of $\FF{-r,-s}$ is therefore associated to its socle (and is therefore singular).  
This identification of the structure may be extended to all $r,s \in \ZZ$ using the symmetries \eqref{eq:FFSymm}.

For chain type \FFms{}, identifying the grade $rs$ \ssv{} as being non-singular, when $r,s \in \ZZ_+$, allows one to identify the appropriate structure in \cref{fig:FeiginFuchsStructures} by ordering all the \ssvs{} by their conformal dimensions.  This does not quite suffice for braid type modules because one also needs to know, at every other horizontal level of \cref{fig:FeiginFuchsStructures}, whether the conformal dimension of the \sv{} exceeds that of the \ssv{} or not, or equivalently which of the two conformal dimensions is
represented by a white circle. Perhaps the easiest rule here is that the \emph{sign} of the difference between the \sv{} and \ssv{} dimensions does not depend upon the horizontal level.  Thus, one can check this sign at the level where the (for $r,s \in \ZZ_+$, non-singular) grade $rs$ vector appears.  Alternatively, if one lets $r = r_0 \bmod{p}$ and $s = s_0 \bmod{p'}$, where $0<r_0<p$ and $0<s_0<p'$, then this sign will be positive whenever $p'r_0 + ps_0 < pp'$ and negative otherwise. With the conformal dimensions of the \ssvs{} chosen to increase towards the right and bottom in \cref{fig:FeiginFuchsStructures}, the first and second braid diagram correspond to the sign being positive and negative, respectively, when $r,s \in \ZZ_+$.

%
\section{Tables of structural results for lattice Kac modules}\label{sec:Gramdata}
%

This section presents examples of the conjectured characters and module structures that we expect to describe 
the scaling limit of the lattice Kac modules. In particular, we detail the evidence obtained by applying the Gram matrix to the states associated with each composition factor. 

Below, we tabulate the results for all models with $1 \le p < p' \le 5$, $0\le k\le 3$ and $0\le d\le 4$,
in both $\xi$-regimes $A$ and $B$. Each table entry contains three pieces of information: 
\begin{enumerate}[label=(\roman*),leftmargin=*] 
\item The value of $r$ guessed from the character analysis. In all cases, the value of $r$ is found to be independent of $d$.  On the other hand, $s$ is always found to equal $d+1$, explaining why it is not included in the tables. \label{it:table-r,s}
\item The conjectured structures of the limiting Virasoro modules. The Kac character $\chit_{r,s}$ found in \ref{it:table-r,s} decomposes as a sum of irreducible characters, each of which corresponds to one composition factor of $\Kac{r,s}$. Each factor is depicted by a black circle and labelled by its conformal weight. For regime $A$, the arrows connecting the black circles are drawn according to \cref{TheConjecture}, which is a main result of this paper and is supported by evidence from character, bilinear form and \cft{} analyses. For regime $B$, the structures are guessed according to the character and bilinear form analysis.
\item The evidence obtained from the bilinear form analysis. Each black circle is assigned either {\scshape (r)}, {\scshape (q)}  or {\scshape (u)}, indicating that the corresponding factor is in the radical $\Rad_{n,k}^d$ of $\stan_{n,k}^d$, in the quotient $\stan_{n,k}^d/\Rad_{n,k}^d$, or that its status is unknown because the grades of its states are too large to be investigated by our computer program. Barred or underlined versions of {\scshape (r)}, {\scshape (q)}  or {\scshape (u)} indicate that the bilinear form used was renormalised: {\scshape (\underline{r})},  {\scshape (\underline{q})} or {\scshape (\underline{u})} for \eqref{eq:underform} and {\scshape (\textoverline{r})}, {\scshape (\textoverline{q})} or {\scshape (\textoverline{u})} for \eqref{eq:overform}.
\end{enumerate}

\bigskip

\begin{center}
\begin{tabular}{|c|c|c||c|c|c|c|c|}
\hline
$k$ & regime & $r$ & $d=0$ & $d=1$ & $d=2$ & $d=3$ & $d=4$\\
\hline \hline
$0$ & $A$ & $1$ & 
\begin{pspicture}[shift=-0.475](-0.2,-0.6)(0.2,0.6)
\pscircle[fillstyle=solid,fillcolor=black](0,0){0.075}\rput(0,0.3){\scriptsize$0$}\rput(0,-0.3){\squ}
\end{pspicture}
&
\begin{pspicture}[shift=-0.475](-0.2,-0.6)(0.2,0.6)
\pscircle[fillstyle=solid,fillcolor=black](0,0){0.075}\rput(0,0.3){\scriptsize$\!\!\!-\tfrac18$}\rput(0,-0.3){\sq}
\end{pspicture}
& 
\begin{pspicture}[shift=-0.475](-0.2,-0.6)(0.95,0.6)
\pscircle[fillstyle=solid,fillcolor=black](0,0){0.075}\rput(0,0.3){\scriptsize$0$}\rput(0,-0.3){\sq}
\rput(0.75,0){\pscircle[fillstyle=solid,fillcolor=black](0.0,0){0.075}\rput(0.0,0.3){\scriptsize$1$}\rput(0.0,-0.3){\sr}}
\psline{->}(0.075,0)(0.675,0)
\end{pspicture}
&
\begin{pspicture}[shift=-0.475](-0.2,-0.6)(0.2,0.6)
\pscircle[fillstyle=solid,fillcolor=black](0,0){0.075}\rput(0,0.3){\scriptsize$\tfrac38$}\rput(0,-0.3){\sq}
\end{pspicture}
& 
\begin{pspicture}[shift=-0.475](-0.2,-0.6)(0.95,0.6)
\pscircle[fillstyle=solid,fillcolor=black](0,0){0.075}\rput(0,0.3){\scriptsize$1$}\rput(0,-0.3){\sq}
\rput(0.75,0){\pscircle[fillstyle=solid,fillcolor=black](0.0,0){0.075}\rput(0.0,0.3){\scriptsize$3$}\rput(0.0,-0.3){\sr}}
\psline{->}(0.075,0)(0.675,0)
\end{pspicture}
\\
\hline
\end{tabular}
\captionof{table}{Module structure for critical dense polymers, $\mathcal{LM}(1,2)$.}\label{tab:12}
\end{center}


\begin{center}
\begin{tabular}{|c|c|c||c|c|c|c|c|}
\hline
$k$ & regime & $r$ & $d=0$ & $d=1$ & $d=2$ & $d=3$ & $d=4$\\
\hline \hline
$0$ & $A$ & $1$  & 
\begin{pspicture}[shift=-0.475](-0.2,-0.6)(0.2,0.6)
\pscircle[fillstyle=solid,fillcolor=black](0,0){0.075}\rput(0,0.3){\scriptsize$0$}\rput(0,-0.3){\sr}
\end{pspicture}
&
\begin{pspicture}[shift=-0.475](-0.2,-0.6)(0.2,0.6)
\pscircle[fillstyle=solid,fillcolor=black](0,0){0.075}\rput(0,0.3){\scriptsize$\!\!\!-\tfrac14$}\rput(0,-0.3){\sr}
\end{pspicture}
&
\begin{pspicture}[shift=-0.475](-0.2,-0.6)(0.2,0.6)
\pscircle[fillstyle=solid,fillcolor=black](0,0){0.075}\rput(0,0.3){\scriptsize$\!\!\!-\tfrac13$}\rput(0,-0.3){\sq}
\end{pspicture}
& 
\begin{pspicture}[shift=-0.475](-0.2,-0.6)(0.95,0.6)
\pscircle[fillstyle=solid,fillcolor=black](0,0){0.075}\rput(0,0.3){\scriptsize$\!\!\!-\tfrac14$}\rput(0,-0.3){\sq}
\rput(0.75,0){\pscircle[fillstyle=solid,fillcolor=black](0.0,0){0.075}\rput(0.0,0.3){\scriptsize$\tfrac74$}\rput(0.0,-0.3){\sr}}
\psline{->}(0.075,0)(0.675,0)
\end{pspicture}
& 
\begin{pspicture}[shift=-0.475](-0.2,-0.6)(0.95,0.6)
\pscircle[fillstyle=solid,fillcolor=black](0,0){0.075}\rput(0,0.3){\scriptsize$0$}\rput(0,-0.3){\sq}
\rput(0.75,0){\pscircle[fillstyle=solid,fillcolor=black](0.0,0){0.075}\rput(0.0,0.3){\scriptsize$1$}\rput(0.0,-0.3){\sr}}
\psline{->}(0.075,0)(0.675,0)
\end{pspicture}
 \\
\hline
\end{tabular}~\captionof{table}{Module structure for $\mathcal{LM}(1,3)$.}\label{tab:13}
\end{center}

\clearpage

\begin{center}
\begin{tabular}{|c|c|c||c|c|c|c|c|}
\hline
$k$ & regime & $r$ & $d=0$ & $d=1$ & $d=2$ & $d=3$ & $d=4$\\
\hline \hline
$0$ & $A$ & $1$  & 
\begin{pspicture}[shift=-0.475](-0.2,-0.6)(0.95,0.6)
\pscircle[fillstyle=solid,fillcolor=black](0,0){0.075}\rput(0,0.3){\scriptsize$0$}\rput(0,-0.3){\sq}
\rput(0.75,0){\pscircle[fillstyle=solid,fillcolor=black](0.0,0){0.075}\rput(0.0,0.3){\scriptsize$2$}\rput(0.0,-0.3){\sr}}
\psline{->}(0.075,0)(0.675,0)
\end{pspicture}
& 
\begin{pspicture}[shift=-0.475](-0.2,-0.6)(0.95,0.6)
\pscircle[fillstyle=solid,fillcolor=black](0,0){0.075}\rput(0,0.3){\scriptsize$0$}\rput(0,-0.3){\sq}
\rput(0.75,0){\pscircle[fillstyle=solid,fillcolor=black](0.0,0){0.075}\rput(0.0,0.3){\scriptsize$1$}\rput(0.0,-0.3){\sr}}
\psline{->}(0.075,0)(0.675,0)
\end{pspicture}
& 
\begin{pspicture}[shift=-0.475](-0.2,-0.6)(0.2,0.6)
\pscircle[fillstyle=solid,fillcolor=black](0,0){0.075}\rput(0,0.3){\scriptsize$\tfrac13$}\rput(0,-0.3){\sq}
\end{pspicture}
& 
\begin{pspicture}[shift=-0.475](-0.2,-0.6)(0.95,0.6)
\pscircle[fillstyle=solid,fillcolor=black](0,0){0.075}\rput(0,0.3){\scriptsize$1$}\rput(0,-0.3){\sq}
\rput(0.75,0){\pscircle[fillstyle=solid,fillcolor=black](0.0,0){0.075}\rput(0.0,0.3){\scriptsize$7$}\rput(0.0,-0.3){\sr}}
\psline{->}(0.075,0)(0.675,0)
\end{pspicture}
& 
\begin{pspicture}[shift=-0.475](-0.2,-0.6)(0.95,0.6)
\pscircle[fillstyle=solid,fillcolor=black](0,0){0.075}\rput(0,0.3){\scriptsize$2$}\rput(0,-0.3){\sq}
\rput(0.75,0){\pscircle[fillstyle=solid,fillcolor=black](0.0,0){0.075}\rput(0.0,0.3){\scriptsize$5$}\rput(0.0,-0.3){\sr}}
\psline{->}(0.075,0)(0.675,0)
\end{pspicture}
\\
\hline
\end{tabular}
\captionof{table}{Module structure for critical percolation, $\mathcal{LM}(2,3)$.}
\end{center}


\begin{center}
\begin{tabular}{|c|c|c||c|c|c|c|c|}
\hline
$k$ & regime & $r$ & $d=0$ & $d=1$ & $d=2$ & $d=3$ & $d=4$\\
\hline \hline
$0$ & $A$ & $1$  & 
\begin{pspicture}[shift=-0.475](-0.2,-0.6)(0.2,0.6)
\pscircle[fillstyle=solid,fillcolor=black](0,0){0.075}\rput(0,0.3){\scriptsize$0$}\rput(0,-0.3){\sr}
\end{pspicture}
& 
\begin{pspicture}[shift=-0.475](-0.2,-0.6)(0.2,0.6)
\pscircle[fillstyle=solid,fillcolor=black](0,0){0.075}\rput(0,0.3){\scriptsize$\!\!\!-\tfrac5{16}$}\rput(0,-0.3){\sr}
\end{pspicture}
& 
\begin{pspicture}[shift=-0.475](-0.2,-0.6)(0.2,0.6)
\pscircle[fillstyle=solid,fillcolor=black](0,0){0.075}\rput(0,0.3){\scriptsize$\!\!\!-\tfrac1{2}$}\rput(0,-0.3){\sr}
\end{pspicture}
& 
\begin{pspicture}[shift=-0.475](-0.2,-0.6)(0.2,0.6)
\pscircle[fillstyle=solid,fillcolor=black](0,0){0.075}\rput(0,0.3){\scriptsize$\!\!\!-\tfrac9{16}$}\rput(0,-0.3){\sq}
\end{pspicture}
& 
\begin{pspicture}[shift=-0.475](-0.2,-0.6)(0.95,0.6)
\pscircle[fillstyle=solid,fillcolor=black](0,0){0.075}\rput(0,0.3){\scriptsize$\!\!\!-\tfrac12$}\rput(0,-0.3){\sq}
\rput(0.75,0){\pscircle[fillstyle=solid,fillcolor=black](0.0,0){0.075}\rput(0.0,0.3){\scriptsize$\tfrac{5}2$}\rput(0.0,-0.3){\sr}}
\psline{->}(0.075,0)(0.675,0)
\end{pspicture}
 \\
\hline
\end{tabular}
\captionof{table}{Module structure for $(p,p') = (1,4)$.}
\end{center}\bigskip


\begin{center}
\begin{tabular}{|c|c|c||c|c|c|c|c|}
\hline
$k$ & regime & $r$ & $d=0$ & $d=1$ & $d=2$ & $d=3$ & $d=4$\\
\hline \hline
$0$ & $A$ & $1$  & 
\begin{pspicture}[shift=-0.475](-0.2,-0.6)(0.95,0.6)
\pscircle[fillstyle=solid,fillcolor=black](0,0){0.075}\rput(0,0.3){\scriptsize$0$}\rput(0,-0.3){\sq}
\rput(0.75,0){\pscircle[fillstyle=solid,fillcolor=black](0.0,0){0.075}\rput(0.0,0.3){\scriptsize$6$}\rput(0.0,-0.3){\sr}}
\psline{->}(0.075,0)(0.675,0)
\end{pspicture}
& 
\begin{pspicture}[shift=-0.475](-0.2,-0.6)(0.95,0.6)
\pscircle[fillstyle=solid,fillcolor=black](0,0){0.075}\rput(0,0.3){\scriptsize$\tfrac1{16}$}\rput(0,-0.3){\sq}
\rput(0.75,0){\pscircle[fillstyle=solid,fillcolor=black](0.0,0){0.075}\rput(0.0,0.3){\scriptsize$\tfrac{65}{16}$}\rput(0.0,-0.3){\sr}}
\psline{->}(0.075,0)(0.675,0)
\end{pspicture}
& 
\begin{pspicture}[shift=-0.475](-0.2,-0.6)(0.95,0.6)
\pscircle[fillstyle=solid,fillcolor=black](0,0){0.075}\rput(0,0.3){\scriptsize$\tfrac12$}\rput(0,-0.3){\sq}
\rput(0.75,0){\pscircle[fillstyle=solid,fillcolor=black](0.0,0){0.075}\rput(0.0,0.3){\scriptsize$\tfrac52$}\rput(0.0,-0.3){\sr}}
\psline{->}(0.075,0)(0.675,0)
\end{pspicture}
& 
\begin{pspicture}[shift=-0.475](-0.2,-0.6)(0.2,0.6)
\pscircle[fillstyle=solid,fillcolor=black](0,0){0.075}\rput(0,0.3){\scriptsize$\tfrac{21}{16}$}\rput(0,-0.3){\sq}
\end{pspicture}
& 
\begin{pspicture}[shift=-0.475](-0.2,-0.6)(0.95,0.6)
\pscircle[fillstyle=solid,fillcolor=black](0,0){0.075}\rput(0,0.3){\scriptsize$\tfrac52$}\rput(0,-0.3){\sq}
\rput(0.75,0){\pscircle[fillstyle=solid,fillcolor=black](0.0,0){0.075}\rput(0.0,0.3){\scriptsize$\tfrac{35}2$}\rput(0.0,-0.3){\su}}
\psline{->}(0.075,0)(0.675,0)
\end{pspicture}
 \\
\hline
\end{tabular}
\captionof{table}{Module structure for the logarithmic Ising model, $\mathcal{LM}(3,4)$.}
\end{center}\bigskip


\begin{center}
\begin{tabular}{|c|c|c||c|c|c|c|c|}
\hline
$k$ & regime & $r$ & $d=0$ & $d=1$ & $d=2$ & $d=3$ & $d=4$\\
\hline \hline
$0$ & $A$ & $1$  & 
\begin{pspicture}[shift=-0.475](-0.2,-0.6)(0.2,0.6)
\pscircle[fillstyle=solid,fillcolor=black](0,0){0.075}\rput(0,0.3){\scriptsize$0$}\rput(0,-0.3){\sr}
\end{pspicture}
& 
\begin{pspicture}[shift=-0.475](-0.2,-0.6)(0.2,0.6)
\pscircle[fillstyle=solid,fillcolor=black](0,0){0.075}\rput(0,0.3){\scriptsize$\!\!\!-\tfrac{7}{20}$}\rput(0,-0.3){\sr}
\end{pspicture}
&
\begin{pspicture}[shift=-0.475](-0.2,-0.6)(0.2,0.6)
\pscircle[fillstyle=solid,fillcolor=black](0,0){0.075}\rput(0,0.3){\scriptsize$\!\!\!-\tfrac{3}{5}$}\rput(0,-0.3){\sr}
\end{pspicture}
&
\begin{pspicture}[shift=-0.475](-0.2,-0.6)(0.2,0.6)
\pscircle[fillstyle=solid,fillcolor=black](0,0){0.075}\rput(0,0.3){\scriptsize$\!\!\!-\tfrac{3}{4}$}\rput(0,-0.3){\sr}
\end{pspicture}
&
\begin{pspicture}[shift=-0.475](-0.2,-0.6)(0.2,0.6)
\pscircle[fillstyle=solid,fillcolor=black](0,0){0.075}\rput(0,0.3){\scriptsize$\!\!\!-\tfrac{4}{5}$}\rput(0,-0.3){\sq}
\end{pspicture}
 \\
\hline
\end{tabular}
\captionof{table}{Module structure for $\mathcal{LM}(1,5)$.}
\end{center}\bigskip


\begin{center}
\begin{tabular}{|c|c|c||c|c|c|c|c|}
\hline
$k$ & regime & $r$ & $d=0$ & $d=1$ & $d=2$ & $d=3$ & $d=4$\\
\hline \hline
$0$ & $A$ & $1$ & 
\begin{pspicture}[shift=-0.475](-0.2,-0.6)(0.95,0.6)
\pscircle[fillstyle=solid,fillcolor=black](0,0){0.075}\rput(0,0.3){\scriptsize$0$}\rput(0,-0.3){\sq}
\rput(0.75,0){\pscircle[fillstyle=solid,fillcolor=black](0.0,0){0.075}\rput(0.0,0.3){\scriptsize$4$}\rput(0.0,-0.3){\sr}}
\psline{->}(0.075,0)(0.675,0)
\end{pspicture}
& 
\begin{pspicture}[shift=-0.475](-0.2,-0.6)(0.95,0.6)
\pscircle[fillstyle=solid,fillcolor=black](0,0){0.075}\rput(0,0.3){\scriptsize$\!\!\!-\tfrac1{5}$}\rput(0,-0.3){\sq}
\rput(0.75,0){\pscircle[fillstyle=solid,fillcolor=black](0.0,0){0.075}\rput(0.0,0.3){\scriptsize$\tfrac{14}{5}$}\rput(0.0,-0.3){\sr}}
\psline{->}(0.075,0)(0.675,0)
\end{pspicture}
& 
\begin{pspicture}[shift=-0.475](-0.2,-0.6)(0.95,0.6)
\pscircle[fillstyle=solid,fillcolor=black](0,0){0.075}\rput(0,0.3){\scriptsize$\!\!\!-\tfrac1{5}$}\rput(0,-0.3){\sq}
\rput(0.75,0){\pscircle[fillstyle=solid,fillcolor=black](0.0,0){0.075}\rput(0.0,0.3){\scriptsize$\tfrac95$}\rput(0.0,-0.3){\sr}}
\psline{->}(0.075,0)(0.675,0)
\end{pspicture}
& 
\begin{pspicture}[shift=-0.475](-0.2,-0.6)(0.95,0.6)
\pscircle[fillstyle=solid,fillcolor=black](0,0){0.075}\rput(0,0.3){\scriptsize$0$}\rput(0,-0.3){\sq}
\rput(0.75,0){\pscircle[fillstyle=solid,fillcolor=black](0.0,0){0.075}\rput(0.0,0.3){\scriptsize$1$}\rput(0.0,-0.3){\sr}}
\psline{->}(0.075,0)(0.675,0)
\end{pspicture}
& 
\begin{pspicture}[shift=-0.475](-0.2,-0.6)(0.2,0.6)
\pscircle[fillstyle=solid,fillcolor=black](0,0){0.075}\rput(0,0.3){\scriptsize$\tfrac{2}{5}$}\rput(0,-0.3){\sq}
\end{pspicture} \\
\hline
\end{tabular}
\captionof{table}{Module structure for the Yang-Lee model, $\mathcal{LM}(2,5)$.}
\end{center}\bigskip


\begin{center}
\begin{tabular}{|c|c|c||c|c|c|c|c|}
\hline
$k$ & regime & $r$ & $d=0$ & $d=1$ & $d=2$ & $d=3$ & $d=4$\\
\hline \hline
$0$ & $A$ & $1$  & 
\begin{pspicture}[shift=-0.475](-0.2,-0.6)(0.95,0.6)
\pscircle[fillstyle=solid,fillcolor=black](0,0){0.075}\rput(0,0.3){\scriptsize$0$}\rput(0,-0.3){\sq}
\rput(0.75,0){\pscircle[fillstyle=solid,fillcolor=black](0.0,0){0.075}\rput(0.0,0.3){\scriptsize$8$}\rput(0.0,-0.3){\su}}
\psline{->}(0.075,0)(0.675,0)
\end{pspicture}
& 
\begin{pspicture}[shift=-0.475](-0.2,-0.6)(0.95,0.6)
\pscircle[fillstyle=solid,fillcolor=black](0,0){0.075}\rput(0,0.3){\scriptsize$\!\!\!-\tfrac1{20}$}\rput(0,-0.3){\sq}
\rput(0.75,0){\pscircle[fillstyle=solid,fillcolor=black](0.0,0){0.075}\rput(0.0,0.3){\scriptsize$\tfrac{119}{20}$}\rput(0.0,-0.3){\su}}
\psline{->}(0.075,0)(0.675,0)
\end{pspicture}
& 
\begin{pspicture}[shift=-0.475](-0.2,-0.6)(0.95,0.6)
\pscircle[fillstyle=solid,fillcolor=black](0,0){0.075}\rput(0,0.3){\scriptsize$\tfrac1{5}$}\rput(0,-0.3){\sq}
\rput(0.75,0){\pscircle[fillstyle=solid,fillcolor=black](0.0,0){0.075}\rput(0.0,0.3){\scriptsize$\tfrac{21}5$}\rput(0.0,-0.3){\sr}}
\psline{->}(0.075,0)(0.675,0)
\end{pspicture}
& 
\begin{pspicture}[shift=-0.475](-0.2,-0.6)(0.95,0.6)
\pscircle[fillstyle=solid,fillcolor=black](0,0){0.075}\rput(0,0.3){\scriptsize$\frac34$}\rput(0,-0.3){\sq}
\rput(0.75,0){\pscircle[fillstyle=solid,fillcolor=black](0.0,0){0.075}\rput(0.0,0.3){\scriptsize$\frac{11}4$}\rput(0.0,-0.3){\sr}}
\psline{->}(0.075,0)(0.675,0)
\end{pspicture}
& 
\begin{pspicture}[shift=-0.475](-0.2,-0.6)(0.2,0.6)
\pscircle[fillstyle=solid,fillcolor=black](0,0){0.075}\rput(0,0.3){\scriptsize$\tfrac{8}{5}$}\rput(0,-0.3){\sq}
\end{pspicture}
\\
\hline

\end{tabular}
\captionof{table}{Module structure for $\mathcal{LM}(3,5)$.}
\end{center}\bigskip


\begin{center}
\begin{tabular}{|c|c|c||c|c|c|c|c|}
\hline
$k$ & regime & $r$ & $d=0$ & $d=1$ & $d=2$ & $d=3$ & $d=4$\\
\hline \hline
$0$ & $A$ & $1$  & 
\begin{pspicture}[shift=-0.475](-0.2,-0.6)(0.95,0.6)
\pscircle[fillstyle=solid,fillcolor=black](0,0){0.075}\rput(0,0.3){\scriptsize$0$}\rput(0,-0.3){\sq}
\rput(0.75,0){\pscircle[fillstyle=solid,fillcolor=black](0.0,0){0.075}\rput(0.0,0.3){\scriptsize$12$}\rput(0.0,-0.3){\su}}
\psline{->}(0.075,0)(0.675,0)
\end{pspicture}
& 
\begin{pspicture}[shift=-0.475](-0.2,-0.6)(0.95,0.6)
\pscircle[fillstyle=solid,fillcolor=black](0,0){0.075}\rput(0,0.3){\scriptsize$\tfrac1{10}$}\rput(0,-0.3){\sq}
\rput(0.75,0){\pscircle[fillstyle=solid,fillcolor=black](0.0,0){0.075}\rput(0.0,0.3){\scriptsize$\tfrac{91}{10}$}\rput(0.0,-0.3){\su}}
\psline{->}(0.075,0)(0.675,0)
\end{pspicture}
& 
\begin{pspicture}[shift=-0.475](-0.2,-0.6)(0.95,0.6)
\pscircle[fillstyle=solid,fillcolor=black](0,0){0.075}\rput(0,0.3){\scriptsize$\tfrac3{5}$}\rput(0,-0.3){\sq}
\rput(0.75,0){\pscircle[fillstyle=solid,fillcolor=black](0.0,0){0.075}\rput(0.0,0.3){\scriptsize$\tfrac{33}5$}\rput(0.0,-0.3){\sr}}
\psline{->}(0.075,0)(0.675,0)
\end{pspicture}
& 
\begin{pspicture}[shift=-0.475](-0.2,-0.6)(0.95,0.6)
\pscircle[fillstyle=solid,fillcolor=black](0,0){0.075}\rput(0,0.3){\scriptsize$\frac32$}\rput(0,-0.3){\sq}
\rput(0.75,0){\pscircle[fillstyle=solid,fillcolor=black](0.0,0){0.075}\rput(0.0,0.3){\scriptsize$\frac{9}2$}\rput(0.0,-0.3){\sr}}
\psline{->}(0.075,0)(0.675,0)
\end{pspicture}
& 
\begin{pspicture}[shift=-0.475](-0.2,-0.6)(0.2,0.6)
\pscircle[fillstyle=solid,fillcolor=black](0,0){0.075}\rput(0,0.3){\scriptsize$\tfrac{14}{5}$}\rput(0,-0.3){\sq}
\end{pspicture}
 \\
\hline
\end{tabular}
\captionof{table}{Module structure for the logarithmic tricritical Ising model, $\mathcal{LM}(4,5)$.}
\end{center}

\raggedright

\end{document}